\documentclass[12pt, notitlepage]{report}

\usepackage{dante}
\usepackage{color,soul}

\begin{document}

\pagestyle{plain}
    \medskip
    \vspace*{2cm}
    \begin{center}
        {\Huge A Mathematical Introduction to \vspace{0.2cm}\\ Geometric Quantization}
        
        \vspace{2cm}
        
        Kadri İlker Berktav\footnote{berktav@metu.edu.tr}$^{,a}$.
        Burak Oğuz\footnote{boguz@ictp.it}$^{,b}$,
        Ömer Önder\footnote{omer.onder@bilkent.edu.tr}$^{,c}$,
        Yunus Emre Sargut\footnote{y.sargut@uky.edu}$^{,d}$,
        Başar Deniz Sevinç\footnote{basardeniz.sevinc@studio.unibo.it}$^{,e}$,
        Deniz Nazif Taştan\footnote{deniznazif.tastan@studio.unibo.it}$^{,e}$.
         
        \bigskip\medskip\centerline{$^a$ \it Department of Mathematics, Middle East Technical University, 06531, Ankara, Türkiye} 
        \smallskip \centerline{$^b$
        \it ICTP, Strada Costiera 11, Trieste 34151, Italy}
        \smallskip \centerline{$^c$ \it Department of Physics, Bilkent University, 06800 Ankara, Türkiye}
        \smallskip \centerline{$^d$ \it Department of Physics and Astronomy, University of Kentucky, Lexington, 40506}  
        \smallskip \centerline{$^e$ \it Dipartimento di Fisica e Astronomia, Universit\`a di Bologna, via Irnerio 46, Bologna 40126, Italy}
    \end{center}
    
    \vspace{1.5cm}
    
    \begin{abstract}
            These notes are based on a series of lectures by Kadri İlker Berktav from May 2024 to November 2024,  providing a detailed exposition of geometric quantization formalism and its essential components. They are organized into three parts: background in symplectic geometry, basics of geometric quantization formalism, and an application related to Edward Witten's work in knot theory and topology. 
    \end{abstract}
\newpage

\section*{Acknowledgment}
This is a collection of notes from a lecture series KİB gave a relatively long time ago, in a galaxy -not- far,
far away. 
\vspace{0.2in}

The main body of the text presents key ideas and standard results gathered from the literature. The arguments we present here are well-known to experts. As a disclaimer, these notes are not intended to offer original or new results, nor do they aim to provide a comprehensive reference list on the topics discussed. Instead, we hope the material in these notes offers a brief introduction and simple guidelines to help non-experts understand the basics of the subject. Therefore, our systematic presentation of various topics may prove useful. Please also be vigilant for typos and errors.
\vspace{0.2in}

 KİB warmly thanks \textbf{Burak Oğuz, Ömer Önder, Yunus Emre Sargut, Başar Deniz Sevinç,} and \textbf{Deniz Nazif Taştan} for their time and collaborative effort in compiling the (very messy) handwritten notes and preparing this current version. KİB appreciates the effort invested in creating beautiful images and diagrams that effectively capture the spirit of our in-class meetings. 
 \vspace{0.2in}
 
 We would like to express our gratitude to all the lecture participants for their valuable contributions, comments, and discussions. Thank you! 
\vspace{0.2in}

Last but not least, we extend our heartfelt gratitude to \textbf{Bayram Tekin} for uniting us and fostering an exemplary research and learning environment. His passion, encouragement, and almost limitless energy have made all of this possible.

\tableofcontents
\newpage
\pagestyle{fancy}

\chapter{Lectures on Symplectic Geometry}\label{ch:symplectic_geometry}
The primary objective of this section is to review the fundamental concepts in symplectic geometry and to investigate their applications. 

We shall focus on two basic applications: \textit{(i)} the formulation of classical mechanics with the aid of symplectic structure on the cotangent bundle (phase space) over a certain base manifold (configuration space), and \textit{(ii)} some applications of the symplectic reduction theorem in gauge theory.

\paragraph{References and our road map.}
We essentially follow \cite{daSilva2008} Ana Cannas da Silva, Lectures on Symplectic Geometry. \cite{HondaTopologicalQFT} Ko HONDA, Lecture notes for MATH 635: Topological Quantum Field Theory.

Topics to be covered and references for further reading are as follows.
\begin{enumerate}
    \item   Symplectic forms, symplectic manifolds and symplectic structure on the cotangent bundle (\cite{daSilva2008} chapter 1 - 2).
    \item  Hamiltonian function, Hamiltonian vector fields and reformulation of Classical Mechanics (\cite{daSilva2008} chapter 18).
    \item   Lie Theory: Lie group, Lie algebras and the Poisson algebra. (\cite{HondaTopologicalQFT} Chapter 1,2, 9.2 and 10.1)
    \item  Moment maps. (\cite{daSilva2008} Chapter 21, 22)
    \item   Symplectic reduction theorem, Noether Principle (\cite{daSilva2008} Chapter 22.3,22.4, 23.1, 24.1)
    \item   Principal $G$-bundles, associated fiber bundle, connection, connection one-form and curvature on a principal $G$-bundle, and an application: The moduli space of flat connections. (\cite{daSilva2008} Chapter 25.1, 25.2, 25.3, 25.4 or F.Schuller online Lectures \cite{SchullerGeometricAnatomy, SchullerGATP_Lec19_21}: Part 19-21.)
\end{enumerate}

\section{Symplectic Vector Spaces}

Let $V$ be an $m$-dimensional vector space over $\mathbb{R}$,
and $\Omega : V \times V \rightarrow \mathbb{R}$ a bilinear map.

\begin{definition}
    $\Omega$ is called \textbf{skew-symmetric} if
    \begin{equation}
        \Omega(u,v) = -\Omega(v,u) \quad \forall u,v \in V.
    \end{equation}
\end{definition}

\begin{theorem}
    Given a skew-symmetric bilinear map $\Omega$ on $V$, with $\dim V < \infty$,  
    there exists a basis $\{u_1, \ldots, u_k, e_1, \dots, e_n, f_1, \dots, f_n\}$ of $V$ such that
    \begin{equation}
        \Omega(u_i, v) = 0 \quad \forall v \in V,
        \quad
        \Omega(e_i, e_j) = \Omega(f_i, f_j) = 0,
        \quad
        \Omega(e_i, f_j) = \delta_{ij} \quad \forall i,j.
    \end{equation}
\end{theorem}

Observe that $V$ can be decomposed into 
\begin{equation}
    V = U \oplus W_1 \oplus W_2 \cdots \oplus W_n,
\end{equation}
where
\begin{align}
    U &= \left\{ u \in V \ | \ \Omega(u, v) = 0 \quad \forall v \in V \right\}, \\
    W &=  \langle e_i, f_i \rangle \quad \text{with} \quad \Omega(e_i, f_i) = 1. 
\end{align}

Also, we have: $\dim{U} = k$ and $\dim{V} = k + 2n$.

\begin{proof}
    Let $U := \left\{ u \in V \ | \ \Omega(u , v) = 0 \quad \forall v \in V \right\}$. Pick a basis $\{u_1, \ldots, u_k\}$ for $U$. Let $W := V - U$ which is the complementary subspace of $U$ in $V$. Then, we have
    \begin{equation}
        V = U \oplus W,
    \end{equation}
    from which we can say: for each $v \in V$, $\exists! \; u \in U$ and $w \in W$ such that $v = u + w$. Also, $ U \cap V = \{0\} $.
\end{proof}

Now, take any nonzero $e_1 \in W$. Then, $\exists f_1 \in W$ such that
$\Omega(e_1, f_1) \neq 0$. We may assume that $\Omega(e_1, f_1) = 1$.
How can we decompose $W$? \\

Let $W_1 := \langle e_1, f_1 \rangle$, a subspace of $W$ spanned by $e_1$ and $f_1$, and let
\begin{equation}
    W_1^\Omega = \{ w \in W \ | \ \Omega(w, v) = 0 \quad \forall v \in W_1 \}.
\end{equation}

\begin{claim}
It follows that
    \begin{equation}
         W_1^\Omega \cap W_1 = \{0\}.
    \end{equation}
\end{claim}
\begin{proof}
    Assume the contrary. Say there exists a nonzero $ z \in W_1^\Omega \cap W_1$. Then, we have
\begin{align}
        (i) \quad & z \in W_1 \implies z = \alpha e_1 + \beta f_1 
        \quad \text{for some} \; \text{nonzero} \quad \alpha, 
        \beta \in \mathbb{R},\\
        (ii) \quad & z \in W_1^\Omega \implies \Omega(z, v_0) 
        = 0 \quad \forall v_0 \in W_1.
    \end{align}

    However, taking $v_0 = e_1$, we get

    \begin{equation}
        0 = \Omega(z, e_1) = \Omega(\alpha e_1 + \beta f_1, e_1)
        = \alpha \Omega(e_1, e_1) + \beta \Omega(f_1, e_1)
        = - \beta,
    \end{equation}
    and taking $v_0 = f_1$ yields
    \begin{equation}
        0 = \Omega(z, f_1) = \Omega(\alpha e_1 + \beta f_1, f_1)
        = \alpha \Omega(e_1, f_1) + \beta \Omega(f_1, f_1)
        = \alpha.
    \end{equation}
    Thus, $\alpha = \beta = 0$ which contradicts our assumption 
that $z$ is nonzero.
\end{proof}

\begin{claim}
    \begin{equation}
        W = W_1 \oplus W_1^\Omega.
    \end{equation}
\end{claim}

\begin{proof}
    Let $v \in W$, then we have $\Omega(v, e_1) = \alpha \neq 0$, and $\Omega(v, f_1) = \beta \neq 0$. Then, we can write
    \begin{equation}
        v = \underbrace{(\beta e_1 - \alpha f_1)}_{\in W_1 } 
        + \underbrace{(v - \beta e_1 + \alpha f_1)}_{\in W_1^\Omega}.
    \end{equation}
    We can go on and let $e_2 \in W_1^\Omega$, $e_2 \neq 0$. There exists an $f_2 \in W_1^\Omega$ such that $\Omega(e_2, f_2) \neq 0$. Assume that $\Omega(e_2, f_2) = 1$. Then, let $W_2 = \text{span of} \; e_2, f_2$. etc... Since $V$ is finite-dimensional, this process must terminate. Hence,
    we obtain
    \begin{equation}
        V = U \oplus W_1 \oplus W_2 \oplus \ldots \oplus W_n.
    \end{equation}
    Thus, we conclude that a skew-symmetric bilinear map $\Omega$ provides a decomposition of $V$ into "orthonormal" summands with respect to $\Omega$.
\end{proof}

\begin{definition}
    Let $V$ be an $m$-dimensional vector space over $\mathbb{R}$, and
    $\Omega : V \times V \rightarrow \mathbb{R}$ a bilinear map.
    Define
    \begin{equation}
    \begin{aligned}
        \tilde{\Omega} : V &\longrightarrow V^*\\
        u &\longmapsto \tilde{\Omega}(\cdot) = \Omega(u, \cdot),
    \end{aligned}
    \end{equation}
    where $\tilde{\Omega} (v) := \Omega(u, v) $. Then, we define the \textbf{kernel} as
    \begin{equation}
        \ker{\tilde{\Omega}} = \{ u \in V \ | \ \Omega(u, v) = U 
        \quad \forall v \in V \}.
    \end{equation}
    
\end{definition}

\begin{definition}
    A skew-symmetric bilinear map $\Omega$ is called 
    \textbf{symplectic} or a \textbf{non-degenerate} if $\tilde{\Omega}$ is bijective. i.e $\ker{\tilde{\Omega}} = U  = \{0\}$.  
    
    Equivalently: $\Omega(u, v) = 0 \; \forall v \in V \implies u = 0$ (non-degeneracy). Then such an $\Omega$ is called a \textbf{symplectic structure} on $V$ and $(V, \Omega)$ is called a \textbf{symplectic vector space}.
\end{definition}

\begin{observation}
Say $\Omega$ is a symplectic structure on $V$. Then, we have $\ker{\tilde{\Omega}} = 0 $, hence $\dim(\ker{\tilde{\Omega}}) = \dim(U) = 0$. Thus, $\dim(V) = 2 n + \dim(U) = 2n$. Also,
\begin{align}
    V &\xlongrightarrow{\cong} V^* \\
    u &\longmapsto \Omega(u, \cdot),
\end{align}

is a bijection. Lastly, $(V, \Omega)$ has a basis
$\{e_1, e_n, \ldots, f_1, f_n\}$ such that 
$\Omega(e_i, f_j) = \delta_{ij}$, 
$\Omega(e_i, e_j) = \Omega(f_i, f_j) = 0$. This
is called a \textit{symplectic basis}.

\end{observation}

\begin{example}
    Write $ (\mathbb{R}^{2n}, \Omega) $ with sympectic basis.

    \textbf{Solution:}
\begin{align}
        e_1 = (1, 0, \cdots, 0) &\quad e_n = (0, 0, \cdots 1, \cdots, 0) ,\\
        f_1 = (0, 1, \cdots, 0) &\quad f_n = (0, 0, \cdots, 1).
    \end{align}
\end{example}

\begin{definition}
    Let $ (V, \Omega) $, $(V', \Omega')$ be two symplectic vector spaces. A linear map $\varphi : V \rightarrow V'$ is called a \textbf{symplectomorphism} if $\varphi^* \Omega' = \Omega$ (pullback).

    \begin{equation}
        \begin{tikzcd}[sep=1.5cm]
            V\times V \arrow[thick,r,"{(\varphi,\varphi)}"]\arrow[thick, rd, "\varphi^*\Omega'"] & V'\times V' \arrow[thick, d,"\Omega'"] \\ 
            {} & \mathbb{R}
        \end{tikzcd}
    \end{equation}
\end{definition}
So, a symplectomorphism is a map relating two symplectic manifolds with symplectic structures $\Omega$ and $\Omega'$:
\begin{equation}
    \phi^* \Omega' := \Omega' \circ (\varphi,\varphi) \implies \phi^*\Omega' (a,b) = \Omega'(\varphi(a),\varphi(b)).
\end{equation}
This then introduces an equivalence relation $\sim$:

\begin{equation}
    (V, \Omega) \sim (V', \Omega') \quad \text{iff} \quad
    \exists \; \text{symplectomorphism} \; \varphi : V \rightarrow V'.
\end{equation}

\begin{corollary}
    Every $2n$ dimensional symplectic vector space $(V, \Omega)$ is
    symplectomorphic to $(\mathbb{R}^{2n}, \Omega_0)$.
\end{corollary}

Observe that the "modular space" of $(V, \Omega)$ up to a
symplectomorphism is just $n \in \mathbb{Z}_{\ge 0}$.

\newpage

\section{Symplectic Manifolds}

Let $M$ be an $n$-manifold, and $\Omega$ a closed 2-form on $M$. Say:

\begin{equation}
    \omega \in \Gamma (\Lambda^2 T^*M) =: \Omega^2(M),
\end{equation}

i.e., for each $p \in M$, $\omega_p : T_p M \times T_p M \rightarrow \mathbb{R}$ is a skew-symmetric bilinear map on $T_p M$. Note that $\omega_p$ smoothly depends on $p$.

\begin{definition}
    $\omega \in \Omega^k (M)$ is a \textbf{closed form} if $dw = 0$, where $d$ is the
    exterior derivative. 
\end{definition}

\begin{example}
    Let $M = \mathbb{R}^2 - \{0, 0\}$ with usual coordinates $(x,y,z)$. Let

    \begin{equation}
        \omega = \frac{x}{x^2 + y^2} dy - \frac{y}{x^2 + y^2} dx .
    \end{equation}

    Show that $\omega$ is a closed form.

    \textbf{Solution:} Let $\dfrac{x}{x^2 + y^2} = f$, $\dfrac{y}{x^2 + y^2}=g$. Then,
    \begin{align*}
        d\omega &= df \wedge dy - dg \wedge dx\\
        &= \left( \frac{x^2 + y^2 - x \cdot 2x}{(x^2 + y^2)^2} dx 
        (\dots) dy\right) \wedge dy
        - \left( \frac{x^2 + y^2 - y\cdot2y}{(x^2 + y^2)^2} dy
        + (\dots) dx \right) \wedge dx\\
        &= \frac{y^2 - x^2}{(x^2 + y^2)^2} dx \wedge dy
        - \frac{x^2 - y^2}{(x^2 + y^2)^2} dy \wedge dx\\
        &= 0.
    \end{align*}

    Thus, $\omega$ is closed.
    
\end{example}
\begin{definition}
    The 2-form $\omega \in \Omega^2 (M)$ is \textbf{symplectic} if $\omega$ is closed and $\omega_p : T_p M \times T_p M \rightarrow \mathbb{R}$ is symplectic (i.e it is non-degenerate: $\omega_p (X, Y) = 0 \; \forall y_p \in T_p M \implies x_p = 0 $ ) for all $p \in M$. 
\end{definition}
\begin{minipage}{0.50\textwidth}
    \begin{remark}
       $\dim T_p M = \dim M$. Hence, if $\omega$ is a symplectic 2-form on $M$, then $\dim T_p M$ must be even, so is $\dim M$.
    \end{remark}
\end{minipage}
\begin{minipage}{0.45\textwidth}
    \centering
    \input{Figures/Part1/p7}
    \hypertarget{fig1}{Figure 1.} A tangent space $T_pM$ at a point $p\in M$.
\end{minipage}

\begin{definition}
    A \textbf{symplectic manifold} is a pair $(M, w)$ where
    \begin{itemize}
        \item $M$ is an even-dimensional manifold.
        \item $w \in \Omega^2 (M)$ is a symplectic  form.
    \end{itemize}
\end{definition}

As an example, consider the following: $(M , w) = (\mathbb{R}^{2n}, w_0)$ with usual coordinate charts $(x_1, \dots, x_n, y_1, \dots, y_n)$, and $w_0 = \sum_{i = 1}^{n} dx_i \wedge dy_i$ is symplectic. 
The basis vectors for $T_p M$ are $ \left\{ \frac{\partial}{\partial x_1}|_p, \dots \frac{\partial}{\partial x_n}|_p, \frac{\partial}{\partial y_1}|_p, \dots \frac{\partial}{\partial y_n}|_p \right\} $. Later, we will see that the identification of the cotangent bundle $T^* M = \bigcup_p T_p^* M$, which is $\mathbb{R}^{2n}$, and the relation between $T^*M$ and the formulation of classical mechanics.

    \begin{remark}
        The ``non-degeneracy" condition is purely algebraic. That is, it corresponds to the existence of an isomorphism
        \begin{align}
            T_pM &\xrightarrow{\cong} T_p^* M \quad \quad \quad \forall p\\
            X_p &\mapsto \omega_p(X_p , \;.).
        \end{align}
    \end{remark}

    \begin{remark}
        ``Closedness" is geometric.
        \begin{itemize}
            \item $\omega$ represents a cohomology class:
            \begin{equation}
                a := [\omega] \in H^2(M; \mathbb{R}).
            \end{equation}
            \item If $M$ is closed (i.e. compact with/without boundary), then the cohomology class $\alpha^n \in H^{2n} (M; \mathbb{R})$ is replaced by the volume form $\omega^n \in \Omega^{2n} (M)$, with $\int_M \omega^n \neq 0$. So $\omega$ cannot be exact. This can easily be seen by letting $\omega = d \eta$. Then, $\int_M \omega = \int_M d \eta = \int_{\partial M = \phi} \eta = 0$ by using Stokes' theorem.
            \item There exist orientable even-dimensional manifolds which do not admit a symplectic structure. A simple example is $S^{2n}$ with $n\geq 2$.
        \end{itemize}
    \end{remark}

Is it true that for any compact even-dimensional manifold, with a non-degenerate $2$-form $\tau \in \Omega^2 (M)$ and a suitable cohomology class $a \in H^2 (M; \mathbb{R})$, there exists a symplectic structure that represents the class $a$ and is isotropic to the given form $\tau$?\\
A question to consider is whether, given the data $(M, \tau, a)$, where $\tau \in \Omega^2(M)$ and $a \in H^2(M; \mathbb{R})$, there exists a symplectic form $\omega$ that satisfies the following conditions:

\begin{itemize}
    \item $(M, \omega)$ is a symplectic manifold,
    \item There is an isotopy $\omega \sim_{\omega_\tau } \tau$,
    \item $[\omega] = a$.
\end{itemize}

The answer to this question is provided by the works of Donaldson (1996) and Jacobs (1994), which present the method for constructing a manifold $M$ with suitable $(\tau, a)$ so that $M$ does not have any symplectic structure.
\begin{example}
    Let $M = \mathbb{C}^n$ with complex coordinates
    $z_1, \dots, z_n$. Define
    \begin{equation}
        w := \frac{i}{2} \sum_{k = 1}^{n} dz_k \wedge d\bar{z}_k.
    \end{equation}
    Observe that $\mathbb{C}^n \simeq \mathbb{R}^{2n}$ together with $z_k = x_k + i y_k$, and the following holds:
    \begin{equation}
        \frac{\partial}{\partial z_k} := \frac{1}{2} \left(
        \frac{\partial}{\partial x_k} - i \frac{\partial}{\partial y_k}
        \right) \quad \text{and} \quad
        \frac{\partial}{\partial \bar{z}_k} := \frac{1}{2} \left(
        \frac{\partial}{\partial x_k} + i \frac{\partial}{\partial y_k}
        \right),
    \end{equation}
    along with
    \begin{equation}
        dz_k = dx_k + i dy_k \quad \text{and} \quad
        d\bar{z}_k = dx_k - i dy_k.
    \end{equation}
    So, we have
    \begin{align*}
        dz_k \wedge d\bar{z}_k &= (dx_k + i dy_k) \wedge (dx_k - i dy_k)\\
        &= - i dx^k \wedge dy^k + i dy^k \wedge dx^k\\
        &= -2 i dx^k \wedge dy^k .
    \end{align*}
    Hence, in the real coordinates, we get
    \begin{equation}
        w := \frac{i}{2} \sum_{k = 1}^{n} dz_k \wedge d\bar{z}_k
        = \sum_{k = 1}^{n} dx_k \wedge dy_k
        = w_0,
    \end{equation}
    which we already know is symplectic.
\end{example}

\begin{example}
    Let $M = S^2 = \left\{ 
    (x, y, z) \in \mathbb{R}^3 \; | \; x^2 + y^2 + z^2 = 1
    \right\}$. Consider the usual inner product $\expval{\cdot , \cdot}$ on $\mathbb{R}^3$
    and restrict to an inner product on $S^2$

    \begin{equation}
        v_p \in T_p S^2 \Longrightarrow p \perp v_p
        \quad \text{i.e.}
        \expval{p, v_p} = 0,
    \end{equation}

    \begin{equation}
        T_p S^2 \simeq \{p\}^\perp 
        :=
        \{ v \in \mathbb{R}^3 | \expval{p, v} = 0 \}.
    \end{equation}

    Define $\omega_p : T_p S^2 \times T_p S^2 \to \mathbb{R}$
    by 

    \begin{equation}
        \omega_p (u, v) := \expval{p, u \times v}.
    \end{equation}

    Is $\omega_p$ a symplectomorphism?\\
    \textbf{Solution:}
    \begin{itemize}
        \item Since $u \times v = - v \times u$, we have, $\omega_p (u,v) = - \omega_p (v,u)$. So skew-symmetricity holds.
        \item Bilinear
        \item $d\omega=0$, since $d\omega$ is a 3-form but
        dim$S^2 = 2$.
        \item If $\expval{p, u \times v} = 0 \quad \forall v \in T_p S^2$
        then by definition of $\times$, $u=0$
    \end{itemize}
    \begin{center}
        \includestandalone[width=0.3\textwidth]{Figures/Part1/p8}
    \end{center}
\end{example}

\begin{definition}
     Let $(M_1 , \omega_1)$ and $(M_2 , \omega_2)$ be 2n-dimensional symplectic manifolds. Let $\varphi: M_1 \rightarrow M_2$ be a diffeomorphism. Then, $\varphi$ is called a \textbf{symplectomorphism} if $\varphi^* \omega_2 = \omega_1$. Here $\varphi^*$ is the pull-back map
    \begin{equation}
        \begin{tikzcd}[sep=1.5cm]
            T_pM_1 \times T_pM_1  \arrow[thick, rd, "\varphi^*\omega_{2q}"]\arrow[thick,r, "{(\varphi_*,\varphi_*)}"] & T_qM_2\times T_qM_2 \arrow[thick, d, "\omega_{2q}"] \\ 
             & \mathbb{R},
        \end{tikzcd}
    \end{equation}
    which sends
    \begin{equation}
        \begin{tikzcd}[sep=1.5cm]
            (x_p,y_p)  \arrow[thick, swap, rd,"{\omega_{2q}\circ(\varphi_*,\varphi_*)}"]\arrow[thick,r, "{(\varphi_*,\varphi_*)}"] &(\varphi_*(x_p),\varphi_*(y_p)) \arrow[thick, d, "\omega_{2q}"] \\ 
             & \omega_{2q}(\varphi_*(x_p),\varphi_*(y_p))
        \end{tikzcd}
    \end{equation}
\end{definition}

Schematically, one can show a symplectomorphism via
\begin{figure}[h!]
    \centering
    \includestandalone[width=0.8\textwidth]{Figures/Part1/p11}\\
    \hypertarget{fig2}{Figure 2.} Diagram of symplectomorphism between two symplectic manifolds, $(M_1,\omega_1)$ and $(M_2,\omega_2)$.
\end{figure}

\begin{remark}
    \begin{enumerate}
        \item As in the case of "the moduli" space of symplectic vector spaces, we would like to classify symplectic manifolds up to symplectomorphism. This can be done locally as \textit{Darboux} theorem suggets. I.e., the dimension (as before) is the only local invariant of symplectic manifolds up to symplectomorphism.
        \item Any symplectic manifold $(M,\omega)$ locally looks like $(\mathbb{R}^{2n},\omega_0)$. I.e., Any symplectic manifold $(M,\omega)$ is locally symplectomorphic to $(\mathbb{R}^{2n},\omega_0)$. There exists a symplectomorphism $\varphi: U \xrightarrow{\sim} V$ where $U \subseteq_{open} M$, $V \subseteq_{open} \mathbb{R}^{2n}$ such that 
        \newline
        $$\varphi^{*} \omega_2 |_{V} = \omega_1 |_{U}, $$
        where
        $\omega_{1} \in \Omega^2 (U), \omega_1 : T_{p}U \times T_p U \rightarrow \mathbb{R} $ and
        $\omega_{2} \in \Omega^2 (V), \omega_1 : T_{q=\varphi(p)}V \times T_q V \rightarrow \mathbb{R}. $
    \end{enumerate}
\end{remark}

\begin{theorem}
    \textbf{Darboux:} Let $(M,\omega)$, $\dim M = 2n$ be a symplectic manifold, $p \in M$. Then there exists a local coordinate chart $(U,x_1,...,x_n,y_1,...,y_n)$ at $p$ s.t.
    \begin{equation}
        \omega = \sum_{i = 1}^{n} dx_i \wedge dy_i \:\: \text{on} \:\: U.
    \end{equation}
\end{theorem}

The local coordinate chart is also called \textit{Darboux's Chart}. Locally $\omega$ looks like a symplectic form $\omega_0$ on $\mathbb{R}^{2n}$. A fact is that there exists no local invariant in symplectic geometry, in contrast to Riemannian geometry, where the curvature provides such invariants.

\section{Symplectic Form on the Cotangent Bundle}
Let us start with some preliminary definitions and properties regarding general vector bundles. 


\begin{definition}
    An \textbf{$n$-dimensional vector bundle} over a base space $X$ is a tuple $(E, \pi, X)$ such that:
    \begin{enumerate}
        \item $E$ and $X$ are both topological manifolds.
        \item $\pi: E\to X$ is a continuous map, called the \textit{projection map}, such that for each $p \in X$, $\pi^{-1}(p)$
        is an n-dimensional vector space over $\mathbb R$ (or $\mathbb C$). One can also denote $\pi^{-1}(p)$ by $E_p$; formally called a \textit{fiber }over p.
        \item Local Trivialization: $E$ locally looks like $U \cross \mathbb{R}^{n}$.
    \end{enumerate}
    \begin{center}
        \includestandalone[width=0.5\textwidth]{Figures/Part1/p13.1} \\ 
        \hypertarget{fig3}{Figure 3.} A vector bundle over $X$. 
    \end{center}
\end{definition}

Locally, this schematic looks like:
\begin{figure}[h!]
    \centering
    \includestandalone[width=0.55\textwidth]{Figures/Part1/p13.2}\\
    \hypertarget{fig4}{Figure 4.} Local picture of the fibers and the base space.
\end{figure}

    One can also see that for each $p\in X$, there exists a neighborhood $U_p$ of $p$ such that
    \begin{enumerate}
        \item There exists a homeomorphism $h_p : \pi^{-1}(U_p)\xrightarrow{\simeq} U \times \mathbb{R}^{n}$,
        \item For each point $q \in U_p$; $h_p$ induces an isomorphism of vector spaces:
        \begin{equation}
            \pi^{-1}(q) \xrightarrow{\simeq} \{ q\} \times \mathbb{R}^{n} \cong \mathbb{R}^n,
        \end{equation}
        i.e. each fiber is isomorphic to $\mathbb{R}^n$.
    \end{enumerate}

\begin{observation}
    \begin{enumerate}
        \item $E = \coprod_{p \in X} E_p $,  and for each $p \in X$, $v_p \in E_p$. That is, each vector $v_p \in \pi^{-1}(p)$ projects  onto a unique point $p$ at which the corresponding fiber $\pi^{-1}(p)$ lives.
        \begin{equation}
            \pi(v_p) = p.
        \end{equation}
        \begin{center}
            \includestandalone[width=0.44\textwidth]{Figures/Part1/p13.3}\\
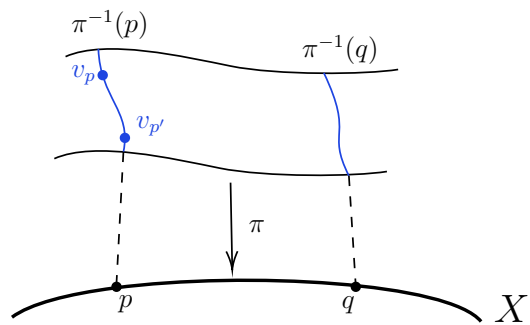
            \hypertarget{fig5}{Figure 5.} Projections of vectors onto points in the base space.
        \end{center}
        \item Same definition will work if we consider ``smooth" structure on given manifolds:
        \begin{enumerate}
            \item Replace ``continuous map" by ``smooth".
            \item Replace ``homeomorphisms" by ``diffeomorphisms", and so on...
        \end{enumerate}
        \item Local trivialization implies that $X$ admits a cover $\{U_X\}$ for which local trivialization conditions hold and $(U_{\alpha},h)$ is called ``local coordinate chart" for each $\alpha$.
    \end{enumerate}
\end{observation}

\newpage

\begin{definition}
    A vector bundle $(E,\pi,X)$ is called \textbf{trivial} if $E \cong X \cross \mathbb{R}^{n}$. In fact, the following diagram commutes:
    \begin{equation}
        \begin{tikzcd}[sep=1.5cm, remember picture]
            |[alias=E]| E \arrow[r,thick,"\varphi"] \arrow[d,thick,"\pi"] & |[alias = Y]| X \times \mathbb{R}^n \arrow[ld,thick,"\pi'"] \\ 
           |[alias = X]| X & 
        \end{tikzcd}
        \tikz[overlay,remember picture]{%
        \node (Y1) [scale=1.5] at (barycentric cs:X=1,E=1,Y=1) {$\circlearrowright$}}
    \end{equation}
\end{definition}

\begin{remark}
    Local trivialization implies that there exists a cover $\{U_\alpha\}_{\alpha \in \Lambda}$ of X such that $E|_{U_{\alpha}}$ is trivial $(E|_{U_\alpha} \cong U_{\alpha} \cross \mathbb{R}^{n})$.
\end{remark}

\begin{definition}
    Let $(E,\pi,X)$ be a vector bundle. A continuous function $s: X \rightarrow \mathbb{R}$ with $\pi \cdot s = \text{id}_X$ is called a \textbf{section}.
\end{definition}

A section $s$ may not be defined globally in general. Hence, we always consider $s$ defined on an open subset $U \subseteq X$:
\begin{equation}
    \begin{aligned}
         s: & U \subseteq_{open} X \rightarrow E,  \\
            & \pi \cdot s = \text{id}_U.
    \end{aligned}
\end{equation}

\begin{figure}[h!]
    \centering
    \includestandalone[width=0.8\textwidth]{Figures/Part1/p14}\\
    \hypertarget{fig6}{Figure 6.} Sections on the bundle over a two-dimensional base space.
\end{figure}

One can use $\Gamma$ to denote the smooth sections as in $\Gamma(TM)$ of the tangent bundle as well as $\Gamma(T^*M)$ of the cotangent bundle. 

\begin{definition}
    Smooth sections $\Gamma(TM)$ of the tangent bundle are called \textbf{smooth vector fields}:
    \begin{equation}
        \begin{aligned}
            V: X &\longrightarrow TM \\
            p &\longmapsto V_p,
        \end{aligned}
    \end{equation}
    where $V_p \in T_p M$.
\end{definition}

\begin{theorem}
    An $n$-dimensional vector bundle $(E,\pi,X)$ is trivial if and only if there exist $n$ global sections $s_1 , ..., s_n$ such that for each $p \in X$, the vectors $s_1 (p), ..., s_n (p)$ are linearly independent, and hence form a basis for $\pi^{-1}(p)$.
    \begin{center}
        \includestandalone[width=0.7\textwidth]{Figures/Part1/p15}\\
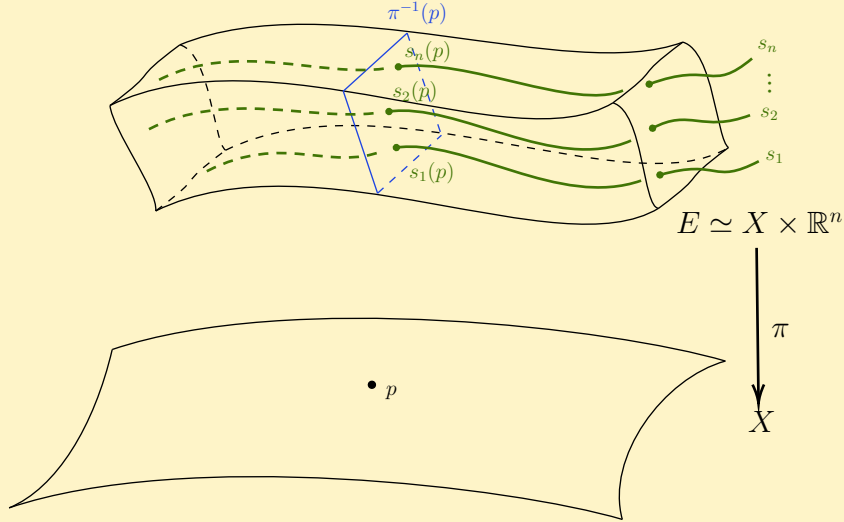
        \hypertarget{fig7}{Figure 7.} Sections on a trivial vector bundle
    \end{center}
\end{theorem}
\begin{proof}
    Suppose that $(E,\pi,X)$ is trivial. Then we have a homeomorphism $\varphi: E \rightarrow X \times \mathbb R^n$ commuting the diagram \begin{equation*}
        \begin{tikzcd}[sep=1.5cm, remember picture]
            |[alias=E]| E \arrow[r,thick,"\varphi"] \arrow[d,thick,"\pi"] & |[alias = Y]| X \times \mathbb{R}^n \arrow[ld,thick,"\pi'"] \\ 
           |[alias = X]| X & 
        \end{tikzcd}
        \tikz[overlay,remember picture]{%
        \node (Y1) [scale=1.5] at (barycentric cs:X=1,E=1,Y=1) {$\circlearrowright$}}
    \end{equation*}Then, for each $i=1,\dots,n,$ we define $s_i:X\rightarrow E$ to be the map sending $p\mapsto \varphi^{-1}(p, e_i),$ where $e_i$ denotes the standard basis element for $\mathbb R^n.$ It is then straightforward to compute
    \begin{align*}
     \pi \circ s_i(p)&=\pi (\varphi^{-1}(p, e_i) ) = (\pi' \circ \varphi) (\varphi^{-1}(p, e_i)=\pi' (p,e_i) = p.
    \end{align*}

    Conversely, given  such  sections $s_1 , ..., s_n$, it suffices to define a map $\varphi: E \rightarrow X \times \mathbb R^n$ by $s_i(p) \mapsto (p,e_i)$. Also, letting $\psi: X \times \mathbb R^n \rightarrow E, \ (p,v=\sum v_ie_i) \mapsto \sum v_i s_i(p)$, it is then straightforward to verify that both $\varphi$ and $\psi$ satisfy the desired properties such that $\psi=\varphi^{-1}$.
\end{proof}
Note that the local triviality condition implies that there exists a cover $ \{U_\alpha\}_{\alpha \in \Lambda} $ of $X$ such that $E|_{U_{\alpha}}$ is trivial.    it follows from the theorem above that there exist  $n$ sections $s_1, ... , s_n$ defined on $U_{\alpha}$ locally for all $\alpha$, with
\begin{equation}
    s_i : U_\alpha \subseteq M \rightarrow E , \ \pi \cdot s_i = \text{id}_{U},
\end{equation}
such that for each $q \in U_\alpha$, \ $s_1 , ..., s_n$ forms a basis for $\pi^{-1}(q)$ (\hyperlink{fig8}{Figure 8.}). That is, each vector bundle admits nowhere dependent $n$ ``local" sections due to \textit{local triviality} condition on $(E,\pi,X)$.\\
\begin{center}
    \includestandalone[width=0.6\textwidth]{Figures/Part1/p16}\\
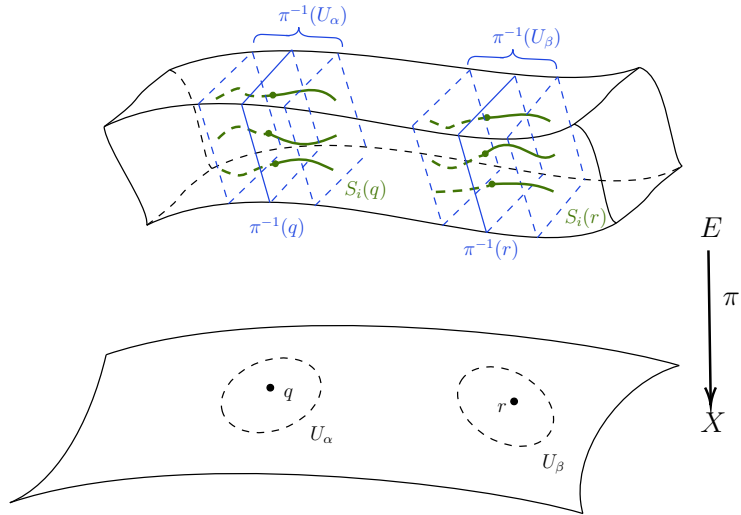
    \hypertarget{fig8}{Figure 8.} Local sections over different open sets $U_\alpha$.
\end{center}

\begin{definition}
    Let $X$ be an $n$-manifold with a coordinate chart $(U,x_1,...,x_n)$. Then the space
    \begin{equation}
        TX = \coprod_{p\in X}T_pX
    \end{equation}
    is the \textbf{tangent bundle} over $X$, with the map
    \begin{equation}
        \pi : TX \to X.
    \end{equation}
    Then the preimage
    \begin{equation}
        \pi^{-1}(p):=T_pX
    \end{equation}
    is said to be a \textbf{fiber} on the bundle. Similarly, one can define the \textbf{cotangent bundle}
    \begin{equation}
        T^*X = \coprod_pT_p^*X,
    \end{equation}
    and another appropriate map 
    \begin{equation}
        \pi: T^*X \to X.
    \end{equation}
    \begin{center}
        \includestandalone[width=0.45\textwidth]{Figures/Part1/p12} \\
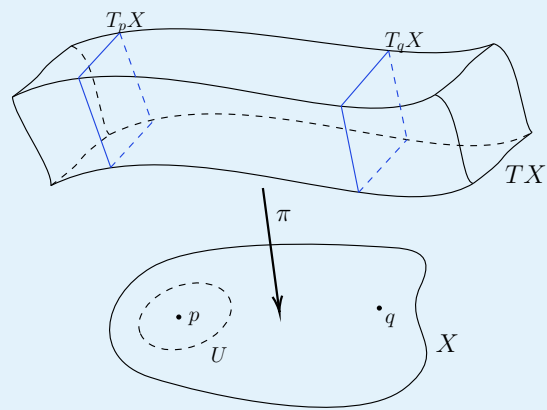
        \hypertarget{fig9}{Figure 9.} A tangent bundle $TX$ over the base manifold, $X$.
    \end{center}
\end{definition}

Now consider the cotangent bundle $\pi : T^{*}X \rightarrow X$. Any covector over $p$ is locally represented by:
\begin{equation}
    \xi = \xi_{i}(p) \: dx^{i} \;\;,\;\; \xi_{i} \; \in C^{\infty}(U;\mathbb R).
\end{equation}

It induces a map:
\begin{equation*}
    \begin{aligned}
        T^{*}U & \longrightarrow \mathbb{R}^{2n},
        \\
        (p,\xi) &\longmapsto (x_1 (p), ..., x_{n}(p), \xi_1(p),...,\xi_{n}(p)).
    \end{aligned}
\end{equation*}
Here, the map is a homeomorphism onto its image, and the last line is a coordinate chart for $T^{*}X$.
\begin{figure}[h!]
    \centering
    \includestandalone[width=0.7\textwidth]{Figures/Part1/p17}\\
    \hypertarget{fig10}{Figure 10.} Mapping from a cotangent bundle to $\mathbb{R}^n$.
\end{figure}

\begin{definition}
    $(x_1 (p), ..., x_{n}(p), \xi_1(p),...,\xi_{n}(p))$ are the \textbf{cotangent coordinates} associated to the coordinates $x_1 , ...,x_n$ on $U$.
\end{definition}
Given another chart $(V,y_1,...,y_n)$ about $p$, on the overlap $p \in U \cap V$, if $\xi \in T_{p}^{*}X$ we have:
\begin{equation*}
    ^{(x)} \xi_i = \frac{\partial y_\tau}{\partial x_i} \xi_\tau \;\;\; , \;\;\; ^{(y)} \xi_j = \frac{\partial x_i}{\partial y_j} \xi_i,
\end{equation*}
which are smooth. This shows that:
\begin{center}
    $\boxed{M = T^{\ast }X  \text{ is a $2n$-dimensional manifold.} }$
\end{center}
\newpage 

Now, let us construct symplectic structure on $M$:\\
1) Given $T^*X =: M \xlongrightarrow{\pi} X$, with $m = (p,\xi)$, where $\xi \in T^{\ast}_p X$,  $p \in X$ , $\pi(m) = p$ and $\xi : T_p X \rightarrow \mathbb{R}$, define a canonical (tautological) 1-form $\alpha \in \Omega^1(M)$ as follows:
For each $m \in M$, let
\begin{equation*}
    \begin{aligned}
        \alpha_M : T_m M &\longrightarrow \mathbb{R}
        \\
        v & \longmapsto \alpha_{m}(v),
    \end{aligned}
\end{equation*}
where $\alpha_m(v) = \xi(\pi_{\ast,m}(v)) $. Note that one can naturally define such an object once the coordinate chart is fixed.
Locally on $T^{*}U$, observe $\alpha=\xi_{i}d x_{i}$ at $m\in M; \text{ hence } \alpha_{m}=\xi_{i}(\pi(m))d x^{i}$, where $\xi=\xi_{i}d x^{i}$ and $\xi_{i} \in C^{\infty}(U,\mathbb{R})$.\\
2) Define a 2-form $\omega$ on M by $\omega \coloneqq -d\alpha$, and locally (on $T^{*}U$), one has
\begin{equation*}
    \omega = -d (\xi_{i} d x^{i}) = d x^{i} \wedge d \xi_{i}, \quad i = 1,2,...,n , 
\end{equation*}
where $\omega = \sum_{i=1}^{n}{dx^{i}\wedge d\xi_{i}}$.
\begin{example}
    $M = T^{*}\mathbb{R}^{n} \cong \mathbb{R}^{2n}$ with the coordinate chart $(U,q_{1},...,q_{n},p_{1},...,p_{n})$, where $q$ is position and $p$ is momentum. A canonical 1-form is $\alpha \coloneqq \sum_{i=1}^{n}p_{i}d q_{i}$, and hence $\omega = \sum_{i}{d q_{i} \wedge dp_{i}}$ $(\omega = -d \alpha)$.
\end{example}
\begin{observation}
    In that case $M=T^{*}\mathbb{R}^{n}$, $\omega$ is an exact form (globally). 
\end{observation}

\section{Hamiltonian Functions and  Vector Fields}

Let $(M,\omega)$ be a symplectic manifold. Then we can define the \textbf{space of smooth vector fields} on $M$ as $X \in \Gamma(TM) =: \mathfrak{X}(M),$ where $X$ is any smooth vector on the thangent bundle.
\begin{definition} 
    We define the \textbf{interior derivative} or the \textbf{contraction} as $\iota_{x}\omega \ (Y) \coloneqq \omega(X,Y)$, $\forall Y \in \mathfrak{X}(M)$. In general, $\iota_{x} : \Omega^{k}(M) \rightarrow \Omega^{k-1}(M)$, and $\alpha \rightarrow \iota_{x}\alpha =\alpha (X,\cdot,\cdot,...)$. 
\end{definition}
\begin{definition}
    Given $f \in C^{\infty}(M)$ (called ``observables"), a vector field $X_{f}$ is called \textbf{Hamiltonian vector field} if $\iota_{X_{f}} = d f$.  Such a function is then called \textbf{Hamiltonian function}.
\end{definition}

\begin{example}
    Consider $(\mathbb{R}^{2},\omega)$, where $\omega = d q \wedge d p$ with coordinates $(q,p)$. We claim that there exists a Hamiltonian vector field 
    \begin{equation*}
        \quad X_{f} \coloneqq -\frac{\partial f}{\partial q}\frac{\partial}{\partial p}+\frac{\partial f}{\partial p}\frac{\partial}{\partial q}.
    \end{equation*}
    Let $Y=\frac{\partial}{\partial p}$, then we have:
    \begin{equation*}
        \begin{aligned}
            \iota_{X_{f}}\omega \left(\frac{\partial}{\partial p}\right) &= \omega \left(X_{f},\frac{\partial}{\partial p}\right) = \omega \left(-\frac{\partial f}{\partial q}\frac{\partial}{\partial p},\frac{\partial}{\partial p}\right)+\omega \left(\frac{\partial f}{\partial p}\frac{\partial}{\partial q},\frac{\partial}{\partial p}\right)\\
            &= \frac{\partial f}{\partial p}\omega \left(\frac{\partial}{\partial q},\frac{\partial}{\partial p}\right) = d_{f}\left(\frac{\partial}{\partial p}\right).
        \end{aligned}
    \end{equation*}
    
    For $Y=\frac{\partial}{\partial q}$, similarly,
    \begin{equation*}
        \iota_{X_{f}}\omega \left(\frac{\partial}{\partial q}\right) = \frac{\partial f}{\partial q} = d_{f} \left(\frac{\partial}{\partial q}\right).
    \end{equation*}
    
    In general, for $(\mathbb{R}^{2n},\omega)$ with coordinates $(q_{1},...,q_{n},p_{1},...,p_{n})$ and symplectic structure $\omega = \sum_{i=1}^{n}{d q_{i}\wedge d p_{i}}$; given a Hamiltonian function $f\in C^{\infty}(\mathbb{R}^{2n})$, the associated Hamiltonian vector field $X_{f}$ is given by,
    \begin{equation*}
        X_{f} = \sum_{i=1}^{n} \frac{\partial f}{\partial p_{i}}\frac{\partial}{\partial q_{i}} -\frac{\partial f}{\partial q_{i}}\frac{\partial}{\partial p_{i}}.
    \end{equation*}
    Here, $X_{f}\in \mathfrak{X}(M)$, where $M=\mathbb{R}^{2n}$.
\end{example}

In fact, $X_{f} \in \textsf{Ham}(M) \subset \mathfrak{X}(M)$, where $\textsf{Ham}(M)$ is the \textbf{space of Hamiltonian vector fields}. 

\begin{center}
    $\qty{X_{g} \in \Gamma(TM) \;|\;  \iota_{X_{g}}\omega = d g \
    \text{for some} \; g \in C^{\infty}(M)}$.
\end{center}

\begin{figure}[h!]
    \centering
    \includestandalone[width=0.7\textwidth]{Figures/Part1/p20.1}\\
    \hypertarget{fig11}{Figure 11.} A flow, $\gamma$, mapping an interval to the manifold $M$.
\end{figure}

Let $X \in \text{Ham(M)}$. Say $X=X_H$ for some smooth function $H$. Then there exists a local flow $\gamma$ of $X$ defined as, $\gamma : (-\varepsilon,\varepsilon) \rightarrow M \quad$ such that
\begin{equation}
    \begin{aligned}    
        \gamma(0)=p \quad &\text{and} \quad \frac{d}{dt}\bigg|_{t=0}\gamma(t)=X_{p}, \\
        \frac{d}{dt}\gamma(t) &= X_{\gamma(t)}.
    \end{aligned}
\end{equation}

Say $\gamma(t)=(q_{1}(t),...,q_{n}(t),p_{1}(t),...,p_{n}(t))$, $\dot{\gamma}(t)=X_{\gamma(t)} \iff \dot{q}_{i}=\frac{\partial H}{\partial p_{i}}, \; \dot{p}_{i} = -\frac{\partial H}{\partial q_{i}}$,  where $\gamma$ is a local flow (integral curve) for $X_{H}$. In other words, we obtain the standard Hamiltonian equations of motion:
\begin{equation*}
    \dot{q}_{i}=\frac{\partial H}{\partial p_{i}}, \quad \dot{p}_{i} = -\frac{\partial H}{\partial q_{i}}.
\end{equation*}

\newpage 
\begin{example}
    $(X,\omega)=(\mathbb{R}^{3},\omega)$ with coordinates $(q_{1},q_{2},q_{3})=q$. Consider a particle of a mass m moving in configuration space $\mathbb{R}^{3}$ with coordinates $q=(q_{1},q_{2},q_{3})$. 
    \begin{center}
        \includestandalone[width=0.4\textwidth]{Figures/Part1/p20.2}
    \end{center}
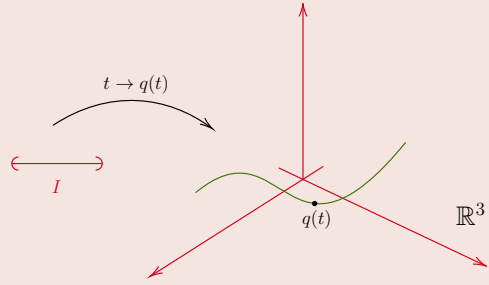
    \begin{itemize}
        \item Newton's $2^{nd}$ law in $\mathbb{R}^{3} \implies$ it moves along a curve $q(t)$ s.t. (consider a vector potential),
            \begin{equation*}
                \quad m\frac{d^{2}q}{dt^{2}} = -\nabla V(q).
            \end{equation*}
        \item Introduce momenta,
            \begin{equation*}
                p_{i}=m\frac{dq_{i}}{dt}, \; i =1,2,3.
            \end{equation*}
        \item $H(p,q) \coloneqq \frac{1}{2m}\abs{p}^{2}+V(q) = \frac{1}{2m}\sum^{3}_{i=1}{p_{i}^{2}}+V(q)$.
    \end{itemize}
    
    Consider $T^{*}\mathbb{R}^{3} \cong \mathbb{R}^{6} \; (\text{the corresponding phase space})$ together with $\omega = \sum_{i=1}^{3}{dq_{i}\wedge dp_{i}}$.
    
    Then, Newton's $2^{nd}$ law in $\mathbb{R}^{3}$ (configuration space) reads as the Hamiltonian equation for phase space in $\mathbb{R}^{6}$,
    \begin{equation*}
        \begin{aligned}
            & \frac{dq_{i}}{dt} = \frac{1}{m}p_{i} = \frac{\partial H}{\partial p _{i}}, \\
            & \frac{dp_{i}}{dt} = m\frac{d^{2}q_{i}}{dt^{2}} = \frac{\partial V(q)}{\partial q_{i}} = - \frac{\partial H}{\partial q_{i}},
        \end{aligned}
    \end{equation*}
    
    where we obtain last equality using $V(q) = H(p,q)-\frac{1}{2m}\abs{p}^{2}$.
\end{example}

\section{Some Lie Theory}
\subsection{Lie Groups}
    \begin{definition}
        A \textbf{Lie group} $(G,\cdot)$ is a group together with a smooth complex manifold structure such that the maps
        \begin{enumerate}
            \item[i)] $\mu: G\times G \longrightarrow G$\quad with \quad $(g,h)\longmapsto g\cdot h$, 
            \item[ii)] $i: G\longrightarrow G$ \quad with \quad $g\longmapsto g^{-1}$,
        \end{enumerate}
    are both smooth.
    \end{definition}
    \begin{example}
        Some examples are: 
        \begin{itemize}
        \item  $G = \mathbb{R}^n$ with usual group structure (pointwise addition $+_{\mathbb{R}^n}$) and with the usual smooth manifold structure is a Lie group.
            \item Likewise, $\mathbb{C}^n$ is a Lie group under addition.
            \item $\mathbb{R}^*=\mathbb{R}-\set{0}$ is a 1-dimensional Lie group under multiplication (view as $GL(1,\mathbb{R})$).
            \item Given Lie groups $G_1,..., G_n$, the product space $G_1\times...\times G_n$ is a Lie group with component-wise multiplication. 
        \end{itemize}
    \end{example}

    \begin{example}
         Let us verify that the \textbf{circle group} $S^1 := \{z \in \mathbb{C} : |z| = 1\}$, with the group operation 
        \begin{equation}
            e^{i\theta_1} \cdot_\mathbb{C} e^{i\theta_2} = e^{i(\theta_1 + \theta_2)},
        \end{equation} indeed admits a Lie group structure. Note that 
        \begin{itemize}
            \item $\lrp{e^{i\theta}}^{-1}=e^{-i\theta}$,
            \item $e^{i\theta}\cdot e^{i\psi} = e^{i(\theta+\psi)}$ \quad (corresponds to rotations), and
            \item there is an identity element $1$, along with other standard conditions satisfied. 
        \end{itemize}
    Therefore, $(S^1,\cdot_\mathbb{C})$ is a group.
    \end{example}
    Moreover, $S^1\subseteq \mathbb{R}^2$ and it \textit{inherits the topology} from $\mathbb{R}^2$ as a subspace topology. To see this:
    \vspace{0.7mm}

    \begin{wrapfigure}{l}{0.4\textwidth}
       \vspace{-\baselineskip} 
       \centering
       \begin{tikzpicture}
            \draw[black, thick] (-2,0) -- (2,0);
            \draw[black,thick] (0,-2) -- (0,2);
            \draw[black] (0,0) circle (1);
            \draw [red,ultra thick,domain=11:78] plot ({cos(\x)}, {sin(\x)});
            \node[red] at (0.4,0.6) {$V$};
            \draw[black, dotted, thick] (0.7,0.7) circle (0.6);
            \node[black] at (1.3,1.3) {$U$};
       \end{tikzpicture}\\
       \hypertarget{fig12}{Figure 12.} An open set $U$ containing portion of $S^1$.
    \end{wrapfigure}
    $V\subseteq S^1$ is open if and only if there exists 
    $U\underset{\mathrm{open}}{\subset}\mathbb{R}^2$ such that $V = S^1 \cap U$. Next, we have to show that $S^1$ has the \textit{structure of a 1-manifold}. Informally, we can say that $S^1$ looks like an open subset in $\mathbb{R}^1$. More formally, however, we say that it is locally homeomorphic to $\mathbb{R}^1$ for all $p \in S^1$. Consider the chart $(U,\varphi)$:
    \begin{align*}
        \varphi : U = S^1-\set{(0,1)} &\longrightarrow \varphi(U) \subseteq \mathbb{R}^1 \\
        (u,v) &\longmapsto \frac{u}{1-v}.
    \end{align*}
    Defining $\xi :=u/(1-v)$,  
    \begin{equation}
        \xi^2 = \frac{u^2}{(1-v)^2} = \frac{1-v^2}{(1-v)^2} = \frac{1+v}{1-v},
    \end{equation}
    where we used the fact that $u$ and $v$ are related by the Pythagorean relation, we get
    \begin{equation}
        v = \frac{\xi^2-1}{\xi^2+1}.
    \end{equation}
    By definition,
    \begin{equation}
        u = \xi(1-v) = \xi\frac{2\xi}{\xi^2+1}.
    \end{equation}
    Hence, the inverse map is
    \begin{align*}  
        \varphi^{-1}: \varphi(U) &\longrightarrow U \\
        \xi &\longmapsto \lrp{\frac{2\xi}{\xi^2+1}, \frac{\xi^2-1}{\xi^2+1}}.
    \end{align*}
    Since both $\varphi$ and $\varphi^{-1}$ are continuous, $\varphi$ is a homeomorphism. Now consider another chart $(V,\psi)$ such that
    \begin{align*}
        \psi: V = S^1-{(0,-1)} &\longrightarrow\psi(V)\subseteq \mathbb{R}^1\\
        (u,v)&\longmapsto \frac{u}{1+v} .
    \end{align*}
    Similar to before, we have
    \begin{equation}
        \eta := \frac{u}{1+v} \implies \eta^2 = \frac{u^2}{(1+v)^2} = \frac{1-v^2}{(1+v)^2} = \frac{1-v}{1+v},
    \end{equation}
    and
    \begin{equation}
        u = \eta(1+v) = \frac{2\eta}{\eta+1} ,\qquad   v = \frac{1-\eta^2}{\eta^2+1},
    \end{equation}
    which gives the inverse map as 
    \begin{align*}
        \psi^{-1}: \psi(V) &\longrightarrow V  \\
        \eta &\longmapsto \lrp{\frac{2\eta}{\eta^2+1},\frac{1-\eta^2}{\eta^2+1}}.
    \end{align*}
    Again, since both are continuous, $\psi$ is a homeomorphism. Then, $\set{(U,\varphi), (V,\psi)}$ form a (maximal) atlas $\mathcal{A}$, which is the main guideline to decipher  the geometry of $S^1$.\\
    
    Now we need to verify that $S^1$ admits a \textit{smooth structure}. To see this, we check whether the transition functions of the atlas $\mathcal{A}$ are smooth. Looking at $\varphi\circ\psi^{-1}$ on the overlap $U\cap V$. Let $\eta \in U\cap V$. Then we get
    \begin{equation}
        \varphi\circ\psi^{-1}(\eta) = \varphi\lrp{\frac{2\eta}{\eta^2+1},\frac{1-\eta^2}{\eta^2+1}} = \frac{\frac{2\eta}{\eta^2+1}}{1-\frac{1-\eta^2}{\eta^2+1}} = \frac{2\eta}{2\eta^2}=\frac{1}{\eta}.
    \end{equation}
    We also have, for $\xi\in U\cap V$, 
    
    \begin{equation}
        \psi\circ\varphi^{-1}(\xi) = \psi\lrp{\frac{2\xi}{\xi^2+1},\frac{\xi^2-1}{\xi^2+1}} = \frac{\frac{2\xi}{\xi^2+1}}{1+\frac{\xi^2-1}{\xi^2+1}}=\frac{1}{\xi}.\footnote{Note that $\eta\neq0$ and $\xi\neq0$ since $\psi^{-1}(0)=(0,1)\not\in U$ and $\varphi^{-1}(0)=(0,1)\not\in V.$}
    \end{equation}
    
    Since both of these are smooth, we have confirmed that $S^1$ is a smooth manifold. Another way to see this is the proposition: $S^1 \cong U(1) \cong \mathbb{R}/\mathbb{Z}\cong SO(2)$. The first isomorphism is the canonical one by the definition of $U(1)$; others follow from the exponential mapping, and Euler's identity correspondence of $e^{i\theta}$ with the rotation matrix of $SO(2)$, respectively. \\
    Now, the only thing left to do is to show \textit{$\mu$ and $i$ are smooth maps}. Starting with $i: S^1 \longrightarrow S^1$:
    \begin{equation}
        \begin{tikzcd}[sep=1.5cm]
            e^{i\theta} \subseteq U \arrow[r, thick, "i"] \arrow[d, thick, "\varphi"] & e^{i\theta} \subseteq V \arrow[d, thick, "\psi"] \\ 
            \mathbb{R} \arrow[r,thick,"\tilde{i}"] & \mathbb{R}
        \end{tikzcd}
    \end{equation}
    Here, we can see that $\tilde{i} = \psi \circ i \circ \varphi^{-1}$. Then,
    \begin{align}
        \tilde{i}(\xi) &= \psi\circ i\circ \varphi^{-1}(\xi) = \psi\circ i\lrp{\frac{2\xi}{\xi^2+1}, \frac{\xi^2-1}{\xi^2+1}}\notag\\
        &=\psi\lrp{\frac{2\xi}{\xi^2+1}, \frac{1-\xi^2}{\xi^2+1}} = \frac{\frac{2\xi}{\xi^2+1}}{1+\frac{1-\xi^2}{\xi^2+1}}= 2\xi \quad \quad \quad (\text{hence differentiable}).
    \end{align}
    Likewise, we can also write $\tilde{i}=\varphi\circ i \circ \psi^{-1}$, where
    \begin{align}
        \tilde{i}(\eta) &= \varphi \circ i \psi^{-1}(\eta) = \varphi \circ i\lrp{\frac{2\eta}{\eta^2+1},\frac{1-\eta^2}{\eta^2+1}} \notag \\
        &= \varphi\lrp{\frac{2\eta}{\eta^2+1}, \frac{\eta^2-1}{\eta^2+1}}=\frac{\frac{2\eta}{\eta^2+1}}{1-\frac{\eta^2-1}{\eta^2+1}} = \eta.
    \end{align}
    \paragraph{The matrix Lie groups.} We begin by defining the \textbf{general linear group},
    \begin{equation}
        GL(n,\mathbb{R}) = \set{A\in\mathbb{R}^{n\times n} : \det A \neq 0},
    \end{equation}
    equipped with the matrix multiplication operation. This group has the dimension $\dim_\mathbb{R}=n^2$. Note that 
    \begin{equation}
        [AB]_{ij} = \sum_k^n a_{ik}b_{kj}
    \end{equation}
    is a polynomial, and thus $\mu$ is a smooth map. Similarly, 
    \begin{equation}
        (A^{-1})_{ij} = \frac{(-1)^{i+j}}{\det A}\lrb{(j,i)-\text{minor of }A},
    \end{equation}
    
    which means that $i$ is also a smooth map. The determinant map $\det: \mathbb{R}^{n\times n}\longrightarrow\mathbb{R}$ which is a continuous map.
    
    One can also consider $GL(n,\mathbb{C})$ which is an open submanifold of $\mathbb{C}^{n\times n}$, and hence has $\dim_\mathbb{R}=2n^2$. Similarly, $GL(n,\mathbb{R})$ is an open submanifold of $\mathbb{R}^{n\times n}$, and thus a manifold of $\dim_\mathbb{R}=n^2$. This follows from the fact that since $\mathbb{R}-\set{0}$ is open, $\det^{-1}(\mathbb{R}-\set{0})=GL(n,\mathbb{R})$ must also be open. \\
    
    The next group is the \textbf{special linear group}.
    \begin{equation}
        SL(n,\mathbb{R}) = \set{A\in\mathbb{R}^{n\times n} : \det A = 1}.
    \end{equation}
    $SL(n,\mathbb{R})$ is a codimension-1 subgroup of $GL(n,\mathbb{R})$, and thus has $\dim_\mathbb{R}SL(n,\mathbb{R})=n^2-1$. This follows from the \textbf{regular value theorem} for the determinant map: $SL(n,\mathbb{R})= det^{-1}(A)$. Similar to before, $SL(n,\mathbb{C})$ is a codimension-1 (with respect to $\mathbb{C}$)subgroup of $GL(n,\mathbb{C})$ with $\dim_\mathbb{R}SL(n,\mathbb{C})=2n^2-2$. \\
    
    The \textbf{orthogonal group} is a subgroup of the general linear group defined as
    \begin{equation}
        O(n) = \set{A\in GL(n,\mathbb{R}): A^TA=\text{id}}=f^{-1}(\text{id}) ,
    \end{equation}
    where $f(A)=A^TA$. The dimension of the orthogonal group is
    \begin{equation}
        \dim_\mathbb{R}O(n)= \frac{n^2-n}{2}.
    \end{equation}
    One can note that for $A\in O(n)$,
    \begin{equation}
        \det\text{id} = \det(A^TA)=(\det A)^2 = 1 \implies \det A =\pm 1.
    \end{equation}
    Thus, all elements of the orthogonal group have determinant $\pm 1$. Those with $\det A = 1$ form the \textbf{Special orthogonal group}, 
    \begin{equation}
        SO(n) = O(n) \cup SL(n,\mathbb{R}) = \set{A\in\mathbb{R}^{n\times n}: A^TA=\text{id}, \det A = 1}.
    \end{equation}
    The dimension of the special orthogonal group is equal to the dimension of the orthogonal group. \\
    Similar to orthogonal matrices, unitary matrices also define a group called the \textbf{unitary group},
    \begin{equation}
        U(n) = \set{A\in\mathbb{C}^{n\times n}:A^*A=\text{id}}\subseteq GL(n,\mathbb{C}),
    \end{equation}
    with dimension $\dim_\mathbb{R}=n^2$. Using the same argument as before, we define the \textbf{special unitary group},
    \begin{equation}
        SU(n) = U(n)\cap SL(n,\mathbb{C}) = \set{A\in\mathbb{C}^{n\times n}: \det A=1, A^*A = \text{id}}
    \end{equation}
    with dimension $\dim_\mathbb{R}SU(n)=n^2-1$. 
    
    Now, let us consider some particular cases. Looking at $SU(2)$, by definition, we have
    \begin{equation}
        SU(2) = \lrc{\begin{bmatrix}
            \bar{\beta} & -\bar{\alpha} \\
            \alpha & \beta 
        \end{bmatrix}\in \mathbb{C}^{2\times 2} : \abs{\beta}^2 + \abs{\alpha}^2 =1}.
    \end{equation}
    Writing $\alpha = x_1+iy_1$ and $\beta=x_2+iy_2$, we get that 
    \begin{equation}
        SU(2) \cong \lrc{(x_1,y_1,x_2,y_2)\in\mathbb{R}^4 : x_1^2+y_1^2+x_2^2+y_2^2=1} = S^3.
    \end{equation}
    Furthermore, $\boxed{SU(2)$ is diffeomorphic to $S^3}$. What about $S^2$? \\
    $S^n$ is simply connected for $n\geq 2$, implying that $S^n$ serves as the universal cover of $\mathbb{RP}^n$. For $n=1$, $\pi_1(S^1)=\mathbb{Z}$. 
    \begin{equation}
        \pi_1(\mathbb{RP}^n)\cong \mathbb{Z}_2, \ \text{for} \ n>1 \implies \mathbb{RP}^n = S^n/\mathbb{Z}_2 \qquad \text{(by Antipodal action)}.
    \end{equation}
    This implies that $S^2 \cong SU(2)/U(1)$ and 
    \begin{equation}
        SU(2)/\mathbb{Z}_2 \cong \mathbb{RP}^3\cong SO(3).            
    \end{equation}

    Then, $SU(2)$ is the double cover  of $SO(3)$. Note that $SO(3)$ is the classical rotation group while $SU(2)$ is the quantum rotation group. \\
    
    Next, consider $SL(2,\mathbb{C})$:
    \begin{equation}
        SL(2,\mathbb{C}) = \lrc{\begin{bmatrix}
            a & b\\
            c & d 
        \end{bmatrix}\in\mathbb{C}^{2\times 2}:ad-bc=1}.
    \end{equation}
    To show that $SL(2,\mathbb{C})$ is in fact a Lie group, we can apply the same processes as we did for $S^1$. It should be noted that  $SL(2,\mathbb{C})$ is the quantum Lorentz group, which is the double cover of the Lorentz group $SO(1,3)$, the symmetry group of special relativity, where
    \begin{equation}
        SO(1,3) \cong SL(2,\mathbb{C})/\mathbb{Z}_2.
    \end{equation}
    Furthermore, one has the identification
    \begin{equation}
        SO(1,3)\cong \mathbb{R}^3\times(S^3/\mathbb{Z}_2).
    \end{equation}
    Now, let us define two charts. Start with $(U,\varphi)$, where 
    \begin{equation}
        U = \lrc{\begin{bmatrix}
        a & b \\ c& d 
        \end{bmatrix}\in SL(2,\mathbb{C}): a\neq 0}.
    \end{equation}
    Note that $U$ is an open subset:
    \begin{align}
        f: SL(2,\mathbb{C}) &\longrightarrow \mathbb{C}^4 \\
            \begin{bmatrix}
            a & b \\ c& d
            \end{bmatrix} &\longmapsto (a,b,c,d),
    \end{align}
    where we inherit the usual topology of $\mathbb{C}^4$. Considering the hypersurface $\lrc{(a,b,c,d) : a\neq 0}$, which is open in $\mathbb{C}^4$, we get
    \begin{equation}
        U = \underbrace{\underbrace{f^{-1}}_\text{cont.}\lrp{\lrc{(a,b,c,d): a \neq 0}}}_\text{open}.
    \end{equation}
    The map, $\varphi$, is then defined by
    \begin{align}
        \varphi: U &\longrightarrow \varphi(U) \\
                \begin{bmatrix}
                    a & b \\ c & d
                \end{bmatrix}&\longmapsto (a,b,c),
    \end{align}
    via the determinant $ad-bc=1 \implies d = \frac{1+bc}{a}$. Then, the inverse map is
    \begin{equation}
        \varphi^{-1}: (a,b,c) \longmapsto \begin{bmatrix}
        a & b \\ c & \frac{1+bc}{a}
        \end{bmatrix}.
    \end{equation}
    The second chart we define is $(V,\psi)$, where 
    \begin{equation}
        V = \lrc{\begin{bmatrix}
        a & b \\ c & d
        \end{bmatrix} \in SL(2,\mathbb{C}): b\neq 0},
    \end{equation}
    and 
    \begin{align}
        \psi: V &\longrightarrow \psi(V) \\
            \begin{bmatrix}
            a & b \\ c & d
            \end{bmatrix} &\longmapsto (a,b,d), 
    \end{align}
    with 
    \begin{equation}
        \psi^{-1}: (a,b,d) \longmapsto \begin{bmatrix}
        a & b \\ \frac{ad-1}{b} & d 
        \end{bmatrix}.
    \end{equation}
    Now, one can check the other conditions for the maximal atlas $\lrc{(U,\varphi), (V,\psi)}$. \\
    
    At last, let us now put forward some key theorems.
    \begin{theorem}
        $S^n$ admits a Lie group structure if $n=0,1$ and 3.
    \end{theorem}
    \begin{theorem}
        \textbf{Hairy Ball Theorem: } There exists no nowhere-vanishing vector field on $S^2$.
    \end{theorem}
    \begin{lemma}
        Let $G$ be a Lie group. Then there exists a nowhere-vanishing vector field on $G$. 
    \end{lemma}
    \begin{theorem}
        Let $G$ be a compact connected Lie group. Then the Euler characteristic vanishes, i.e. $\chi(G)=0$.
    \end{theorem}

\subsection{Lie Algebras}
    \begin{definition}
         A \textbf{Lie algebra}, $\mathfrak{g}$, is a vector space over a field $K$ equipped with a bilinear map, $[\cdot,\cdot]: \mathfrak{g}\times\mathfrak{g}\longrightarrow \mathfrak{g}$, called the \textbf{Lie bracket} such that
        \begin{enumerate}
            \item $[x,y]=-[y,x] \quad\forall x,y\in\mathfrak{g}$ (skew-symmetry),
            \item $[x,[y,z]]+[y,[z,x]]+[z,[x,y]]=0 \quad\forall x,y,z\in\mathfrak{g}$ (\textbf{Jacobi's identity}).
        \end{enumerate}
    \end{definition}
    \begin{example}
        As an example, consider the space $\Gamma(TG)$ of smooth vector fields on $G$.
        \begin{itemize}
            \item $\Gamma(TG)$ is an infinite dimensional $C^\infty(G)$-module (or $\mathbb{R}$-vector space).
            \item The Lie bracket is defined as the usual commutator in differential geometry,
                \begin{equation}
                    [x,y] := xy-yx = [x,y]_o.        
                \end{equation}
        \end{itemize}
        With these two conditions, $(\Gamma(TG), [\cdot,\cdot]_o)$ forms a Lie algebra. 
    \end{example}
    \begin{definition}
        Let $(G,\cdot)$ be a Lie group. Then for any $g\in G$, one can define a map
        \begin{align}
            l_g = G &\longrightarrow G \\
                  h &\longmapsto l_g(h):=g\cdot h,
        \end{align}
        which is called a \textbf{left-translation}. One can similarly define a right-translation. 
    \end{definition}
    \begin{observation}
        \begin{enumerate}
            \item Since the product operation is smooth, $l_g$ is also smooth.
            \item $l_g$ is invertible: It's inverse is defined as $(l_g)^{-1}=l_{g^{-1}}\quad \forall g\in G$. Taking the inverse is a smooth operation and thus, $l_{g^{-1}}$ is also smooth. This means that $l_g$ is a diffeomorphism. ($l_g: G \xlongrightarrow{\cong}G$).
            \item Since $(l_g)^{-1}$ exists, $l_g$ is an isomorphism considering $G$ as a set. However, it is not a group isomorphism as $l_g(h\cdot h')=g\cdot h \cdot h'$ does not necessarily equal $gh \cdot gh'$. 
        \end{enumerate}
    \end{observation}
    
    Since $l_g$ is a diffeomorphism, we can define the push-forward $l_{g^*}$. 
    \begin{center}
        \includestandalone[width=0.4\textwidth]{Figures/Part1/p30.1}\\ 
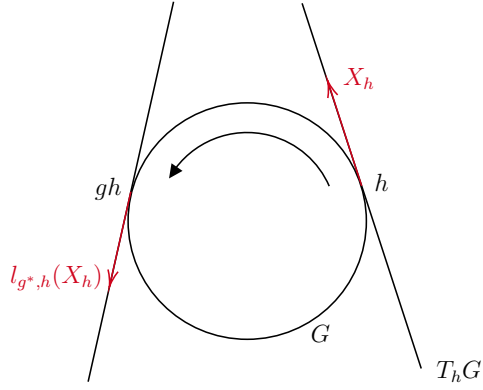
        \hypertarget{fig13}{Figure 13.} Pushforward map of left-translation.
    \end{center}
    \begin{equation}
        l_{g^*,h}: T_hG \longrightarrow T_{l_g(h)}G = T_{gh}G.
    \end{equation}
    \begin{definition}
        A vector field $X \in \Gamma(TG)$ is called \textbf{left-invariant} if 
        \begin{equation}
            l_{g^*}X = X,
        \end{equation}
        that is; for any $h\in G$,
        \begin{equation}
            l_{g^*,h}(X_h) = X_{gh}.
        \end{equation}
    \end{definition}

    We denote the \textbf{space of left-invariant vector fields} by $L\Gamma(TG)$. This is indeed a $C^\infty(G)$-submodule of $\Gamma(TG)$. 

    \begin{observation}
        Let $f\in C^\infty(G)$ and $X\in L\Gamma(TG)$. Then for $h,g\in G$, we have
        \begin{equation}
           \boxed {(l_{g^*}X_h)f = X_{gh}f}.
        \end{equation}
        Looking at the left-hand side,
        \begin{equation}
            (l_{g^*}X_h)f = X_h(f\circ l_g) = (X(f\circ l_g))(h).
        \end{equation}
        Now the right-hand side reads as
        \begin{equation}
            X_{gh}f = (X_f)(gh) = (X_f\circ l_g)(h).
        \end{equation}
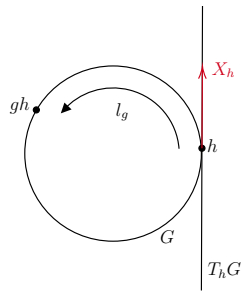
        \begin{center}
            \includestandalone[width=0.2\textwidth]{Figures/Part1/p30.2}
        \end{center}
        Then, we reach an important characterization (to decide whether a vector field is left-invariant or not):
        \begin{equation}
            \boxed{X\in L\Gamma(TG) \iff X(f\circ l_g) = X_f\circ l_g.}
            \label{eq:star}
        \end{equation}
    \end{observation}
    \begin{claim}
        $(L\Gamma(TG), [\cdot,\cdot]_o)\subseteq (\Gamma(TG), [\cdot,\cdot])$ is a \textbf{Lie subalgebra}. 
    \end{claim}
    \begin{proof}
        We need to show that for $x,y\in L\Gamma(TG)$, $[x,y]_o\in L\Gamma(TG)$. Here, we will make use of (\ref{eq:star}). As we stated before, $L\Gamma(TG)$ is a $C^\infty(G)$-submodule of $\Gamma(TG)$ or a $\mathbb{R}$-subvector space of $\Gamma(TG)$. Then, let $X,Y$ be left-invariant vector fields. For $g\in G$ and $f\in C^{\infty}(G)$, we have
        \begin{align}
            [X,Y]_o(f\circ l_g) &= (XY-YX)(f\circ l_g) \notag\\ 
                                &= X(Y(f\circ l_g))-Y(X(f\circ l_g))\notag\\
                                &= X(Yf \circ l_g)- Y(Xf \circ l_g)\notag\\ 
                                &= X(Yf)\circ l_g - Y(Xf)\circ l_g\notag\\
                                &= (XY-YX)(f)\circ (l_g).
        \end{align}
        This means that
        \begin{equation}
            [X,Y]_o (f\circ l_g) = [X,Y]_o f \circ l_g,
        \end{equation}
        which, by (\ref{eq:star}), implies that $[X,Y]_o\in L\Gamma(TG)$. 
    \end{proof}
    \begin{claim}
        $T_eG \cong L\Gamma(TG)$. 
        An immediate corrollary from this proposition is: $\dim L\Gamma(TG)=\dim T_eG = \dim G <\infty$\footnote{When $\dim G<\infty$.} even if $\Gamma(TG)$ is an infinite dimensional Lie algebra.
    \end{claim}
    \begin{proof}
         We define a map 
        \begin{align*}
            \varphi : T_eG &\longrightarrow L\Gamma(TG) \\ 
                    A &\longmapsto \varphi(A),
        \end{align*}
        where for each $g\in G$, $\varphi (A)_g := l_{g^*}A$ defines a vector field. We need to show that $\varphi(A)\in \Gamma(TG)$ is indeed a left-invariant vector field, i.e. to show that $l_{h^*}\varphi(A)_g = \varphi(A)_{hg}$. 
        \begin{center}
            \includestandalone[width=0.4\textwidth]{Figures/Part1/p31}
        \end{center}
        To see this, let us look at how it behaves when it acts on a function $f\in C^{\infty}(G)$:
        \begin{equation}
            l_{h^*}(l_{g^*}A)f = (l_{g^*}A)(f\circ l_{h})= A(f\circ l_h\circ l_g) = A(f\circ l_{hg}) = (l_{(hg)^*}A)f,
        \end{equation}
        \begin{equation}
            \therefore l_{h^*}\varphi(A)_g = \varphi(A)_{hg}.
        \end{equation}
        Hence, $\varphi(A)$ is a left-invariant vector field. Furthermore, its inverse is given by
        \begin{align*}
            \psi: L\Gamma(TG) &\longrightarrow T_eG \\
                    X &\longmapsto X_e.
        \end{align*}
        Then, we need to check that 
        \begin{enumerate}
            \item[(i)] $(\psi\circ\varphi)A = \psi(\varphi(A))=\varphi(A)_e=l_{e^*}A = A$,
            \item[(ii)] $(\varphi \circ\psi)(X)=\varphi(X_e)=X$.  
        \end{enumerate}
        The second point arises from the fact that for each $g\in G$;
        \begin{equation}
            \varphi(X_e)_g = l_{g^*}X_e = X_{ge} = X_g. 
        \end{equation}
        Hence, $T_eG \cong L\Gamma(TG)$ as a vector space.
    \end{proof}
    
    Our main result is then that there exists a \textit{Lie algebra homomorphism}\footnote{Recall that  $\beta$ is called a \textbf{Lie algebra homomorphism} if 
    \begin{equation}
        \beta: \mathfrak{g}\longmapsto \mathfrak{h}
    \end{equation}
    such that $\forall X,Y \in \mathfrak{g}$,
    \begin{equation}
        \beta ([X,Y]_\mathfrak{g}) = [\beta(X),\beta(Y)]_\mathfrak{h}.
    \end{equation}}
    \begin{equation}
        (T_eG,[\cdot,\cdot]) \cong (L\Gamma(TG), [\cdot,\cdot]_o), 
    \end{equation}
    where
    \begin{equation}
        [A,B] = [\varphi(A),\varphi(B)]_o\bigg\vert_e.
    \end{equation}
    \paragraph{The matrix Lie algebras.}
    \begin{table}[h!]
        \centering
        \begin{tabular}{c|c}
            Lie group	&  Lie algebra \\
            \hline
            $GL(n,\mathbb{R})$ & $\mathfrak{gl}(n,\mathbb{R})$ \\
            $SO(n)$ & $\mathfrak{so}(n)$ skew symmetric matrices \\
            $U(n)$ & $\mathfrak{u}(n)$ skew Hermition matrices \\ 
            $SL(n,\mathbb{C})$ & $\mathfrak{sl}(n,\mathbb{C})$ traceless matrices
        \end{tabular}
    \end{table}
    
    Calculating such Lie algebras from definitions is lengthy and quite messy. Instead, we use a \emph{practical method} via the notion of the exponential map. \\ 

    Let us try computing the Lie algebra of $SU(2)$ using the exponential map.
        \begin{equation}
            SU(2) = \set{A\in GL(2,\mathbb{C}) : AA^* = \text{id}, \det A = 1}.
        \end{equation}
        Let $M$ be a smooth manifold. The idea of exponential maps is that any "sufficiently close" two points in $M$ can be joined through a geodesic given by the exponential map. \\ 
        \begin{center}
            \includestandalone[width=0.6\textwidth]{Figures/Part1/p33} \\ 
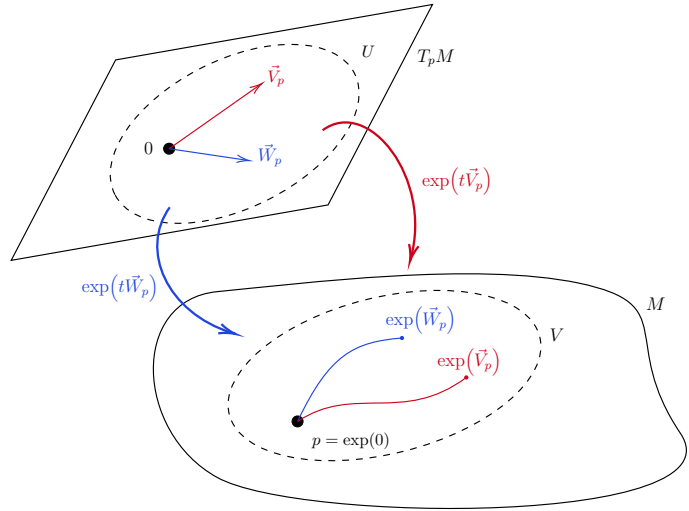
            \hypertarget{fig14}{Figure 14.} Exponential map from a tangent space to the manifold.
        \end{center}
    Here one can define a local diffeomorphism: $U \longrightarrow v := \exp(U)$ such that
    \begin{equation}
        \begin{aligned}
            \gamma(t) &:= \exp(t\v{v}_p), \\ 
            \gamma(0) &= p, \\ 
            \gamma(1) &= \exp(\v{v}_p).
        \end{aligned}
    \end{equation}
    Then any point $q \in V$ can be "hit" by a geodesic $\gamma$ starting from the point $p$. For this, the data we need are the initial point $p$ and the initial velocity $\v{v}_p$. \\ 

    Similarly, since $\left(\mathfrak{g},[\cdot,\cdot]_0\right) \simeq \left(T_eG, [\cdot,\cdot]\right)$, we have a surjection
    \begin{equation}
        \exp : \mathfrak{g} \longrightarrow G.
    \end{equation}
    For $\mathfrak{su}(2)$, the Lie algebra of $SU(2)$, we let $X\in \mathfrak{su}(2)$. Then, by the surjective map defined above, we get $\exp(t X)\in SU(2)$. So, we have,
    \begin{enumerate}
        \item $\exp(tX)\exp(tX)^* = \text{id}$ and,
        \item $\det\exp(tX)$.
    \end{enumerate}
    Now, observe that we have
    \begin{equation}
        \exp(tX) = e^{tX} = \text{id} + tX + \frac{(tX)^2}{2!} + \cdots.  = \text{id} + tX + \mathcal{O}(t^2).
    \end{equation}
    Then we have
    \begin{equation}
        (\text{id} + tX)\bar{(\text{id} + tX)^T} = (\text{id}+tX)(\text{id} + tX^*) = \text{id} + tX^* + tX + \mathcal{O}(t^2) = \text{id}.
    \end{equation}
    Thus, if $\text{id} + t(X^*+X)=\text{id}$ for all $t$, then we must have 
    \begin{equation}
        X^* = -X.
    \end{equation}
    Now, recall that for any $X\in \mathbb{R}^{n\times n}$, $\det(e^X)=e^{\tr X}$. To prove this, we make the following observations. 
    \begin{enumerate}
        \item[i)] Any complex (real) square matrix $X$ can be triangularised, i.e., there exists a non-singular complex square matrix $A$ such that $AXA^{-1}$ is upper triangular. 
        \item[ii)] $X$ and $AXA^{-1}$ share the same eigenvalues:
        \begin{equation}
        \begin{aligned}
            \det(\lambda\text{id} - AXA^{-1}) &= \det(\lambda AA^{-1}-AXA^{-1}) 
            \\
            &= \det(A(\lambda\text{id}-X)A^{-1})
            \\
            &= \det(\lambda\text{id} - X).
        \end{aligned}
        \end{equation}
        \item[iii)] Then, the eigenvalues $\lambda_1,\dots,\lambda_n$ of $AXA^{-1}$ (also those of $X$), must lie along its diagonal. Due to this observation, it is enough to assume $X$ is an upper triangular matrix with eigenvalues $\lambda_1,\cdot,\lambda_n$. 
            \begin{equation}
                X = \begin{bmatrix}
                    \lambda_1 & & \star \\
                     & \ddots & \\ 
                     0 & & \lambda_n
                \end{bmatrix}.
             \end{equation}
    \end{enumerate}
    Then, we have 
    \begin{align}
        e^X &= \sum_k \frac{1}{k!}X^k = \sum_k \frac{1}{k!} \begin{bmatrix}
                    \lambda_1^k & & \star \\
                     & \ddots & \\ 
                     0 & & \lambda_n^k
                \end{bmatrix} \notag \\ 
            &= \begin{bmatrix}
                \sum_k \frac{1}{k!}\lambda_1^k & & \star \\ 
                 & \ddots & \\
                 0 & & \sum_k \frac{1}{k!}\lambda_n^k 
                \end{bmatrix} = \begin{bmatrix}
                e^{\lambda_1} & & \star \\ 
                & \ddots & \\ 
                0 & & e^{\lambda_n}
            \end{bmatrix}.  
    \end{align}
    Therefore,
    \begin{equation}
        \det(e^X) = \prod_{i=1}^n e^{\lambda_i} = e^{\sum_i\lambda_i} = e^{\tr X}.
    \end{equation}
    With this, the second statement then reads as
    \begin{equation}
        1 = \det e^{tX} = e^{\tr(tX)} \implies \tr(tX)=t\tr{X} = 0 \quad\forall t.
    \end{equation}
    \begin{equation}
        \therefore \tr X = 0.
    \end{equation}
    Then the Lie algebra $\mathfrak{su}(2)$ is
    \begin{equation}
        \mathfrak{su}(2) = \set{X\in \mathbb{C}^{2\times2} : X^*=-X, \tr X=0}.
    \end{equation}
    The bases of this algebra are given by the \textit{Pauli matrices} with some scale factors.
    \begin{equation}
        i\sigma_1 = u_1 = \begin{pmatrix}
            0 & i \\ 
            i & 0 
        \end{pmatrix} \quad,\quad i\sigma_2 = u_2 = \begin{pmatrix}
            0 & -1 \\ 
            1 & 0 
        \end{pmatrix} \quad,\quad i\sigma_3 = u_3 = \begin{pmatrix}
            i & 0 \\ 
            0 & -i
        \end{pmatrix}.
    \end{equation}
    The \textit{structure relations} are :
    \begin{equation}
        \begin{aligned}
            [u_3,u_1] = 2u_2, \qquad
            [u_1,u_2] = 2u_3, \qquad
            [u_2,u_3] = 2u_1. 
        \end{aligned}
    \end{equation}
    \begin{remark}
        There is a Lie algebra isomorphism $\mathfrak{su}(2) \simeq \mathfrak{so}(3)$. However, their corresponding Lie groups, $SU(2)$ and $SO(3)$, are \emph{not} isomorphic. In fact, $SU(2)$ is the \textit{double cover} of $SO(3)$. In some sense, $SO(3)$ corresponds to 3D rotations in the macroworld while $SU(2)$ corresponds to infinitesimal rotations in the microworld. 
    \end{remark}
    \begin{example}
        Compute the Lie algebra $\mathfrak{sl}(2,\mathbb{C})$ of $SL(2,\mathbb{C})=\set{A\in \mathbb{C}^{2\times 2}: \det A =1}$. \\
        Consider the determinant function $\det : GL(n,\mathbb{R})\longrightarrow \mathbb{R}$ (or $\mathbb{C}$ instead of $\mathbb{R}$). Then, we have
        \begin{equation}
            \det\nolimits_{*,A}: T_nGL(n,\mathbb{R})\longrightarrow T_{\det A}\mathbb{R}.
        \end{equation}
        In particular, consider $A=\text{id}$. Then,
        \begin{equation}
            \det\nolimits_{*,\text{id}} : T_\text{id} GL(n,\mathbb{R}) \longrightarrow T_1 \mathbb{R} \simeq \mathbb{R}.
        \end{equation}
        We then claim, $\det\nolimits_{*,\text{id}}(X)=\Tr X$. To prove this, let $X\in T_\text{id} GL(n,\mathbb{R})$. Then\textbf{} there exists a curve $\gamma: (-\varepsilon,\varepsilon)\longrightarrow GL(n,\mathbb{R})$ such that $\gamma(0)=\text{id}$ and $d\gamma /d t \vert_{t=0}=X$ with $\gamma(t)=\exp(tX)$. Then, by the definition of the pushforward, we have
        \begin{equation}
            f_{*,p}(X_p) = \deriv{}{t}\bigg\vert_{t=0}f \circ \gamma.
        \end{equation}
        This gives
        \begin{align}
            \det\nolimits_{*,\text{id}}(X) = \deriv{}{t}\bigg\vert_{t=0}\det\circ\gamma = \deriv{}{t}\bigg\vert_{t=0}\det(e^{tX})= \deriv{}{t}\bigg\vert_{t=0}e^{t\tr(X)} = \tr(X)e^{t\tr(X)}\bigg\vert_{t=0}=\tr(X).
        \end{align}
    \end{example}

\newpage
\section{Poisson Algebra}
Let $(M,\omega)$ be a symplectic manifold, and let $f \in C^\infty(M)$. We then define the Hamiltonian vector field $X_f$ of $f$ by 
\begin{equation}
    \iota_{X_f} \omega = df.
\end{equation}
Then we have
\begin{equation}
   \begin{aligned}
    C^\infty(M) &\longrightarrow \textsf{Ham}(M) \subset \Gamma(TM),\\
    f &\longmapsto X_f.
\end{aligned} 
\end{equation}
Here, $\textsf{Ham}(M)$ is the space of Hamiltonian vector fields as we defined before. Now, define a bilinear map $\{ \cdot, \cdot\} : C^{\infty}(M) \times C^\infty(M) \longrightarrow C^{\infty}(M)$ by
\begin{equation} \label{Poisson bracket related with Hamiltonian vector fields}
    \{ f, g \} := - \omega(X_f, X_g).
\end{equation}
This bilinear map is called the \textbf{Poisson bracket} and satisfies the following properties.
\begin{lemma}
    \begin{enumerate}
    \item $\{f,g\} = - \{g, f\} $,
    \item $ \{f, \{ g, h \} \} + \text{cyclic permutations} =0$,
    \item $ \{fg, h\} = f \{g, h\} + \{f,h\} g $.
\end{enumerate}
\end{lemma}
With the Poisson bracket, one has a Lie algebra structure. Namely, $\big( C^\infty(M) , \{ \cdot, \cdot \} \big)$ is a Lie algebra where the Lie algebra bracket is identified with the Poisson bracket 
\begin{equation}
    [\cdot, \cdot] \equiv \{ \cdot, \cdot \}.
\end{equation}

\begin{claim}
$\big( \textsf{Ham}(M), [\cdot, \cdot ]_0 \big)$ is a Lie algebra, where $[\cdot, \cdot]_0$ is the bracket that descends from the commutator of vector fields in $\big(  \Gamma(TM), [\cdot, \cdot ]_0 \big)$. 
\end{claim}

\begin{proof}
To prove this claim, it suffices to show that the commutator of two Hamiltonian vector fields $X_f$ and $X_g$ gives rise to a Hamiltonian vector field for any two $f,g \in C^\infty(M)$. Namely, we need to show that, for any $X_f, X_g \in \textsf{Ham}(M)$, \begin{equation}
    [X_f, X_g]_0 \in \textsf{Ham}(M).
\end{equation}
For a moment, suppose that $[X_f, X_g]_0 \equiv X_h$ for some $h \in C^{\infty}(M)$. How does $h$ relate to $f$ and $g$? A very natural guess would be $h \equiv \{f,g\}$ so that 
\begin{equation} 
    [X_f, X_g]_0 = X_{\{f,g\}}.
\end{equation}
In other words, the Hamiltonian vector field corresponding to $\{f,g\}$ is given by the commutator $[X_f, X_g]$. While this appears quite intuitive and appealing, how can we show it to be true? The proof uses two things:

\begin{enumerate}
    \item By Cartan's identity, 
        \begin{equation}
            \mathcal L_X \alpha = d \ \iota_X \alpha + \iota_X \ d \alpha.
        \end{equation}
    If we choose $X = X_f \in \textsf{Ham}(M)$, and $\alpha = \omega$ with $\omega$ being the symplectic form on $M$, we get 
        \begin{equation}
            \mathcal L_{X_f} \omega = d (\underbrace{\iota_{X_f} \omega}_{df} ) + \iota_{X_f} \ \underbrace{d\omega}_{0 },
        \end{equation}
    and since $d^2 = 0$ holds on any object, we have 
        \begin{equation}
            \boxed{\mathcal L_{X_f} \omega = 0,}
        \end{equation}
    meaning that Hamiltonian vector field $X_f$ preserves the symplectic structure $\omega$ in the sense that $\omega$ is invariant under the Hamiltonian flow generated by $X_f$. Remember also that $\mathcal L_{X_f} f = 0$. That is, Hamiltonian vector field preserves its Hamiltonian function. To see this, one can imagine $f$ physically as the Hamiltonian, which commutes with itself and is preserved under time translations, the latter being the physical interpretation of Hamiltonian flows. Mathematically, one proves $\mathcal L_{X_f} f = 0$ as
        \begin{equation}
            \mathcal L_{X_f} f = X_f(f) = df(X_f) = \iota_{X_f} \ \omega(X_f) = \omega( X_f, X_f) = 0.
        \end{equation}
    \item For any $X,Y \in \Gamma(TM)$, one has 
        \begin{equation}
            \iota_{[X,Y]} = [\mathcal L_X, \iota_Y].
        \end{equation}
    Then, 
        \begin{equation}
        \begin{aligned}
            \iota_{[X_f,X_g]} \omega &= [\mathcal L_{X_f} , \iota_{X_g} ] \omega \\
            &= \mathcal L_{X_f} ( \underbrace{\iota_{X_g} \omega}_{\eta} ) - \iota_{X_g} ( \underbrace{\mathcal L_{X_f} \omega}_{0 \text{ (by item 1)} } ) \\
            &= d ( \iota_{X_f} \iota_{X_g} \omega) - \iota_{X_f } d (\underbrace{\iota_{X_g} \omega}_{dg} ) \\
            &= d \left[ \iota_{X_f} (\omega (X_g, \cdot) ) \right] \\
            &= d \big( \omega( X_g, X_f) \big) \\
            &= d ( \{f,g\}) \\
            &= \iota_{X_{ \{f,g\} }} \omega.
        \end{aligned}
        \end{equation}
\end{enumerate}
\end{proof}

Thus, we have established the following result. $\textsf{Ham}(M), [\cdot, \cdot]$ is a Lie algebra with

\begin{equation}
    \begin{aligned}
    [\cdot, \cdot]: \textsf{Ham}(M)\times \textsf{Ham} (M) &\longrightarrow \textsf{Ham} (M)\\
    (X,Y) &\longrightarrow [X, Y] = X_{\qty{f,g}}.
    \end{aligned}
\end{equation}

Note that there exists a natural map $\varphi$ from the two Lie algebras 
\begin{equation}
\begin{aligned}
    \varphi: \quad \big( C^\infty(M), \{\cdot, \cdot\} \big) \ \ \ &\longrightarrow \ \ \ \big( \textsf{Ham}(M), [\cdot, \cdot]_0 \big) \\
    f \ \ \ \hspace{10.5mm} &\longmapsto \ \ \ \hspace{10.5mm} X_f \\
    \{g,h\} \ \ \ \hspace{7mm} &\longmapsto \ \ \ X_{ \{g,h\} } = [X_g, X_h].
\end{aligned}
\end{equation}
\begin{itemize}
    \item $\varphi$ is a Lie algebra homomorphism as the algebra structure is preserved under $\varphi$:
    \begin{equation}
        \varphi(\qty{f,g}) = X_{\qty{f,g}} = [X_f, X_g] = [\varphi (f), \varphi(g)].
    \end{equation}
    \item $\varphi$ assigns to $f$ a differential operator $X_f$.
    \item The condition $``X_{ \{f,g\} } = [X_f, X_g]"$ resembles \emph{Dirac's quantization condition}, where one takes the Poisson bracket $\{f,g\}$ and makes it an operator $\widehat{\{f,g\}}$ such that 
        \begin{equation}
            \widehat{\{f,g\}} = \frac{1}{i\hbar} [ \hat{f} , \hat{g} ],
        \end{equation}
    where the quantities with hats are operators acting on the Hilbert space $\mathcal H$ of the quantum theory, and $\hbar$ is the quantization parameter (Planck's constant), governing the transition from classical phase space to the quantum Hilbert space.
\end{itemize}
One may suppose that this is the end of Geometric Quantization, and that the "deformation" above does the job, but there are some loose ends.
\begin{enumerate}
    \item What exactly is the suitable Hilbert space on which $\hat f$ acts? 
    \item Is the assignment $f \mapsto X_f$ suitable for our purposes?
\end{enumerate}
One problem with the assignment $f \mapsto X_f$ is that it is not injective. This is because if $f = c$ for some constant $c$, then $0 = df = \iota_{X_f} \omega( \cdot) = \omega(X_f, \cdot)$. However, $\omega$ is a non-degenerate bilinear form, so $X_f$ must be 0 identically! Hence, the map has a kernel, where any constant maps to the "0-operator". \\

To remedy the situation above and provide better construction, we will need some additional structures. To this end, we will first introduce some useful concepts in the following sections.
\newpage

\section{Group Actions on Manifolds}
\subsection{Lie Group Actions}

\begin{definition}

Let $(G,\cdot)$ be a Lie group, and $M$ be a smooth manifold. We define the \textbf{action of $G$ on $M$} as a smooth map:
\begin{equation}
    \begin{aligned}
        G \times M &\xlongrightarrow{\triangleright} M\\
        (g,p) &\longmapsto g \triangleright p
    \end{aligned}
\end{equation}
The operation $\triangleright$ is a smooth \textbf{left-action}, and the following two properties hold: $\forall p \in M, \; \forall g_1, g_2 \in G$
\begin{equation}
    1) \ e \triangleright p = p , \quad \quad 2) \ g_1 \triangleright (g_2 \triangleright p) = (g_1 \cdot g_2) \triangleright p,
\end{equation}
where $\cdot$ denotes the multiplication on the group $G$, and $e$ is the identity element.
\end{definition}
Similarly, one can define a smooth right-action by $(g,p) \longmapsto p \triangleleft g$. 
\begin{definition}

Let $(G, \cdot_G)$ and $(H, \cdot_H)$ be two Lie groups, and $\rho : G \longrightarrow H$ be a group homomorphism. Also let $M,N$ be two smooth manifolds, acted upon by $G,H$ as $G \times M \xlongrightarrow[]{\blacktriangleright} M$ and $H \times N \xlongrightarrow[]{\triangleright} N$. A function $f : M \longrightarrow N$ is \textbf{$\rho$-equivariant} if  the   diagram
    \begin{equation}
        \begin{tikzcd}[sep=1.5cm]
            G\times M \arrow[thick, r, "{(\rho,f)}"] \arrow[thick, d, "\blacktriangleright"] & H\times N \arrow[d,thick, "\triangleright"] \\ 
            M \arrow[thick, r, "f"] & N \arrow[to path={node[midway,scale=2] {$\circlearrowright$}}]{}
        \end{tikzcd}
    \end{equation}
commutes. One then calls such $f$ a \textbf{$\rho$-equivariant map}. The above diagram commuting means
\begin{equation}
    \begin{tikzcd}[sep=1.5cm]
        (g,p) \arrow[thick,r] \arrow[thick,d] & (\rho(g),f(p)) \arrow[d,thick] \\ 
        g\blacktriangleright p \arrow[thick, r] & f(g\blacktriangleright p) = \rho(g)\triangleright f(\rho)
    \end{tikzcd}
\end{equation}
The equality at the bottom right is the condition of equivariance.
    
\end{definition}

\begin{definition}

Let $G \longrightarrow \text{Diff}(M)$ be a left action $g \longmapsto g \triangleright \cdot$.

\begin{enumerate}
    \item For any point $p$, the \textbf{orbit} of the $G$-action on $p$ is defined as the set
    \begin{equation}
        \mathcal O_p = \qty{ q \; | \; q = g \triangleright p \ \ \forall g \in G }.
    \end{equation}
    \newpage
    \textbf{Example}: Consider the $SO(2)$ action on the $\mathbb R^2$ plane (or equivalently, the $S^1$ action on the $\mathbb C$ plane) 
    \begin{equation}
        A \triangleright p := A \cdot 
        \begin{bmatrix}
            p_1 \\p_2
        \end{bmatrix},  
    \end{equation}
    
    where $\cdot$ denotes matrix multiplication. Recall that $S^1 \simeq SO(2)$ from the homomorphism $e^{i\theta} \mapsto \begin{pmatrix}
        \cos\theta & \sin\theta \\ -\sin\theta & \cos\theta
    \end{pmatrix}$, so on the complex plane, this action corresponds to multiplication by a phase.
    \item The \textbf{stabilizer} subgroup of the orbit group is the set 
    \begin{equation}
        \text{stab}(p) = \{ g \in G \ | \ g \triangleright p = p \}.
    \end{equation}
    The action $\triangleright$ on the manifold $M$ is called \textbf{free} if the stabilizer subgroup is trivial, namely $\text{stab}(p) = \{ e \} \ , \ \forall p \in M$, with $e$ being the identity element of $G$.

    \textbf{Example:    }    
    Let $\mathbb C^* = \mathbb R^2 - \{0,0\}$. The $S^1$ action on $\mathbb C^*$ is free 
    \begin{equation}
        \text{stab} \Bigg
        ( \begin{bmatrix}
            0 \\0
        \end{bmatrix} \Bigg)  = G  
        \Longrightarrow e^{i\theta} z = r e^{i(\theta+\psi)},
    \end{equation}
    so that $e^{i\theta} z = z$ only when $\theta = 2\pi k \in \text{ker}(\exp)$.

    \item Define an equivalence relation using the $G$ action. Let $p, q \in M$:
    \begin{equation}
        p \sim q  \Longleftrightarrow \exists \ g \in G \ \text{  s.t. } \ q = g \triangleright p \ \ ( \text{i.e}. \;q \in \mathcal O_p).
    \end{equation}
    For example, in $U(1)$ gauge theory, two gauge potentials $A$ and $A'$ are declared the same if there exists a $g \in U(1)$ such that $A' = g \triangleright A$. Having defined the equivalence relation, one can also consider the quotient \\ $M / \sim \ := \ M / G$, which is called the \textbf{orbit space}.
\end{enumerate}
    
\end{definition}
\paragraph{Smooth $\mathbb R$-action on a smooth manifold:} Given a smooth manifold $M$, let $X$ be a complete vector field. For each $p \in M$, there is a unique flow

\begin{center}
    \includestandalone[width=0.7\textwidth]{Figures/Part1/p46.2}\\ 
    \hypertarget{fig15}{Figure 15.} A one-parameter flow $\rho_t$ on $M$.
\end{center}

where $\rho_t : M \longrightarrow M$ is a map $p \longmapsto \rho_t(p)$ such that 

\begin{itemize}
    \item $\rho_t$ is smooth,
    \item $\rho_t^{-1} = \rho_{t^{-1}}$ is also smooth (hence $\rho_t \in \text{Diff}(M), \; \forall t$),
    \item by uniqueness of $\rho_t$, we have $\rho_{t+s} = \rho_t \circ \rho_s$.
\end{itemize}

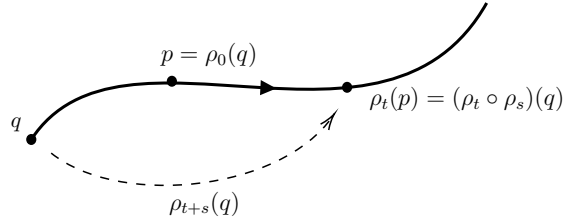
\begin{center}
    \includestandalone[width=0.46\textwidth]{Figures/Part1/p46.3}
\end{center}

Then, we can equip the family  $\{ \rho_t \; | \; t \in \mathbb R \}$ with the structure of a group 
using $\circ$. Such a family is called the \textbf{one-parameter group of diffeomorphisms}. We also have the group homomorphism $\rho: \mathbb R \longrightarrow \text{Diff}(M)$ acting as $t \longmapsto \exp(tX)$, which gives a smooth $\mathbb R$ action on $M$.

\begin{center}
    \includestandalone[width=0.66\textwidth]{Figures/Part1/p47}\\
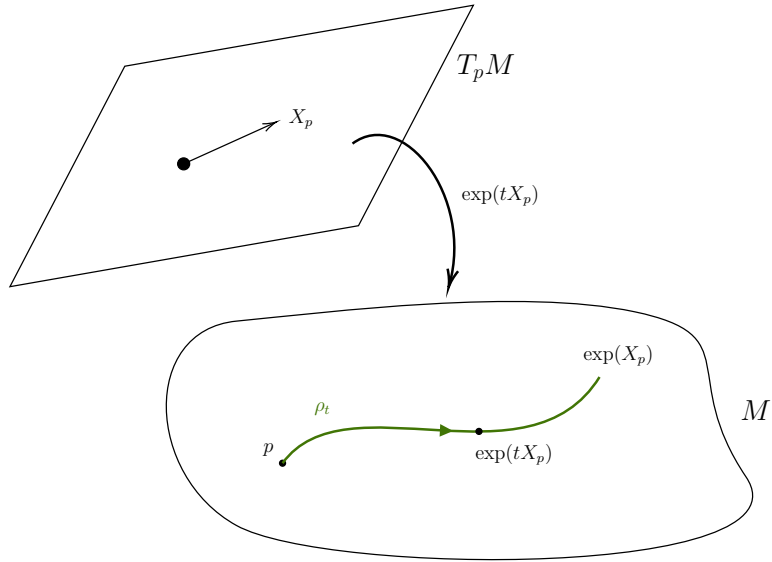
    \hypertarget{fig16}{Figure 16.} A smooth $\mathbb{R}$-action on $M$ through the one-parameter flow acting as an exponential map.
\end{center}

Furthermore, we have the bijection 
\begin{equation}
\begin{aligned}
    \{ \text{complete vector fields} \} &\longleftrightarrow \{ \text{smooth } \mathbb R \text{ actions on } M \} \\
    X \hspace{15mm} \hspace{1.7mm} &\longmapsto  \hspace{15mm} \exp (tX) \\
    X_p = \frac{d}{dt} \big|_{t=0} \psi_t(p) \hspace{5mm} \hspace{1.7mm} \hspace{-.8mm} &\longmapsfrom  \hspace{5mm} \psi_t
\end{aligned}
\end{equation}
for each $p$. 

\subsection{Symplectic Actions}
Let us now discuss symplectic actions. 

\begin{definition}
    
On a symplectic manifold $(M,\omega)$, the \textbf{action} $G \curvearrowright M$ is defined by a map:
\begin{equation}
\begin{aligned}
    \psi : G &\longrightarrow \text{Diff}(M)  \\
    g \ &\longmapsto \  \psi_g.
\end{aligned}
\end{equation}
$\psi$ is a \textbf{symplectic action} if $\psi : G \longrightarrow \text{Symp}(M) \subset \text{Diff}(M)$, namely, $\psi_g$ is a diffeomorphism such that $\psi_g^* \omega = \omega$, so it respects the symplectic form of $M$.

\end{definition}

\begin{definition}
A symplectic action $\psi$ of $\mathbb R$ (or $S^1$) on $(M,\omega)$ is \textbf{Hamiltonian} if the vector field generated by $\psi$ is Hamiltonian. 
\end{definition}

\begin{example}
Consider $(M,\omega) = (\mathbb R^2, dx\wedge dy)$ with the standard $x, y$ coordinates. Let 
\begin{equation}
    X = - \frac{\partial}{\partial y} \in \Gamma(T\mathbb R^2),
\end{equation}
and 
\begin{equation}
\begin{aligned}
    \psi^X : \mathbb R &\longrightarrow \text{Diff}(\mathbb R^2) \\
    t \ &\longmapsto \ \rho_t^X = \exp(tX)
\end{aligned}
\end{equation}
be the smooth action generated by $X$. Define $\rho_t^X(p) = \begin{pmatrix}
    \alpha_t(p) \\ \beta_t(p)
\end{pmatrix}$. Recall that $\rho_0^X(p) = p$ and $\frac{d}{dt} \rho_t^X(p) = X_{\rho_t^X(p)}$, so $\alpha_0(p) = p^1$, and $\beta_0(p) = p^2$ while 
\begin{equation}
    \dot\alpha_t(p) = 0 \quad ; \quad \dot\beta_t(p) = -1,
\end{equation}
since $X = 0 \partial_x - 1 \partial_y = (0,-1)$. Integrating this, we obtain 
\begin{equation}
    \rho_t^X(p) = \begin{pmatrix}
        p^1 \\ p^2-t
    \end{pmatrix},
\end{equation}
hence, the orbit $\mathcal O_p$ is all the vertical lines $\mathcal O_p = \{ \rho_t(p) \; | \; t\in \mathbb R \}$:
\begin{equation}
    p \sim q = \Bigg\{ \begin{pmatrix} p^1 \\ p^2-t \end{pmatrix} \Bigg| \ t \in \mathbb R \Bigg\}
\end{equation}
which is the line at $x= p^1$. If we look at $\cup_{p \in \mathbb R^2} \mathcal O_p$, we get all the lines parallel to the $y$-axis. 

Now, let us show that $X$ is a Hamiltonian vector field, that is to say, it obeys
\begin{equation}
    \iota_X \omega = df \quad \text{ for some }f.
\end{equation}
Looking at 
\begin{equation}
    \iota_X \omega(Y) = \omega \Big( - \frac{\partial}{\partial y} , Y \Big) = dx (Y) \quad \text{ for any }Y \in \Gamma(T\mathbb R^2),
\end{equation}
we see that $f(x,y) = x$ does the job, namely $\iota_X \omega = dx$. Hence, we conclude that the above action is Hamiltonian.
\end{example}

It is instructive to repeat the process for $(M, \omega) = (S^2, \omega = d\theta \wedge dh)$. Here $\theta$ and $h$ are the components of cylindrical coordinate system. By following similar steps, one can see that each orbit is a horizontal cycle and the action is Hamiltonian.

We have understood how to define a Hamiltonian action for $\mathbb R$ (or $S^1$) on a symplectic manifold $(M,\omega)$ using the 1-parameter group of diffeomorphisms. How about the Hamiltonian action of a more generic group $G$? Which notion is the right analogy to "Hamiltonian functions"? To build this, we need further tools.

\subsection{Adjoint and Coadjoint Actions}

Let $G$ be a Lie group, and $\mathfrak{g}$ its Lie algebra. We have seen before that the Lie algebra $\mathfrak g$ is  isomorphic to the space of left-invariant vector fields $\mathfrak g = L\Gamma(TG)$, which follows from the fact that the  Lie algebra isomorphism
\begin{equation}
    \big( L\Gamma(TG) , [\cdot,\cdot]_0 \big) \cong \big( T_e G, [\cdot,\cdot] \big).
\end{equation}
 Observe that any Lie group $G$ acts on itself by \textit{conjugation}:
\begin{equation}
\begin{aligned}
    \psi : G &\longrightarrow \text{Diff}(G) \\
    g &\longmapsto (\psi_g : G \to G),
\end{aligned}
\end{equation}
where $\psi_g : h \longmapsto ghg^{-1}$. Let us look at the pushforward of $\psi_g$, namely $(\psi_g)_{*,e}$, acting as
\begin{equation}
\begin{aligned}
    (\psi_g)_{*,e} : \underbrace{T_eG}_{\mathfrak g} & \xlongrightarrow{(d\psi_g)_e} \underbrace{T_{geg^{-1}} G}_{\mathfrak g},
\end{aligned}
\end{equation}
so we have induced a natural isomorphism from $\mathfrak g$ to $\mathfrak g$. We denote $(\psi_g)_{*,e} \in GL(\mathfrak g) $ by $$ \text{Ad}_g  : \mathfrak g \xlongrightarrow{\cong} \mathfrak g.$$

We then get the action of $G$ on the Lie algebra by this adjoint map
\begin{equation}
\begin{aligned}
    \text{Ad} : G & \longrightarrow GL(\mathfrak g)\\
    g &\longmapsto \text{Ad}_g,
\end{aligned}
\end{equation}
where $\text{Ad}_g : \mathfrak g \xrightarrow[]{\cong} \mathfrak g$. This is the \textbf{adjoint representation} of $G$ on the Lie algebra $\mathfrak g$. 

We calculate its pushforward $\text{Ad}_{*,e}$, which we call $\text{ad}$
\begin{equation}
    \text{ad} : T_e G \Longrightarrow T_{\text{Ad}_e} GL(\mathfrak g),
\end{equation}
or, equivalently
\begin{equation}
\begin{aligned}
    \text{ad}: \mathfrak g &\longrightarrow \mathfrak{gl(g)} \\
    X &\longmapsto \text{ad}_X,
\end{aligned}
\end{equation}
where $\text{ad}_X: \mathfrak g \to \mathfrak g$.

\begin{center}
    \includestandalone[width=0.82\textwidth]{Figures/Part1/p51}\\
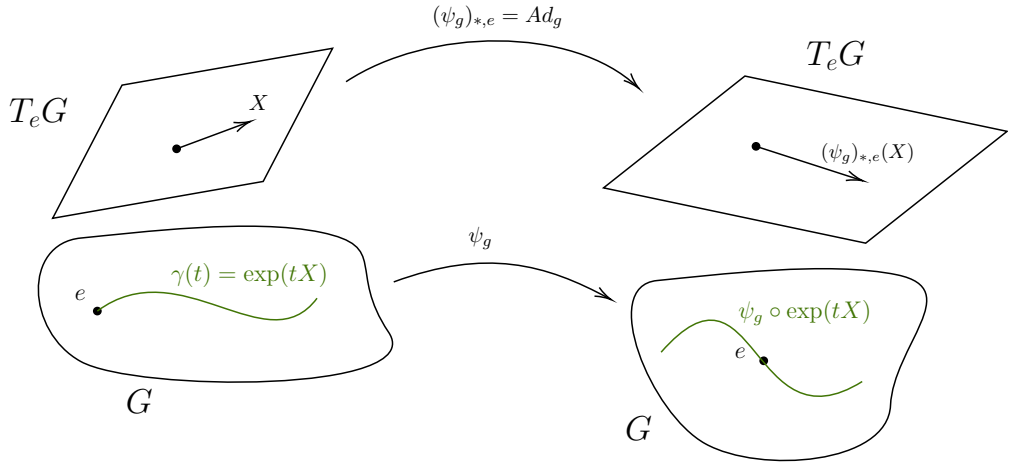
    \hypertarget{fig17}{Figure 17.} The action of a group endomorphism, $\psi_g$, and its adjoint action, $(\psi_g)_{*,e}.$
\end{center}

Now, let $Y \in \mathfrak g$ and look at 
\begin{equation}
\begin{aligned}
    \text{ad}_X(Y) &= \text{Ad}_{*,e}(X)(Y) \\
    &= \frac{d}{dt} \bigg|_{t=0} \text{Ad} \big( \exp(tX) \big) (Y)  \\
    &= \frac{d}{dt} \bigg|_{t=0} \text{Ad}_{\exp(tX)}(Y)\\
    &= \frac{d}{dt} \bigg|_{t=0} \exp(tX) Y \exp(-tX) \\
    &= [X, Y]_0 .
\end{aligned}
\end{equation}
To get the last sign, one can either take the derivative straightforwardly, or use the Baker-Campbell-Hausdorff expansion to leading order in $t$, which are essentially equivalent.

\begin{example}
    
For $G$ a matrix Lie group, let $g \in G$ 
\begin{equation}
\begin{aligned}
    (\psi_g)_{*,e}(X) &= \frac{d}{dt} \bigg|_{t=0} \psi_g \circ \gamma(t) \\
    &= \frac{d}{dt} \bigg|_{t=0} \psi_g\big( \exp(tX) \big) \\
    &= \frac{d}{dt} \bigg|_{t=0} g \ \exp(tX) \ g^{-1} \\
    &= g X g^{-1},
\end{aligned}
\end{equation}
so $\text{Ad}_g: G \longrightarrow GL(\mathfrak g)$ acts as $g \longmapsto \text{Ad}_g$, where $\text{Ad}_g(X) = g X g^{-1}$, and thus 
\begin{equation}
\begin{aligned}
    \text{ad}: \mathfrak g &\longrightarrow \text{End}(\mathfrak g) \\
    X &\longmapsto \text{ad}_X,
\end{aligned}
\end{equation}
where $\text{ad}_X : \mathfrak g \longrightarrow \mathfrak g$, \ $Y \longmapsto \text{ad}_X(Y) = [X,Y]_0$. This is the adjoint representation of the Lie algebra $\mathfrak g$ on itself.
\end{example}

\begin{definition}

Consider a natural pairing between $\mathfrak g^*$ and $\mathfrak g$ (where $\mathfrak g^* = \text{Hom}(\mathfrak g, \mathbb R))$:
\begin{equation}
\begin{aligned}
    \langle \ , \ \rangle : \mathfrak g^* \times \mathfrak g &\longrightarrow \mathbb R \\
    (\xi, X) &\longmapsto \langle \xi, X \rangle = \xi(X).
\end{aligned}
\end{equation}
For some $\xi \in \mathfrak g^*$, we define $\text{Ad}^*_g \xi$ by 
\begin{equation}
\begin{aligned}
    \text{Ad}^*_g : \mathfrak g^* &\longrightarrow \mathfrak g^* \\
    \xi &\longmapsto \text{Ad}_g^* \xi \in \text{Hom}(\mathfrak g, \mathbb R).
\end{aligned}
\end{equation}
For $X \in \mathfrak g$, one has 
\begin{equation}
    (\text{Ad}^*_g \xi)(X) := \xi \big( \text{Ad}_{g^{-1}}(X) \big).
\end{equation}
We define the \textbf{coadjoint representation} of $G$ on $\mathfrak g^*$ to be the map
\begin{equation}
\begin{aligned}
    \text{Ad}^* : G &\longrightarrow GL(\mathfrak g^*) \\
    g  &\longmapsto ( \text{Ad}^*_g : \mathfrak g^* \xlongrightarrow[]{\cong} \mathfrak g^*).
\end{aligned}
\end{equation}
    
\end{definition}
 
\section{Moment Map}
We are now ready to define the so-called \textbf{moment map}. Recall that for  a symplectic manifold $(M,\omega)$, a symplectic action $\psi : M \longmapsto \text{Symp}(M) \subset \text{Diff}(M)$, with $\psi^*_g \omega = \omega$ for all $g\in G$, has been given for the cases $G = \mathbb R, \text{ or }S^1$ only. Now we would like to do the same for a generic group.

\begin{definition}

The action $\psi$ above is called a \textbf{Hamiltonian action} if there exists a map 
\begin{equation}
    \mu : M \longrightarrow \mathfrak g^*,
\end{equation}
called the \textbf{moment map} (or more appropriately the momentum map), such that 
\begin{enumerate}
    \item[(1)] $d \mu^X = \iota_{X^\sharp} \omega$, for all $X\in \mathfrak g$ where for each $X$,
    \begin{equation}
    \begin{aligned}
        \mu^X : M &\longrightarrow \mathbb R \\
        p &\longmapsto \mu^X(p),
    \end{aligned}
    \end{equation}
    such that $\mu^X(p) := \langle \mu(p), X \rangle$. The $\mu^X$ are called the \textit{components} of $\mu$ along $X$. That is to say, each component of $\mu$ is an ordinary Hamiltonian function. \\
    \big($\mu^X \rightsquigarrow X_{\mu^X} =: X^\sharp$, s.t. $f \mapsto X_f$, $df = \iota_{X_f} \omega$\big)

    $X^\sharp$ is a vector field on $M$ generated by one parameter family $\{ \exp(tX) \; | \; t \in \mathbb R \}$, i.e., $\mathfrak g \longrightarrow \Gamma(TM)$, $X \longmapsto X^\sharp$. We have 
    \begin{equation}
        X^\sharp_p f = \frac{d}{dt} \bigg|_{t=0 } f \big( \underbrace{ \overbrace{\exp(tX)}^{\in G} \overbrace{\blacktriangleright}^{\text{left-action } G \curvearrowright M} p }_{\in M} \big).
    \end{equation}
    \item[(2)] $\mu$ is $G$-equivariant with respect to the $G$ action on $M$ and with respect to the coadjoint action $\text{Ad}^*$ of $G$ on $\mathfrak g^*$:
    \begin{equation}
        \begin{tikzcd}[sep=1.7cm]
            G\times M \arrow[thick, r, "\psi_g (\text{ or } \blacktriangleright)"] \arrow[thick, d, "\mu", font=large] & H\times N \arrow[thick, d, "\mu"]\\ 
            \mathfrak{g}^* \arrow[thick, r, "\mathrm{Ad}^*"] & \mathfrak{g}^* \arrow[to path={node[midway,scale=2] {$\circlearrowright$}}]{}
        \end{tikzcd}
    \end{equation}
    where $\text{Ad}^*_g \circ \mu = \mu \circ \psi_g$ ($\mu$ must be compatible with both actions). 
\end{enumerate}
Such tuple will be called $(M,\omega,G,\mu)$ a \textbf{Hamiltonian $G$-space}.

\end{definition}

\begin{remark}
    Hamiltonian action for connected Lie groups can be recast by means of a \textbf{co-moment map}:
    \begin{equation}
        \mu^* : \mathfrak{g} \longrightarrow C^\infty (M).
    \end{equation}
\end{remark}

\begin{observation}
When $G = \mathbb R$ or $S^1$, $\mathfrak g \cong \mathbb R \cong \mathfrak g^*$, so $\mu$ reduces to $\mu : M \longrightarrow \mathbb R$ satisfying
\begin{enumerate}
    \item 
    \begin{equation}
        \begin{aligned}
            \mu^X : M &\longrightarrow \mathbb R\\
            p &\longmapsto \mu^X(p) = \mu(p) \cdot 1 = \mu(p).
        \end{aligned}
    \end{equation}
    So, $\mu^X = \mu$, $\forall X\in \mathfrak g \in \mathbb R$.
    \item $X^\sharp$ turns out to be the standard vector field generated by the $\mathbb R(\text{or} \;S^1)$-action: $\iota_{X^\sharp} \omega = d \mu$.
    \item $\mu$-equivariance $\equiv$ invariance of $\mu$, in the sense of
    \begin{equation}
        \mathcal L_{X^\sharp} \mu = X^\sharp \mu = d\mu(X^\sharp) = \iota_{X^\sharp} \omega(X^\sharp) = 0. 
    \end{equation}
\end{enumerate}
\end{observation}

As a result, we recovered $\mu$ as a usual Hamiltonian function, and the previous $\mathbb{R}$-action on $M$.

\begin{example}
    
Consider $(\mathbb C, \omega)$ where $\omega = \frac{i}{2} dz \wedge d\overline z = dx \wedge dy = r dr \wedge d\theta$, or if we write $z = x + iy$, $\omega = dx \wedge dy = r dr \wedge d\theta$, where $(r,\theta)$ is the polar coordinates of the plane parameterized by $(x,y)$. Let $S^1 \xlongrightarrow[]{\psi} \text{Diff}(\mathbb C)$ be the standard $S^1$-action on $\mathbb C$ given as
\begin{equation}
\begin{aligned}
    \psi_t : \mathbb C &\xlongrightarrow{\cong} \mathbb C \quad \quad \quad (\forall t \in S^1) \\
    z &\longmapsto t \cdot z,
\end{aligned}
\end{equation}
where $\psi_t(z) := t\cdot z = t z$, and in the last term, the multiplication is the ordinary complex multiplication. Observe that $T_e S^1 \cong \mathfrak g \cong \mathbb R$ is the Lie algebra of the circle group. The dual algebra $\mathfrak g^*$ is clearly again $\mathbb R$ as a vector space $\mathfrak g^* \cong \mathbb R$. \\
\end{example}

\begin{claim}

$\psi$ is a Hamiltonian action with a moment map
\begin{equation}
\begin{aligned}
    \mu: \mathbb C &\longrightarrow \mathfrak g^* \cong \mathbb R \\
    z &\longmapsto - \frac{1}{2} |z|^2 + c, \quad c \in \mathbb R.
\end{aligned}
\end{equation}
Using the $(r,\theta)$ coordinate chart, we have $\mu(z = re^{i\theta}) = -\frac{1}{2} r^2 + c$. As discussed before, $\mu^X = \mu$ for all $X \in \mathfrak g \cong \mathbb R$. Hence, $\mu$ reduces to the usual Hamiltonian function. \\
\end{claim}

\begin{claim}
$\mu$ is the Hamiltonian function for $X^\sharp := \frac{\partial}{\partial \theta}$, i.e. $\iota_{X^\sharp} \omega = d\mu$. Note that in the setup above, $X^\sharp$ reduces to the standard vector field generated by $\mathbb R$- or $S^1$-action. \\
\end{claim}

\begin{proof}

Note that $d\mu = -r dr$. Then, observe that
\begin{equation}
\begin{aligned}
    \iota_{X^\sharp} \omega &= \omega( X^{\sharp}, \cdot) = \omega \left( \frac{\partial}{\partial \theta} , \cdot  \right) \\
    &= rdr \otimes d\theta \left(\frac{\partial }{\partial \theta} , \cdot\right) - rdr \otimes d\theta \left(\cdot , \frac{\partial}{\partial \theta}\right) \\
    &= -rdr(\cdot) = d\mu.
\end{aligned}
\end{equation}
(Since $\mathcal L_{X^\sharp} \mu = 0$, $\mu$ is equivariant.) Thus, we have shown the desired equality. \\
\end{proof}

\begin{observation}

\begin{enumerate}
    \item The same result holds for $(\mathbb C^n, \omega)$ with $\omega = \frac{i}{2} \sum_i^n dz_i \wedge d\overline z_i$ or writing $z_i = x_i+iy_i$, $\omega = \sum_i^n dx_i \wedge dy_i = \sum_i^n r_i dr_i \wedge d\theta_i$. The moment map in this case is
    \begin{equation}
    \begin{aligned}
        \mu : \mathbb C^n &\longrightarrow \mathbb R  \\
        (z_1,\cdots z_n) &\longmapsto -\frac{1}{2} \sum_i^n |z_i|^2 + c.
    \end{aligned}
    \end{equation}
    Take $\psi$ to be the $S^1$-action on $\mathbb C^n$, it is again a Hamiltonian action. In that case, we have
    \begin{equation}
    \begin{aligned}
        d\mu &= -\dfrac{1}{2 } d \left( \sum_i^n r_i^2 \right) = - \sum_i^n r_i dr_i, \\
        X^\sharp &= \frac{\partial }{\partial \theta_1} + \cdots +\frac{\partial}{\partial \theta_n}, \\
        i_{X^\sharp} \omega &= - \sum_i^n r_i dr_i = d\mu.
    \end{aligned}
    \end{equation}
    \item When $c= \frac{1}{2}$, we have
    \begin{equation}
        \mu^{-1}(0) = \qty{ z \in \mathbb C^n \ \bigg| \ |z|=1 } \cong S^{2n-1}.
    \end{equation}
    Quotiening out by the $S^1$-action, we obtain
    \begin{equation} \label{mu inverse / S1 equals CPn-1}
        \mu^{-1}(0) \ / \  S^1 = S^{2n-1} \ / \ S^1 = \mathbb {CP}^{n-1}.
    \end{equation}
    \textbf{Special case:} For the case of $n=2$, we have the \textbf{Hopf fibration}:
    \begin{equation}
    \begin{aligned}
        S^1 \hookrightarrow\ &S^3 = \mu^{-1}(0) \\
        & \downarrow \\
        & S^2 = \underbrace{S^3\ / \ S^1}_{\mu^{-1}(0) \ / \ S^1} = \mathbb{CP}^1.
    \end{aligned}
    \end{equation}
    Considering $S^3 \cong SU(2)$ and $S^1 \cong U(1)$ as groups, we have the coset $S^2 \cong SU(2) /  U(1)$. In fact, the spheres in general can be defined as $S^{n-1} \cong SO(n)  / SO(n-1)$. 
    
    Equation (\ref{mu inverse / S1 equals CPn-1}) is related to the symplectic reduction theorem, which will be discussed later. We will see that, via the moment map, one can form a principal $G$-bundle. 
\end{enumerate}
\end{observation}
Recall also that $SU(2)$ is  the double cover of $SO(3)$, meaning that $SU(2) / \mathbb Z_2 \cong  SO(3)$. Moreover, $\mathbb R P^3 (\cong S^3/\mathbb Z_2)$ is diffeomorphic to $SO(3)$,   hence admits a group structure.\\
\newpage 
\begin{example}
    
Consider $\mathbb R^6$ with coordinates $(x_1,x_2,x_3,y_1,y_2,y_3)$ and with the usual symplectic form $\omega = \sum_{i=1}^3 dx_i \wedge dy_i$. Consider the $\mathbb R^3$-action on $\mathbb R^6$ given by translations
\begin{equation}
\begin{aligned}
    \mathbb R^3 &\xlongrightarrow[]{\psi} \text{Symp}(\mathbb R^6)\\
    \Vec{a} &\longmapsto \psi_{\Vec{a}},
\end{aligned}
\end{equation}
where 
\begin{equation}
    \psi_{\Vec{a}}(\Vec{x},\Vec{y}) := ( \Vec{x} + \Vec{a} , \Vec{y}). 
\end{equation}
Then we have
\begin{equation}
\begin{aligned}
    &(1)\quad \quad  X^\sharp = a_1 \frac{\partial}{\partial x_1} + a_2 \frac{\partial}{\partial x_2} + a_3 \frac{\partial}{\partial x_3} \quad \text{ for } X = \vec{a}, \\
    &(2)\quad \quad \mu: \mathbb R^6 \longrightarrow \mathbb R^3 \cong \mathfrak g^* \text{  is defined by  } \mu( \vec{x}, \vec{y})  := \vec{y} \text{  with components  } \\
    & \hspace{4cm} \mu^{\vec{a}} : \mathbb R^6 \longrightarrow \mathbb R \\
    &\hspace{4.26cm} (\vec{x}, \vec{y}) \longmapsto \mu(\vec{x}, \vec{y}) \cdot \vec a = \vec y \cdot \vec a,
\end{aligned}
\end{equation}
where $\cdot$ denotes the usual inner product: $\vec y \cdot \vec a = y_1a_1 + y_2a_2 + y_3a_3$. In this setup, $\vec y$ is the \textit{momentum vector} corresponding to \textit{position vector} $\vec x$, and $\mu$ is the \textit{linear momentum} map. \\
\end{example}

\begin{example}
    
Consider $SO(3)$-action on $\mathbb R^3$ by rotation 
\begin{equation}
    SO(3) \xlongrightarrow[]{\psi} \text{Diff}(\mathbb R^3).
\end{equation}
This lifts to a symplectic action $\psi$ on the cotangent bundle $T^* \mathbb R^3 \equiv \mathbb R^6$. The infinitesimal version of this action is
\begin{equation}
\begin{aligned}
    \mathbb R^3 &\longrightarrow \Gamma_{\text{symp}}(T^* \mathbb R^3) \\
    \vec a &\longmapsto d\psi (\vec a) ,
\end{aligned}
\end{equation}
where for any pair $(\vec x, \vec y) \in \mathbb R^6$, we have 
\begin{equation}
    d\psi(\vec a)(\vec x, \vec y)  = (\vec a \times \vec x, \vec a \times \vec y),
\end{equation}
with $\times$ the cross product $\vec a \times \vec x = \sum_{ijk} \varepsilon_{ijk} a_i x_j \hat{e}_k$ with $a_i$ the $i$th component of $\vec a$, $\varepsilon_{ijk}$ is the Levi-Civita symbol, and $\hat{e}_k$ is a basis of $\mathbb R^3$. We thus define the moment map
\begin{equation}
\begin{aligned}
    \mu : \mathbb R^6 \hspace{1mm} &\longrightarrow\hspace{3mm} \mathbb R^3 \cong \mathfrak g^* \\
    (\vec x, \vec y) &\longmapsto \mu(\vec x, \vec y) = \vec x \times \vec y,
\end{aligned}
\end{equation}
with components 
\begin{equation}
    \mu^{\vec a}(\vec x,\vec y) = \mu(\vec x,\vec y) \cdot \vec a = (\vec x \times \vec y) \cdot \vec a.
\end{equation}
In this context, $\mu$ is called the \textit{angular momentum}.
\end{example}

\begin{remark}

The name "moment map" comes from its role as the generalization of \emph{linear} and \emph{angular momenta} in classical mechanics.

\end{remark}

\section{Symplectic Reduction Theorem and Noether Principle}

\paragraph{Recalling our setup.} We defined the notion of moment map and Hamiltonian $G$-space as follows:

Let $(M,\omega)$ be a symplectic manifold, and $G$ a Lie group. Let $\psi: G \longrightarrow \text{Symp}(M)$ be a symplectic action of $G$ on $M$, that is to say, for each $g \in G$, $\psi_g: M \xlongrightarrow[]{\cong} M$ is a diffeomorphism such that $\psi^* \omega = \omega$ namely, we want for each $X_p,Y_p \in T_pM$ the following to hold
\begin{equation}
    \psi^*_\omega( X_p,Y_p) := \omega( \psi_{g*,p}(X_p), \psi_{g*,p}(Y_p) ) = \omega (X_p,Y_p).
\end{equation}
Diagrammatically, we have
\begin{equation}
\begin{tikzcd}[sep=1.7cm]
    T_pM\times T_pM \arrow[r,thick,"{\psi_{g,*},\psi_{g,*}}"]\arrow[d,thick, "\psi_\omega^*"] & T_qM\times T_qM \arrow[ld, thick, "\omega"]\\ 
    \mathbb{R} & {}
\end{tikzcd}
\end{equation}
where $q = \psi_g(p)$. Then, $\psi$ is called a \textbf{Hamiltonian action} if there exists a moment map
\begin{equation}
\begin{aligned}
    \mu : M &\longrightarrow \mathfrak g^* \\
    p &\longmapsto \mu(p),
\end{aligned}
\end{equation}
satisfying:
\begin{enumerate}
    \item $\boxed{\text{For each } X \in \mathfrak g$, $i_{X^\sharp} \omega = d\mu^X}$ with $X^\sharp$, a vector field on $M$, generated by $\{ \exp(tX): t\in \mathbb R \}$ and a $G$-action on $M$. That is to say, each "component" $\mu^X$ gives a Hamiltonian function for $X^\sharp$. Recall that $\mu^X$, the component of $\mu$ along $X$, is defined as
    \begin{equation}
    \begin{aligned}
        \mu^X : M &\longrightarrow \mathbb R \\
        p &\longmapsto \langle \mu^*(p) , X \rangle ,
    \end{aligned}
    \end{equation}
    where we use the natural pairing between $\mathfrak g^*$ and $\mathfrak g$.

    To see that each component gives a Hamiltonian function, observe that there indeed exists a natural map
    \begin{equation}
        \mathfrak g \longrightarrow \Gamma(TM), \;
        X \mapsto X^\sharp,
    \end{equation}
    where for each $p \in M$ and $f \in C^\infty(M)$, we define
    \begin{equation}
        X^\sharp_p f = \frac{d}{dt} \Big|_{t=0} f \circ \gamma := \frac{d}{dt} \Big|_{t=0} f \underbrace{\Big( \exp(tX) \bullet p  \Big)}_{:= \gamma(t)} .
    \end{equation}

    \begin{center}
     \includestandalone[width=0.6\textwidth]{Figures/Part1/p61}\\
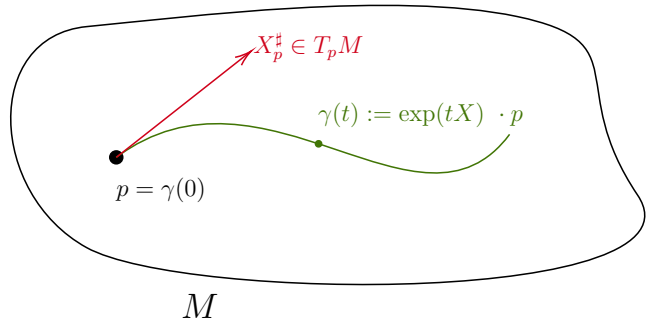
     \hypertarget{fig18}{Figure 18.} The point, $p$, is taken by the path through multiplication with exponentiation of tangent vector. 
    \end{center}

    The first equality is the usual definition of a vector field acting on a $C^{\infty}(M)$ function, and the second is where the setup comes into play. In the right-hand side, $\bullet$ is the left action of $\exp(tX)\in G$ on the point $p$. 
    \item $\boxed{\mu-\text{equivariance}}$ with respect to the $G$-action and co-adjoint action $\text{Ad}^*$. Namely, for each $g \in G$, the following diagram commutes 
    \begin{equation}
        \begin{tikzcd}[sep=1.5cm]
            M \arrow[thick,r,"\psi_g"] \arrow[thick,d,swap,"\mu"] & M \arrow[thick,d,"\mu"]
            \\
            \mathfrak{g} \arrow[thick,r, "\mathrm{Ad}_g^*"] &  \mathfrak{g}\arrow[to path={node[midway,scale=1.5] {$\circlearrowright$}}]{}
        \end{tikzcd}
    \end{equation}
    This means that $\mu \circ \psi_g = \text{Ad}_g^* \circ \mu$, where 
    \begin{equation}
    \begin{aligned}
        \text{Ad}^* : G &\xlongrightarrow[\phantom{as}]{} \text{GL}(\mathfrak g^*) \\
        g \hspace{0.5mm} &\longmapsto \phantom{a} \text{Ad}^*_g,
    \end{aligned}
    \end{equation}
    and $\text{Ad}^*_g : \mathfrak g^* \xlongrightarrow[]{\cong} \mathfrak g^*$ is an isomorphism of the dual Lie algebra $\mathfrak g^*$.
\end{enumerate}
The space $(M,\omega)$ equipped with a Hamiltonian action (symplectic action + the existence of moment map above) $G$ on $M$ and corresponding moment map $\mu$ is called a Hamiltonian $G$-space, denoted by $(M,\omega, G,\mu)$. \\

\begin{theorem}

\textbf{(Marsden-Weinstein-Meyer) Symplectic Reduction Theorem:} \\
Let $(M, \omega, G, \mu)$ be a Hamiltonian $G$-space for a compact Lie group $G$. Let $i: \mu^{-1}(0) \hookrightarrow M$ be the natural inclusion map. Assume that the $G$-action is free on $\mu^{-1}(0) \subset M$. Then
\begin{enumerate}
    \item The orbit space $\mu^{-1}(0)/G$ is a manifold,
    \item $\pi: \mu^{-1}(0) \longrightarrow \mu^{-1}(0)/G$ is a principal $G$-bundle,
    \item $\omega$ on $M$ induces a symplectic form $\tilde\omega$ on $\mu^{-1}(0)/G$ such that $i^* \omega = \pi^* \tilde\omega$.
\end{enumerate}
Diagrammatically,
\begin{equation}
    \begin{tikzcd}[sep=1.5cm]
       \boxed{ i^*\omega = \pi^*\tilde{\omega}} \arrow[dr, thick, dashrightarrow, no head] & {} & {} & \boxed{\omega }\arrow[lll, thick, swap, dashrightarrow,"\text{pull-back via }i "] \arrow[dl, thick, dashrightarrow, no head]\\ 
        {} & \mu^{-1}(0) \arrow[r, thick, hook, "i"] \arrow[d, thick, "\pi"] & M & {} \\
        \boxed{\tilde{\omega}} \arrow[r,thick, dashrightarrow, no head] \arrow[uu, thick, "\text{pull-back via } \pi", dashrightarrow, swap] & \mu^{-1}(0)/G & {} & {}
    \end{tikzcd}
\end{equation}
The data $(\mu^{-1}(0)/G, \tilde\omega)$ is called the \textbf{reduced space}, or the \textbf{symplectic quotient} with respect to the $G$-action and the moment map $\mu$.
    
\end{theorem}

There will be an extensive discussion on principal $G$-bundles later, but roughly speaking a $G$-bundle is
\begin{itemize}
    \item Fiber bundle $\pi :P \longrightarrow M$ such that each fiber is a ``copy" of $G$, i.e., for each $p\in M$, $\pi^{-1}(\pi(p)) = G$.
    \item Furthermore, the orbit $\mathcal O_p$ of point $p$ under $G$ is isomorphic to $G$, $\mathcal O_p \cong G$, since the action is free. Namely, if $p,q \in \pi^{-1}(p)$, then there exists a unique $g \in G$ such that $q = g \bullet p$.

   \begin{center}
    \includestandalone[width=0.7\textwidth]{Figures/Part1/p62}\\
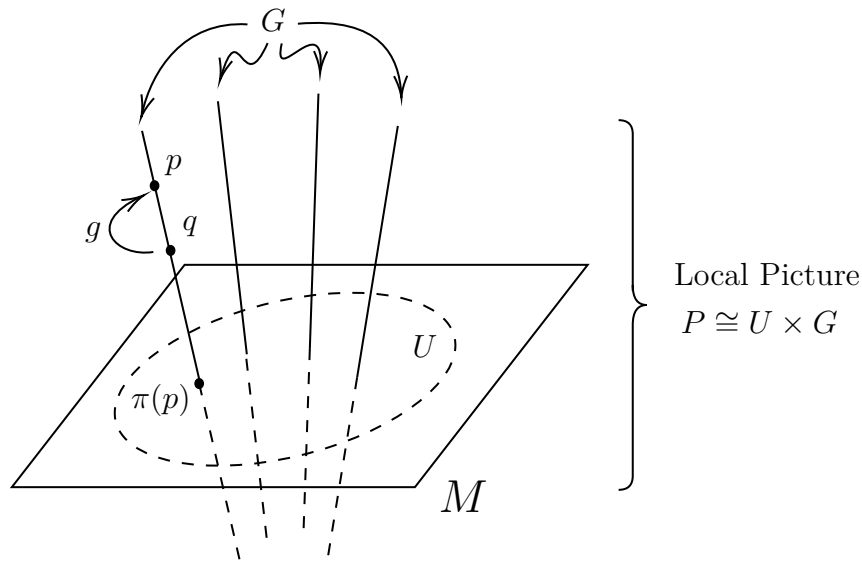
        \hypertarget{fig19}{Figure 19.} Local picture of a principal $G$-bundle.
    \end{center}

\end{itemize}

\begin{example}
    
We discussed $(\mathbb C^n, \omega)$ with $S^1$-action before, where $\omega = \frac{i}{2} \sum_{j=1}^n dz_j \wedge d \overline z_j$ and $z_i = x_i + i y_i$, $\overline z_i = x_i - iy_i$ are the coordinates of $\mathbb C^n$. The free action of $S^1$ is given by ordinary complex multiplication
\begin{equation}
    g \bullet z = g\bullet (z_1,\cdots z_n) := (g z_1, \cdots g z_n), \quad \forall z \in \mathbb C^n.
\end{equation}
Define a moment map $\mu: \mathbb C^n \longrightarrow \mathfrak g^* \cong \mathbb R$ by 
\begin{equation}
    \mu(z) := - \frac{|z|^2}{2} + \frac{1}{2}.
\end{equation}
We observe that $(\mathbb C^n, \omega, U(1) \cong S^1, \mu)$ forms a Hamiltonian $G$-space with $G = U(1)$. The inverse of the moment map reads
\begin{equation}
    \mu^{-1}(0) = \big\{ z \in \mathbb C^n \big| \ |z|=1 \big\} \cong S^{2n-1} \hookrightarrow \mathbb C^n,
\end{equation}
and 
\begin{equation}
    \mu^{-1}(0)/G = S^{2n-1}/S^1 \cong \mathbb{CP}^{n-1}.
\end{equation}
By the symplectic reduction theorem, we have a principal $G$-bundle:
\begin{equation}
\begin{aligned}
    S^1 \ \curvearrowright \ \mu^{-1}(0) &= S^{2n-1} \hookrightarrow \mathbb C^n \\
    \pi &\Bigg\downarrow \\
    \mu^{-1}(0)&/S^1 \cong \mathbb{CP}^{n-1}.
\end{aligned}
\end{equation}
For the case $n=2$, we get the Hopf fibration
\begin{equation}
\begin{aligned}
    S^1 \ \curvearrowright \ &  S^3 \\
    \pi &\Bigg\downarrow \\
    S^3&/S^1 \cong S^2.
\end{aligned}
\end{equation}
\end{example}

\begin{theorem}
\textbf{(Noether):} 

Let $(M,\omega, G,\mu)$ be a Hamiltonian $G$-space. A function $f: M \longrightarrow \mathbb R$ is $G$-invariant (in the sense of $\mathcal L_{X^\sharp} f = 0$, with $X^\sharp$ a vector field generated by the $G$-action on $M$) or a conserved quantity if and only if $\mu$ is constant on the trajectories of the Hamiltonian vector field $X_f$ of $f$ (in the sense of $\mathcal L_{X_f} \mu^X = 0$ $\forall X \in \mathfrak g$).

\end{theorem} 

\begin{proof}
    
Let $X_f$ be the Hamiltonian vector field of $f$ (i.e., $i_{X_f} \omega = df$). Let $X \in \mathfrak g$, and 
\begin{equation}
\begin{aligned}
    \mu^X: M &\longrightarrow \mathbb R  \\
    p &\longmapsto \langle \mu^X(p), X \rangle,
\end{aligned}
\end{equation}
be the component of $\mu$ along $X$ so that $\mu^X(p) \in \mathfrak g^*$. 

Then we have 
\begin{equation}
\begin{aligned}
    \mathcal L_{X_f} \mu^X = X_f \mu^X &= d\mu^X(X_f) = i_{X_f} d\mu^X \\
    &= i_{X_f} i_{X^\sharp} \omega \quad \Big( = \omega(X^\sharp, X_f) = - \omega(X_f, X^\sharp) \Big) \\
    &= -i_{X^\sharp} i_{X_f} \omega \\
    &= -i_{X^\sharp} df  \\
    &= df(X^\sharp) \quad \Big( = -X^\sharp f \Big) \\
    &= -\mathcal L_{X^\sharp} f.
\end{aligned}
\end{equation}
In the first line, we use the pairing of $d\mu^X$, a one-form, with $X_f$, a vector field, to produce a scalar; in the second, we use the definition of the interior product; in the third, we use $d\mu^X = i_{X^\sharp} \omega$ for all $X \in \mathfrak g$; in the fourth, we use the anti-symmetry of the $\omega$, the symplectic form;  in the fifth, we use the definition of a Hamiltonian vector field $i_{X_f}\omega = df$; in the sixth, we again use the definition of interior product; and in the seventh line, we use the construction of $df$ via $i_{X^\sharp}$. Therefore, we have
\begin{equation}
    \mathcal L_{X_f} \mu^X = 0 \Longleftrightarrow \mathcal L_{X^\sharp} f = 0,
\end{equation}
which proves the theorem. 
\end{proof} 

\begin{definition}

A $G$-invariant function $f: M \longrightarrow \mathbb R$ is called an \textbf{integral of motion} of $(M,\omega, G, \mu)$.

If $\mu$ is constant along the trajectories of a Hamiltonian vector field $X_f$ for $f$, then the corresponding one-parameter group of diffeomorphisms $\{ \exp(tX_f) \ | \ t\in \mathbb R\}$ is called a \textbf{symmetry} of $(M,\omega, G, \mu)$.
\end{definition}

\begin{observation}
Noether's theorem says that there exists a one-to-one correspondence between symmetries of the system and integrals of motion. 
\end{observation}

\begin{example}
\begin{enumerate}
    \item Translation invariance $\Longleftrightarrow$ conservation of linear momentum. The moment map is given as $\mu: \mathbb R^6 \longrightarrow \mathbb R^3$, $(\vec{x},\vec{y}) \longmapsto \vec{y}$, with components
\begin{equation}
\begin{aligned}
    \mu^{\vec{a}} : \mathbb R^6 &\xlongrightarrow[\phantom{asd}]{} \mathbb R \\
    (\vec x, \vec y) &\longmapsto  \mu(\vec x, \vec y) \cdot \vec a= \vec y \cdot \vec a.   
\end{aligned}
\end{equation}
    \item Rotational invariance $\Longleftrightarrow$ conservation of angular momentum. The moment map is given by $\mu: \mathbb R^6 \longrightarrow \mathbb R^3$, $(\vec{x},\vec{y}) \longmapsto \vec{x} \times \vec{y}$, with components
    \begin{equation}
    \begin{aligned}
        \mu^{\vec{a}} : \mathbb R^6 &\xlongrightarrow[\phantom{asd}]{} \mathbb R  \\
        (\vec x, \vec y) &\longmapsto \mu(\vec x, \vec y) \cdot \vec{a} =  (\vec x \times \vec y) \cdot \vec a.
    \end{aligned}
    \end{equation}
    \item Time translation invariance $\Longleftrightarrow$ conservation of energy.
\end{enumerate}
\end{example}

\section{Theory of Principal $G$-bundles}

\subsection{Some Terminology and Background}
\begin{definition}
    Let $G$ be a Lie group, a bundle $(P,\pi,M)$ is called a \textbf{principal $G$-bundle} if:
    \begin{enumerate}
        \item $P$ is a right $G$-space (i.e., there exists a right smooth $G$-action on $P$):
        \begin{equation}
            \begin{aligned}
                P \times G &\longrightarrow P
                \\
                (p, g) &\longmapsto p \cdot g.
            \end{aligned}
        \end{equation}
        \item $G$-action on $P$ is free (i.e., $\text{stab}(p) = \{e\}$ $\forall p \in P$).
        \item One can write a diagram of the \textbf{bundle map}
        \begin{equation}
            \begin{tikzcd}[sep= 1.5cm]
                P \arrow[d, thick, "\pi"] & P \arrow[d, thick, "\pi'"] \\ 
                M & P/G \arrow[to path={node[midway,scale=1.5] {$\cong$}}]{}
            \end{tikzcd}
        \end{equation}
        
        where $P/G$ is the orbit space and 
        \begin{equation}
            \begin{aligned}
                \pi' : P &\longrightarrow P/G \; \;\text{quotient map}
                \\
                x &\longrightarrow [x].
            \end{aligned}
        \end{equation}
    \end{enumerate}
\end{definition}

\begin{definition}
A \textbf{bundle map} between two bundles $(E \xlongrightarrow[]{\pi} M) \longrightarrow (E'\xlongrightarrow[]{\pi'} M')$ is a pair of smooth maps:
\begin{equation}
    f: E \longrightarrow E' \;\; \text{and} \;\; g: M \longrightarrow M'
\end{equation}
such that the diagram
\begin{equation}
    \begin{tikzcd}[sep = 1.5cm]
        E \arrow[d,thick, "\pi"] \arrow[r, thick, "f"] & E' \arrow[d, thick, "\pi'"] \\
        M \arrow[r,thick,"g", swap] & M'
        \arrow[to path={node[midway,scale=1.5] {$\circlearrowright$}}]{}
    \end{tikzcd}
\end{equation}
commutes: $\pi'\circ f = g\circ \pi$.

The third item in the previous definition gives
\begin{equation}
    \begin{tikzcd}[sep= 1.5cm]
        P \arrow[d, thick, "\pi"] \arrow[r, thick, "\text{id}_p"] & P \arrow[d, thick, "\pi'"] \\ 
        M \arrow[r, thick, "\cong", swap]& P/G \arrow[to path={node[midway,scale=1.5] {$\circlearrowright$}}]{}
    \end{tikzcd}
\end{equation}
which can also be written as 
    \begin{equation}
        \begin{tikzcd}[sep=1.5cm, remember picture]
            {} & |[alias=E]| P \arrow[ld,thick,"\pi", swap] \arrow[rd, thick, "\pi'"]& {} \\
            |[alias = Y]| M \arrow[rr,thick,"\cong", swap] & & |[alias = X]| P/G & 
        \end{tikzcd}
        \tikz[overlay,remember picture]{%
        \node (Y1) [scale=1.5] at (barycentric cs:X=1,E=1,Y=1) {$\circlearrowright$}}
    \end{equation}
\end{definition}

\begin{observation}
    Here, we make some immediate observations. Given a principle $G$-bundle, $(P, \pi,M)$, $(P \xlongrightarrow[]{\pi} M)$:
\begin{enumerate}
    \item Since $G$-action is free we have: \textbf{$\theta_p \cong G$} $\forall p \in P$. Indeed, we define:
    \begin{equation}
        \begin{aligned}
            \theta_p &\xlongrightarrow{\varphi} G
            \\
            q &\longrightarrow g_q
            \\
            p &\longrightarrow e.
        \end{aligned}
    \end{equation}

As a side note, if $q\in\theta_p$ then there exists a unique $g_q \in G$ such that $p \cdot g_q = q$. Furthermore, since the action is free $p\cdot g = p \iff g=e$. Otherwise the map would be ill-defined.\\
Conversely, we define:
\begin{equation}
    \begin{aligned}
        \psi: G &\longrightarrow \theta_p 
        \\
        g &\longrightarrow p \cdot g.
    \end{aligned}
\end{equation}
Notice that:
\begin{equation}
    \begin{aligned}
        &\psi \circ \varphi \;  (q) = \psi(g_q) = p \cdot g_q = g 
        \\
        & \varphi \circ \psi \; (g) = \varphi(p \cdot g =q) = g
    \end{aligned}
\end{equation}
implies $\theta_p \cong_\varphi G $.

\item Since
\begin{equation}
    \begin{tikzcd}[sep= 1.5cm]
        P \arrow[d, thick, "\pi"] & P \arrow[d, thick, "\pi'"] \\ 
        M & P/G \arrow[to path={node[midway,scale=1.5] {$\cong$}}]{}
    \end{tikzcd},
\end{equation}

each fiber $(\pi')^{-1} ([x]) \cong \theta_{p\in(\pi')^{-1} ([x])} \cong G \; \;\forall p \;.$ 

Each fiber is a copy of $G$.
\hyperlink{fig19}{Figure 19} shows the local picture of a principal $G$-bundle.

Here, both $p$ and $q$ lie on the same fiber $\pi(p)$ whenever there exists a unique $g\in G$ such that $q=p\circ g$. This construction will be much clearer later.
\end{enumerate}
\end{observation}

\begin{definition}
    Let 
    \begin{equation*}
        \begin{tikzcd}[sep=1.2cm]
            P \arrow[r, thick, "\cdot G"] & P \arrow[d,thick, "\pi"] \\ {} & M 
        \end{tikzcd} \qquad\text{and}\qquad
        \begin{tikzcd}[sep=1.2cm]
            P' \arrow[r, thick, "\cdot'G'"] & P'\arrow[d, thick, "\pi'"] \\ {} & M'
        \end{tikzcd}
    \end{equation*}
     be two principal bundles. Also let $\rho : G \longrightarrow G'$ be a Lie Group homomorphism. Then \textbf{a principle bundle map} is a pair of smooth maps $(u,f)$,

    \begin{equation}
        u: P\longrightarrow P' \;\; , \;\; f: M\longrightarrow M'
    \end{equation}

    such that the diagram

    \begin{equation}\label{eq:longdiagram}
        \begin{tikzcd}[sep=1.5cm, remember picture]
            |[alias=A]| p \arrow[r, thick, "u"] &|[alias=B]| p' \\
           |[alias=C]| p \arrow[u,thick, "\cdot G"] \arrow[r,thick, "u"] \arrow[d, thick, swap, "\pi"] & |[alias=D]|p' \arrow[u, thick, swap, "\cdot' G'"] \arrow[d, thick, "\pi'"] \\ 
           |[alias=E]| M \arrow[r, thick, "f"] &|[alias=F]| M'
        \end{tikzcd}
        \tikz[overlay,remember picture]{%
        \node (Y1) [scale=1.5] at (barycentric cs:A=1,B=1,C=1,D=1) {$\circlearrowright$};
        \node (Y2) [scale=1.5] at (barycentric cs:C=1,D=1,E=1,F=1) {$\circlearrowright$}}
    \end{equation}
    commutes. 
\end{definition}

The upper square in Eq.~\ref{eq:longdiagram} comes from the "principal" nature of a bundle, i.e., compatibility condition with group action:
    \begin{equation}
        u(p\cdot g) = u(p) \cdot' \rho(g) \ \  \forall g \in G,
    \end{equation}
where $\cdot $ is the $G$-action, and $\cdot'$ is the $G'$-action. The bottom diagram in the definition comes from the "bundle" structure, and being in fact a bundle map.
    \begin{equation}
         \pi' \circ u = f \circ \pi.
    \end{equation}
Therefore, one can say that a principal bundle map is a bundle map compatible with "group actions".
\begin{definition}
    A principal $G$-bundle 
    \begin{equation*}
        \begin{tikzcd}[sep = 1.2cm]
            P \arrow[r, thick, "\cdot G"] & P \arrow[d, thick, "\pi"] \\ {} & M
        \end{tikzcd}
    \end{equation*}
    is called \textbf{trivial}, if there exists a principal bundle map
    \begin{equation}
        \varphi : P \longrightarrow M \times G.
    \end{equation} 
    Here, we take $M \xlongrightarrow{\text{id}_M} M$. In the above notation: $u \longrightarrow \varphi$ and $f \longrightarrow id_M$.
    \begin{equation}
        \begin{tikzcd}[sep=1.5cm, remember picture]
            |[alias=A]| p \arrow[rr, thick, "\varphi"]& {} & |[alias=B]| M\times G \\ 
            |[alias=C]| p \arrow[rd, thick, swap, "\pi"] \arrow[u, thick, "\cdot G"] \arrow{rr}{\cong}[swap, thick]{\varphi} & {} & |[alias=D]| M\times G \arrow[u,thick, swap, "\circ G"] \arrow[ld, thick, "\pi'"] \\ 
            {} & |[alias=E]| M & {}
        \end{tikzcd}
        \tikz[overlay,remember picture]{%
        \node (Y1) [scale=1.5] at (barycentric cs:A=1,B=1,C=1,D=1) {$\circlearrowright$};
        \node (Y2) [scale=1.5] at (barycentric cs:C=1,D=1,E=1) {$\circlearrowright$}}
    \end{equation}
\end{definition}
Note that:
\begin{enumerate}
    \item $G$-action on $P$ induces right $G$-action on $M \times G$:
    \begin{equation}
        \begin{aligned}
            (M\times G)\times G &\longrightarrow M\times G
            \\
            ((x,g),\Tilde{g}) &\longrightarrow (x,g\Tilde{g}).
        \end{aligned}
    \end{equation}
    i.e. \; $(x,g) \circ \Tilde{g} := (x,g\Tilde{g}) \;\; \forall \Tilde{g} \in G$.
    \item $P \xlongrightarrow[\varphi]{\Tilde{\circ}}$ is a diffeomorphism. 
    \item $\pi' (x,g) = x$ \; $\forall (x,g) \in M \times G$.
\end{enumerate}

\begin{theorem}
    A principal $G$-bundle 
    \begin{equation}
        \begin{tikzcd}[sep=1.2cm]
            P \arrow[r, thick, "\cdot G"] & P \arrow[d, thick, "\pi"] \\ {} & M    
        \end{tikzcd}
    \end{equation} 
    is trivial if and only if there exists a smooth (global) section $\sigma : M \longrightarrow P \;\;\;(\pi \circ\sigma = \text{id}_M)$.
\end{theorem}
Note that there was a similar theorem we saw before when discussing the theory of vector bundles.
\begin{proof}
    For the $\implies$ part of the theorem, assume that the bundle  is trivial. Then, there exists a principal bundle map, $ P \xlongrightarrow[\cong(\text{diffeo.})]{\varphi}M\times G$, with the commuting diagram

\begin{equation}
    \begin{tikzcd}[sep=.8cm, remember picture]
        |[alias=A]| p \arrow[rr, thick, "\varphi"]& {} & |[alias=B]| M\times G \\ 
        |[alias=C]| p \arrow[rd, thick, swap, "\pi"] \arrow[u, thick, "\cdot G"] \arrow{rr}{\cong}[swap, thick]{\varphi} & {} & |[alias=D]| M\times G \arrow[u,thick, swap, "\circ G"] \arrow[ld, thick, "\pi'"] \\ 
        {} & |[alias=E]| \arrow[ul, dashed, thick, "\sigma", bend left=60]M & {}
    \end{tikzcd}
\end{equation}
where
\begin{itemize}
    \item[a)] $\varphi$ is a diffeomorphism.
    \item[b)] $\pi' \circ \varphi = \pi$.
    \item[c)] $\varphi(p\cdot g) = \varphi(p) \circ g$.
\end{itemize}
Then, we define $\sigma : M \longrightarrow P$ by:
\begin{equation}
    \sigma(x) := \varphi^{-1}(x,e) \;\; \text{where} \; \; (x,e) \in (\pi')^{-1} (x),
\end{equation}
with $\varphi^{-1}$ being smooth. Notice that:
\begin{equation*}
\begin{aligned}
    (\pi \circ \sigma)(x) &= \pi(\varphi^{-1}(x,e)) = (\pi' \circ \varphi \circ \varphi^{-1})(x,e) = x,
    \\
    \pi \circ \sigma &= \text{id}_M.
\end{aligned}
\end{equation*}
where the middle equality in the first line comes from $(b)$.

For the $\impliedby$ part, we assume that there exists a section $\sigma: M\longrightarrow P$, where $\sigma$ is smooth, and $\pi \circ \sigma = \text{id}_M$. We then need to construct a principle bundle map (which is equipped with compatibility conditions and diffeomorphism) such that $P \underbrace{\cong}_{\text{diffeo.}} M \times G$.
\end{proof}

Now, let $p \in P$. Then $\sigma(\pi(p))$ and $p$ lie in the same fiber over $\pi(p)$. Then,     there exists a unique $g_{\sigma}(p) \in G$ such that:
    \begin{equation}
        \sigma(\pi(p)) \cdot g_{\sigma}(p) = p,
    \end{equation}
    where $g_\sigma(p)$ depends both on $\sigma$ and $p$. More explicitly,
    \begin{equation}
        \begin{aligned}
            g_\sigma: P &\longrightarrow G\\
            p &\longrightarrow g_\sigma(p).
        \end{aligned}
    \end{equation}
    \begin{center}
        \includestandalone[width=0.45\textwidth]{Figures/Part1/p71} \\
        \hypertarget{fig20}{Figure 20.} Local picture of a trivial principal $G$-bundle with a local section, $\sigma$.
    \end{center}
    Then, we define
    \begin{equation}\label{eq:ssdiagramtwopart}
        \begin{tikzcd}[sep=1.5cm, remember picture]
            |[alias=A]| p \arrow[r, thick, "\varphi_\sigma"] & |[alias = B]| M\times G \\ 
            |[alias = C]| p \arrow[u, thick, "\cdot G"] \arrow[r, thick, "\varphi_\sigma"] \arrow[d, thick, "\pi"] & |[alias = D]| M\times G \arrow[u, thick, "\circ G"] \arrow[ld, thick, "\pi'"] \\ 
            |[alias=E]| M & {}
        \end{tikzcd}
        \tikz[overlay,remember picture]{%
        \node (Y1) [scale=1.2] at (barycentric cs:A=1,B=1,C=1,D=1) {II};
        \node (Y2) [scale=1.2] at (barycentric cs:C=1,D=1,E=1) {I};}
    \end{equation}
    \begin{equation}
        \begin{aligned}
            \varphi_{\sigma}: P &\longrightarrow M \times G
            \\
            P &\longrightarrow (\pi(p) , g_{\sigma}(p)).
        \end{aligned}
    \end{equation}
    Clearly, $\varphi_\sigma$ is smooth (since both $\pi$ and $g_{\sigma}$ are smooth). We need to show the commutativity of the diagrams I, II. Before, let us put forward a lemma and prove it.
    \begin{lemma}
        $g_{\sigma}(p \cdot g) = g_{\sigma}(p) \cdot g \;\; \forall g \in G$, where the right hand side has the multiplication in $G$.
    \end{lemma}
    \begin{proof}
        To prove this lemma, consider the setup in \hyperlink{fig20}{Figure 20.} For all $ g \in G$; $p \cdot g$ lies in the fiber over $\pi(p)$. Then, there exists a unique (uniqueness comes from the freeness of the action) $ g_{\sigma}(p \cdot g) \in G$ such that:
        \begin{equation}
        \label{eq:1.239}
            \sigma(\pi(p)) \cdot g_{\sigma}(p\cdot g) = p \cdot g.
        \end{equation}
        On the other hand we have:
        \begin{equation}
        \label{eq:1.240}
            \sigma(\pi(p)) \cdot g_\sigma (p) = p.
        \end{equation}
        Act on \ref{eq:1.240} by $g$
        \begin{equation}
        \begin{aligned}
            \label{eq:1.241}
            (\sigma(\pi(p))\cdot g_{\sigma}(p))\cdot g &= p \cdot g
            \\
            &= \sigma(\pi(p)) \cdot \underbrace{g_{\sigma}(p) g }_{\in \mathbb{C}},
        \end{aligned}
        \end{equation}
        where we have used the definition of action. Then, by using \ref{eq:1.239} we have:
        \begin{equation}
            \sigma(\pi(p)) \cdot g_{\sigma}(p) \cdot g \underbrace{=}_{\text{by} \: \ref{eq:1.241}} p \cdot g \underbrace{=}_{\text{by} \: \ref{eq:1.239}} \sigma(\pi(p)) \cdot g_{\sigma}(p \cdot g) 
        \end{equation}
        Since the action is free, we get:
        \begin{equation}
            g_{\sigma}(p) g = g_{\sigma}(p \cdot g). 
        \end{equation}
        $\forall g \in G$ and $\forall p \in P$.
    \end{proof}
    Now, we can see that the sub-diagram I in Eq.~\ref{eq:ssdiagramtwopart} commutes since
    \begin{equation*}
        \begin{aligned}
            (\pi' \circ \varphi_{\sigma}) (p \cdot g) &= \pi' \left(\pi (p\cdot g), g_{\sigma} (p \cdot g) \right)
            \\
            &= \pi(p\cdot g)
            \\
            & \Longrightarrow \pi' \circ \varphi_{\sigma} = \pi.
        \end{aligned}
    \end{equation*}
    Sub-diagram II also commutes:
    \begin{equation*}
        \begin{aligned}
            \varphi_{\sigma}(p \cdot g) &= \left( \pi(p \cdot g),g_{\sigma}(p\cdot g)\right)
            \\
            &\underbrace{=}_{\text{Lemma}} (\pi(p)), g_{\sigma}(p)g)
            \\
            &= (\pi(p),g_{\sigma}(p)) \circ g
            \\
            &= \varphi_{\sigma}(p)\circ g.
        \end{aligned}
    \end{equation*}
    Here $p \cdot g$ and $p$ lie in the same fiber over $\pi(p)$. Also, on the third line $\circ$ comes from the definition of induced action. So, $\varphi_\sigma$ is a principal bundle map. 
    \begin{claim}
        $\varphi_{\sigma}$ is a diffeomorphism.
    \end{claim}
    \begin{proof}
        Firstly, $\varphi_\sigma$ is smooth. To prove the claim, we need to define the smooth inverse of $\varphi_{\sigma}$. To do so, we define
        \begin{equation*}
            \psi_{\sigma}:\;\; M \times G \longrightarrow P,
        \end{equation*}
        by
        \begin{equation*}
            \psi_{\sigma}(x,g) := \sigma(\pi(x,g)) \cdot g = \sigma(x) \cdot g,
        \end{equation*}
        where $\sigma(x) \in P$, and both $\sigma$ and $\cdot$ are smooth.
        \begin{center}
            \includestandalone[width=0.40\textwidth]{Figures/Part1/p73} \\ 
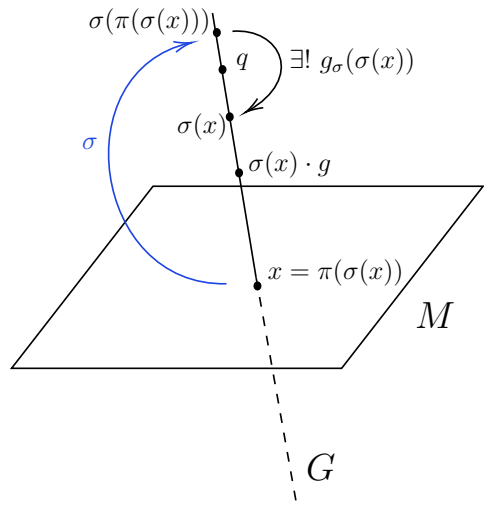
            \hypertarget{fig21}{Figure 21.} Action of a section, $\sigma$, on a fiber.
        \end{center}
        First, observe that $\sigma(x)$ and $\sigma\left(\pi(\sigma(x))\right)$ lie in the same fiber over $x$. So, there exists a unique $g_{\sigma}(\sigma(x)) \in G$ such that:
        \begin{equation*}
            \underbrace{\sigma(\pi(\sigma(x)))}_{\sigma(x) \: \text{since} \ \ \pi \circ \sigma=id_M} \cdot \ g_{\sigma}(\sigma(x)) = \sigma(x).
        \end{equation*}
        Freeness implies that $g_{\sigma}(\sigma(x)) = e$. Then, for $(x,g) \in M \times G$ we have:
        \begin{equation*}
            \begin{aligned}
                (\varphi_\sigma \circ \psi_\sigma)(x,g) &= \varphi_\sigma (\sigma(x)\cdot g)
                \\
                &= ( \pi(\sigma(x)\cdot g), \underbrace{g_\sigma(\sigma(x) )}_{\text{lemma}}\cdot g)
                \\
                &= (\underbrace{\pi(\sigma(x))}_{id_M}, \underbrace{g_{\sigma}(\sigma(x))}_{e} \cdot g)
                \\
                &= (x,g),
                \end{aligned}
        \end{equation*}
        and
        \begin{equation*}
            \begin{aligned}
                (\psi_\sigma \circ \varphi_\sigma)(p) &= \psi_{\sigma}(\pi(p)), g_{\sigma}(p))
                \\
                &= \sigma(\pi(p)) \cdot g_{\sigma}(p)
                \\
                &= p,
            \end{aligned}
        \end{equation*}
        by definition of $g_\sigma:$ unique element $g_{\sigma}(p) \in G$ such that $\sigma(\pi(p)) \cdot g_{\sigma}(p)=p$.
        \begin{equation*}
            \therefore P \xlongrightarrow[\varphi_\sigma]{\cong} M\times G 
        \end{equation*}
        is a diffeomorphism.
    \end{proof}   
\begin{example}
    \textbf{Frame bundle} over a smooth $\textit{n}$-manifold: $LM \xlongrightarrow{\pi}M$. For each $x \in M$,
    \begin{enumerate}
        \item \begin{equation*}
            L_x M := \{(e_{1}^{x}, e_{2}^{x},\cdots,e_{n}^{x}) \; | \; e_{1}^{x}, e_{2}^{x},\cdots,e_{n}^{x} \; \text{is a basis for} \; T_xM  \} \cong GL(n;\mathbb{R}),
        \end{equation*}
        and $LM := \amalg_{x \in M} L_x M$. Note that $\dim(LM) = \dim M + \dim M^2$. $LM$ inherits smooth atlas from $M$. Here the map $\pi$ is given as
        \begin{equation*}
        \begin{aligned}
            \pi: LM &\longrightarrow M
            \\
            (e_1 , ..., e_n) &\longrightarrow x \;\; (\text{unique $x \in M $ such that $e_1,...,e_n$ form a basis for $T_x M$}).
        \end{aligned}
        \end{equation*}
        So, $LM\xlongrightarrow{\pi} M$ is a smooth bundle.

        \item Construct a right $GL(n;\mathbb{R})$-action on $LM$;
        \\
        Let 
        \begin{equation*}
            g \in GL(n;\mathbb{R}) = \{A \in \mathbb{R}^{n \times n} \; | \; \det A \neq 0\}.
        \end{equation*} 
        Denote $g$ by $g = [g^{i}_{\:j}]$ where $i,j = 1\cdots,n$. For $(e_1,\cdots,e_n) \in T_x M$ define:
        \begin{equation*}
           (e_1,\cdots,e_n) \cdot g := (e_1,\cdots,e_n)  [g^{i}_{\:j}] = (g^{i}_{\:1} e_i \:,\: g^{i}_{\:2} e_i \:,\cdots, g^{i}_{\:n} e_i ),
        \end{equation*}
        which follows from matrix multiplication.
        i.e.:
        \begin{equation*}
            \left((e_1,\cdots,e_n) \cdot g\right)_K = g^{l}_{\:k}e_l \;\;\; l,k = 1,\cdots,n.
        \end{equation*}
        Note that the action is free.
        \item 
            \begin{equation}
                \begin{tikzcd}[sep=1.2cm, remember picture]
                   |[alias=A]| LM \arrow[d, thick, "\pi"] & |[alias=B]|LM \arrow[d, thick, "\pi'"] \\ 
                    |[alias=C]|M & |[alias=D]|LM/GL(n;\mathbb{R})
                \end{tikzcd}
                \tikz[overlay,remember picture]{%
                \node (Y1) [scale=1.2] at (barycentric cs:A=1,B=1,C=1,D=1) {$\cong$};}
            \end{equation}    
            $(e_1 , ...,e_n) \sim (\Tilde{e}_{1},...,\Tilde{e}_{n})$ if and only if there exists a $g \in GL(n;\mathbb{R})$ such that:
        \begin{equation}
        \label{eq:1.244}
            \Tilde{e}_K = g_{k}^{\: \: l} e_l.
        \end{equation}
        We define:
        \begin{equation*}
            \begin{aligned}
                M &\longrightarrow LM/GL(n;\mathbb{R})
                \\
                x &\longrightarrow [e_1 ,\cdots,e_n] \in T_x M,
            \end{aligned}
        \end{equation*}
        assigning local basis whose transformation rule in the overlaps is determined by equation \ref{eq:1.244}.
        \\
        Clearly smooth, inverse map $[e_1 ,...,e_n] \longrightarrow \underbrace{x}_{\exists! x \in M}$ smooth implying that:
        \begin{equation*}
            M \underbrace{\cong}_{\text{diffeo}} LM/GL(n;\mathbb{R}).
        \end{equation*}
        Therefore, $L\xlongrightarrow{\pi}M$ frame bundle is a principal $GL(n;\mathbb{R})$- bundle over a smooth $n$-manifold $M$.
        (Later we shall see its associated vector bundle is isomorphic to the "tangent bundle" over $M$.) Note that the construction is a suitable setup for the theory of General Relativity. Moreover, the formalism serves as a "background geometry" for general relativity.
    \end{enumerate}

\end{example}

\begin{remark}
    Let
    \begin{equation}
        \begin{tikzcd}[sep=1.2cm]
            P\arrow[r, thick, "\cdot G"] & P  \arrow[d, thick, "\pi"] \\ {} & M
        \end{tikzcd}
    \end{equation}
    be a principal $G$-bundle. Then, let $\sigma: U \subseteq_{\text{open}} M \longrightarrow P$ be a local section. Then $\sigma$ defines a local trivialization, i.e.: $P \cong U \times G$ locally.
        \begin{center}
            \includestandalone[width=0.55\textwidth]{Figures/Part1/p76}\\ 
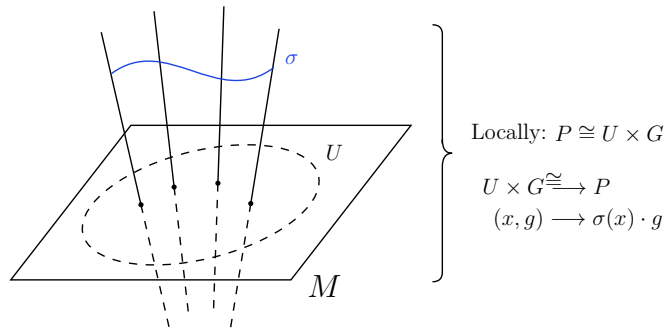
            \hypertarget{fig22}{Figure 22.} Local trivialization of a principal $G$-bundle.
        \end{center}
\end{remark}
This is just as we did in global case.
\subsection{Associated Bundles}
\begin{definition}
    Let $(P,\pi,M)$ be a principal $G$-bundle.
    \begin{equation}
        \begin{tikzcd}[sep=1.2cm, remember picture]
           G \arrow[r, thick, hook] & |[alias=A]| P \arrow[d, thick, "\pi"] & |[alias=B]|P \arrow[d, thick, "\pi'"] \\ 
          {} &  |[alias=C]|M & |[alias=D]|LM/GL(n;\mathbb{R})
        \end{tikzcd}
        \tikz[overlay,remember picture]{%
        \node (Y1) [scale=1.2] at (barycentric cs:A=1,B=1,C=1,D=1) {$\cong$};}
    \end{equation}    

    Let $F$ be a smooth manifold which is a left $G$-space. i.e.: There exists a left $G$-action on $F$.
    \begin{equation*}
    \begin{aligned}
        \cdot : G \times F &\longrightarrow F
        \\
        (g,f) &\longrightarrow g \cdot f.
    \end{aligned}
    \end{equation*}
    Then, we define the \textbf{associated bundle} as follows:
    \begin{enumerate}
        \item Define an equivalence relation on $P\times F$:
        \begin{equation*}
            (p,f) \sim (p' , f') \iff \exists g \in G,
        \end{equation*}
        such that:
        \begin{equation*}
            \begin{aligned}
                p' &= p \cdot g \;\;(\text{there exists a right $G$-action on $P$}),
                \\
                f' &= g^{-1} \cdot f \;\;(\text{there exists a left $G$-action on $F$}).
            \end{aligned}
        \end{equation*}
        \item Consider the quotient space with respect to the relation above:
        \begin{equation*}
            P_F := \frac{P \times F}{\sim}.
        \end{equation*}
        An observation is that a fiber $P$ over $m\in M$ is just $G$. But, after quotioning out, roughly speaking, we attach $F$ as a fiber of $P_F$ over $m$.
        \item Define the projection map:  $P_F\xlongrightarrow{\pi_F}M$ where $\pi_F \left([p,f]\right):= \pi(p)$. Note that $\pi_F$ is well-defined: (Since the "inputs" of $\pi_F$ are equivalence classes, we need to show that $\pi_F$ depends only on the equivalence class, not a representative chosen in a class.)
        Take $[p,f] \in P_F$. Consider another representative $[p' , f'] = [p,f]$ (i.e. $(p',f') \in [p,f]$), and have $p' = p \cdot g$, $f=g^{-1} \cdot f$, for some $g \in G$. Then we have:
        \begin{equation*}
            \pi_F([p',f']) = \pi_{F}([p\cdot g, g^{-1} \cdot f]) = \pi(p \cdot g) =  \pi(p) = \pi_F ([p,f]),
        \end{equation*}
        where in the last equality we have used the fact that $p \cdot  g$ and $p$ lie in the same fiber over $\pi(p)$.
    \end{enumerate}

\newpage 

    Then, 
    \begin{equation}
        P_F := P\times F/ \sim \ \ \xlongrightarrow{\pi_F} M,
    \end{equation}
    is called the \textbf{associated bundle} for 
    \begin{equation}
        G \hookrightarrow F \xlongrightarrow{\pi} M.
    \end{equation}
\end{definition}
\begin{example}
    Consider the following principal $G$-bundle: Frame bundle:
    \begin{equation}
        \begin{tikzcd}[sep=1.5cm]
            LM \arrow[r, thick, "\cdot GL(n; \mathbb{R})"] & LM \arrow[d, thick, "\pi"] \\ 
            {} & M
        \end{tikzcd}
    \end{equation}
    where $n$ = $\dim M$, and $\forall x \in M$:
    \begin{enumerate}
        \item[(a)] $L_x M = \{(e_1,...,e_n)\;|\; e_1 ,...,e_n  \; \text{forms a basis for} \; T_x M\}$
        ($LM = \amalg_{x\in M} L_x M$).
        \item[(b)] $\pi(e_1,...,e_n) = x$.
        \item[(c)] the right $GL(n,\mathbb{R})$- action on $LM$ is given by: For each $(e_1,...,e_n) \in T_x M$, $g = g^{i}_{\;j} \in GL(n;\mathbb{R})$,
        \begin{equation*}
            (e_1,...,e_n)\cdot [g]^{i}_{\; j} = (g^{i}_{\;1} e_i , g^{i}_{\;2} e_i,...,g^{i}_{\;k} e_i,...,g^{i}_{\;n}e_i).
        \end{equation*}
        (i.e. $\underbrace{\left((e_1,...,e_n)\cdot g^{i}_{\;j}\right)_k}_{e'_{k}} = g_{k}^{\; i}e_i$), \;\; $i=1,...,n$.
    \end{enumerate}
    Take $F = \mathbb{R}^{n}$. Define the left $GL(n;\mathbb{R})$-action on $F$;
    For each $f=[f^1,...,f^n]^{T} \in \mathbb{R}^{n}$, $g^{i}_{\;j} \in GL(n,\mathbb{R})$ we define: 
    \begin{equation*}
        g^{i}_{\;j} \cdot \begin{bmatrix}
            f^1 \\ \vdots \\f^n
        \end{bmatrix}
        := \begin{bmatrix}
            (g^{-1})^{1}_{\;k} f^k \\ \vdots \\ (g^{-1})^{a}_{\;k}f^k \\ \vdots \\
            (g^{-1})^{n}_{\; k}f^k
        \end{bmatrix}.
    \end{equation*}
    i.e.:
    \begin{equation}
        (g^i_{j}\cdot f)^{a} \underbrace{:=}_{a^\mathrm{th}  \ \mathrm{ comp.}}  (g^{-1})^a_kf^k.
    \end{equation}
    Then the associated bundle is:
    \begin{equation}
        LM_{\mathbb{R}^n} \xlongrightarrow{\pi_{\mathbb{R}^n}} M,
    \end{equation}
    where 
    \begin{equation*}
        LM_{\mathbb{R}^{n}} = \frac{LM \times \mathbb{R}^n}{\sim} . 
    \end{equation*}
\end{example}
\begin{observation}
    \begin{equation}
        \begin{tikzcd}[sep=1.2cm, remember picture]
           |[alias=A]| LM_{\mathbb{R}^n} \arrow[d, thick, "\pi_{\mathbb{R}^n}"] & |[alias=B]|TM \arrow[d, thick, "\pi'"] \\ 
            |[alias=C]|M & |[alias=D]|M
        \end{tikzcd}
        \tikz[overlay,remember picture]{%
        \node (Y1) [scale=1.2] at (barycentric cs:A=1,B=1,C=1,D=1) {$\cong$};}
    \end{equation}    
    which also means that the following diagram commutes.
    \begin{equation}
        \begin{tikzcd}[sep=1.5cm, remember picture]
           |[alias=A]| LM_{\mathbb{R}^n} \arrow[rd, thick, "\pi_{\mathbb{R}^n}"] \arrow[rr, thick, "\varphi"] & {} & |[alias=B]|TM \arrow[ld, thick, "\pi'"] \\ 
           {} & |[alias=C]|M & {}
        \end{tikzcd}
        \tikz[overlay,remember picture]{%
        \node (Y1) [scale=1.2] at (barycentric cs:A=1,B=1,C=1) {$\circlearrowright$};}
    \end{equation}
    This requires the existence of a diffeomorphism, $\varphi: LM_{\mathbb{R}^n}\xlongrightarrow{\cong} TM$ that is compatible with the bundle structure. Indeed, we may naturally define:
    \begin{equation*}
        \begin{aligned}
            LM_{R^{n}} &\xlongrightarrow{\varphi} TM
            \\
            [e,f] &\longmapsto\underbrace{f^{\mu}e_{\mu}}_{:= X_x \in T_x M},
        \end{aligned}
    \end{equation*}
    where the left hand side is given by:
    \begin{equation*}
        \begin{aligned}
            e &= (e_1 , ... , e_n
            )  \in T_x M
            \\
            f &= \begin{bmatrix}
                \vdots
            \end{bmatrix}
            \in \mathbb{R}^{n}.
        \end{aligned}
    \end{equation*}
    The right-hand side is the component version of vector $X_x \in T_X M$, where on the overlaps the transformation law for components is:
    \begin{equation}
    \label{1.255}
        f'^{a} = (g^{-1})^{a}_{\; k}f^{k},
    \end{equation}
    and for basis is:
    \begin{equation*}
        e^{'k} = g^{i}_{\; k}e_i .
    \end{equation*}
    Conversely, given $X = f^{\mu}e_{\mu} \in TM$, we assign:
    \begin{equation*}
        X \xlongrightarrow{\psi} [e,f],
    \end{equation*}
     where
    \begin{itemize}
        \item $\psi$ is also smooth.
        \item $\psi = \varphi^{-1}$.
    \end{itemize}
    Here, the assignment is independent of coordinates as we identify all $(e,f)$ which are transformed as in \ref{1.255}.
   
\end{observation}

\begin{observation}
    \begin{enumerate}
        \item The action of "the structure group" $GL(n;\mathbb{R})$ is not seen when we directly consider the tangent bundle. $TM \longrightarrow M$. But, if we consider the corresponding associated bundle $LM_{\mathbb{R}^{n}} \longrightarrow M$, we manifestly uncover \textit{the underlying symmetry} given by a certain structure group-action (in that case $GL(n;\mathbb{R})$). Note that the framework serves as an underlying geometry for general relativity.
        \item We may consider some other structure groups! Take $SO(n)$, for instance, the same setup above yields "rotational invariance".
        \item For example, in special relativity, the underlying symmetry group above is the "Lorentz group", i.e., if we take $G = O(n;1)$, then one recovers special relativity.
        \item Taking $GL(n,\mathbb{R})$, 
        \begin{equation*}
            \begin{aligned}
                &F := (\mathbb{R}^{n})^{\times p} \times (\mathbb{R}^{n*} )^{\times q}
                \\
                &\underbrace{\mathbb{R}^{n} \times \cdots \times \mathbb{R}^{n}}_{p} \times \underbrace{\mathbb{R}^{*n} \times \cdots \times \mathbb{R}^{n*}}_{q}
            \end{aligned}
            ,
        \end{equation*}
         The first line denotes a $(p,q)$-tensor bundle, and the corresponding associated bundle is isomorphic to the $(p,q)$-tensor bundle. One can check Schuller's Lecture Notes for more details.
        \item Taking $F = \mathbb{R}$ gives rise to the scalar density bundle.
    \end{enumerate}
\end{observation}
\begin{claim}
    A principal bundle map $(\varphi, \psi)$ induces a bundle map between corresponding associated bundles with the same fiber $F:$
\end{claim}
Let 
    \begin{equation}
        \begin{tikzcd}
            P \arrow[r, thick, "\cdot G"] & P \arrow[d, thick, "\pi"] \\ 
            {} & M
        \end{tikzcd} \qquad \text{and} \qquad 
        \begin{tikzcd}
            P \arrow[r, thick, "\cdot' G'"] & P' \arrow[d, thick, "\pi'"] \\ 
            {} & M
        \end{tikzcd}
    \end{equation}
be two principal bundles, and let $\rho: G\longrightarrow G'$ be a Lie group homomorphism. We assume that $(\varphi,\psi)$ is a principal bundle map, i.e.
\begin{equation}
    \begin{tikzcd}[sep = 1.2cm, remember picture]
        |[alias = A]| P \arrow[r, thick, "\varphi"] & 
        |[alias = B]| P' \\ 
        |[alias = C]| P \arrow[u, thick, "\cdot G"] \arrow[r, thick, "\varphi"] \arrow[d, thick, "\pi"] & |[alias = D]| P' \arrow[u, thick, "\cdot'G'", swap] \arrow[d, thick, "\pi'"] \\ 
        |[alias = E]| M \arrow[r, thick, "\psi"] &  
        |[alias = F]| M
    \end{tikzcd}
    \tikz[overlay,remember picture]{%
        \node (Y1) [scale=1.5] at (barycentric cs:A=1,B=1,C=1,D=1) {$\circlearrowright$};
        \node (Y2) [scale=1.5] at (barycentric cs:C=1,D=1,E=1,F=1) {$\circlearrowright$};
        }
\end{equation}
such that
\begin{itemize}
    \item[(a)] $\varphi: P \longrightarrow P'$, $\psi: M\longrightarrow M$ smooth map,
    \item[(b)] $\pi' \circ \varphi = \psi \circ \pi$,
    \item[(c)] $\varphi (p\cdot g) = \varphi(p) \cdot' \rho (g) \;\; \forall g \in G$.
\end{itemize}

Then, letting $F$ be a manifold, and $P_F \xlongrightarrow{\pi_F}M$, $P'_F \xlongrightarrow{\pi'_F}M$  the associated bundles,  $\pi_F ([p,f]) = \pi(p) \; \forall p \in P, f \in F$. We define an \textbf{associated bundle map} (induced from principal bundle map) as a smooth pair $(\Tilde{\varphi},\Tilde{\psi})$.
\begin{equation*}
\begin{aligned}
    \Tilde{\varphi}:P_F &\longrightarrow P'_{F},
    \\
    \Tilde{\psi}: M &\longrightarrow M,
\end{aligned}
\end{equation*}
such that:
\begin{enumerate}
    \item $\Tilde{\varphi} := [\varphi(p),f]\;\;\; \forall [p,f ]\in P_F$
    \item $\Tilde{\psi}(p) := \psi(x) \;\;\; \forall x \in M$
\end{enumerate}
\begin{observation}
    These items, by definition, imply that the diagram
    \begin{equation}
        \begin{tikzcd}[sep=1.5cm, remember picture]
            |[alias = A]| P_F \arrow[r,thick, "\tilde{\varphi}"] \arrow[d, thick, "\pi_F", swap] & 
            |[alias = B]| P_F' \arrow[d, thick, "\pi_F'"] \\
            |[alias = C]| M \arrow[r, thick, "\tilde{\psi}"] & 
            |[alias = D]| M
        \end{tikzcd}
        \tikz[overlay,remember picture]{%
            \node (Y1) [scale=1.5] at (barycentric cs:A=1,B=1,C=1,D=1) {$\circlearrowright$};
        }
    \end{equation}
    commutes:
    \begin{equation}
        \pi'_F \circ \Tilde{\varphi} = \Tilde{\psi} \circ \pi_F.
    \end{equation}
\end{observation}
Indeed; let $[p,f] \in P_F$ then we have:
\begin{equation}
    \begin{aligned}
        (\pi_F' \circ \Tilde{\varphi})([p,f]) &= \pi'_F (\Tilde{\varphi([p,f])}) = \pi'_F ([\varphi(p),f]) \;\;\; &\text{by definition of $\Tilde{\varphi}$}
        \\
        &=\pi'(\varphi(p)) \;\;\; &\text{by definition of $\pi'_F$}
        \\
        &= (\psi \circ \pi)(p) \;\;\; &\text{by b.}
        \\
        &= \psi(\pi(p))
        \\
        &= \Tilde{\psi}(\pi(p)) \;\;\; &\text{by definition of $\Tilde{\psi}$} 
        \\
        &= (\Tilde{\psi} \circ \pi_F)([p,f]) \;\;\; &\text{by definition of $\pi_F$}
    \end{aligned}
\end{equation}
hence
\begin{equation*}
    \pi'_F \circ \Tilde{\varphi} = \Tilde{\psi} \circ \pi_F.
\end{equation*}
Now, let us present some results and definitions for associated bundles:
\begin{definition}
    An associated bundle is called \textbf{trivial} if the underlying principal bundle map is trivial.
\end{definition}
\begin{theorem}
    A trivial associated bundle is a trivial fiber bundle.
\end{theorem}
\begin{theorem}
    \label{eq:theorem12}
    Let
    \begin{equation}
        \begin{tikzcd}
            G \arrow[hook, r, thick] & P \arrow[d, thick, "\pi"] \\ {} & M    
        \end{tikzcd}
    \end{equation}
    be a principal $G$-bundle. Let $P_F \xlongrightarrow{\pi_F} M$ be its associated bundle for a manifold $F$. Then, the spaces
$$
\left\{
  \begin{array}{l}
    \text{The sections } \sigma: M \to P_F \\
    \text{of an associated bundle}
  \end{array}
\right\}
 \qquad \text{ and }\qquad
\left\{
  \begin{array}{l}
    {\text{$F$-valued, smooth maps on } P} \\
    \phi: P \to F
  \end{array}
\right\}
$$
are in one-to-one correspondence.
\end{theorem}
i.e. given $\sigma$, one can construct $\phi_\sigma:P \longrightarrow F$ and conversely, given $\phi: P \longrightarrow F$, one can construct $\sigma_\phi : M \longrightarrow P_F$, ($\pi_F \circ \sigma_\phi = id$) in a suitable way.
\begin{remark}
    That theorem is crucial when discussing the notion of a "covariant derivative" on a principal $ G$-bundle or an associated bundle.
    \begin{equation}
        \begin{aligned}
            \nabla_Y \sigma \in \Gamma(P_F) \quad \text{where} \; Y \in TM, \sigma \in \Gamma(P_F)
        \end{aligned}
    \end{equation}
\end{remark}
This can be done  by following the  steps below:
\begin{itemize}
    \item Using Theorem \ref{eq:theorem12}, pass to $P \xlongrightarrow{\pi} M$.
    \item Construct  on $P$ (via corresponding $1$-form on $P$ and $Y \in TP$).
    \item Pull it down to $M$.
\end{itemize}
\subsection{Connection, Covariant Derivative, and Curvature on a Principal  Bundle}

\textbf{Our setup:} Let 
    \begin{equation}
        \begin{tikzcd}[sep=1.1cm]
            G \arrow[r, thick, "\cdot G"] & P \arrow[d, thick, "\pi"] \\ {} & M    
        \end{tikzcd}
    \end{equation}
be a principal $G$-bundle. Then
\begin{enumerate}
    \item[(a)] $P$ is a right $G$-space.
    \item[(b)] the action is free.
    \item[(c)] 
    \begin{equation}
        \begin{tikzcd}[sep=1.2cm, remember picture]
            |[alias = A]| P \arrow[d, thick, "\pi", swap] & 
            |[alias = B]| P \arrow[d, thick, "\pi'"] \\
            |[alias = C]| M & 
            |[alias = D]| P/G
        \end{tikzcd}
        \tikz[overlay,remember picture]{%
            \node (Y1) [scale=1.5] at (barycentric cs:A=1,B=1,C=1,D=1) {$\cong$};
        }
    \end{equation}
    where $M\cong P/G$ together with the diagram
    
    \begin{equation}
        \begin{tikzcd}[sep = 1.5cm, remember picture]
            {} & |[alias = A]| P \arrow[ld, thick] \arrow[rd, thick, "\pi'"] & {} \\ 
            |[alias = B]| M \arrow[rr, thick, "\cong"] & {} & |[alias = C]| P/G 
        \end{tikzcd}
        \tikz[overlay, remember picture]{%
            \node (Y1) [scale=1.5] at (barycentric cs:A=1,B=1,C=1) {$\circlearrowright$};
        }
    \end{equation}
\end{enumerate}
\begin{observation}
    Let us recall an observation we made when discussing the moment map. In accordance with the setup above, there exists a natural map
    \begin{equation*}
        \begin{aligned}
            i: T_e G \cong \mathfrak{g} &\longrightarrow \Gamma(TP)
            \\
            X &\longmapsto X^{\sharp},
        \end{aligned}
    \end{equation*}
    which is defined as follows: For each $p \in P$, and $\forall f \in C^{\infty}(p)$, we set
    \begin{equation*}
        X^{\sharp}_p f := \frac{d}{dt} \bigg|_{t=0} f(\gamma(t)),
    \end{equation*}
    where $\gamma(t):= p \cdot \exp(tX)$, $\cdot$ is $G$-action on $P$, and
    \begin{equation*}
        \begin{aligned}
            \exp: \mathfrak{g} &\longrightarrow G
            \\
            X &\longmapsto \exp(tX)
        \end{aligned}
    \end{equation*}
    
    \begin{center}
     \includestandalone[width=0.6\textwidth]{Figures/Part1/p85}\\
     \hypertarget{fig23}{Figure 23.}$\exp$ map and local picture. 
    \end{center}
    
    In fact, $i$ induces an isomorphism onto some subspace.
    \begin{enumerate}
        \item[(a)] $T_e G \cong \Gamma(V P)$, $VP \subseteq TP$ (will be clear later),
        \\
        where \begin{equation*}
            \begin{aligned}
                i^{-1}: \Gamma(VP) &\longrightarrow \mathfrak{g} \;\;\; \text{a natural map}
                \\
                X^{\#} &\longmapsto X \;\;\;(\text{will be clear later})
            \end{aligned}
        \end{equation*}
        \item[(b)] $i$ is a Lie algebra homomorphism:
        \begin{equation}
        \begin{aligned}
            \underbrace{i([X,Y]_\mathfrak{g})}_{= [X,Y]^{\#}} &= [i(X),i(Y)]_0 \;\;\;\;\; X,Y,[X,Y]\in T_e G
            \\
            &= [X^{\#},Y^{\#}]_0
        \end{aligned}
        \end{equation}
    \end{enumerate}
\end{observation}
\begin{definition}
    Let $P \xlongrightarrow{\pi} M$ be a principal $G$-bundle.
    \begin{equation}
        \begin{tikzcd}
            P \arrow[r, thick, "\cdot G"] & P \arrow[d, thick, "\pi"] & {} & TP \arrow[d, thick, "\pi_*"] \\ 
            {} & M & {} & TM
        \end{tikzcd}
    \end{equation}
    We define the \textbf{vertical bundle} $VP \subseteq TP$ as follows: For each $p\in P$, we let
    \begin{equation*}
        \begin{aligned}
            V_p P &:= \ker(\pi_{*})
            \\
            & =\{X_p \in T_p M \; |\; \pi_{*,p}(X_p)=0 \in T_{\pi(p)}M  \},
        \end{aligned}
    \end{equation*}
    where $V_p P$ is the vertical subspace at $p$. And, we have
    \begin{equation*}
        \begin{aligned}
            \pi_{*,p}: T_p P &\longrightarrow T_{\pi(p)}M
            \\
            X_P &\longmapsto\pi_{*,p}(X_p)
        \end{aligned}
    \end{equation*}
\end{definition}

\begin{remark}
    \begin{enumerate}
        \item Try to realize on which "layer" of the structure we study. That is on $P$, \textit{not} on $M!$ $\Longrightarrow$ $X_p \in T_p P$
            \begin{center}
              \includestandalone[width=0.6\textwidth]{Figures/Part1/p86.1}\\
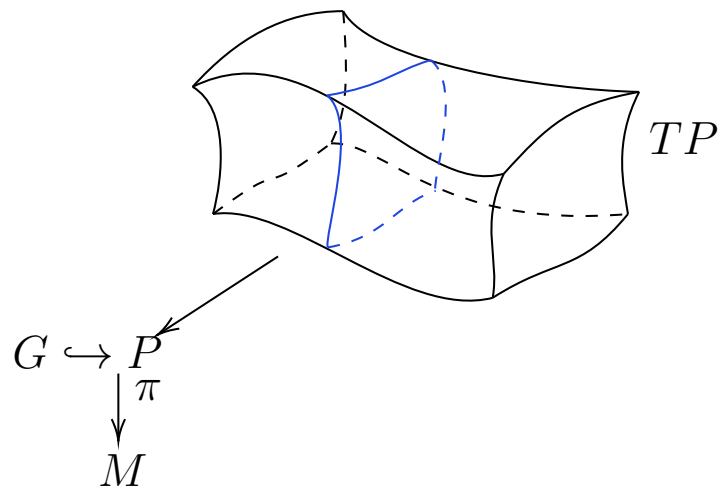
              \hypertarget{fig24}{Figure 24.} Principal $G$-bundle from a tangent bundle of $P$.
        \end{center}
        
        \newpage 
        
        \item In the local picture:
            \begin{center}
                \includestandalone[width=0.5\textwidth]{Figures/Part1/p86.2}\\ 
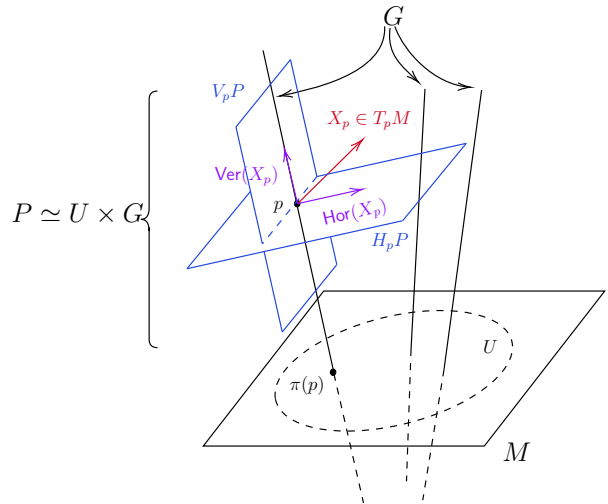
                \hypertarget{fig25}{Figure 25.} Decomposition of a vector in terms of its vertical and horizontal components which lie on their respective spaces, $\mathsf{Ver}(X_p)$ and $\mathsf{Hor}(X_p)$
            \end{center}
        Note that $X_p$ will have two components: "\textsf{Hor}($X_p$)" and "\textsf{Ver}($X_p$)". Formal definitions will be given later; we can simply write
        \begin{equation*}
            X_p \in V_p P \Longleftrightarrow \mathsf{Hor}(X_p)=0 \;\;\; (\text{no horizontal component})
        \end{equation*}
        \item The choice of so-called "horizontal space $H_p P$" at $p \in P$, (as in (2)) determines how to "connect" points in the neighboring fibers! This leads to the following definition.
    \end{enumerate}
    
\end{remark}
\begin{definition}
    A \textbf{connection} on $P$ is a distribution of subspaces $H_p P$ of $T_p P$ such that
    
    \begin{enumerate}
        \item $T_p P = H_p P$ $\oplus$ $V_p P$ \;\;\;\:$\forall p \in P$,
        \item $(\cdot g)_*$ $(H_p P) = H_{p\cdot g} P \;\;\; \forall p \in O, \forall g\in G$.

        \begin{center}
            \includestandalone[width=0.5\textwidth]{Figures/Part1/p87.1} \\
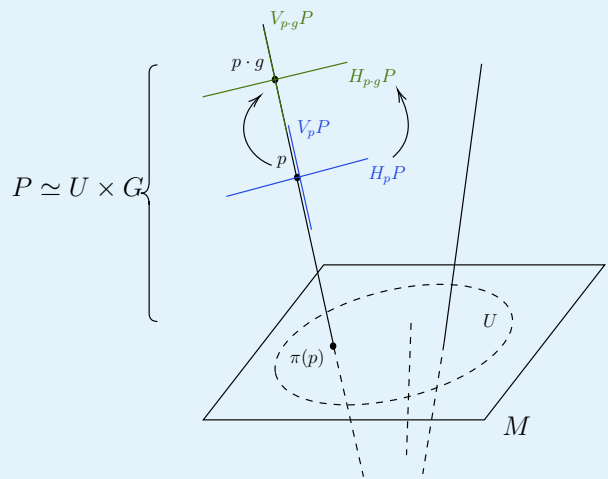
            \hypertarget{fig25}{Figure 25.} Connection between horizontal subspaces at two different points on the fiber.
        \end{center}
        Note that:
        \begin{equation*}
            \begin{aligned}
                P &\xlongrightarrow[\text{diffeo.}]{\cdot g} P \;\; , \;\; T_p P \xlongrightarrow[\cong,\text{isom.}]{(\cdot g)_*} T_p P
                \\
                (T_p P = H_p P \oplus V_p &P)\rightsquigarrow p \longrightarrow p\cdot g \rightsquigarrow (T_p P = H_p P \oplus V_p P)
            \end{aligned}
        \end{equation*}
        \begin{center}
            \includestandalone[width=0.23\textwidth]{Figures/Part1/p87.2}
        \end{center}
        What we impose here is that $(\cdot g)_*$ maps $H_p P$ to $H_{p\cdot g}P$ isomorphically (a compatibility condition with $G$-action at "push forward" level).
        \item $\forall X_p \in T_p P$, there exists a unique decomposition.
        \begin{equation*}
            X_p = \underbrace{\mathsf{Ver}(X_p)}_{ \in V_p P} \;\; + \;\; \underbrace{\mathsf{Hor}(X_p)}_{ \in H_p P},
        \end{equation*}
        where $\mathsf{Ver}(X_p)$ is \textbf{vertical subspace} at $p$ and $\mathsf{Hor}(X_p)$ is the \textbf{horizontal subspace} at $p$.
    \end{enumerate}
\end{definition}

    \begin{lemma}
        For all $ p \in P$ and  $X \in T_e G$ the above $X_{p}^{\#}$ is an element of $V_p P$.
        \begin{equation}
        \begin{tikzcd}[sep = 1.5cm]
            P \arrow[d, thick, "\pi"] & {} \\ 
            M \arrow[r, thick, "f"] & \mathbb{R} \text{ (or } \mathbb{C})
        \end{tikzcd}
        \end{equation}
    \end{lemma}
    \begin{proof}
    We need to show that $X_{p}^{\#} \in \text{ker}\pi_*$. Let $f \in C^{\infty}$. Then we consider $\pi_{*,p}(X_{p}^{\sharp}) \in T_{\pi(p)}M$ and try to show that $\pi_{*,p}(X_{p}^{\sharp})f= 0 \;\; \forall f.$
    \begin{equation}
        \begin{tikzcd}[row sep = 1.5 cm, column sep = 0.2cm]
            { T_pP} \arrow[d, "{\pi_{*,p}}"] & \ni & X_p^\sharp \arrow[d, maps to] \\
            T_{\pi(p)}M  & {} & {\pi_{*,p}(X_p^\sharp)}      
\end{tikzcd}
    \end{equation}
    We have:
    \begin{equation*}
        \begin{aligned}
            \pi_{*,p}(X_{p}^{\#})f &= X_{p}^{\#} (f \circ \pi) \;\;\;\; \text{where} \;\; f\circ \pi \in C^{\infty}(p)
            \\
            &= \frac{d}{dt} \bigg|_{t=0} f(\pi(p \cdot \exp(tX)))
            \\
            &= \frac{d}{dt} \bigg|_{t=0} (f \circ \pi)(p)
            \\
            &= 0 \;\;\; \forall f
        \end{aligned}
    \end{equation*}
    Here, in the first line, we have used the definition of $\pi_*$, and in the second line, we have used the definition of $X_{p}^{\sharp}$. Note that on the third line, $f(\pi(p))$ is independent of $t$ since $p$ and $p\cdot\exp(tX)$ lie on the same fiber.
\end{proof}
\begin{corollary}
    $i : \mathfrak{g} \xlongrightarrow{\cong} V_p P \subset T_p P \;\; \forall p \in P$. Given a basis $e_1,...,e_n$ for $\mathfrak{g}$, $\forall p \in P$  $e_{1}^{\#},...,e_{n}^{\#}$ forms a basis for $V_p P$.
\end{corollary}
    Given a connection, one can define a Lie-algebra valued one-form on $P$, say $\omega \in \Omega^{1}(p) \otimes \mathfrak{g}$ that captures the behavior of the distribution of $H_p P$.
\begin{definition}
    Formally, given a connection, i.e. a distribution $T_p P = H_p P \oplus V_p P \;\; \forall p \in P,$ we define for each $p\in P$,
    \begin{equation*}
        \begin{aligned}
            \omega_p : T_p P &\longrightarrow \mathfrak{g} \;\;\;\;\;\;             \\
            X_p &\longmapsto i^{-1}(\mathsf{Ver}(X_p)) \;\;\;\; \forall X_p \in T_p P.
        \end{aligned}
    \end{equation*}
   Note that $ \omega \in \Omega^{1}(p) \otimes \mathfrak{g}$,  where $\Omega^{1}(p)\otimes \mathfrak{g}$ denotes the space of $\mathfrak{g}$-valued one-forms on $P$. In other words, $\omega \in \Gamma (T^{*}P) \otimes \mathfrak{g} $. Then, for each $p\in P$ we have
    \begin{equation*}
        \boxed{\omega_p (X_p) := i^{-1}(\mathsf{Ver}(X_p))}
    \end{equation*}
    where
    \begin{equation*}
        i : \mathfrak{g} \xlongrightarrow{\cong} \Gamma(V_p P) \subset \Gamma(TP).
    \end{equation*}
    ($i$ maps $\mathfrak{g}$ isomorphically onto $\Gamma(VP)$). Call such form $\omega$ \textbf{the connection one-form.}
\end{definition}

\begin{observation}
    Conversely, given a connection one-form $\omega \in \Omega^{1}(P) \otimes \mathfrak{g}$, we can introduce a connection as follows:
    \begin{equation}
        \forall p \in P, \text{ set}\ \ H_{p} \coloneqq \ker(\omega_{p}), 
    \end{equation}
    where
    \begin{equation}
        \begin{aligned}
            \omega_{p}: T_{p}P &\longrightarrow \mathfrak{g} \\
            X_{p} &\longmapsto \iota^{-1}(\mathsf{Ver}(X_{p})).
        \end{aligned}
    \end{equation}
    So $X_{p} \in \ker(\omega_{p}) \iff \mathsf{Ver}(X_{p}) = 0$, thus gives the required distribution.
\end{observation}

\newpage
\begin{theorem}
    Let $P \xrightarrow{\pi} M$ be a principal $G$-bundle, and let $\omega$ be a connection one-form on $P$. Then:
    
    \begin{enumerate}
        \item For all $p \in P$ and $X \in \mathfrak{g}$,
        \[
            \omega_{p}(X^{\sharp}_{p}) = X,
        \]
        \begin{flushleft}
        where
        \[
        \begin{aligned}
            \mathfrak{g} \xrightarrow[\iota]{\cong} \Gamma(VP), \qquad &\omega_{p}: T_{p}P \to \mathfrak{g}, \\
            X \mapsto X^{\sharp} \qquad &Y \mapsto \iota^{-1}(\mathsf{Ver}(Y)) \\
            X^{\sharp}_{p} \in V_{p}P; \quad \iota^{-1}(\mathsf{Ver}(X^{\sharp}_{p})) &= \iota^{-1}(X^{\sharp}_{p}) \Rightarrow\; \omega_{p}(X^{\sharp}_{p}) = X.
        \end{aligned}
        \]
        \end{flushleft}
        \item For all $g \in G$ and $X_{p} \in T_{p}P$,
        \begin{equation}
            (\cdot g)^{*}\omega(X_{p}) = (\mathrm{Ad}_{g^{-1}})_{*}\,\omega_{p}(X_{p}),
        \end{equation}
        where
        \begin{equation}
            \begin{aligned}
                \mathrm{Ad}_{g}: G &\to G \\
                h &\mapsto g h g^{-1}, \\
                (\mathrm{Ad}_{g^{-1}})_{*}: \mathfrak{g} &\to \mathfrak{g}. \\[0.4em]
                &\rule{6cm}{0pt}
            \end{aligned}
        \end{equation}
        \begin{equation}
            \begin{tikzcd}[sep=1.5cm]
                P \arrow[r, thick, "\cdot g"] & P \arrow[d, thick, "\pi"] & {} & T_pP \arrow[r,thick, "\cdot g_*"] \arrow[rd, thick, dashed, swap, "\cdot g^*\omega"]& T_{p\cdot g}P \arrow[d, thick, "\omega_q"] \\ 
                {} & M & {} & {} & \mathfrak{g} 
            \end{tikzcd}
        \end{equation}
    \end{enumerate}
\end{theorem}

Let us look at \emph{the role of local sections in that setup:} 

\noindent Let 
\begin{equation}
    \begin{tikzcd}[sep = 1.5cm]
        P \arrow[r, thick, "\cdot G"] & P \arrow[d, thick, "\pi"] \\ 
        {} & M
    \end{tikzcd}
\end{equation}
be a principle G-bundle. Also let $\omega$ be a connection one-form, where $\omega$ is geometric object on the principal bundle $\omega \in \Omega^{1}(P)\otimes\mathfrak{g}, \;\mathfrak{g}\text{-valued one-form on } P$. Our aim is to obtain familiar geometric objects by pulling $\omega$ down to $M$. Given a local section $\sigma = U \subset M \to P, \;\pi \;\circ \sigma = \text{id}_{u}$.
Then we have,
\begin{enumerate}
    \item 
        \begin{equation}
            \begin{tikzcd}[sep=1.5cm]
                P \arrow[d, thick, "\pi"] & \boxed{\omega} \arrow[l, dashed, thick, no head] \\ 
                M \arrow[u, thick, dashed, "\sigma", bend left = 60]
            \end{tikzcd}
        \end{equation}
        $\omega \in \Omega^{1}(P)\otimes\mathfrak{g}$. Consider its pull back via $\sigma$: 
        \begin{equation}
            \sigma^{*}\omega \in \Omega^{1}(\mathnormal{u}) \otimes \mathfrak{g}
        \end{equation}
        Denote it by $\omega^{\mathnormal{u}} \coloneqq \sigma^{*}\omega$ called the \textbf{Yang-Mills field}. Via the local section, one can pull the geometric structure upstairs to downstairs to obtain a geometric object on the base manifold.
    \item Define "local" trivialization of $P$ by: (as we did before)
        \begin{equation}
            \begin{aligned}
                \varphi: U \cross G &\xrightarrow{\subseteq} P \\
                (x,g) &\to \sigma(x)\cdot g 
            \end{aligned}
        \end{equation}
        \begin{center}
            \includestandalone[width=0.50\textwidth]{Figures/Part1/p90.1}\\ 
            \hypertarget{fig26}{Figure 26.} Trivialization of $P$.
        \end{center}
    \item Let $\sigma_{1}: U_{1} \to P$ and $\sigma_{2}: U_{2} \to P$ be two local sections. Then both give a local trivial case where on the overlap $U_{1} \cap U_{2}$, we get:
    \begin{center}
        \includestandalone[width=0.66\textwidth]{Figures/Part1/p90.2}\\
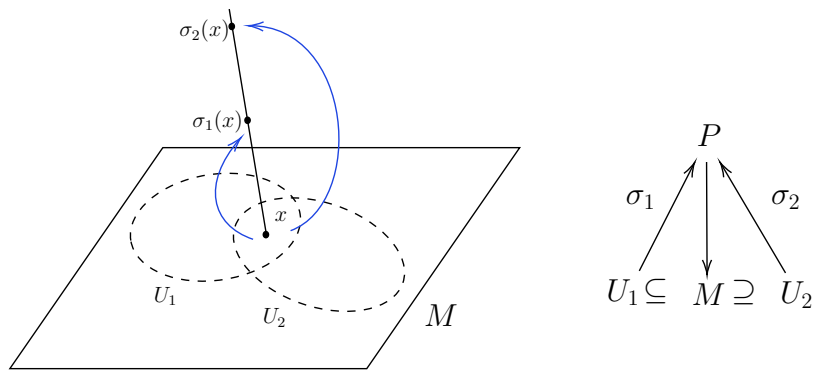
        \hypertarget{fig27}{Figure 27.} Trivialization of $P$ on an overlap $U_1\cap U_2$.
    \end{center}
    Note that both $\sigma_{1}(x)$ and $\sigma_{2}(x)$ lie on the same fiber over $x \in U\subseteq M,\; (U \coloneqq U_{1} \cap U_{2})$.
    Then for all $x$ in $U$ there exists $ g_{x} \in G$ such that $\sigma_{2}(x) = \sigma_{1}(x)\cdot g_{x}$ for fixed $\sigma_{1}$ and $\sigma_{2}$.

Define a map,
\begin{equation}
    \begin{aligned}
        U_{1} \cap U_{2} \xrightarrow{g} G \\
        x \to g_{x}
    \end{aligned}
\end{equation}
called the \textbf{gauge map}.
\end{enumerate}

\begin{example}
    When $G$ is a matrix Lie group, under the gauge transformations above, we have
    \begin{equation}
        \omega \mapsto g \cdot \omega \cdot g^{-1} + dg \dot g^{-1}
    \end{equation}
    where $\cdot$ denotes matrix multiplication.
\end{example}

\begin{example}
    Consider the fiber bundle, $LM \xrightarrow{\pi}M$, $M$ is a manifold where $\forall x \in M$, $L_{x}M = \qty{(c_{1},..,c_{n})|e_{1},..,e_{n} \; \text{forms a basis for} \; T_{x}M}$. Take  $G = GL(n,\mathbb{R})$ and
    \begin{equation}
        \begin{tikzcd}[sep = 1.5cm]
            LM \arrow[d, thick, "\pi"] \\ 
            M \arrow[u, thick, dashed, "\sigma", bend right = 60]
        \end{tikzcd}
    \end{equation}
    Any choice of local coordinate chart $x=(x^{i})_{i=1,..,n}$ for $M$ yields a local section $\sigma$:
    \begin{equation}
        \begin{aligned}
            \sigma :  U &\longrightarrow LM \\
            m &\longmapsto \left(\frac{\partial}{\partial x^{i}}|_{m},..,\frac{\partial}{\partial x^{n}}|_{m}\right).
        \end{aligned}        
    \end{equation} 
    Then the Yang-Mills field $\omega \coloneqq\sigma^{*}\omega$ is a one form on $U$ with the $\mathfrak{g}$-valued components
    \begin{equation}
        \Gamma^{i}_{jk} \coloneqq (\omega^{\mathnormal{u}})^{i}_{jk}, \quad i,k = 1,..,n.
    \end{equation}
    on the base manifold, we recover the $\Gamma$-coefficients.
\end{example}

\begin{definition}
    Let 
    \begin{equation}
        \begin{tikzcd}[sep = 1.5cm]
            P \arrow[r, thick, "\cdot G"] & P \arrow[d, thick, "\pi"]  & {\boxed\omega} \arrow[l, thick, dashed, no head]\\ 
            {} & M & {}
        \end{tikzcd}
    \end{equation}
    be a principal G-bundle, and $\omega$ a connection one-form on $P$.
    Let $\phi$ be a $A$- valued $k$-form on $P$. i.e.
    \begin{equation}
    \begin{aligned}
        \phi \in \Gamma(\Lambda^{k}T^{*}P)\otimes \mathfrak{g} = \Omega^{k}(P) \otimes A \\
        \phi : \underbrace{TP \cross ... \cross TP}_{\text{k-times}} \rightarrow A 
    \end{aligned}
    \end{equation}
    \newpage
    Then we define the \textbf{covariant exterior derivative} $D\phi \in \Omega^{k+1}(P) \otimes A$ as follows:
    $\forall x_{1},..,x_{k+1} \in TP$,
    \begin{equation}
        D\phi(x_{1},..,x_{k+1}) \coloneqq d\phi(\mathsf{Hor}(x_{1}),..,\mathsf{Hor}(x_{k+1})), 
    \end{equation}
    where $d\phi$ is the usual exterior derivative, and $\omega$ provides $\mathsf{Hor}(-)$,  and hence the splitting $TP = HP \oplus VP$.
\end{definition}

\begin{definition}
    Let 
    \begin{equation}
        \begin{tikzcd}[sep = 1.2cm]
            P \arrow[r, thick, "\cdot G"] & P \arrow[d, thick, "\pi"]  & \boxed{\omega} \arrow[l, thick, dashed, no head]\\ 
            {} & M & {}
        \end{tikzcd}
    \end{equation}
    be a principal $G$-bundle with a connection one-form $\omega$. The \textbf{curvature}, $\Omega$, of the connection is the $\mathfrak{g}$-valued 2-form on $P$ given by $\Omega \in \Omega^{2}(P) \otimes \mathfrak{g} \; \text{ such that } \;\Omega \coloneqq D\omega$, i.e.
    \begin{equation}
        \begin{aligned}
                \forall p \in P, \; \; \Omega_{p} : T_{p}P \times T_{p}P &\longrightarrow g \qquad (\Omega \in \Gamma(\Lambda^{2}T^{*}P)\otimes \mathfrak{g})\\
                (X_{p},Y_{p}) &\longmapsto \Omega_{p}(X_{p},Y_{p}),
        \end{aligned}
    \end{equation}
    where $\Omega(X,Y) = D\omega(X,Y)=d\omega(\mathsf{Hor}(X),\mathsf{Hor}(Y)). \; \forall X,Y \in T_{p}P$.
\end{definition}
\begin{claim}
    We can write the curvature two-form as $\boxed{\Omega = d\omega + \omega \wedge \omega}$, where 
    \begin{equation}\label{eq:omerstar1}
        \omega \wedge \omega(X,Y) \coloneqq [\omega(X),\omega(Y)]
    \end{equation}
    with $\omega(X),\omega(Y) \in g$ and $[.,.]$ is the Lie bracket on $g \cong T_{e}G$. 
    
    Now recall that $(T_{e}G,[.,.]) \cong (L\Gamma(TG),[.,.]_{0})$, where $\cong$ is a Lie algebra isomorphism, $[.,.]_{0}$ is the usual commutator bracket on $\Gamma(TG)$, and $L\Gamma(TG)$ space of left-invariant vector fields on $G$. Also, $L\Gamma(TG) \subseteq \Gamma(TG)$, which inverts a Lie bracket from $\Gamma(TG)$.
\end{claim}
\begin{proof}
    We then prove Eq.~\ref{eq:omerstar1} via the following:
    \begin{enumerate}
        \item[(i)] 
            \begin{equation}
                \begin{aligned}
                    T_{e}G &\xlongrightarrow[\cong]{\varphi} L\Gamma(TG) \\
                    A &\longmapsto \varphi(A), 
                \end{aligned}
            \end{equation}
            where $\forall g \in G, \; \varphi(A)_{g}\coloneqq \ell_{g_{*}}(A)$ (Moreover, for all $g \in G, \ell _{g}:G \xrightarrow{\text{diffeo}} G$ is a left translation.)
        \item[(ii)]
            \begin{equation}
                [A,B] \coloneqq [\varphi(A),\varphi(B)]_{0}|_{e} \quad \forall A,B \in T_{e}G.
            \end{equation}
            Here $[\eta,\xi]_{0}=\eta\xi - \xi\eta \quad \forall \xi,\eta \in \Gamma(TG)$.
    \end{enumerate}
\end{proof}

\begin{remark}
    When $G$ is a matrix Lie group, the curvature components are: $\Omega^{i}_{j}=d\omega^{i}_{j}+\omega^{i}_{k} \wedge \omega^{k}_{j}$ where each $\omega$ is 1-form on $P$.
\end{remark}

\begin{proof}
    Note that
    \begin{enumerate}
        \item[(1)]It is clear that $\Omega$ is a $C^{\infty}(P)$-module.
        \item[(2)] $\omega$ yields a splitting /distribution: $TP = VP \oplus HP$ where $HP \coloneqq \text{ker}(\omega)$ and $\forall p\in P$, 
        \begin{equation}
            \begin{aligned}
                \omega_{p}: T_{p}P &\longrightarrow g  \\
                X_{p} &\longmapsto \iota^{-1}(\mathsf{Ver}(X_{p})), \\
            \end{aligned}
        \end{equation}
        where 
        \begin{equation}
            \begin{aligned}
                \;\iota : g &\xlongrightarrow{\cong}\Gamma(VP)\\
                X &\longmapsto X^{\sharp} \; (\text{Lie algebra homomorphism}).
            \end{aligned}
        \end{equation}
    \end{enumerate}
    Consider now, case by case:
    \begin{itemize}
        \item[(a)] Let $X,Y \in \Gamma(TP)$ be vertical. Then there exists $A,B \in \mathfrak{g}$ such that $X=A^{\sharp}$, $Y=B^{\sharp}$.
        Observe that the left-hand side is
        \begin{equation}
        \begin{aligned}
            \Omega^i_j = \Omega(X,Y) &= \Omega(A^{\sharp},B^{\sharp}) \\
            &= D\omega(A^{\sharp},B^{\sharp})=d\omega(\mathsf{Hor}(A^{\sharp}),\mathsf{Hor}(B^{\sharp})) = 0
        \end{aligned}
        \end{equation}
        
        Since both $\mathsf{Hor}(A^{\sharp})=\mathsf{Hor}(B^{\sharp})=0$. Similarly, the right-hand side is
        \begin{equation}
        \begin{aligned}
            d \omega^i_j + \omega^j_k\wedge \omega^k_j &= \omega \wedge \omega (X,Y) + d\omega(X,Y) = d\omega(A^{\sharp},B^{\sharp}) + [\omega(A^{\sharp}),\omega(B^{\sharp})] \\
            &= \underbrace{XA^{\sharp}(\omega(B^{\sharp})) - B^{\sharp}(\omega(A^{\sharp})) - \omega([A^{\sharp},B^{\sharp}]_{\circ})+[A,B]}_{(\star)}.
        \end{aligned}
        \end{equation}
        Here $\omega(A^{\sharp})=A$ and $\omega(B^{\sharp})=B$, and $[\cdot,\cdot]$ is Bracket on $\mathfrak{g}$. Then we further get

        \begin{equation}
            \begin{aligned}
               (\star)  &= A^{\sharp}(B) - B^{\sharp}(A) - \omega([A^{\sharp},B^{\sharp}]_{0})+[A,B]\\ &= A^{\sharp}B-B^{\sharp}A - \omega([A,B]^{\sharp})+[A,B] \\
                &=-[A,B]+[A,B] = 0,
            \end{aligned}
        \end{equation}where $A,B \in \mathfrak{g}$, and $[A,B]^{\sharp}=[A^{\sharp},B^{\sharp}]_{0}$ since $\iota$ is a Lie algebra homomorphism.
        \item[(b)] Assume both $X,Y\in \Gamma(TP)$ are horizontal. Then, the left-hand side is 
        \begin{equation}
            \Omega(X,Y)= D\omega(X,Y)=d\omega(\mathsf{Hor}(A^{\sharp}),\mathsf{Hor}(B^{\sharp}))=d\omega(X,Y),
        \end{equation}
        and the right-hand side is
        \begin{equation}
            d\omega(X,Y)=+\omega \wedge \omega (X,Y) = d\omega(X,Y) + [\omega(X),\omega(Y)]_{g}=d\omega(X,Y),
        \end{equation}
        where $X,Y \in \ker{(\omega)}$; hence, $[\omega(X),\omega(Y)]=0$.
        \item[(c)] W.l.o.g let $X=HP, \; Y\in VP$. Then  the left-hand side is 
        \begin{equation}
            \Omega(X,Y) = d\omega(\mathsf{Hor}(A^{\sharp}),\mathsf{Hor}(B^{\sharp}))=0
        \end{equation}
        since hor$(X)=X$ and hor$(Y=0)$. Similarly, the right-hand side is 
        \begin{equation}
            d\omega(X,Y) + \omega \wedge \omega(X,Y) = \underbrace{d\omega(X,Y) + [\omega(X),\omega(Y)]_{q}}_{(\star')}
        \end{equation}
        where $\omega(X)=0$ since $X \in \ker{(\omega)}=HP$. So we get
        \begin{equation}
          (\star')  = X\omega(Y)-Y\omega(X)-\omega([X,Y]_{0}).
        \end{equation}
        Here $\omega(A^{\sharp})=A$, and $\omega(Y) = A (\in \mathfrak{g}) = 0$, $\omega(X)=0$ and $[X,Y]_{0} = XY-YX$. So we have
        \begin{equation}
        \begin{aligned}
          (\star')  &=- \omega([X,A^{\sharp}]_{0}) \\
            &= 0
        \end{aligned}        
        \end{equation}
        $\therefore \quad \Omega= d\omega + \omega \wedge \omega$
    \end{itemize} 
\end{proof}
\begin{observation}
    Let 
    \begin{equation}
        \begin{tikzcd}[sep = 1.2cm]
            P \arrow[r, thick, "\cdot G"] & P \arrow[d, thick, "\pi"]  & \boxed{\omega} \arrow[l, thick, dashed, no head]\\ 
            {} & M & {}
        \end{tikzcd}
    \end{equation}
    be a principal G-bundle equipped with a connection one-form $\omega$. Let $\sigma: U \to P$ be a local section. We saw before,
    \begin{enumerate}
        \item[(1)] $\Gamma \coloneqq \sigma^{*}\omega \; \in \; \Omega^{1}(M) \otimes\mathfrak{g}$ is a $\mathfrak{g}$-valued one-form on $M$, which we called the "Yang-Mills field" (with the gauge fields $A$).
        \item[(2)] $R \coloneqq \sigma^{*}\omega \; \in \; \Omega^{2}(M) \otimes \mathfrak{g}$ (or just $F$)
        \begin{enumerate}
            \item When $G=U(1)$; $F$ is the electromagnetic field tensor.
            \item When $G=GL(\dim{(M)},\mathbb{R})$; $F$ is the usual Riemann tensor, say $R$.
        \end{enumerate}
        On $P$, we have $\Omega= d\omega+\omega \wedge\omega$.
        But on $M$ (recovering the Riemann geometry when $G=GL$), one computes
        \begin{equation}
        \begin{aligned}
            \sigma^{*}\Omega &= \sigma^{*}(d\omega+\omega \wedge\omega) = \sigma^{*}(d\omega) + \sigma^{*}(\omega \wedge \omega) \\
            &= \sigma^{*}(d\omega) + \sigma^{*}(\omega) \wedge \sigma^{*}(\omega) \quad (\text{Let }\Gamma \coloneqq \sigma^{*}\omega)
        \end{aligned}
        \end{equation}
        We get the familiar expression: $R = d\Gamma + \Gamma \wedge \Gamma = \partial \Gamma^{\cdot\cdot}_{\cdot\cdot}-\Gamma^{\cdot\cdot}_{\cdot\cdot}+\Gamma^{\cdot\cdot}_{\cdot\cdot}\Gamma^{\cdot\cdot}_{\cdot\cdot}-\Gamma^{\cdot\cdot}_{\cdot\cdot}\Gamma^{\cdot\cdot}_{\cdot\cdot}$, where $\Gamma^{\cdot\cdot}_{\cdot\cdot}$ are the \emph{Christoffel symbols} in the components $R^{i}_{ijk}$ of the Riemann curvature tensor.
    \end{enumerate}
\end{observation}

\subsection{Covariant Derivative on Associated Vector Bundles}
Recall the setup: Let
\begin{equation}
    \begin{tikzcd}[sep=1.2cm]
        P \arrow[r, thick, "\cdot G"] & P \arrow[d, thick, "\pi"] \\ 
        {} & M
    \end{tikzcd}
\end{equation}
be a principal G-bundle. We defined the following geometric objects on $P$:
\begin{equation}
    \begin{tikzcd}[sep=1.2cm]
        {} & {} & \boxed{\omega} \arrow[ld, thick, dashed, no head] \\ 
        P \arrow[r, thick, "\cdot G"] & P \arrow[d, thick, "\pi"] & {\boxed\Omega} \arrow[l, thick, dashed, no head] \\ 
        {} & M & {}
    \end{tikzcd}
\end{equation}
\begin{enumerate}
    \item Given a connection on $P$, i.e., the splitting/distribution on $TP: TP = HP \oplus VP$, where $HP$ is the horizontal bundle and $VP$ is the vertical bundle defined as $ V_{p}P \coloneqq \qty{X_{p} \in TP: \; \Pi _{*,p}(X_{p})=0 } = \ker{(\Pi_{*})} \ \forall p \in P.$
    \item There exists a natural map
    \begin{equation}
        \begin{aligned}
            \iota: \mathfrak{g} \doteq T_{e}G \xlongrightarrow{\cong} \Gamma(VP) \subset \Gamma(TP), \ X \longmapsto X^{\#},
        \end{aligned}
    \end{equation}
    where, $\forall p \in P ,  \forall f  \in C^{\infty}(P)$, we define 
    \begin{equation}
        X^{\sharp}_{p}f = \frac{d}{dt}\bigg\vert_{t=0} f(p \cdot \exp(tX)).
    \end{equation}
    Here $\cdot$ is the right $G$-action, and $\{ \exp(tX) \}$ is the one-parameter group of diffeomorphisms generated by $X\in \mathfrak{g}$.
    \item Then we define: $\forall p \in P$,
    \begin{equation}
     \begin{aligned}
        \omega_{p}: T_{p}P &\longrightarrow \mathfrak{g} \\
        X & \longmapsto \iota^{-1}(\mathsf{Ver}(X)).
     \end{aligned}      
    \end{equation}
    I.e., $\omega \in \Omega^{1}(P)\otimes \mathfrak{g}$.
    \item Conversely, given a connection one-form $\omega \in \Omega^{1}(P)\otimes \mathfrak{g}$, we define a connection by letting $H_{P}P \coloneqq\ker{(\omega)} \; \forall p\in P$.
    \item Curvature two form on $P$ $\Omega \in \Omega^{2}(P)\otimes \mathfrak{g}$ is defined by
    \begin{equation}
        \Omega = D\omega = d\omega \; \circ \mathsf{Hor}
    \end{equation}
    I.e., $\forall p\in P, \;\forall X,Y \in T_{P}(P), \ \; \Omega(X,Y)=D\omega(X,Y)=d\omega(\mathsf{Hor}(X),\mathsf{Hor}(Y))$
    \item We showed that $\Omega = d\omega + \omega \wedge \omega$, where $\omega\wedge\omega \coloneqq [\omega(X),\omega(Y)]_{\mathfrak{g}}$.
\end{enumerate}

    Consider the data 
    \begin{equation}
        \begin{tikzcd}[sep=1.2cm]
            {} & {} & \boxed{\omega} \arrow[ld, thick, dashed, no head] \\ 
            P \arrow[r, thick, "\cdot G"] & P \arrow[d, thick, "\pi"] & \boxed{\Omega} \arrow[l, thick, dashed, no head] \\ 
            {} & M & \boxed{D} \arrow[lu, thick, dashed, no head]
        \end{tikzcd}
    \end{equation}
   such that  $\omega\in \Omega^{1}(P) \otimes \mathfrak{g}, \; \Omega \in \Omega^{2}(P) \otimes \mathfrak{g}, \text{ and }D\phi \in \Omega^{k+1}(P) \otimes A$  defined by $D\phi = d\phi \circ \mathsf{Hor}$, where $\phi \in \Omega^{k}(P) \otimes A$ and 
    \begin{equation}
        \begin{aligned}
            \mathsf{Hor}: \Gamma(TP) &\longrightarrow \Gamma (HP) \\
            X &\longrightarrow \mathsf{Hor}(X).
        \end{aligned}
    \end{equation}
    
Let $F$ be a "vector" space equipped with a  left $G$-action. Then consider an associated "vector" bundle $P_F \xlongrightarrow{\pi_F}M$ which consists of the following data:
\begin{itemize}
    \item $P_{F} =P \cross F / \sim$ where $(p,f) \sim (p',f')$ if and only if there exists a $g \in G$ such that $p'=p\cdot g$ (as $P$ is a right $G$-space) and $f'=g^{-1}\cdot f$ via the right $G$-action on $F$.
    \item $\pi_{F}([p,f]) = \pi(p)$ $\forall[p,f] \in P_{F}$.
\end{itemize}

Since each fibre is a vector space, i.e., a linear space,
    \begin{center}
        \includestandalone[width=0.4\textwidth]{Figures/Part1/p98} \\
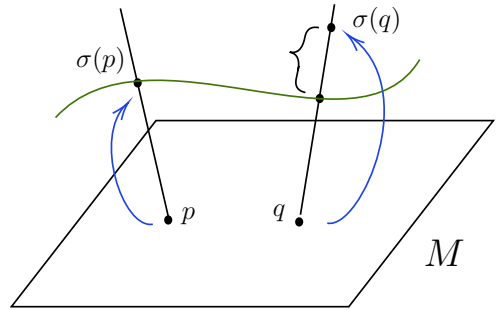
        \hypertarget{fig28}{Figure 28.} Change in the section given a direction.
    \end{center}
it makes sense to discuss \emph{how a section $\sigma$ changes in the direction of some $X \in TP$.}

Our aim is then to construct the so-called "covariant derivative" notion:
    \begin{equation}
        \nabla_{X}\sigma \in \Gamma(P_{F}),
    \end{equation}
    where $X \in TM$, $\sigma \in \Gamma(P_{F})$ such that
    \begin{itemize}
        \item $\nabla_{(fX+Y)}\sigma = f\nabla_{X}\sigma+\nabla_{Y}\sigma \quad \forall  f \in C^{\infty}(M)$
        \item $\nabla_{X}(\sigma+\psi)=\nabla_{X}\sigma+\nabla_{X}\psi$
        \item $\nabla_{X}f\sigma = f \nabla_{X}\sigma + (X_{f})\sigma$
    \end{itemize}

The question is how to construct such an object. 

We want to realize $\sigma \in \Gamma(P_{F})$ as a $F$-valued function $S$ on $M$, i.e.,
\begin{equation}
    \begin{aligned}
        S : U \subseteq M \longrightarrow F \\
        \text{with} \quad \pi_{F} \circ S = \textit{id}_{u}
    \end{aligned}
\end{equation}

\begin{itemize}
    \item[STEP 1:] Using the "correspondence" between $P$ and $P_{F}$, pass from $P$ to $P_{F}$, i.e., given $\sigma \in \Gamma(P_{F})$, we consider the map $\phi_{\sigma}: P \to F$.
    To this end, recall the following theorem:
\end{itemize}

\begin{theorem}
The space of sections $\sigma:M \to P_{F}$ of an associated bundle is in one-to-one correspondence with that of $F$-valued $G$-equivariant maps on $P$: Namely, those maps 
\begin{equation}
    \phi: P \to F
\end{equation}
with the commuting diagram
    \begin{equation}
    \begin{tikzcd}[sep = 1.2cm, remember picture]
        |[alias = A]| P \arrow[r, thick, "\phi"] \arrow[d, thick, "\cdot G"] & 
        |[alias = B]| F \arrow[d, thick, "\cdot G"] \\ 
        |[alias = C]| P \arrow[r, thick, "\phi"]& 
        |[alias = D]| F
    \end{tikzcd}
    \tikz[overlay,remember picture]{%
        \node (Y1) [scale=1.5] at (barycentric cs:A=1,B=1,C=1,D=1) {$\circlearrowright$};
        }
\end{equation}

I.e., $\forall g \in G$, $\forall p \in P$, we have
\begin{equation}
    \phi(p\cdot g) = g^{-1} \cdot \phi(p) 
\end{equation}
\end{theorem}
    
\begin{itemize}
    \item[STEP 2:] Define it for $\phi_{\sigma}$ and define everything on $P$ (i.e., $\forall X \in T_pP$).
    \item[STEP 3:] Pull everything down to $M$. Define it on $M$ (i.e., $\forall Y \in T_{P}M$).
\end{itemize}

\begin{proof}
    ("$\Leftarrow$") Given a $G$-invariant map $\phi: P \to F$,   we construct the corresponding section $\sigma_{P}: U\subseteq M \to P_{F}$ as follows:
    
    \begin{equation}
        \begin{tikzcd}[sep = 1.5cm]
            P_F \arrow[d, thick, "\pi_F"] \\ 
            M\ni U \arrow[u, thick, dashed, "\sigma_\phi", bend right = 60, swap]
        \end{tikzcd}
    \end{equation}
    
    \begin{equation}
        \begin{aligned}
            \sigma_{\phi}: U\subseteq M &\to P_{F} = P \cross F / \sim \\
            x &\mapsto [p,\phi(p)] \quad \text{where} \; p \in \pi^{-1}(X)
        \end{aligned}
    \end{equation}
    i.e., $\forall X \in U \subseteq M$, $\sigma_{\phi} \coloneqq [p,\phi(p)]$, for any $p\in \pi^{-1}(X)$

    \begin{center}
        \includestandalone[width = 0.3\textwidth]{Figures/Part1/p99}
    \end{center}
    
    \begin{itemize}
        \item[$i)$] $\sigma_{\phi}$ is well defined: Take another $q\in \pi^{-1}(x)$. Since $p$ and $q$ lie in the same fiber over $X$, there exists a unique $g \in G$ such that $q = p \cdot g$
        Then we have,
        \begin{equation}
            [q,\phi(q)] = [ p \cdot g, \phi( p \cdot g)] = [ p \cdot g, g^{-1}\cdot \phi(p)] = [p,\phi(p)],
        \end{equation}
        where the second equality holds due to $G$-invariance and the third equation holds by the definition of "$\sim$".
        \item[$ii)$] $\pi_{F} \circ \sigma_{\phi}= \text{id}_{u}$: let $X \in U \subseteq M$. Then we have,
        \begin{equation}
            \begin{aligned}
                (\pi_{F} \circ \sigma_{\phi})(x) &= \pi_{F}([p,\phi(p)]) \quad \text{for} \; p \in \pi^{-1}(x) \\
                &= \pi(p) \quad \text{by definition of} \; \pi_{F} \\
                &= x.
            \end{aligned}
        \end{equation}
    \end{itemize} 
    Therefore, for a given $\phi$ above, we construct the corresponding section $\sigma_{\phi}:U\subseteq M \to P_{F}$.\\

    ("$\Rightarrow$") Given a section $\sigma:U\subseteq M \to P_{F}$. $(\pi_{F} \circ \sigma_{\phi}= \text{id}_{u})$ 
    
    We construct a $G$-invariant map $\phi_{\sigma}: P \to F$ as follows.
    \begin{itemize}
        \item[(1)] Define $\phi_{\sigma}:P\to F$ as 
        \begin{equation}
            \begin{tikzcd}[sep = 1.5cm]
               P_F \arrow[d, thick, "\pi_F"] &  P \arrow[d, thick, "\pi"] \\ 
               M \arrow[u, thick, dashed, "\sigma", bend right= 60] &  M
            \end{tikzcd}
        \end{equation}
        \begin{equation}
            \forall p \in P; \; \phi_{\sigma} \coloneqq \iota^{-1}(\sigma(\pi(p))) = f_{p}, \quad \text{where} \; \sigma(\pi(p)) = [p,f_{p}].
        \end{equation}
        \item[(2)] $\phi_{\sigma}$ is $G$-invariant: For this, we need to show that 
        \begin{equation}
            \begin{tikzcd}[sep = 1.5cm, remember picture]
                |[alias = A]| P \arrow[r, thick, "\phi_\sigma"] \arrow[d, thick, "\cdot G", swap] & 
                |[alias = B]| F \arrow[d, thick, "\cdot G"] \\ 
                |[alias = C]| P \arrow[r, thick, "\phi_\sigma"]& 
                |[alias = D]| F
            \end{tikzcd}
            \tikz[overlay,remember picture]{%
                \node (Y1) [scale=1.5] at (barycentric cs:A=1,B=1,C=1,D=1) {$\circlearrowright$};
                }
    \end{equation}
    commutes.
    \end{itemize}
    Observe that
    \begin{equation}
        \begin{aligned}
            \phi_{\sigma}(p\cdot g) &= \iota^{-1}_{p \circ g}(\sigma(\pi(p \cdot g))) \\
            &= \iota^{-1}_{p \circ g}(\sigma(\pi(p))) \quad (\text{since $p$ and $g$ are in the same fiber on} \; \Pi(p).) \\
            &= \iota^{-1}_{p \circ g}([p,f_{p}]) \\
            &= \iota^{-1}_{p \circ g}([p \cdot g, g^{-1} \cdot f_{p}]) \quad (\text{by definition of} \sim) \\
            &= g^{-1} \cdot f_{p} \quad (\text{by definition of} \;\phi_{\sigma}) \\
            &= g^{-1} \cdot \phi_{\sigma}(p).
        \end{aligned}
    \end{equation}
    Therefore, $ \quad \phi_{\sigma}(p\cdot g) = g^{-1} \cdot \phi_{\sigma}(p)$.
    
    Finally, we need to check that no information is lost:
    \begin{itemize}
        \item[(a)] Given $\sigma$, check $\sigma_{\phi_{\sigma}} = \sigma$
        \begin{equation}
            \begin{tikzcd}[row sep = 0.2 cm, column sep = 1.5cm]
                \sigma \arrow[dd, thick, equal, "?", swap] \arrow[rd, thick, bend left = 10] & {} \\ 
                {} & \phi_\sigma \arrow[ld, thick, bend left = 10] \\
                \sigma_{\phi_\sigma} & {}
            \end{tikzcd}
        \end{equation}
        \item[(b)] Given $\phi$, check $\phi_{\sigma_{\phi}} = \phi$
        \begin{equation}
            \begin{tikzcd}[row sep = 0.2 cm, column sep = 1.5cm]
                \phi \arrow[dd, thick, equal, "?", swap] \arrow[rd, thick, bend left = 10] & {} \\ 
                {} & \sigma_\phi \arrow[ld, thick, bend left = 10] \\
                \phi_{\sigma_\phi} & {}
            \end{tikzcd}
        \end{equation}
    \end{itemize}
    Indeed,
     \begin{itemize}
        \item[(a)] Let $x \in M$;
        \begin{equation}
            \begin{aligned}
                \sigma_{\phi_{\sigma}} &= [p,\phi_{\sigma}(p)] \quad \text{for any} \; p \in \pi^{-1}(x) \\
                &= [p,f_p] \quad f_{p} \in F \\
                &= \sigma(\pi(p)) \quad \text{by definition and} \; \pi(p) = x \\
                &= \sigma(x).
            \end{aligned}
        \end{equation}
        \item[(b)] Let $p\in P$. Then we have,
        \begin{equation}
            \begin{aligned}
                \phi_{\sigma_{\phi}}(p) &= \iota^{-1}_{p}(\sigma_{\phi}(\Pi(p))) \\
                &= \iota^{-1}_{p}([p,\phi(p)]) \\
                &= \phi(p).
            \end{aligned}
        \end{equation}
    \end{itemize}
\end{proof}
\newpage
Here, given $\sigma: U \subseteq M \to P_{F}$, with $\pi_{F} \circ \sigma = \text{id}_{U}$, we have constructed
    \begin{itemize}
        \item[($i$)] 
        \begin{equation}
            \phi_{\sigma}: P \to F
        \end{equation} 
        where $\phi_{\sigma}(p)= \iota^{-1}_{p}(\sigma(\pi(p)))$, which is $G$-equivariant.
        
        (Later, we shall see the "pull-back" of $\phi_{\sigma}$ makes much more sense! At this stage, it is a sort of an "abstract" field on $P$ that will be used to "produce" more familiar notions/fields on $M$.)

        Notice also that if $x\in U$, $\sigma(x)=[p,f]$ for some representatives $p\in P$ and $f\in F$. Then $\phi_{\sigma}$ above can be viewed as a tool providing a communication between $P$ and the fibre space $F.$ Using that tool, $\sigma$ will indeed be realized as a $F$-valued map $s:U \rightarrow F$ on $ M$ satifying $\pi_{F} \circ \sigma = \text{id}_{U}$. For example, when $F$ is a vector space, say $F:=T_xM$, $s$ is the map that assigns to each $x\in U$ a unique vector $s(x) \in F$.

        \item [($ii$)] Given a $G$-equivariant $\phi: P \to F$, we construct a section
        \begin{equation}
            \sigma_{\phi}: M \to P_{F},
        \end{equation}
        where $\sigma_{\phi}(X) = [p,\phi(p)]$ for some $p\in \pi^{-1}(X)$.
    \end{itemize}


Here we can make an important observation:
\begin{observation}
    Let 
    \begin{equation}
        \begin{tikzcd}[sep=1.2cm]
            P \arrow[r, thick, "\cdot G"] & P \arrow[d, thick, "\pi"] & \boxed{\omega} \arrow[l, thick, dashed, no head] \\ 
            {} & M & {}
        \end{tikzcd}
    \end{equation}
    be a principle $G$-bundle and  $\omega$ a connection one-form on $P$.

    Let $\phi: P \to F$ be a $G$-equivariant map, where $F$ is a vector space with a left $G$-action.

    I.e., $\phi: P \to F$ is a map such that  the diagram
    \begin{equation}
            \begin{tikzcd}[sep = 1.5cm, remember picture]
                |[alias = A]| P \arrow[r, thick, "\phi"] \arrow[d, thick, "\cdot G", swap] & 
                |[alias = B]| F \arrow[d, thick, "\cdot G"] \\ 
                |[alias = C]| P \arrow[r, thick, "\phi"]& 
                |[alias = D]| F
            \end{tikzcd}
            \tikz[overlay,remember picture]{%
                \node (Y1) [scale=1.5] at (barycentric cs:A=1,B=1,C=1,D=1) {$\circlearrowright$};
                }
    \end{equation} commutes:
    $\phi(p\cdot g) = g^{-1} \cdot \phi(p)$    $\forall p \in P,$ $\forall g \in G$
\end{observation}

Consider the one-parameter group of diffeomorphisms $\qty{\exp(tX) \;| \;X \in g}$. Then
the $G$-equivarience of $\phi$ implies $\underbrace{\phi(p \cdot \exp(tX)) = \exp(-tX) \cdot \phi(p)}_{(\star)}$

It follows from the equation \((\star)\) above that we obtain
\[
\underbrace{
\left.\frac{d}{dt}\right|_{t=0}\phi\bigl(p\cdot \exp(tX)\bigr)
}_{%
\substack{
\text{(we saw before)}\\[3pt]
\eqqcolon\, X^{\sharp}_{p}\phi \\[6pt]
\begin{aligned}
\iota:\; \mathfrak{g} &\to\Gamma(VP) \\
        X&\mapsto X^{\sharp}
\end{aligned}
}}
\;=\;
\underbrace{
\left.\frac{d}{dt}\right|_{t=0}\exp(-tX)\cdot \phi(p)
}_{\substack{
= -\,X\cdot \phi(p)\\[4pt]
X \in \mathfrak{g},\; \phi(p)\in F,\\
\text{and}\; \phi \text{ is an $F$-valued function},\\
\cdot\text{ is the left action of }\mathfrak{g}\text{ on }F\\
\text{induced by the left }G\text{-action on }F
}}
\]

Then we have
\begin{equation}
    X_{p}^{\sharp} \phi = - X \cdot \phi(p).
\end{equation}
Or equivalently,
\begin{equation}
    \forall p \in P, \quad d\phi(X_{p}^{\sharp}) = -\omega(X_{p}^{\sharp})\phi(p) \quad \text{where} \; X_{p}^{\sharp} \in V_{p}P
\end{equation}

In sum,  we simply showed:
\begin{equation}
   \boxed{d\phi(X) + \omega(X)\phi = 0 \quad \forall X \in \Gamma(VP)} \label{eq:omerstar}
\end{equation}

Before we proceed, it should be noted here that
    \begin{itemize}
        \item[1)] 
        \begin{equation}
            \begin{aligned}
                \omega_{p}: T_{p}P &\to \mathfrak{g} \\
                X_{p} &\mapsto \iota^{-1}(\mathsf{Ver}(X_{p}))
            \end{aligned}
        \end{equation}
        \item[2)] $\omega(X^{\sharp})=X \in \mathfrak{g}$ (note that $\iota(X)=X^{\sharp}$)
    \end{itemize}

If we insert a "vertical" vector field into $\omega$, then by using $\iota$, one can recover the element in $\mathfrak{g}$ from which the vertical field is obtained. I.e., let $Y\in V_{p}P \subset T_{p}P$ be a vertical vector at $p \in P$. Then we have:
    \begin{itemize}
        \item[(a)] By $\iota: \mathfrak{g} \xlongrightarrow{\cong } \Gamma(VP), \; \exists A \in \mathfrak{g}$ such that $\iota(A) = A^{\sharp}=Y$.
        \item[(b)] By definition of $\omega$, $\omega(Y)=\omega(A^{\sharp})=A$ (as $\omega = \iota^{-1})$.
    \end{itemize}

Now, we give a proposition generalizing the situation in Eqn. \ref{eq:omerstar} to any vector fields on $P$ --- not just vertical ones---: 
\begin{proposition}
For all $X_{p} \in T_{p}P$, we have
    \begin{equation}
        D\phi(X_{p}) = d\phi (X_{p}) + \omega(X_{p})\phi(p) ,
    \end{equation}
    where $\phi: P\to F$ as above and $D_{X}\phi \coloneqq D\phi(X)$.
\end{proposition}

\newpage
\begin{proof}
    To prove this, consider case by case,
    \begin{itemize}
        \item[(a)] if $X\in V_{p}P$ (i.e. hor$X$=0), then, the right-hand side is zero by Eq.~\ref{eq:omerstar}. The left-hand side has LHS = $D\phi(X) = d\phi(\mathsf{Hor}(X)) = 0$ , where the first equality holds by definition of exterior covariant derivative, and hor$(X) =0 $
        \item[(b)] If $X \in H_{p}P$, then,
            \begin{itemize}
                \item[$\cdot$] RHS = $d\phi(X) + \omega(X)\phi = d\phi(X)$, where $\omega(X) = 0$ since $X \in H_{p}P = \ker(\omega)$.
                \item[$\cdot$] LHS = $D\phi(X) = d\phi(\mathsf{Hor}(X)) = d\phi(X)$, because $X$ is already horizontal.
            \end{itemize}
    \end{itemize}
\end{proof}

What we have done so far is that we completed step 1 and step 2:
\begin{itemize}
    \item 
    \begin{equation}
       \begin{tikzcd}[sep=1.2cm]
        P_F \arrow[d, thick, "\pi_F"] \\ M \arrow[u, thick, dashed, "\sigma", bend right=50, swap]           
       \end{tikzcd} \qquad \xlongrightarrow{\text{by theorem}} \qquad \begin{tikzcd}[sep=1.2cm]
            {} & {} & \boxed{\omega} \arrow[ld, thick, dashed, no head] \\ 
            P \arrow[r, thick, "\cdot G"] & P \arrow[d, thick, "\pi"] & \boxed{\Omega} \arrow[l, thick, dashed, no head] \\ 
            {} & M & \boxed{D} \arrow[lu, thick, dashed, no head]
        \end{tikzcd} 
    \end{equation}
    
    That is, given a section $\sigma \in \Gamma(P_{F})$, we consider the corresponding map 
    $\phi_{\sigma}: P \to F$, which is $G$-equivariant.

    \item Let $X \in \Gamma(TP)$. Then, by Proposition, we defined the covariant derivative of $\phi_{\sigma}$ \underline{on $P$} in the direction of $X$:
    \begin{equation}
        D_{X}\phi_{\sigma} \coloneqq D\phi_{\sigma}(X) = d\phi_{\sigma}(X) + \omega(X)\phi_{\sigma}
    \end{equation}

    \item Still, it remains to define the covariant derivative \underline{on $M$}. To this end, we need to \emph{pull everything we constructed down onto $M$}:
    \begin{equation}
        \begin{tikzcd}[column sep = 0.5cm, row sep = 1.5cm]
            \boxed{\phi_\sigma: P \to F} \arrow[rd, dashed, no head] & {} & \boxed{\omega} \arrow[dashed, ld, no head] & {} \\ 
            {} & P_F \arrow[d, thick, "\pi_F"] & \boxed{D_x} \arrow[l, dashed, no head] \\ {} & M \arrow[u, thick, dashed, "\sigma", bend left=50, swap] & {}
        \end{tikzcd}
    \end{equation}
    \newpage
    Let $\phi: U \subseteq M \to P$ be a local section. Then we have:
    \begin{itemize}
        \item[$i)$] $\varphi^{*}\phi_{\sigma} = \phi_{\sigma} \circ \varphi: U \to F$ ($F$-valued map on $M$) denoted also by $s\coloneqq \varphi^{*}\phi_{\sigma}$.
        \begin{equation}
            \begin{tikzcd}[sep = 1.5cm, remember picture]
                |[alias = A]| U \arrow[r, thick, "\varphi"] \arrow[rd, dashed, "\text{pull-back}", swap] & |[alias = B]| P \arrow[d, thick, "\phi_\sigma"] \\ 
                {} & |[alias = C]| F & 
            \end{tikzcd}
            \tikz[overlay,remember picture]{%
                \node (Y1) [scale=1.5] at (barycentric cs:A=1,B=1,C=1) {$\circlearrowright$};
                }
    \end{equation}

        \item[$ii)$] $\varphi^{*}\omega \in \Omega^{1}(M) \otimes \mathfrak{g}$ denoted also by $\omega^{\mathnormal{u}}\coloneqq\varphi^{*}\omega$, which is Yang-Mills field.
        \item[$iii)$] $\varphi^{*}(D\phi_{\sigma}) \in \Omega^{1}(M) \otimes F$, where $D\phi_{\sigma} \in \Omega^{1}(M) \otimes F$.
    \end{itemize}
\end{itemize}

Then observe that for any $Y\in \Gamma(TM)$, we have 
\begin{equation}
    \begin{aligned}
        \underbrace{\varphi^{*}(D\phi_{\sigma})(Y)}_{\eqqcolon \nabla_{Y}s} &= \varphi^{*}(d\phi_{\sigma}+\omega\cdot\phi_{\sigma})(Y) \\
        &= \varphi^{*}d\phi_{\sigma}(Y) +\varphi^{*}(\omega\cdot\phi_{\sigma})(Y)  \\
        &= d\underbrace{(\varphi^{*}\phi_{\sigma})}_{\eqqcolon s}(Y) +\underbrace{(\varphi^{*}\omega)}_{\omega^{\mathnormal{u}} \in \Omega^{1}(M) \otimes\mathfrak{g}}(Y)\underbrace{\varphi^{*}\phi_{\sigma}}_{s},
    \end{aligned}
\end{equation}
where $s: U \subseteq M \to F$ is as above. Finally, we have
\begin{equation}
   \boxed{ \nabla_{Y}s = ds(Y)+\omega^{\mathnormal{u}}(Y)s.}
\end{equation}
    \begin{center}
        \includestandalone[width = 0.4\textwidth]{Figures/Part1/p105.1}\\ 
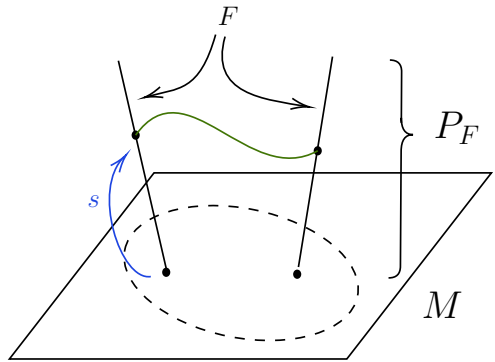
        \hypertarget{fig29}{Figure 29.} Local picture of a mapping $s:U\rightarrow F$
    \end{center}
But we still need to check that such an object, in fact, satisfies the conditions of being a covariant derivative that we imposed at the beginning, which is actually straightforward to verify! 

    \begin{center}
        \includestandalone[width = 0.55\textwidth]{Figures/Part1/p105.2}\\ 
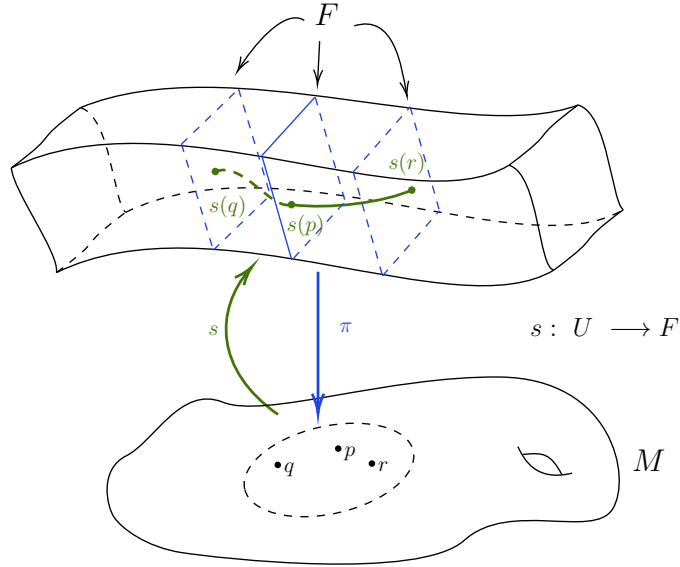
        \hypertarget{fig30}{Figure 30.} Schematic for the map $s:U\rightarrow F$
    \end{center}

Let us make a side remark before the next topic.
\begin{remark}
    Given a "local" section $\sigma\in\Gamma(P):$,
    \begin{equation}
        \sigma: U \subseteq M\to P \quad \text{with} \; \pi \circ \sigma = \text{id}_{\mathnormal{u}}.
    \end{equation}
    When we pull-back $\omega$ by $\sigma$, we in fact obtain
    \begin{equation}
        \sigma^{*}\omega \in \Omega^{1}(U) \otimes \mathfrak{g},
    \end{equation}
    where $\Omega^{1}(U)$ is a \textit{local }$\mathfrak{g}$-valued one form on $U \subseteq M$.
Such local data can indeed give rise to a \textit{global} $\mathfrak{g}$-valued one form $\sigma^{*}\omega$ on $M$ by gluing patches over the overlap.

    \begin{center}
        \includestandalone[width = 0.4\textwidth]{Figures/Part1/p106}\\ 
        \hypertarget{fig31}{Figure 31.} Gauge transformation.
    \end{center}

    Consider another local section $\psi: V \to P$. Let $x\in U \cap V$; since $\sigma(x)$ and $\psi(x)$ lie in the same fiber, there exists a unique $g_{x} \in G$ such that $\psi(x)= \sigma(x) \cdot g_{x}$.

    In fact, there exists a natural map, called a \textbf{gauge transformation},
    \begin{equation}
        \begin{aligned}
            U \cap V &\to G \\
            x &\mapsto g_{x} \; \text{(given above)}
        \end{aligned}
    \end{equation}

    Recall, when pulling everything down to $M$ under a gauge transformation, we have
    \begin{equation}
        \sigma^{*}\omega \mapsto \underbrace{g^{-1}\cdot\sigma^{*}\omega\cdot g + dg\cdot g^{-1}}_{\psi^{*}\omega} \quad \text{on the overlap.} 
    \end{equation}
    Define an equivalence relation:
    \begin{equation}
        \sigma^{*}\omega \sim \psi^{*}\omega \iff \exists g\in G \;\text{such that} \; \psi^{*}\omega=g^{-1}\cdot\sigma^{*}\omega\cdot g + dg\cdot g^{-1}
    \end{equation}
    Using such a relation, glue $U$ and $V$ over the overlap(s) and obtain a \textit{global} $\mathfrak{g}$-valued one-form on $M$.
\end{remark}


\subsection{Parallel Transport and Holonomy on a Principal $G$-Bundle}

\paragraph{Construction of the parallel transport map.} Consider the following geometric setup:
    \begin{equation}
        \begin{tikzcd}[sep=1.2cm]
            P \arrow[r, thick, "\cdot G"] & P \arrow[d, thick, "\pi"] & {\boxed\omega }\arrow[l, thick, dashed, no head] \\ 
            {} & M & {}
        \end{tikzcd}
    \end{equation}
Here $\omega$ is a connection one-form on $P$, defining a distribution of $TP$
\begin{equation}
    TP = HP \oplus VP, \quad \text{where} \; HP\coloneqq \ker(\omega) .
\end{equation}

    \begin{center}
        \includestandalone[width = 0.35\textwidth]{Figures/Part1/p107.1}\\ 
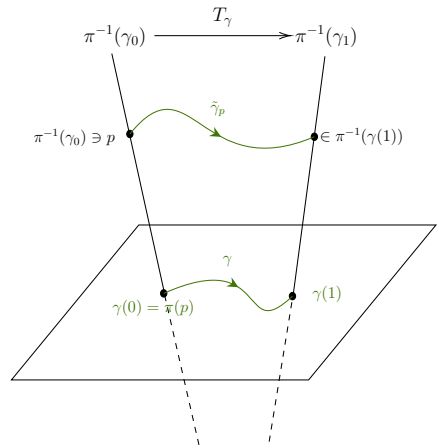
        \hypertarget{fig32}{Figure 32.} Parallel transport from $\pi^{-1}(\gamma_0)$ to $\pi^{-1}(\gamma_1)$.
    \end{center}
Given a curve $\gamma: [0,1] \to M$, we want to define:
\begin{itemize}
    \item An unique horizontal curve $\tilde{\gamma}_{p}$ for an initial point $p\in\pi^{-1}(\gamma(0))$,
    \item The \textbf{parallel transport map}
    \begin{equation}
        \begin{aligned}
            T_{\gamma}: \pi^{-1}(\gamma(0)) &\longrightarrow \pi^{-1}(\gamma(1)) \\
            p &\longmapsto \tilde{\gamma}_{p}(1).
        \end{aligned}
    \end{equation}
\end{itemize}

\begin{definition}
    Let 
    \begin{equation}
        \begin{tikzcd}[sep=1.2cm]
            P \arrow[r, thick, "\cdot G"] & P \arrow[d, thick, "\pi"] \\ 
            {} & M
        \end{tikzcd}
    \end{equation}
    
    be a principal $G$-bundle, $\gamma: [0,1] \to M$ a curve in $M$.

    The \textbf{horizontal lift} $\tilde{\gamma}_{p}$ of $\gamma$ with initial point $p \in \pi^{-1}(\gamma(0))$ (where for each choice of $p$, there exists a unique horizontal lift  $\tilde{\gamma}_{p}$) is an unique curve in $P$ such that
    
      \begin{center}
        \includestandalone[width = 0.45\textwidth]{Figures/Part1/p107.2}\\ 
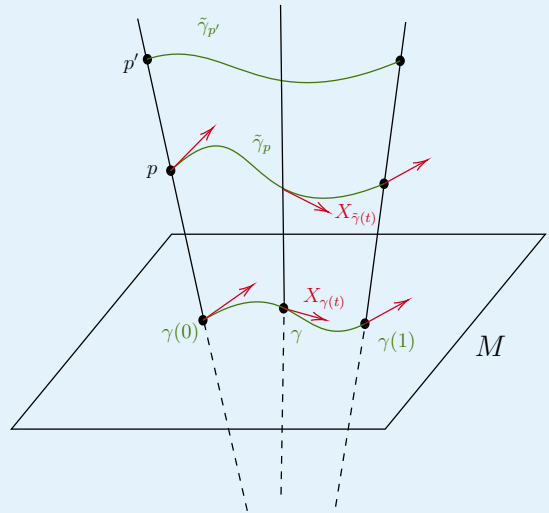
        \hypertarget{fig33}{Figure 33.} Parallel transport of vectors.
    \end{center}
    
    \begin{itemize}
        \item $\gamma = \pi \circ \tilde{\gamma}_{p}$.
        \item For $X_{\tilde{\gamma}_{p}(t)}\in H_{\tilde{\gamma}_{p}(t)}P$, ver$(X_{\tilde{\gamma}_{p}(t)}) = 0$ $\forall X_{\tilde{\gamma}_{p}(t)} \in T_{X_{\tilde{\gamma}_{p}(t)}}P$, each vector along $\tilde{\gamma}_{p}(t)$ is horizontal.
        \begin{equation}
            \begin{tikzcd}[sep = 1.5cm]
                T_{\tilde{\gamma}_p(t)}P \arrow[d, thick, "\pi_{*,\tilde{\gamma}_p(t)}"] \\ T_{\gamma(t)}M
            \end{tikzcd}
        \end{equation}
        \item For $X_{\tilde{\gamma}_{p}(t)} \in T_{X_{\tilde{\gamma}_{p}(t)}}P$, $\pi_{*,\tilde{\gamma}_{p}(t)}(X_{\tilde{\gamma}_{p}(t)})=X_{\gamma(t)} \in T_{{\gamma}(t)}M$.
    \end{itemize}
\end{definition}

So we have a couple of questions concerning this situation:
\begin{itemize}
    \item How to construct/characterize such a lift?
    \item What is the governing ODE for defining such a lift $\tilde{\gamma}_{p}(t)$?
\end{itemize}
KEY: A connection indeed defines such a lift with initial point $\tilde{\gamma}_{p}(0)$. Here are the essential steps:
      \begin{center}
        \includestandalone[width = 0.45\textwidth]{Figures/Part1/p108.1}\\ 
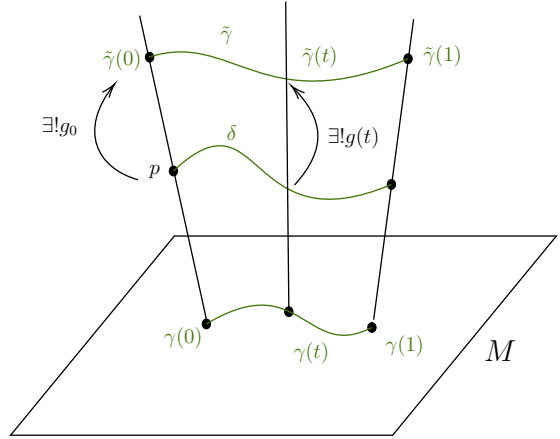
        \hypertarget{fig34}{Figure 34.} Lift of a curve.
    \end{center}
\begin{itemize}
    \item[\text{STEP 1:}] Take an arbitrary curve $\delta: [0,1]\to P$ starting from $p\in \pi^{-1}(\gamma(0))$ such that $\pi \circ \delta = \gamma$
    \item[\text{STEP 2:}] Find a suitable curve  $g: [0,1] \to G$ such that $\tilde{\gamma}(t) \coloneqq \delta(t)\cdot g(t)$ where $\cdot$ is a right $G$-action on $P$. 
\end{itemize}

So, the problem of defining $\tilde{\gamma}$ reduces to finding a suitable ODE that defines the curve $g(t)$.

\begin{theorem}
    For a matrix Lie group $G$, the ODE defining $g(t)$ above is,
    \begin{equation}
        g^{-1}(t)\cdot\omega_{\delta(t)}(X_{\delta(t)}) \cdot g(t) +  g^{-1}(t)\cdot \dot{g}(t) = 0 \quad 
    \end{equation}
    or equivalently,
    \begin{equation}
        \dot{g}(t) = -\omega_{\delta(t)}(X_{\delta(t)}) \cdot g(t) \label{eq:omerinsidetheorem1}
    \end{equation}
    where "$\cdot$" is the matrix multiplication. 
\end{theorem}

\begin{observation}
    Pick a local chart $(U,x)$ on $M$ together with a local section $\sigma: U \to P$. 
      \begin{center}
        \includestandalone[width = 0.5\textwidth]{Figures/Part1/p108.2}\\
        \hypertarget{fig35}{Figure 35.} Schematic for $\delta=\sigma\circ\gamma$.
    \end{center}
    
    Then we have
    \begin{itemize}
        \item $\delta \coloneqq\sigma \circ\gamma$, and hence $\sigma_{*}(X_{\gamma(t)}) = X_{\delta(t)} \; \forall t$, where $X_{\gamma(t)} \in T_{\gamma(t)}M$ and $X_{\delta(t)} \in T_{\gamma(t)}P$.
        \item The ODE above becomes: (let's focus on the term with the connection one-form)
        \begin{equation}
        \begin{aligned}
            \omega_{\delta(t)}(X_{\delta(t)}) &= \omega_{\delta(t)}(\sigma_{*,\gamma(t)}(X_{\gamma(t)}))  \quad \text{for} \; X_{\gamma(t)}\in T_{\gamma(t)}M \\
            &= \underbrace{(\sigma^{*}\omega)_{\gamma(t)}}_{\text{Yang-Mills field}}(X_{\gamma(t)}) \\
            &= \omega^{\mathnormal{u,\sigma}}_{\gamma(t)}(X_{\gamma(t)}) \\
            &= \Gamma_{\mu}(\gamma(t))\underbrace{X^{\mu}_{\gamma(t)}}_{\dot{\gamma}^{\mu}(t)}, 
        \end{aligned}
        \end{equation} 
    \end{itemize}
    where the first equation holds due to the  condition above, the second equation holds by the definition of pull-back, and we use the component form of the equation in the last line with the map
    \begin{equation}
        \begin{aligned}
            \Gamma_{\mu}: M &\to \mathbb{R} \\
            x &\mapsto \Gamma_{\mu}(x).
        \end{aligned}
    \end{equation}
    Hence, Eq.~\ref{eq:omerinsidetheorem1} becomes
    \begin{equation}
        \dot{g}(t) = -\Gamma_{\mu}(\gamma(t))\dot{\gamma}^{\mu}(t) g(t), \label{eq:omerinsidetheorem2}
    \end{equation}
    with the initial condition $g(0) = g_{0}$.
\end{observation}

Now, the question is how to solve Eq.~\ref{eq:omerinsidetheorem2}. We can try the following.
\begin{itemize}
    \item[(1)] Integrate Eq.~\ref{eq:omerinsidetheorem2}, and get
    \begin{equation}
        g(t) = g_{0} - \int^{t}_{0}{\Gamma_{\mu}(\gamma(\lambda))\dot{\gamma}^{\mu}(\lambda) g(\lambda)d\lambda}.
    \end{equation}
    \item[(2)] Substitute recursively; i.e., substitute $(1)$ to the right-hand side of $(1)$,
    \begin{equation}
        \begin{aligned}
            g(t) &= g_{0} - \int^{t}_{0}{d\lambda_{1}}\Gamma_{\mu}(\gamma(\lambda_{1}))\dot{\gamma}^{\mu}(\lambda_{1}) \Big[ {g_{0} - \int_{0}^{t}{\Gamma_{\mu}(\gamma(\lambda_{2}))\dot{\gamma}^{\mu}(\lambda_{2}) g(\lambda_{2})d\lambda_{2}}}  \Big] \\
            &= g_{0} - g_{0}\int_{0}^{t}{d\lambda_{1}\Gamma(\lambda_{1})} + \int^{t}_{0}{d\lambda_{1}\int^{\lambda_{1}}_{0}{d\lambda_{2}}{\Gamma(\lambda_{1})\Gamma(\lambda_{2})}g(\lambda_{2})}.
        \end{aligned}
    \end{equation}
    \item[(3)] Repeat the same argument. Then we have 
    \begin{equation}
         \begin{aligned}
            g(t) = &g_{0}\Big[ 1_{g}-\int_{0}^{t}{d\lambda_{1}\Gamma(\lambda_{1})}+
            \underbrace{
            \int_{0}^{t}{d\lambda_{1}\int^{\lambda_{1}}_{0}{d\lambda_{2}\Gamma(\lambda_{1})\Gamma(\lambda_{2})}}}_{\frac{1}{2}\int_{0}^{t}{d\lambda_{1}\int^{t}_{0}{d\lambda_{2}}}\Gamma(\lambda_{1})\Gamma(\lambda_{2})} \\
            &-\underbrace{\int_{0}^{t}{d\lambda_{1}\int^{\lambda_{1}}_{0}{d\lambda_{2}\int_{0}^{\lambda_{3}}{d\lambda_{3}\Gamma(\lambda_{1})\Gamma(\lambda_{2})\Gamma(\lambda_{3}) }}}}_{\frac{1}{6}\int^{t}_{0}\int^{t}_{0}\int^{t}_{0}\cdots} 
            +\dots \Big] \\
            &= \mathcal{P} \exp{-\int^{t}_{0}{d\lambda \Gamma(\lambda)}} g_{0},
        \end{aligned}
    \end{equation}
    where $\mathcal{P}$ is the \textit{path ordering} that controls the order of the iterated integral. 
    
    We may also denote the integral term above simply by  $\int^{t}_{0}{\omega d\lambda}$, with $\omega$ a connection one-form. Additionally, $\Gamma(\lambda)$ was
    \begin{equation}
        \begin{aligned}
            \Gamma(\lambda) &\coloneqq \Gamma_{\mu}(\gamma(\lambda))\dot{\gamma}^{\mu}(\lambda) \\
            &= (\sigma^{*}\omega)_{\mu}(\gamma(\lambda))\dot{\gamma}^{\mu}(\lambda) \\
            &= (\sigma^{*}\omega)_{\gamma(\lambda)}(X_{\gamma(\lambda)}),
        \end{aligned}
    \end{equation}
    called the \textit{Yang-Mills field connection one-form} on $M$.
\end{itemize}

Then we define any horizontal lift $\tilde{\gamma}$ as follows:
\begin{equation}
    \tilde{\gamma} \coloneqq \delta(t) \cdot g(t) = (\sigma \circ \gamma)(t) \cdot \mathcal{P} \exp{-\int^{t}_{0}{d\lambda \Gamma(\lambda)}} g_{0},
\end{equation}
with $g(0) = g_{0}$.\\ 

Hence, the \textbf{parallel transport map} is (as claimed) given as 

  \begin{center}
        \includestandalone[width = 0.5\textwidth]{Figures/Part1/p110}\\ 
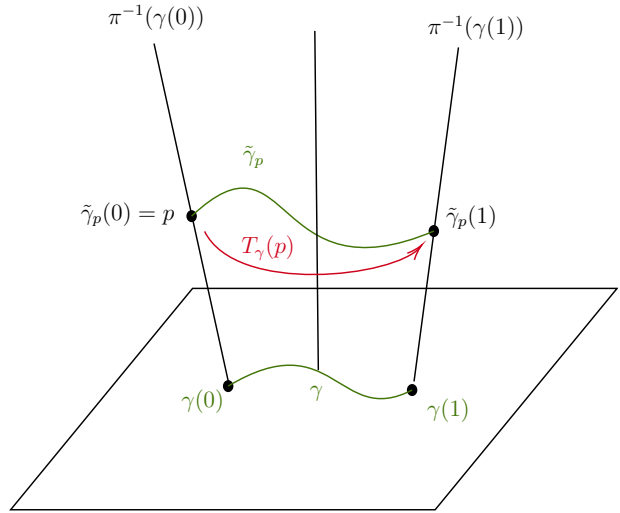
        \hypertarget{fig36}{Figure 36.} Schematic for the parallel transport, $T_\gamma(p)$, over a curve $\gamma$.
    \end{center}
Here
\begin{equation}
    \begin{aligned}
        T_{\gamma}: \pi^{-1}(\gamma(0)) &\xlongrightarrow{\cong} \pi^{-1}(\gamma(t)) \\
        p &\longmapsto \tilde{\gamma}_{p}(1),
    \end{aligned}
\end{equation}
where there exists a unique lift $\tilde{\gamma}$ starting from $p \in \pi^{-1}(\gamma(0))$. For each $p \in \pi^{-1}(\gamma(0))$, $\tilde{\gamma}$ is defined uniquely as above.

\begin{remark}
    Since $(\cdot g)_{*} \; H_{p}P = H_{p \cdot g}P \quad \forall p \in P$ and $\forall g \in G$, $T_{\gamma}$ is a bijection.   
\end{remark}

\paragraph{The case of loops and the Holonomy groups.} Using the above setup, let us consider now the parallel transport map acting on loops: Given a loop $\gamma$ in $M$,

  \begin{center}
        \includestandalone[width = 0.5\textwidth]{Figures/Part1/p111}\\
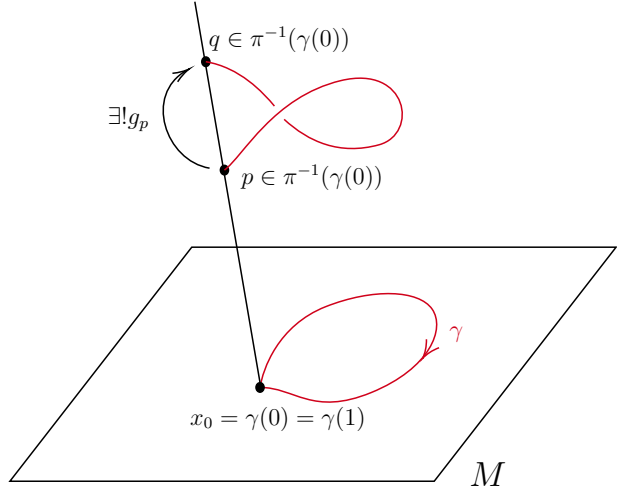
        \hypertarget{fig37}{Figure 37.} Parallel transport along loops.
    \end{center}

\begin{itemize}
    \item[1)] The lift $\tilde{\gamma}$ of a loop $\gamma$ is \textit{not} necessarily a loop at $p$.
    \item[2)] However, it still returns to the same fibre over the base point $x_{0}$ of the loop $\gamma$, i.e.,
    \begin{equation}
        \begin{aligned}
            &p = \tilde{\gamma}(0) \in \pi^{-1}(x_{0}), \\
            &q = \tilde{\gamma}(1) \in \pi^{-1}(x_{0}),
        \end{aligned}
    \end{equation}
    where $x_{0} \coloneqq \gamma(0) = \gamma(1)$. But it may land some other point $q \neq p$ in $\Pi^{-1}(x_{0})$. 
     \item[3)] So, there exists a unique $g_{p,\gamma} (\text{or just} \; g_{p}) \in G$ such that $q = p \cdot g_{p}$, there exists a map,
     \begin{equation}
         \begin{aligned}
             \pi^{-1}(x_{0}) &\longrightarrow G \\
             p &\longmapsto g_{p,\gamma}.
         \end{aligned}
     \end{equation}
     Therefore, the parallel transport map $T_{\gamma}$ reduces to,
     \begin{equation}
         \begin{aligned}
             T_{\gamma}: \pi^{-1}(\gamma(0)=x_{0}) &\longrightarrow \pi^{-1}(\gamma(1)=x_{0}) \\
             p &\longmapsto \tilde{\gamma}_{p}(1),
         \end{aligned}
     \end{equation}
     where $\tilde{\gamma}_{p}(1)=  p \cdot g_{p}$ for some $g_{p} \in G$.
     \item[4)] We define 
     \begin{equation}
         \mathsf{Hol}_{p}(\omega) \coloneqq \qty{g_{p,\gamma}:\tilde{\gamma}_{p}(1) = p \cdot g_{p}} = \qty{g_{p,\gamma}: p \sim p \cdot g_{p,\gamma}}
     \end{equation}
     This is the \textbf{holonomy group} of $\omega$ at $p \in \pi^{-1}(x_{0})$. Or we can just introduce an equivalence relation $p \sim q$ in $\pi^{-1}(x_{0})$ if and only if there exists a unique $g_{p} \in G$ such that $g=p\cdot g_{p}$ as we denoted in the second equality.
     It is subgroup of $G$ such that $$p\sim q \implies \mathsf{Hol}_{p}(\omega) = \mathsf{Hol}_{q}(\omega)$$
\end{itemize}

\newpage
Some properties of $\mathsf{Hol}_{p}(\omega)$:
\begin{itemize}
    \item[(1)] The \textbf{restricted holonomy group} at $p$ is the subgroup $\mathsf{Hol}^{0}_{p}(\omega)$ coming from horizontal lifts of contractible loops $\gamma$.
    \item[(2)] If $M$ and $P$ are connected, then $\mathsf{Hol}_{p}(\omega)$ depends on $p$ up to conjugacy, i.e., for any $q= p \cdot g$.
    \begin{equation}
        \mathsf{Hol}_{p}(\omega) = g^{-1} \mathsf{Hol}_{p}(\omega) g.
    \end{equation}
    
    \item[(3)] $\mathsf{Hol}^{0}_{p}(\omega) \subseteq G$ is connected Lie subgroup.
    \item[(4)] There exists a surjective map $\pi_{1}(M) \to \mathsf{Hol}^{0}_{p}(\omega)/\mathsf{Hol}_{p}(\omega)$
    \item[(5)] If $M$ is simply connected (i.e. $\pi_{1}(M) =0$), then $ \mathsf{Hol}_{p}(\omega) = \mathsf{Hol}^{0}_{p}(\omega)$.
    \item[(6)] $\omega$ is flat if and only if $\mathsf{Hol}^{0}_{p}(\omega)$ is trivial.
\end{itemize}

\section{A Glimpse at Witten's Setup}

We essentially consider the $(2+1)$ dimensional $SU(2)$ Yang-Mills theory with purely Chern-Simons functional: Let us consider a principal $SU(2)$-bundle over a $3$-manifold $M$ (in particular, we will take $M = S^3$) equipped with a connection one-form $\omega$ on $P$. That is, we have the diagram
\begin{equation}
    \begin{tikzcd}[sep=1.2cm]
        P \arrow[r, thick, "\cdot SU(2)"] & P \arrow[d, thick, "\pi"] & \boxed{\omega} \arrow[l, thick, 
        dashed, no head] \\ 
        {} & M & {}
    \end{tikzcd}
\end{equation}
Here, $\omega \subset \Omega^1 (P) \otimes g$ and $\mathcal{A}$ is the \textit{space of all connections} on $P$ Note that since $SU(2)$ is connected, any principal bundle over $M$ is trivial. Now, consider the Chern-Simons functional
\begin{equation}
    \text{CS}: \mathcal{A} \xlongrightarrow{} \mathbb{R}/\mathbb{Z} \cong S^1,
\end{equation}
where
\begin{equation}
    \text{CS}(A) : = \frac{k}{4\pi} \int_M \Tr{\omega \wedge d \omega + \frac{2}{3}  \omega \wedge \omega \wedge \omega}
    \quad \quad
    k \in \mathbb{Z}.
\end{equation}
The above integrand is independent of the metric. Hence, it produces a "generally covariant" theory.
\begin{observation}
    Critical points of the CS functional form the set
    \begin{equation}
        \{ \omega \in \mathcal{A} \;\big|\; \underbrace{d\omega + \omega \wedge \omega}_{\Omega \in \Omega^2 (p) \otimes g} = 0 \}.
    \end{equation}
    That is, the critical points are the \textit{flat connections}. We will work on the moduli space of flat connections in more detail!
\end{observation}

Here are some more essential observations. From now on, we will denote the connection $1$-form by $A$.

\begin{theorem}
    (Atiyah-Bott I)\\ 

    $\mathcal{A}$ above is infinite-dimensional symplectic manifold together with a symplectic form $\omega_0: T_A \mathcal{A} \times T_A \mathcal{A} \xrightarrow{} \mathbb{R}$ given by

    \begin{equation}
        \omega_0 (a, b) = \int_\Sigma \Sigma_{i,j} \underbrace{a_i \wedge b_j \langle X_i, X_j \rangle}_{\in \Omega^2 (\Sigma)},
    \end{equation} where the inner product is $ G$-invariant, 
    and

    \begin{equation}
        a = \sum_i a_i \otimes X_i \;, \quad b = \sum a_i \otimes X_i  \in (\Omega^1_{\mathsf{Hor}} (p) \otimes \mathfrak{g})^G,
    \end{equation} such that $X_1, \cdots, X_n$ is a basis for $\mathfrak{g}$.


    \begin{center}
        \includestandalone[width=0.3\textwidth]{Figures/Part1/p114}\\
        \hypertarget{fig38}{Figure 38.} A Riemannian surface.
    \end{center}
    On \hyperlink{fig38}{Fig. 38}, $\Sigma$ is a 2-dimensional Riemannian manifold, indeed a \textit{Riemann surface}. Here, we have cut out $M$ by $\Sigma$. 
\end{theorem}

\begin{observation}
    Consider, 
    \begin{equation}
        \begin{tikzcd}
            P \arrow[r, thick, "SU(2)"] & P \arrow[d, thick, "\pi"] \\ 
            {} & M
        \end{tikzcd}  ,  
    \end{equation}

    and a diffeomorphism $f: P \xlongrightarrow{\cong} P$ commuting with the $G$-action which gives rise to a diffeomorphism $f_0: M \xlongrightarrow{\cong}M$. This is something we did before.
    \begin{equation}
    \begin{tikzcd}[sep = 1.5cm, remember picture]
        |[alias = A]| P \arrow[r, thick, "f"] & 
        |[alias = B]| P \\ 
        |[alias = C]| P \arrow[u, leftrightarrow, thick, "\cdot g"] \arrow[r, thick, "f"] \arrow[d, thick, "\pi"] & 
        |[alias = D]| P \arrow[u, thick, "\cdot g", swap] \arrow[d, thick, "\pi"] \\ 
        |[alias = E]| M \arrow[r, thick, dashed, "f_0"] &  
        |[alias = F]| M
    \end{tikzcd}
    \tikz[overlay,remember picture]{%
        \node (Y1) [scale=1.5] at (barycentric cs:A=1,B=1,C=1,D=1) {$\circlearrowright$};
        }
\end{equation}
\end{observation}

\begin{definition}
    A diffeomorphism $f: P \xrightarrow{\cong} P$ commuting with $G$-action is called a \textbf{gauge transformation} if the induced map $f_0 = \text{id}_M$. We denote by $\mathcal{G}$, the group of all gauge transformations.
\end{definition}

Note that a connection $1$-form $A$ on $P$ defines a connection via $HP = \text{ker} (A)$. i.e.

\begin{equation}
    TP = VP \oplus HP
\end{equation}

If $f \in \mathcal{G}$, then $f_{*, p} : T_p P \xlongrightarrow{\cong} T_{f(p)} P$ induces an action of $\mathcal{G}$ on $\mathcal{A}$.

\begin{equation}
    \begin{tikzcd}[sep=1.5cm]
        P \arrow[r, thick, "f"] \arrow[d, thick] & P \arrow[d, thick] \arrow[r, thick, squiggly] & f_*=TP \rightarrow TP  \\ 
        M \arrow[r, thick, "f_0=\text{id}"] & M & {}
    \end{tikzcd}
\end{equation}
where $TP=VP\oplus HP$ and $f_*$ takes $VP\oplus HP$ to $VP\oplus HP_f$. That is, \emph{$\mathcal{G}$ acts on $\mathcal{A}$ by $f_*$.}

\begin{theorem}
    (Atiyah-Bott II)\\
    
    $(\mathcal{A}, \omega_0, \mathcal{G}, \mu)$ is indeed a Hamiltonian $\mathcal{G}$-space with the moment map
\begin{align}
        \mu : & \ \mathcal{A} \xlongrightarrow{} \mathfrak{g}^*_{\mathcal{G}} \nonumber \\
        & A \longmapsto \Omega_A \quad \text{ curvature }  of A.
    \end{align}
\end{theorem}

    By the symplectic reduction theorem, we get
    \begin{equation}
        M : = \mu^{-1}(0) / \mathcal{G}
    \end{equation}
    is a symplectic manifold. This is the moduli space of flat connections mod gauge transformations.

\begin{theorem}
    (Atiyah-Bott III)\\

    $\mathcal M$ above is a finite-dimensional symplectic manifold. 
\end{theorem}
This is where geometric quantization comes into play, since we now have a suitable symplectic manifold to quantize. Later, we shall elaborate on the following arguments.

\begin{itemize}
    \item $\displaystyle  Z(M) = \int_{\mathcal M} DA' \exp(i\text{CS})$ is the partition function. The integral is over the moduli space of all flat connections. $Z(M)$ yields a $3$-manifold invariant!
    \item More generally, by introducing a functional $\theta(A)$ associated to a connection one-form $A$ on $P$, one can construct an invariant for the data definig $\theta(A)$ as follows:

    \begin{equation}
        Z(M, \theta) = \int DA' \exp(i\text{CS}) \theta(A). 
    \end{equation}

    To derive a knot invariant, we let $M=S^3$, and $C$ a knot in $S^3$, with $G = SU(2)$, such that  
    \begin{equation}
        \theta (A) := \Tr_R (\mathsf{Hol}_C(A)) \quad \left(\text{recall our shorthand notation} \; \theta(A):= \Tr_R P \exp \oint_C A \right),
    \end{equation}

    where $R$ is a certain representation of $SU(2)$.
    \begin{center}
        \includestandalone[width=0.2\textwidth]{Figures/Part1/p116.1}\\ 
        \hypertarget{fig39}{Figure 39.} A trefoil knot.
    \end{center}
    \item

    In accordance with Sigma model (or  TQFT in the sense of Atiyah). Let
    \begin{equation}
        M = (M_1\amalg M_2)/\Sigma
    \end{equation}
    where $M_1$ and $M_2$ are glued together along their boundaries. I.e., take a homeomorphism $\Sigma \xrightarrow{\varphi} \Sigma$ and define

    \begin{equation}
        M = M_1 \amalg M_2/x \sim \varphi(x) \quad \forall x \in \mathbb{Z}.
    \end{equation}

    \begin{center}
        \includestandalone[width=0.7\textwidth]{Figures/Part1/p116.2} \\ 
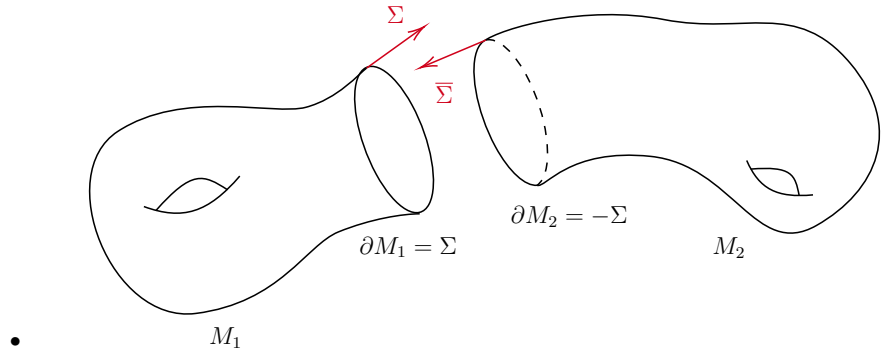
        \hypertarget{fig40}{Figure 40.} A manifold $M$ cut by a Riemannian surface $S$, yielding two manifolds $M_1,M_2$.
    \end{center}

    Then we have

    \begin{equation}
        Z(M) = \langle Z(M_1, Z(M_2)\rangle,
    \end{equation}

    where $Z_\Sigma$ is a vector space associated to $\Sigma$, the \textit{quantum Hilbert space}, and $ Z(M_1) \in Z_\Sigma$ is a vector in that quantum Hilbert space; $Z(M_2) \in Z_{\Bar{\Sigma}} = Z_\Sigma^*$ is a covector in the dual space of $Z_\Sigma$, with $Z_\Sigma^* = \mathsf{Hom} (Z_\Sigma; \mathbb{C})$. Note that the identification $Z_{\Bar{\Sigma}} = Z_\Sigma^*$ is by the axioms of TQFT.  

    In fact, a TQFT $Z$ is a functor in the sense of category theory, as we will see later!
    
    \begin{center}
        \includestandalone[width=0.5\textwidth]{Figures/Part1/p117}\\
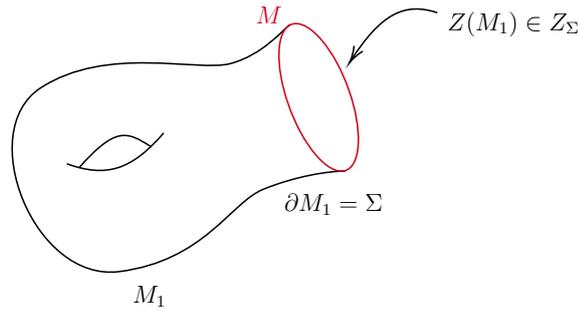
        \hypertarget{fig41}{Figure 41.} The manifold $M_1$ with a vector space $Z_\Sigma$ assigned to its boundary, $\partial M_1 = \Sigma$.
    \end{center}
    
    The moduli space $\mathcal M$ of flat connections serves as the desired symplectic manifold assigned to $\Sigma$ so that one can construct the quantum Hilbert space $Z_\Sigma$ by means of the geometric quantization procedure designed for special symplectic manifolds.

\end{itemize}

Let's give a quick preview of what is coming next.

\begin{itemize}
    \item We shall see $Z_\Sigma$ will be the space of holomorphic section of a "certain" complex line bundle over $M$.
    \item By using dimensionality of $Z_\Sigma$, we will derive some relation in terms of a partition function $Z$ depending on some possible line configuration.
    \item That derived relation turns out to be the Skein-like defining relations for the Jones polynomials.
    \item That framework can be extended to any $3$-manifold, not just $S^3$!
\end{itemize}

\chapter{Lectures on Geometric Quantization}
The main goal of this section is to understand how to appropriately assign a suitable Hilbert space to a given symplectic manifold \( M \) based on a specific set of axioms and the formalism called \textit{geometric quantization}. 

More formally, we aim to construct a Lie algebra representation of a Poisson algebra—the \textit{algebra of observables}— defined as the space of smooth functions on a given symplectic manifold \(M\)-often called the \textit{phase space}, equipped with a Poisson bracket. 

Additionally, we will examine the Sigma model of quantum field theory (QFT) and the formulation of topological quantum field theory (TQFT) in the sense of  Atiyah. \\

\noindent \textbf{Main references:}

Ko HONDA, Lecture notes for MATH 635: Topological Quantum Field Theory \cite{HondaTopologicalQFT}.

Matthias BLAU, Symplectic Geometry and Geometric Quantization \cite{BlauGQ}.\\

\noindent We will primarily reference Honda’s lecture notes \cite{HondaTopologicalQFT}, specifically chapters 9.2 to 12. The topics we will cover include:  
\begin{enumerate}
    \item Towards the quantization: Poisson Structure, vector bundles, connection and covariant derivative on a vector bundle.
\item 	Axioms of Geometric Quantization, construction of a (pre)quantum line bundle and a suitable Lie algebra representation on (pre)quantum Hilbert space which is a space of certain sections of the (pre)quantum line bundle.
\item Path integrals, Sigma model and formulation of TQFT.
\end{enumerate}

\newpage
\section{Main Ingredients and Towards the Quantization}

Recall the Poisson algebra: Given a symplectic manifold $(M , \omega)$, let $f \in C^\infty (M)$. We defined the \textit{Hamiltonian vector field} $X_f$ for $f$ by the equation
\begin{equation}
    \iota_{X_f} \omega (\cdot) = df (\cdot)
    \quad
    (\text{the one-form $\iota_{X_f} \omega$ is exact}).
\end{equation}

By studying the corresponding \textit{local flow} $\gamma_f$ 
of Hamiltonian vector field
$X_f$, we rederived 
Hamilton equations as outlined in the example below.

\begin{example}
   Consider  $(\mathbb{R}^{2n}, \omega)$ with $\omega = \sum_i^n dq_i \wedge dp_i$, and the
     coordinates $(q_1, \dots, q_n, p_1, \dots, p_n)$.

Given $H \in C^\infty (\mathbb{R}^{2n})$, its Hamiltonian vector field $X_H$ is given by

\begin{equation}
    X_H = \sum_i \pdv{H}{p_i} \pdv{}{q_i} - \pdv{H}{q_i} \pdv{}{p_i}.
\end{equation}

Consider its local flow $\gamma_H (t) = (q_1 (t), \dots, q_n(t), p_1(t), \dots, p_n(t) )$
which gives (by dynamics of flow)

\begin{equation}
    \dot{\gamma}_H = X_H \iff \dot{q}_i = \pdv{H}{p_i}, \quad
    \dot{p}_i = - \pdv{H}{q_i}.
\end{equation}

We have a natural map between vector spaces
\begin{align}\label{naturalmap in example}
    C^\infty (M) &\xlongrightarrow{} \textsf{Hom}(M) \subseteq \Gamma(TM) \nonumber \\
    f &\longmapsto X_f.
\end{align}

Define a bilinear map $\qty{\cdot,\cdot}: C^\infty (M) \times C^\infty(M) \xrightarrow{} C^\infty (M)$ by

$$
\forall f, g \in C^\infty(M), \quad \{f, g\} := - \omega (X_f, X_g).
$$
\end{example}

\begin{observation}
    \begin{itemize}
        \item[$1)$] $\qty{\cdot, \cdot}$ is skew symmetric and satisfies the Jacobi identity. Hence,
    $\qty{\cdot, \cdot}$ gives a Lie algebra structure on $C^\infty(M)$. That is,

    $(C^\infty (M), \; \qty{\cdot, \cdot}$ is a Lie algebra (which is called \textit{Poisson algebra}).
    The vector space is infinite-dimensional. It is an algebra with point-wise multiplication
    $(f \cdot g)(x) = f(x) g(x), \quad \forall x \in M$
    \item[$2)$]
    $\textsf{Hom}(M) \subseteq \Gamma(TM)$ is a $C^\infty (M)$-\text{submodule}.
    It inherits a Lie bracket $[.,.]_0$ from $(\Gamma(TM), \; [\cdot,\cdot]_{0})$
    where $[\cdot,\cdot]_{0}$ is the usual commutator:
    \begin{equation}
        [X,Y]_{0} = XY - YX.
    \end{equation}
    \item[$3)$] \textbf{Lemma}: $\forall X_{f},X_{g} \in \textsf{Hom}(M)$,
    \begin{equation*}
        [X_{f},X_{g}] = X_{\qty{f,g}} \in \textsf{Hom}(M),
    \end{equation*}
    hence, $[\cdot,\cdot]_{0}$ defines a Lie algebra structure on $\textsf{Hom}(M)$, i.e., $(\textsf{Hom}(M),[\cdot,\cdot]_{0})$ is also a \textit{Lie algebra}.
    \item[$4)$] The map in (\ref{naturalmap in example})
    is a Lie algebra homomorphism:
    \begin{equation}
    \begin{aligned}
        (C^\infty (M), \; \qty{\cdot, \cdot}) &\xlongrightarrow{\varphi} (\textsf{Hom} (M), \; [\cdot,\cdot]_{0})\\
        f &\longmapsto X_f\\
        \qty{f,y} &\longmapsto X_{\qty{f, g}} = [X_f, X_g]
        = [\varphi(f), \varphi(g)].
    \end{aligned}
    \end{equation}
    \end{itemize}
    
\end{observation}

    Observe that
    \begin{itemize}
        \item $\varphi$ assigns to $f$ a "differential operator".
        \item The condition $[\varphi(f), \varphi(g)] = \varphi(\{f,g\})$
        looks like \textit{Dirac's quantization condition!} That is,
        \begin{equation}
            [\hat{f}, \hat{g}] = - i \hbar \widehat{\qty{f,g}}.
        \end{equation}
    \end{itemize}

This is the first 'primitive step' towards geometric quantization! However, we have several loose ends:

\begin{itemize}
    \item What is the suitable quantum Hilbert space $\mathcal{H}$ on which 
    $\hat{f}$ above acts?
    (In this case, it is $X_f$)
    \item 
    Is the argument $f \xrightarrow{}X_f$ suitable? If it is, in what sense? We
    later see that we need compatibility with axioms of quantization.. We'll see that one indeed needs to modify it. To determine the $\mathcal{H}$ above, we need to study a particular vector bundle $E$ over $(M,\omega)$ (in fact, a complex line bundle over $M$). The modification of $\varphi$ is then related to a choice of a suitable connection on $E$.
\end{itemize}

\section{Connection on a Vector Bundle}

Let $E \xrightarrow{\pi} M$ be a vector bundle of rank $k$, where rank $k$ means each fibre is a vector space of dimension $k$.
 Let $s: U \subseteq M \rightarrow E$ be a \textit{local section},  $s\in \Gamma(E,U)$, with $\pi \circ s = id_U$. \\ 
Recall that every vector bundle is locally trivial. That is, there exists a cover $\{U_\alpha\}$ of $M$ such that $E|_{U_\alpha}$ is trivial (i.e. $E|_{U_\alpha} \simeq U_\alpha \times \mathbb{C}^k$). Equivalently, there exists $k$ local section $s_1^\alpha, \cdots, s_k^\alpha$ on $U_\alpha$ with
\begin{equation}
    s_i^\alpha: U_\alpha \xrightarrow{} E, \quad \pi \circ s_i^\alpha = \text{id}_{U_\alpha},
\end{equation}
such that $\forall p \in U_\alpha$, $s_1^\alpha(p), \cdots, s_k^\alpha(p)$ forms a basis for $\pi^{-1}(p)$.

\begin{center}
        \includestandalone[width=0.45\textwidth]{Figures/Part2/p4} \\ 
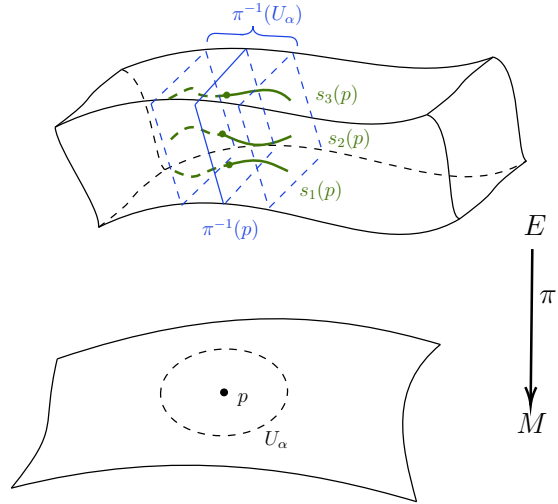
        \hypertarget{fig42}{Figure 42.} Local sections on a bundle.
\end{center}

Let $X \in \Gamma (TM)$. We want to differentiate $s$ at $p$ in the direction of $X_p \in \Gamma(E, U)$.

\begin{definition}
    A \textbf{connection} or \textbf{covariant derivative} $\nabla$ assigns to every vector field $X \in \Gamma(TM)$ a differential operator
    \begin{equation}
        \nabla_X : \Gamma(E) \xrightarrow{} \Gamma(E),
    \end{equation} satisfying the following conditions:
    \begin{itemize}
        \item $\nabla_X$ is $\mathbb{C}$-linear:
        $\nabla_X (c_1s_1 + c_2s_2) = c_1 \nabla_Xs_1 + c_2\nabla_2$.
        \item $\nabla_X$ is $C^\infty(M)$-linear:
        $\forall f,g \in C^\infty(M), \quad f \nabla_X s + g \nabla_y s$.
        \item Leibniz rule:
        $\nabla_X fs = f\nabla_X s + (Xf) s
        \quad \forall f \in C^\infty(M), \forall s \in \Gamma(E)$.
    \end{itemize}
\end{definition}

\begin{remark}
    $\nabla_X$ depends on $X$ \textit{tensorially}; $\nabla_{\qty{\cdot}} \qty{\cdot}$ is such that the subscript slot behave like a tensor and the other slot behaves like a derivation.
\end{remark}
\begin{definition}
    Let $E \xrightarrow{} M$ be a vector bundle of rank $k$, $E|_U$ a local trivialization with sections $s_1, \cdots, s_k$ defined on $U$, and $X \in \Gamma(TM)$.
Let $s \in \Gamma(E, U)$ be an arbitrary section on $U$. Then we have
\begin{equation}
    s = \sum_{i=1}^k f_i s_i, \quad f_i \in C^\infty (U).
\end{equation}
This is because $\forall p \in U, \; s_1(p), \cdots, s_k(p)$ forms a basis for $\pi^{-1}(p)$.

Since $s(p)$ is a vector in $\pi^{-1}(p)$, it must be spanned by $s_1(p), \dots, s_k (p)$. So, $s(p) = f_1 (p) s_1(p) + \dots + f_k(p) s_k(p)$. A better notation is
\begin{equation}
    s = f_i s^i,
\end{equation}which is called the \textit{Einstein summation} convention. With this notation, we define a \textbf{flat connection} as 
\begin{equation}
    \nabla_X s := \sum_i (Xf_i)s_i = (Xf_i) s^i.
\end{equation}
\end{definition}
Therefore, when we differentiate, we differentiate component-wise.

\begin{observation}
    \begin{itemize}
        \item[$1)$]
        \begin{equation}
            \nabla_X s = (X f_i) s^i = df_i (X) s^i, \quad
            \text{where}
            \quad
            s = 
            \begin{pmatrix}
                f_1\\
                \vdots\\
                f_X
            \end{pmatrix}
            = f_i,
        \end{equation} and hence, for a trivial connection, we write
        \begin{equation}
            \boxed{\nabla_{\cdot} \cdot = d.}
        \end{equation}
    Or in matrix form, we can write
    \begin{equation}
        \nabla_X s =
        \begin{pmatrix}
            df_1(X)\\
            \vdots\\
            df_k(X)
        \end{pmatrix},
    \end{equation}
    where $d$ is the usual \textit{exterior derivative}.
    \item[$2)$] In general, flat connections \textit{do not exist globally} on $M$, but they \textit{always exist locally}. When $E \rightarrow M$ is a trivial bundle, it exists!
    \item[$3)$] By patching flat connections, one can construct a globally defined connection on $M$.

    Patching, by definition, is a suitable \textit{equivalence relation} in accordance with the base change and the behaviour of the connection under this base change. We saw a similar discussion before. We shall also see how a connection behaves under a so-called gauge change!
    \end{itemize}
\end{observation}
\begin{observation}
    Let $E \xrightarrow{M}$ be a rank $k$ vector bundle, $\nabla$ and $\nabla'$ connections on $M$.

    Let $X \in \Gamma(TM)$ and $s \in \Gamma(E)$, then  ($\forall f \in C^\infty (M)$) we have
    \begin{equation}
    \begin{aligned}
        (\nabla_X - \nabla_X^{'}) (fs) &= \nabla_X(fs) - \nabla^{'}_X(fs)\\
        &= f \nabla_X s + (Xf)s - f \nabla'_Xs - (Xf) s\\
        &= f (\nabla_X s - \nabla^{'}_X s)\\
        &= f (\nabla_X - \nabla^{'}_X)s.
    \end{aligned}
    \end{equation}
    That means, $\nabla_X - \nabla^{'}_X$ is \textit{tensorial} in $s$. Therefore, if we take local trivialization sections $(U, s^1, \cdots, s^k)$ for $E|_U$, we get
    \begin{equation}
        (\nabla_X - \nabla^{'}_X)s^i = a^i_{j} s^j.
    \end{equation}
    Here $a^{i}_{j} : U \rightarrow \mathbb{C}$ is a $k \times k$ matrix of functions. In particular, we let
    \begin{equation}
        a^i_{j} := A^i_{j} (X), \quad A^i_{j} \in \Omega^1 (M).
    \end{equation}
\end{observation}

What does it all tell us?
\begin{itemize}
    \item[$1)$] Any two connections on $M$ (locally) differ by a $k \times k$ matrix of $1$-form $A = (A^i_{j})$, where $A^i_{j} \in \Omega^1 (M), \; \forall i,j$.
    \item[$2)$] In particular, if we take a trivial connection $\nabla_X^{'} = d$, for an arbitrary connection $\nabla$, we write
    \begin{equation}
        \boxed{\nabla = d + A,}
    \end{equation}
    where $A = (A^i_{j})$, $A^i_{j} \in \Omega^1 (M), \; \forall i,j$. That is, $\forall X \in \Gamma(TM),\; s \in \Gamma(E)$ we have
    \begin{equation}
        \nabla_X s = ds (X) + A(X) s.
    \end{equation}
    For basis field $s^i$, we have
    \begin{equation}
        \nabla_X s^i = A^i_{j} (X) s^j.
    \end{equation}

    \section{Gauge Change}

    Let $E \rightarrow M$ be a rank $k$ vector bundle, $E|_U$ and $E|_V$ be local trivialization with sections $\{s_1, \dots, s_k\}$ and $\{\bar{s}_1, \dots, \bar{s}_k\}$ respectively.

\begin{center}
        \includestandalone[width=0.45\textwidth]{Figures/Part2/p7}\\ 
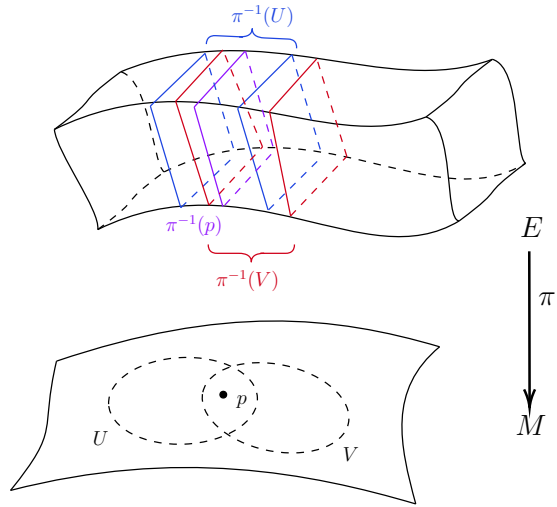
        \hypertarget{fig43}{Figure 43.} Local trivializations on a bundle. 
\end{center}

    Then, on $U \cap V$, we have
    \begin{equation}
        \bar{s}_i = g^i_{j} s^j,
    \end{equation}
    where $g^i_{j} : U \cap V \xrightarrow{} \mathbb{C}$. In the matrix form

    \begin{equation}
        \bar{s} =
        \begin{pmatrix}
            \bar{s}^1\\
            \vdots\\
            \bar{s}^k
        \end{pmatrix}
        =
        \begin{pmatrix}
            g^1_{\; 1} & g^1_{\; 2}& \dots\\
            \vdots& g^i_{j} & \vdots\\
            g^k_{\; 1} & \dots & g^k_{\; k}
        \end{pmatrix}
        \begin{pmatrix}
            s^1\\
            \vdots\\
            s^k
        \end{pmatrix}.
    \end{equation}

    Here, the matrix $g \in GL(k, \mathbb{C})$.
\end{itemize}

On $U\cap V$, $\forall X \in \Gamma(TM)$ we have
\begin{equation}
\begin{aligned}
    \nabla_X \bar{s}^i &= \nabla_X g^i_{j} s^j\\
    &= g^i_{j} \nabla_X s^j + (X g^i_{j})s^j\\
    &= g^i_{j} A^j_{\; \mu} s^\mu + dg^i_{j} s^j\\
    &= g^i_{j} A^j_{\; \mu} (g^{-1})^\mu_{\; \nu} \bar{s}^\nu + d g^i_{j} (X) (g^{-1})^j_{\; \sigma} \bar{s}^\sigma\\
    &= (g \cdot A(X) \cdot g^{-1})^i_{\; \nu} \bar{s}^\nu + (dg(X) \cdot g^{-1})^i_{\; \sigma} \bar{s}^\sigma\\
    &= \underbrace{(g \cdot A(X) \cdot g^{-1} + d g(X) \cdot g^{-1} )^i_{\; \nu}}_{\bar{A}(x)^i_{\; \nu}} \bar{s}^\nu.
\end{aligned}
\end{equation}
Therefore, under the gauge transformation 
\begin{equation}
    g^i_{j}: U \cap V \xrightarrow{}\mathbb{C}, \quad (\text{or} \; g:U \cap V \to \mathbb{C}^{k} \equiv GL(k, \mathbb{C}) = G),
\end{equation}

we have

\begin{equation}
    A \xrightarrow{} \underbrace{g \cdot A \cdot g^{-1} + dg \cdot  g^{-1}}_{\eqcolon g^{*} A}.
\end{equation}

Note that when we "patch" connections on the overlap, we define an equivalence relation with respect to the above gauge action. Namely, we have:
\begin{equation*}
    A \sim A' \quad \iff \quad
    A' = g^* A = g\cdot A\cdot g^{-1} + dg \cdot g^{-1},
\end{equation*}
for some $g$, where at each point $p \in U \cap  V$, $g(p) \in GL(k, \mathbb{C})$.

\begin{definition}
    For all $X, Y, \text{ and } Z \in \Gamma(TM), \; Z \in \Gamma(E)$,
    the \textbf{curvature }of a connection $\nabla$ is given by
    \begin{equation}
        R(X, Y) Z = \nabla_X \nabla_Y Z - \nabla_Y \nabla_X Z - \nabla_{[X,Y]} Z.
    \end{equation}
\end{definition}

\begin{claim}
    \begin{equation}
        R(X, Y) = (dA + A \wedge A) (X, Y)
    \end{equation}
    in local trivialization. Here 
    \begin{equation}
        A \wedge A (X, Y) := [A(X), A(Y)]_0
        = A(X) \cdot A(Y) - A(Y) \cdot A(X).
    \end{equation}
\end{claim}

\begin{proof}
    Direct computation yields the desired result:

    We start by letting $X, Y \in \Gamma(TM)$ and $s_1, \dots, s_k \in \Gamma(E,U)$ be local trivial sections. It is enough to examine $R(X,Y) s^i$. Observe that
    \begin{equation}
    \begin{aligned}
        \nabla_X \nabla_Y s^i &= \nabla_X (A(Y)^i_{j} s^j)\\
        &= A(Y)^i_{j} \nabla_X s^j + (X A(Y)^i_{j}) s^j\\
        &= A(Y)^i_{j} A(X)^j_{\; k} s^k + (X A(Y)^i_{j}) s^j\\
        &= (A(Y) \cdot A(X))^i_{\; k} s^k + (X A(Y)^i_{j})s^j.
    \end{aligned} 
    \end{equation}
    Similarly,
    \begin{equation}
    \begin{aligned}
        \nabla_X \nabla_Y s^i &= (A(X) \cdot A(Y))^i_{\;k} s^k + (Y A(X)^i_{j}) s^j\\
        \nabla_{[X, Y]}s^i &= A([X, Y])^i_{j} s^j.
    \end{aligned} 
    \end{equation}
    Then we have
    \begin{equation}
    \begin{aligned}
        R(X,Y) s^i &= \underbrace{\left[X (A(Y)^i_{j}) - Y (A(X))^i_{j} - A([X, Y])^i_{j} \right]}_{\text{by definition} \; dA^i_{j} (X, Y), \; A^i_{j} \in \Omega^1 (M)} s^j
        + \left[ A(X) \cdot A(Y) - A(Y) \cdot A(X) \right]^i_{j} s^j\\
        &= \left[ dA(X,Y) + (A \wedge A) (X, Y) \right]^i_{j} s^j.
    \end{aligned}
    \end{equation}
\end{proof}

Consider the gauge transformation $g: U \cap V \xrightarrow{} G$ as before. We saw that the connections $A_1$ and $A_2$ are related via $g$ as follows. On $U \cap V$, we have
\begin{equation}
    A_2 = g A g^{-1} + dg g^{-1}.
\end{equation}
Then we get 
\begin{equation}
\begin{aligned}
    R' &= dA_2 + A_2 \wedge A_2\\
    &= d(gAg^{-1} + dg g^{-1}) + (gAg^{-1} + dg g^{-1}) \wedge (gAg^{-1} + dg g^{-1})\\
    &= g (dA + A \wedge A) g^{-1}\\
    &= g R g^{-1}.
\end{aligned}
\end{equation}
So, under gauge transformations, $R$ transforms as 
\begin{equation}
    R \longmapsto R' = \underbrace{g R g^{-1}}_{R'}.
\end{equation}

\newpage
\section{Geometric Quantization Formalism}

The critical question we ask ourselves is: "In what sense do we quantize?"

Formulate the notion of the \textit{quantization} in the following sense:
Let $(M, \omega)$ be a symplectic manifold of dimension $2n$ with $(C^\infty (M), \qty{\cdot,\cdot}$ being its Poisson algebra. Then, we want the following correspondence:

\begin{table}[h]
\centering
\begin{tabular}{|c|c|c|c|}
\hline
\textit{Concepts} & \textit{Classical System} & \textit{Quantum System}   \\ \hline
Phase space & $(M, \omega)$ & Complex Hilbert space $\mathcal{H}$ \\ \hline
States & $p \in M$ & The ray $[\psi]$ in $\mathcal{H}$   \\ \hline
Algebra of observables & $f \in C^\infty (M)$ & $Q(f): \mathcal{H} \xrightarrow{} \mathcal{H}$  \\ \hline
Structure on the algebra & $\{\cdot , \cdot\}$ & $[\cdot,\cdot]$  \\ \hline
\end{tabular}
\caption{The Correspondence}
\end{table}

Here, we denote the \textit{quantization procedure} by $Q$. It assigns to a function $f$ a linear operator on $\mathcal{H}$, 
written $Q(f) \in \text{End}(\mathcal{H})$.

$\psi \sim \eta$ in $\mathcal{H}$ if and only if $\exists \lambda \in \mathbb{C}$ s.t. $\psi = \lambda \eta$. We call $\psi$ a \textbf{wave-function}.

On the quantized structure of algebra, the commutator is encoded in a Lie algebra homomorphism
\begin{equation}
    (C^\infty(M), \{\cdot , \cdot \}) \xlongrightarrow{} (\text{End} (\mathcal{H}), [\cdot, \cdot]),
\end{equation}
which in fact gives rise to a Lie algebra representation of $C^\infty (M)$ on $\mathcal{H}$. 

\textit{Geometric quantization (GQ)} provides us a suitable mathematical formalism to define such an assignment $Q$ above (based on Dirac's ideas). It has the following properties.

\begin{itemize}
    \item[$1)$] Applicable to any finite-dimensional symplectic manifold $(M,\omega)$.
    \item[$2)$] Keep track of the symmetries of the classical system with a Hamiltonian $G$-action, in the sense that the quantum states form an \textit{irreducible representation} of $G$. This is the so-called \textit{"irreducibility condition"}.
    \item[$3)$] \textbf{$1{\text{st}}$ step of GQ} consists of  the \textit{Prequantization} and introduction of "$\mathcal{H}_{\text{pre}}$". The construction involves a line bundle $\mathcal{L}$, a connection, and certain sections of this bundle. However, we lack the irreducibility condition. In fact, we will set $\mathcal{H}_{\text{pre}} = \Gamma(\mathcal{L})$, where $\Gamma(\mathcal{L})$ is the space of smooth sections on $\mathcal{L}$.
    \item[$4)$] \textbf{$2{\text{nd}}$ step of GQ:} Polarization of the space of functions on $M$.
    \begin{itemize}
        \item[a)] Aim is to select a suitable subspace of functions on $M$ to be quantized so that the irreducibility condition holds.
        \item[b)] In accordance with the polarization $P$ (a particular choice of a $n$-dim sub-bundle $P$ of $TM$, which is $2n$-dim), one restricts $\mathcal{H}_{\text{pre}} = \Gamma(\mathcal{L})$ to $\mathcal{H} = \Gamma_p(\mathcal{L})$, which are the $P$-polarized sectors of $\mathcal{L}$ so that the elements of $\mathcal{H}$ form irreducible representations of $G$. Here, by $P$-polarized sections, we mean
        \begin{equation}
            \Gamma_P(\mathcal{L}):=\{ s \in \Gamma(\mathcal{L}): \nabla_X s = 0, \quad \forall X \in \Gamma(p) \subseteq \Gamma(TM)\}.
        \end{equation}
        The space of sections of $L$ that are covariantly constant along $P \subseteq TM$ (for example, when $M$ is a K\"{a}hler manifold, $\Gamma(\mathcal{L})_P$ is the space of holomorphic sections).
        \item[c)] Independent of the choice of $P$.
        
    \end{itemize}
\end{itemize}

We need to discuss further topics to study these steps, such as almost complex structures, complex structures, and Kähler structures. You can check DaSilva's book [] for details of these topics.

In our current study, we mainly focus on the $ 1^ {\text{st}}$ step and make a brief comment on the $2^{\text{nd}}$ step.

\begin{definition}
    Let $(M, \omega)$ be a classical system and $\mathcal{A}$ a sub-algebra of $C^\infty(M)$. The pair $(\mathcal{H}, {Q})$ is called the \textbf{quantum system associated to $(M, \omega)$} if
    \begin{itemize}
        \item[$0)$] $\mathcal{H}$ is a complex separable Hilbert space:
        \begin{itemize}
            \item[$\cdot$] Its elements $\psi$ are the \textit{quantum wave functions}.
            \item[$\cdot$] The rays $\{\lambda \psi| \lambda \in \mathbb{C}\}$ are the \textit{quantum states}.
            \item[$\cdot$] For all $ f \in \mathcal{A}, \; Q(f) \in \text{End} (\mathcal{H})$.
        \end{itemize}
        \item[$1)$] $\mathbb{C}$-linearity: $Q(c, f+g) = c Q(f) + Q(g), \; \forall c \in \mathbb{C}$.
        \item[$2)$] $Q(f=1) = id_{\mathcal{H}}$.
        \item[$3)$] $Q^*(f) = Q(f)$. I.e., $Q(f)$ is a self-adjoint operator on $\mathcal{H}$.
        \item[$4)$] Quantum Condition: $[Q(f), Q(g)] = -ih Q(\{f,g\})$.
        \item[$5)$] Irreducibility Condition: 

        Suppose that $\{f_1, \dots, f_n\}$ is a complete set of observables, meaning that a function commuting with all $f_i$ must be constant
        \begin{equation}
            \{f, f_i\} = 0, \; \forall i \iff f = \text{constant}.
        \end{equation}

        Then $\{Q(f_1), \dots , Q(f_n)\}$ is also complete set of operators. 
        
        Equivalently, quantum states must form an irreducible representation of $G$. That is, if $R$ is an irreducible representation of $G$, then $R$ can be realized as a quantum Hilbert space.
    \end{itemize}
\end{definition}

\begin{example}
    In this prototype example, we basically examine why one needs to impose the above-mentioned axioms. Take
    \begin{equation}
        (M, \omega) = (\mathbb{R}^{2n}, \sum_i dqi \wedge dp_i),
    \end{equation}

    with the usual coordinates $(q_1, \cdots , q_n, p_1, \cdots, p_n)$.

    Then we have

    \begin{itemize}
        \item Given $H \in C^\infty (M)$, its Hamiltonian vector field $X_H$ is given by

        \begin{equation}
            X_H = \sum_i \pdv{H}{p_i} \pdv{}{q_i} - \pdv{H}{q_i} \pdv{}{p_i},
        \end{equation}

        and $\forall f, g \in C^\infty(M)$
        \begin{equation}
        \begin{aligned}
            \{f, g\} &= - \omega (X_f, X_g), \quad \text{recall}\; \iota_{X_{f}}=df\\  
            &= - df (X_g)\\
            &= dg(X_f)\\
            &= X_f g\\
            &= \sum_i \pdv{f}{p_i} \pdv{g}{q_i} - \pdv{f}{q_i} \pdv{g}{p_i}.
        \end{aligned} 
        \end{equation}
        where in the first step, we have used $\iota_{X_f} \omega = df$,

        \item For $q_j$ and $p_i$, we have 

        \begin{equation}
            X_{q_j} = - \pdv{}{p_j} \quad \text{and} \quad X_{p_i} = \pdv{}{q_i},
        \end{equation}

        with Hamilton equations

        \begin{equation}
            \dot{q}_i = \pdv{H}{p_i} \quad \text{and} \quad
            \dot{p}_i = - \pdv{H}{q_i}.
        \end{equation}

        \item Observe that
        \begin{align}
            \{q_i, p_j\} &= X_{q_i} p_j = - \pdv{}{p_i} p_j = - \delta_{ij},\\
            \{q_i, q_j\} &= X_{q_i} q_j = 0 =  \{p_i, p_j\},\\
            \{p_i, q_j\} &= X_{p_i} q_j = \pdv{q_i}{q_j} = \delta_{ij}.
         \end{align}

         The set $S := \{p_1, \cdots, p_n, q_1, \cdots, q_n\}$ forms a complete set. Hence, we must have
         \begin{align}
             [Q(q_i), Q(q_j)] &= [\hat{q}_i, \hat{q}_j] = - i \hbar Q(\overbrace{\qty{q_i, q_j}}^{=0}) = 0,\\
             [\hat{p}_i, \hat{p}_j] &= - i \hbar \{\widehat{{p_i, p_j}}\} = 0,\\
             [\hat{p}_i, \hat{q}_j] &= - i \hbar \{\widehat{{p_i, q_j}}\} = - i \hbar \delta_{ij}.
         \end{align}

         These commutation relations define a Lie algebra, which is called \textbf{Heisenberg algebra}.

         By Schur's lemma, for the set $S$ with the above commutation relations, we observe:
         \begin{equation*}
             \text{Condition 5} \equiv \text{Finding irreducible representations of Heisenberg algebra}
         \end{equation*}
         That is why we called it the irreducibility condition.

         Using the Stone-von Neumann theorem, we note that any irreducible representation of the Heisenberg algebra is equivalent to $L^2 (\mathbb{R}^{n}) =: \mathcal{H}$, with $q_i$ and $p_j$ represented by
         \begin{equation}
         \begin{aligned}
            \mathbb{R}^{2n} &\xlongrightarrow{Q} \text{End} \left( L^2 (\mathbb{R}^n) \right)\\
            p_i &\longmapsto Q (p_i) \eqcolon \hat{p}_i\\
            q_j &\longmapsto Q(q_j) \eqcolon \hat{q}_j,
         \end{aligned}
         \end{equation}
         where $\forall \psi \in L^2(\mathbb{R}^{n})$
         \begin{equation}
         \begin{aligned}
             \hat{q}_i (\psi(x)) &\coloneq x^i \psi (x),\\
             \hat{p}_i (\psi(x)) &\coloneq- i \hbar \pdv{\psi(x)}{x^i}.
         \end{aligned}
        \end{equation}
         So we have indeed recovered the \textit{Schrödinger picture}.
        
    \end{itemize}
    \end{example}

    Now, we aim to extend the above construction to an arbitrary symplectic manifold $(M, \omega)$ by using a suitable mathematical analogue of $\mathcal{H}$ and the assignment $Q$. To this end, we will present the so-called \textit{geometric quantization formalism.}

    Before providing the actual formalism, let's start with a naive attempt:  let $(M,\omega)$ be a symplectic manifold and $(C^\infty(M), \qty{\cdot,\cdot})$ its Poisson algebra. We saw before that there exists a natural Lie algebra homomorphism
    \begin{equation}
    \begin{aligned}
        (C^\infty (M), \qty{\cdot,\cdot}) &\xlongrightarrow{\varphi_0} (\textsf{Hom}(M), [\cdot,\cdot]_0)\\
        f &\longmapsto X_f\\
        \qty{f,g} &\longmapsto X_{\qty{f,g}} = [X_f, X_g] = [\varphi(f), \varphi(g)].
    \end{aligned}
    \end{equation}
    Firstly, we can naively update the assignment as 
    \begin{equation}
       {\varphi_1}: f \longmapsto - i \hbar X_f.
    \end{equation}

    Then it should be clear to see that  ${\varphi}_1$ above satisfies the  quantization axioms $(1), (3),$ and $(4)$; however, $(2)$ is not satisfied.

    On the other hand, we previously saw that $\varphi_0$ is not injective and observed that any constant function (not just $f=1$) is mapped to $X_f = 0$. Indeed, if $f=c$ is a constant function, then we have
    \begin{equation}
    \begin{aligned}
        df = 0 &\iff \omega(X_f, \cdot) = df (\cdot) = 0\\
        &\iff \; X_f = 0  \quad(\text{since } \omega \ \text{is non-degenerate} ).
    \end{aligned}
    \end{equation}
    The second attempt is to let
    \begin{equation}
        {\varphi_2}: f \longmapsto - i \hbar X_f - f.
    \end{equation} You can check that Axiom $(2)$ is now satisfied. However, Axiom $(4)$ will no longer hold.

    Note that we have not defined any suitable $\mathcal{H}$ on which either of $\varphi_0, \varphi_1 \;\text{or} \; \varphi_2$ acts. A suitable assignment and the definition of $\mathcal{H}$ can be given by modifying the following case. (This is where geometric quantization will come into play.)

    In our third (more serious) attempt, we proceed as follows:

    If $M = T^* N$ for some $n$-manifold $N$, there exists a natural symplectic form $\omega$ on $M$ given by $\omega = d \theta$ globally, where $\theta$ is the canonical one-form. We then define the assignment by
    \begin{equation}
       {\varphi_3}: f \longmapsto\varphi_3 (f) := - i \hbar X_f - \theta (X_f) - f,
    \end{equation} which in fact satisfies all the axioms except the irreducibility condition $(5)$. To generalize the above assignment, we construct a suitable line bundle: Let $(M, \omega)$ be a symplectic manifold of dimension $2n$. Then Observe the following,
    \\
    
\begin{minipage}{0.55\textwidth}
    \begin{itemize}
        \item Since $\omega$ is a closed $2$-form, it defines a de Rham cohomology class $[\omega] \in H^2_{\text{dR}}(M)$.
        \item Since $\omega$ is closed, by Poincaré lemma, $\omega$ is \textit{locally exact}, meaning that there exists an open cover  $\{U_\alpha\}$ of $M$ such that $\omega = dA_\alpha$ on $U_\alpha$ and $A_\alpha \in \Omega^1 (U_\alpha)$.
    \end{itemize}
\end{minipage}
\hspace{0.5cm}
\begin{minipage}{0.4\textwidth}
    \begin{center}
        \includestandalone[width=0.7\textwidth]{Figures/Part2/p17}\\ 
        \hypertarget{fig44}{Figure 44.} A line bundle $L$ over $M$. 
    \end{center}
\end{minipage}   

\begin{itemize}
\item If $\left[ \frac{\omega}{2\pi\hbar} \right] \in H^2(M;2\pi \mathbb{Z})$, i.e. $\dfrac{\omega}{2\pi \hbar}$ represent an \textit{integral cohomology class}\footnote{It means $\displaystyle \int_{\Sigma} \dfrac{\omega}{2\pi \hbar} \in \mathbb{Z}$, with $\Sigma$ a closed oriented 2-manifold in $M$.}, then one can construct a particular complex line bundle $\mathcal{L}$ with a connection $\nabla$ as follows:
\end{itemize}
        
        Let $\{U_i\}$ be a local trivializing cover for $\mathcal{L}$ so that $\mathcal{L}|_{U_i}$ is trivial and $\omega$ is locally exact on each $U_i$. That is, we have
        \begin{equation}
            \omega = dA_i \;\text{on} \; U_i, \text{ for some } \; A_i \in \Omega^1(U_i).
        \end{equation}
        \begin{itemize}
            \item[$\boxed{a}$] Define a connection $\nabla$ as follows: on each $U_\alpha$
            \begin{equation}
                \nabla := d - \frac{i}{\hbar} A_\alpha.
            \end{equation}
            On each $U_\alpha$, connection one form, $-\frac{i}{\hbar}A_\alpha$, is given by $A_\alpha$ above.
            \item[$\boxed{b}$]
            On the overlap $U_\alpha \cap U_\beta$, one has local trivial sections $(U_\alpha, s_\alpha), (U_\beta, s_\beta)$ such that
            \begin{equation}
                s_\alpha: U_\alpha \xrightarrow{} \mathcal{L} \quad \text{and} \quad
                s_\beta: U_\beta \xrightarrow{} \mathcal{L},
            \end{equation}
            which can be seen on \hyperlink{fig45}{Figure 45}.
\begin{center}
        \includestandalone[width=0.4\textwidth]{Figures/Part2/p18}\\ 
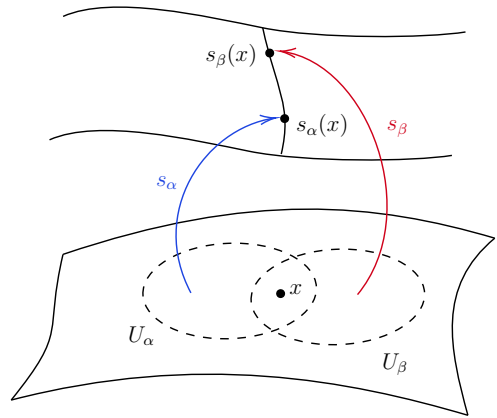
        \hypertarget{fig45}{Figure 45.} Two covers $U_\alpha$ and $U_\beta$ with local sections mapping to the same fiber. 
\end{center}

            Then, a gauge transformation
            \begin{equation}
                g: U_\alpha \cap U_\beta \xrightarrow{} S^1 \subset \mathbb{C},
            \end{equation}
            can be defined by
            \begin{equation}
                g(x) := e^{-\frac{i}{\hbar}f_{\alpha \beta}(x)},
            \end{equation}
            for some $f_{\alpha\beta} \in C^\infty (U_\alpha \cap U_\beta)$.

    In fact, on $U_\alpha \cap U_\beta$, \ $dA_\alpha = \omega = d A_\beta$. Hence,
    \begin{equation}
        d(A_\alpha - A_\beta) = 0 \quad \text{on} \quad U_\alpha \cap U_\beta,
    \end{equation}
    meaning it is closed on $U_\alpha \cap U_\beta$. Then, by the Poincaré lemma, $A_\alpha - A_\beta$ is locally exact. So, it can be written as
    \begin{equation}
        A_\alpha - A_\beta = df_{\alpha \beta},
    \end{equation}
    for some $f_{\alpha \beta} \in C^\infty(U_\alpha \cap U_\beta)$ or $C^\infty (V_{\alpha \beta}), \text{with} \; V_{\alpha \beta} \subset (U_\alpha \cap U_\beta)$. 

    \newpage
    As a result, we obtain
    \begin{equation}
        A_\alpha = A_\beta + df_{\alpha \beta},
    \end{equation}inducing the gauge transformation. Moreover, this relation defines how $U_\alpha$ and $U_\beta$ patch together. 

    Recall that under a gauge transformation $g$, we have
    \begin{equation} 
    \begin{aligned}
        g \cdot - \frac{i}{\hbar} A_\beta &= g \left( - \frac{i}{\hbar} \right) g^{-1} + dg g^{-1}  \quad \text{when} \quad g = e^{-if_{\alpha\beta}}\\
        &= - \frac{i}{\hbar} A_\beta + \frac{i}{\hbar} e^{-i f_{\alpha \beta}}df_{\alpha \beta} e^{f_{\alpha \beta}}\\
        &= - \frac{i}{\hbar} (A_\beta + df_{\alpha \beta}).
    \end{aligned}
    \end{equation}
    So,
    \begin{equation}
        \underbrace{g \cdot A_\beta}_{A_\alpha} = A_\beta + df_{\alpha \beta}
        \quad
        (A_\beta \xrightarrow{} A_\beta + df_{\alpha \beta}).
    \end{equation}
    \item[$\boxed{c}$] Now, consider the curvature two-form $F_A = dA + A \wedge A$. On $U_\alpha$, we have
    \begin{equation}
        F_{A_\alpha} = d\left( - \frac{i}{\hbar} A_\alpha \right)
        - \frac{i}{\hbar} A_\alpha \wedge - \frac{i}{\hbar} A_\alpha
        = - \frac{i}{\hbar} d A_\alpha.
    \end{equation}
    Since $1 \times 1$ matrices commute, this holds. Furthermore, recall that under the gauge transformation $g: U_\alpha \cap U_\beta \xrightarrow{}S^1$. We have
    \begin{equation}
        F_{A_\beta} = g \cdot F_{A_\alpha} = g F_{A_\alpha}g^{-1} = F_{A_\alpha}
        \quad 
        \text{on}
        \; \;
        U_\alpha \cap U_\beta.
    \end{equation}
    Therefore, $d A_\alpha$ is invariant under gauge change.

    This implies that $\{dA_\alpha\}$ can be patched into a closed $2$-form on $M$. It defines a cohomology class, called the \textit{first Chern class} of $\mathcal{L}$, and denoted by $c_1 (L) = \frac{i}{2\pi} [F_A] \in H^2_{\text{dR}} (M; \mathbb{R})$.
    \end{itemize}

    Now, calculate $F_{A_\alpha}$ above explicitly. By definition we have, let $(U_\alpha, s_\alpha)$ be a local trivialization section as above. Let $X, Y \in \Gamma(TM)$. Then on $U_\alpha$, we have
    \begin{equation}
    \begin{aligned}
        F_{A_\alpha} &= \nabla_X \nabla_Y s_\alpha - \nabla_Y \nabla_X s_\alpha \nabla_{[X, Y]} s_\alpha = \nabla_X \left( - \frac{i}{\hbar} A_\alpha (Y) s_\alpha \right)
        - \nabla_Y \left( - \frac{i}{\hbar} A_\alpha (X) s_\alpha \right)
        - \nabla_{[X,Y]} s_\alpha \\
        &=- \frac{i}{\hbar} A_\alpha(Y) \nabla_X s_\alpha - \frac{i}{\hbar} X(A_\alpha(Y)) s_\alpha + \frac{i}{\hbar} \big[ A_\alpha (X) \underbrace{\nabla_Y s_\alpha}_{-\frac{i}{\hbar} A_\alpha(Y) s_\alpha} + Y (A_\alpha(X) s_\alpha) \big] + \frac{i}{\hbar} A_\alpha ([X, Y]) s_\alpha \\
        &= \frac{1}{\hbar^2} \underbrace{\left[ - A_\alpha (Y) A_\alpha(X) s_\alpha + A_\alpha(X) A_\alpha(Y) s_\alpha \right]}_{=0 \; \text{since one-dimensional objects commute.}} - \frac{i}{\hbar} \left[ X(A_\alpha(Y)) - Y(A_\alpha(X)) - A_\alpha ([X, Y]) \right] s_\alpha \\
        &= - \frac{i}{\hbar} dA_\alpha (X,Y) s_\alpha = - \frac{i}{\hbar} \omega(X, Y) s_\alpha.
    \end{aligned}
    \end{equation}
    Therefore, on $U_\alpha$, we have
    \begin{equation}
        \boxed{
            F_A = - \frac{i}{\hbar} \omega.
        }
    \end{equation}
    We have the following theorem:
    \begin{theorem}\label{theorem_prequantum line bundle}
        Let $\omega$ be a close $2$-form on $M$ such that $\left[ \frac{\omega}{2\pi \hbar} \right] \in H^2 (M; 2 \pi \mathbb{Z}) \subseteq H^2_{\text{dR}}(M; \mathbb{R})$.

        Then there exists a complex line bundle $L \xrightarrow{} M$ and a connection $\nabla$ such that,
        \begin{itemize}
            \item Locally, $\nabla:= d - \frac{i}{\hbar}A_\alpha$ where $\omega = dA_\alpha$.
            \item Locally, $F_A = - \frac{i}{\hbar} \omega .$
        \end{itemize}
    \end{theorem}

    This line bundle is called the \textbf{pre-quantization line bundle}.

    \paragraph{Construction of GQ assignment.}
    
    Let $(M,\omega)$ be a symplectic manifold with $\left[ \frac{\omega}{2\pi \hbar} \right] \in H^2 (M;\mathbb{Z})$. Then, by the theorem above, there exists a complex line bundle $L \xrightarrow{\pi} M$ and a connection $\nabla$ such that $F_A = - \frac{i}{\hbar} \omega$. Then we have the following descriptions.

    \begin{itemize}
        \item Define $\mathcal{H}_{\text{pre}} := \Gamma(L)$ as the \textit{pre-quantization Hilbert space} and the inner product on $\mathcal{H}_{\text{pre}}$ is given by
        \begin{equation}
            \langle s_1,s_2 \rangle := \left( \frac{1}{2\pi \hbar} \right)^n
            \int_M H_X (s_1(X), s_2(X)) \;\mu .
        \end{equation}

        Here $H_\alpha$ is a Hermitian metric on each fiber $\mathbb{C}$ of $\mathcal{L}$, and $\mu$ is the Liouville measure.

        \item Define the \textit{geometric quantization assignment} to be the map
        \begin{align*}
            (C^\infty (M), \qty{\cdot,\cdot}) &\xlongrightarrow{Q} (\text{End} (\Gamma(L)), \;[.,.])\\
            f &\longmapsto Q(f),
        \end{align*}
        where
        \begin{equation}
            \boxed{
                Q(f) := - i\hbar \nabla_{X_f} - f.
              }
        \end{equation}
        One can check that this construction satisfies all the quantization conditions from $(0)$ to $(4)$ except $(5)$. To achieve $(5)$ as well, we still need to modify $\mathcal{H}_{\text{pre}}$ and $Q(f)$!

        \item The above assignment is a Lie algebra homeomorphism in the sense that
        \begin{equation}
            [Q(f), Q(g)] = - i \hbar Q(\{f, g\}).
        \end{equation}

        If we take $Q(f) = \nabla_{X_f} - \frac{i}{\hbar}f$, then we would have
        \begin{equation}
            [Q(f), Q(g)] := Q(\{f,g\}).
        \end{equation}
    \end{itemize}

    Indeed; $\forall f, g \in C^\infty(M)$ and for all $s \in \Gamma(L)$, we have 
    \begin{align}
        [Q(f), Q(g)] s &= [- i \hbar \nabla_{X_f} - f, -i \hbar \nabla_{X_g} - g]s \notag\\
        &= [- i \hbar \nabla_{X_f}, - i \hbar \nabla_{X_g}] s + \underbrace{ [-f, - i \hbar \nabla_{X_g}] s}_{(1)} + \underbrace{[-i\hbar \nabla_{X_f}, -g]}_{(2)} + \underbrace{[f,g]s}_{=0 \; \text{as} \; fgs - sgf = 0}.
    \end{align}
    Observe that
    \begin{align}
        (1) = [-f, - i \hbar \nabla_{X_g}] s &= i\hbar(f\nabla_{X_g} s - \nabla_{X_g} (fs)) \notag \\
        &= i \hbar (\cancel{f\nabla_{X_g} s} - \cancel{f \nabla_{X_g} s} - (X_g f) s )\notag \\
        &= -i \hbar (X_g f)s.
    \end{align}
    Similarly,
    \begin{align}
        2) = [i\hbar \nabla_{X_f}, g] = i\hbar(X_f g)s.
    \end{align}
    Then we have
    \begin{equation}
        [Q(f), Q(g)] = - \hbar^2 [\nabla_{X_f}, \nabla_{X_g}]s + i \hbar (X_f g - X_g f)s.
    \end{equation}
    Let's recall the definition of $F_A$ for convention
    \begin{equation}
        F_A(X_f, X_g) = \nabla_{X_f} \nabla_{X_g} - \nabla_{X_g} \nabla_{X_f} - \nabla_{[X_f, X_y]}.
    \end{equation}
    
    Also, recall the following
    \begin{equation}
        \{f, g\} = - \omega(X_f, X_g) = \omega(X_g, X_f) = \iota_{X_g} \omega(X_f)
        = dg (X_f) = X_f g.
    \end{equation}

    Using them we have
    \begin{equation}
    \begin{aligned}
        [Q(f), Q(g)] &= - \hbar^2 [\nabla_{X_f}, \nabla_{x_g}]s + i \hbar (X_f g - X_g f) s\\
        &= - \hbar^2 (F_A(X_f, X_g) + \nabla_{[X_f, X_g]})s + 2 i \hbar \{f, g\}\\
        &= i \hbar \underbrace{\omega (X_f, X_g)}_{-\{f,g\}}s - \hbar^2 \nabla_{[X_f, X_g]}s + 2 i \hbar \{f, g\}, \quad \text{using} \; F_A = - \frac{i}{\hbar} \omega.
    \end{aligned}
    \end{equation}
    Thus,
    \begin{equation}
        [X_f, X_g] = X_{\{f,g\}} = - \hbar^2 \nabla_{X_{\{f,g\}}} + i \hbar \{f,g\} = - i \hbar \underbrace{(-i \hbar \nabla_{X_{\{f,g\}}} - \{f,g\})}_{Q(\{f,g\})}.
    \end{equation}

\newpage 
    \begin{example}
        Let's recover the prototype example using the above framework: $(n=2)$

        Let $(M, \omega) = (\mathbb{R}^2, \omega = dq \wedge dp)$. Then, we have

        \begin{itemize}
            \item We saw before
            \begin{equation}
                X_q = - \pdv{}{p} \quad \text{and} \quad X_p = \pdv{}{q}.
            \end{equation}
            \item Note that $\mathbb{R}^2 \cong T^* \mathbb{R}$. Therefore, it admits a natural symplectic form $\omega = dq \wedge dp$ which is \textit{globally exact}. In fact, $\omega = d\theta$, where $\theta = q \wedge dp$ is a canonical $1$-form. Therefore, it defines a trivial cohomology class 
            \begin{equation}
                [\omega]=0 \in H^2 (M; \mathbb{Z}) \subset H^2_{\text{dR}} (M; \mathbb{R}).
            \end{equation}

            So, by Theorem \ref{theorem_prequantum line bundle}, there exists a suitable prequantum line bundle $L \xrightarrow{\pi} \mathbb{R}^2$, which is just the trivial one $L = \mathbb{R}^2 \times \mathbb{C}$, and a connection $\nabla$ is as defined above.

            Note that connection $1$-form $A$ is given by $A = - \dfrac{i}{\hbar} q dp $.

            \item We quantize $f\in C^\infty (\mathbb{R}^2)$ by the $GQ$ assignment
            \begin{equation}
                f \xmapsto{\; \; Q \; \;} Q(f) := - i \hbar \nabla_{X_f} - f,
            \end{equation}            
            where
            \begin{equation}
                \nabla_{X_f} = d - \frac{i}{\hbar} q dp.
            \end{equation}
        \end{itemize}

        For $f=q$, we have:
        \begin{equation}
        \begin{aligned}
            \hat{q} := Q(q) &= -i \hbar \nabla_{X_q} - q\\
            &= -i \hbar \nabla_{-\pdv{}{p}} - q\\
            &= -i\hbar \left( -\pdv{}{p} + \frac{i}{\hbar q} \right) - q  \\
            &= i \hbar \pdv{}{p} +\cancel{q} - \cancel{q},
        \end{aligned} 
        \end{equation}
        where the third line is done by using: \text{by using} $$ \nabla_{-\pdv{}{p}s} = \underbrace{ds \left(-\pdv{}{p}\right)}_{-\pdv{}{p}s} - \frac{i}{\hbar} q \underbrace{dp \left(-\pdv{}{p}\right)}_{-1} \forall s \in \Gamma(L).$$ Therefore, we have
        \begin{equation}
            \boxed{
                q \longmapsto \hat{q} := i \hbar \pdv{}{p}.
            }
        \end{equation}

        For $f=p$, we have
        \begin{equation}
        \begin{aligned}
            \hat{p} := Q(p) &= - i \hbar \nabla_{X_p} - p\\
            &= - i \hbar \nabla_{\pdv{}{q}} - p\\
            &= - i \hbar (\pdv{}{q} + 0) - p\\
            &= - i \hbar \pdv{}{q} - p.
        \end{aligned}
        \end{equation}
        Therefore, we obtain
        \begin{equation}
           \boxed{ p \longmapsto \hat{p} = - i \hbar \pdv{}{q} - p.}
        \end{equation}
    \end{example}

    \begin{remark}
        Even if we recover many familiar objects, the definition of the assignment is \textit{not} fully completed.
    \end{remark}

    \begin{remark}
        Recall that to obtain the complete quantization framework, we need to \textit{restrict the section $\Gamma(L)$} to the so-called \textit{P-polarized sections}, and \textit{modify $Q(f)$ and $\mathcal{H}_{\text{pre}}$} accordingly.
    \end{remark}

    \section{Sigma Model and Topological Quantum Field Theory}

    \subsection{Energy Functional}

    Let $(M, \gamma)$ and $(N,g)$ be a Riemannian manifold of dimension $m$ and $n$ respectively. Here $g$ is the metric. In local coordinates, we write
    \begin{equation}
        \gamma = (\gamma_{\alpha \beta})_{\alpha,\beta = 1, \cdots,m},
        \quad \quad
        g = (g_{ij})_{i,j = 1, \cdots, n}.
    \end{equation}

    Let $f \in \text{Map}(M,N)$ be a smooth map. We define the \textbf{energy functional} as follows.
    \begin{align}
        \text{Map}(M,N) &\xrightarrow{} \mathbb{R},\
        f \longmapsto E(f),
    \end{align}
    where
    \begin{equation}
           \boxed{ E(f) := \int_M \norm{df}^2 dM.}
    \end{equation}

    Here the measure is $dM :=\sqrt{\det{\gamma}} dx^1 \wedge \dots \wedge dx^m$ in local coordinates $(x^\alpha)_{\alpha = 1, \cdots,m}$ on $M$.

    Note that $\norm{\cdot}$ above involves the metric on $T^* M$ and $f^{-1} T^*N$. That is, let $(x^1, \cdots, x^m)$ and $(f^1, \cdots,f^n)$ be local coordinate charts on $M$ and $N$ respectively. Then one can show that $\forall x\in M$,
    \begin{equation}
        \norm{df}^2(x) = \gamma^{\alpha\beta} (x) g_{ij} (f(x)) \pdv{f^i (x)}{x^\alpha} \pdv{f^j}{x^\beta}.
    \end{equation}

    Here $g = (g_{ij})$ is the metric on $N$, and $\gamma^{\alpha \beta} := (\gamma_{\alpha \beta})^{-1}$ is the inverse metric on $M$.

    We evaluate RHS by taking orthonormal basis $e_{1},\dots,e_{m}$ for $T_{x}M$, then
    \begin{equation}
        \norm{df}^{2}(x) = \sum_{i=1}^{m}\expval{f_{x}(e_{i}),f_{x}(e_{i})},
    \end{equation}

    where $f_{x}(\cdot)$ is push forward. Hence, we write: $\forall f \in \text{Map} (M,N)$
    \begin{equation} \label{energy funcl in local}
        E(f) := \int_M \gamma^{\alpha \beta}(x) g_{ij} (f(x)) \pdv{f^i(x)}{x^\alpha} \pdv{f^j (x)}{x^\beta} dM.
    \end{equation}

    \begin{example}
        Take $M:= [0,1]$ and $N:=N$ for some $n$-dimensional Riemannian manifold as above. Also, take $f:[0:1] \rightarrow N$ a smooth curve in $N$.
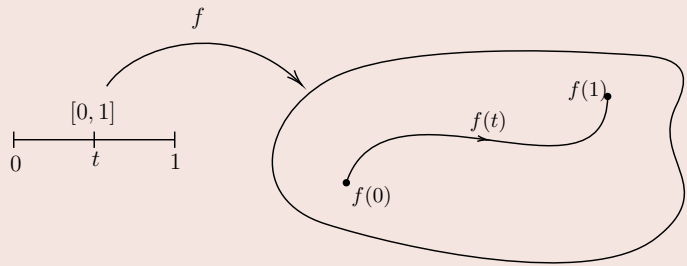
\begin{center}
            \includestandalone[width=0.60\textwidth]{Figures/Part2/p26}
\end{center}
        \begin{equation}
            f(t)
            =
            \begin{pmatrix}
                f^1(t)\\
                \vdots\\
                f^n(t)
            \end{pmatrix},
            \quad
            \text{where}
            \;
            f^i (t) := g^i_{\; 0} f(t).
        \end{equation}

         Then, $E(f)$ reduces to 
        \begin{equation}
            E(f) = \int_0^1 \underbrace{g_{ij}(f(t)) \dv{f^i}{t} \dv{f^j}{t}}_{\langle \dot{f}, \dot{f}\rangle_N} dt
            \quad \left(\dot{f}^i = \dv{f^i}{t} \right).
        \end{equation}

        This is a familiar version of the energy functional for paths in $N$.
\end{example}
Using the above setting, the corresponding \textit{Euler-Lagrange equations for Eqn. (\ref{energy funcl in local})} above can be given as follows:
        \begin{lemma}
            Euler-Lagrange equation for $i^{\text{th}}$ function $f^i$ (component-wise)  
            \begin{equation}
                \frac{1}{\sqrt{\det \gamma}}
                \pdv{}{x^\alpha}
                \left( \sqrt{\det \gamma} \gamma^{\alpha \beta} \pdv{}{x^\beta} f^i \right)
                +
                \gamma^{\alpha \beta} (x) \Gamma^i_{jk} (f(x))
                \pdv{f^j}{x^\alpha} \pdv{f^k}{x^\beta} = 0.
            \end{equation}
            
            Here $\Gamma$'s are given by               
            \begin{equation}
                \Gamma^i_{jk} = \frac{1}{2} g^{il} (-\partial_l g_{jk} + \partial_j g_{kl} + \partial_k g_{jl}).
            \end{equation}
            The solutions are called \textit{Harmonic maps}. In other words, the solutions are critical points of $E(f)$, where
            \begin{equation}
                \mathsf{crit}(f)=\{f: \delta E(f) = 0\}.
            \end{equation}
        \end{lemma}

    \begin{observation}
        \begin{itemize}
        
        \item Euler-Lagrange equations depend on the metric \textit{on the source} $(\gamma^{\alpha \beta})$ and the \textit{target} via $\Gamma^i_{jk}$.
        
        \item Consider the previous example: $M = [0,1]$ and $f : [0,1] \xrightarrow{} N$ smooth curve.

        For $E(f) = \int_M g_{ij} \dot{f}^i \dot{f}^j dt$: For $f^k$ lemma reduces to
        \begin{equation}
            \ddot{f}^k + \Gamma^k_{ij} (f(x)) \dot{f}^i \dot{f}^j = 0.
        \end{equation}

        This is the usual "geodesic equation" obtained by the variation $\delta E(f)$. It can also be obtained by the standard Euler-Lagrange equations
        \begin{equation}
            \pdv{\mathcal{L}}{q_i} - \dv{}{t} \pdv{\mathcal{L}}{\dot{q}_i} = 0
            \quad \text{where} \; 
            \mathcal{L} (f, \dot{f}) = g_{ij} \dot{f}^i \dot{f}^j.
        \end{equation}

        \item Why the name \textit{harmonic}?

        Consider the case with a smooth map $f:\mathbb{R}^2 \xrightarrow{}\mathbb{R}$. Then the corresponding Euler-Lagrange equation (by Lemma) is given by
        \begin{equation}
            \triangle f := \sum_{i=1}^2 \pdv[2]{f}{x_i} = 0,
        \end{equation}
        where $\triangle = \text{div} (\text{grad} (\cdot))$ is usual Laplacian on $\mathbb{R}^{2}$ with $(x_{1},x_{2})$. This can be found by using
        \begin{equation}
            \sqrt{\det{\gamma}} = 1, \quad \Gamma^i_{jk} = 0 \; \; \forall i, j, k \quad \text{and} \; \gamma^{ab} = \delta^{ab}.
        \end{equation}
        Then one gets
        \begin{equation}
            0 = \delta^{ab} \pdv{}{x^\alpha} \left( \pdv{f}{x^\beta} \right) = \delta^{ab} \frac{\partial^2 f}{\partial x^\alpha \partial x^\beta} = \sum_{i=1}^2 \pdv[2]{f}{(x^\alpha)}.
        \end{equation}
        \end{itemize}
    \end{observation}

    \subsection{Feynman Path Integral}

    Now we focus on the case of paths $\gamma: I \xrightarrow{} (M,g)$, with $\gamma(0) = a, \gamma(t) = b$. Then the energy functional is given by 
    \begin{equation}
        E(\gamma) = \int_0^1 g_{ij} (\gamma) \dot{\gamma}^i \dot{\gamma}^j dt.
    \end{equation}
    In classical mechanics, (by using D'Alambert's principle), we want to consider a particular path $\gamma$ that minimizes the action $E$ (i.e. $\gamma$ such that $\delta E|_{\gamma} = 0$). Here $\gamma$ must be a geodesic.

    Now, in the "quantum" case, to each path $\gamma$, one assigns a "probability function" $\exp\left(\frac{i}{\hbar} E(\gamma)\right)$ and integrate over the space $\mathcal{A} (a,b)$ of all paths connecting $a$ and $b$:

    \begin{equation}
        \int_{\mathcal{A}} \exp \left( \frac{i}{\hbar} E \right) d\mu (\gamma).
    \end{equation}

    Here $d\mu (\gamma)$ is some measure on $\mathcal{A}$. This integral is called the \textbf{Feynman path integral}.

    \begin{remark}
        In general, $\mathcal{A}$ is infinite-dimensional, and there exists no rigorously defined measure $\mu$ on $\mathcal{A}$. However, one can still use some asymptotic approaches.
    \end{remark}

    \begin{remark}
        When $\hbar \xrightarrow{} 0 $, one obtains classical minimizing paths.
    \end{remark}

    \subsection{Sigma Model}

    We will define the model using path integrals.
    Let $C_X := \text{Map} (X, M)$. Then we have:

    (1) If $\partial X = \phi$, then the partition function is
    \begin{equation}
        \forall f \in C_X, \quad Z(X) := \int_{C_X} \exp \left( \frac{i E(f)}{\hbar}  \right) \; \in \mathbb{C}.
    \end{equation}

    (2) If $\partial X = Y \neq \phi$, define $Z(X)$ on $C_Y := \text{Map} (Y, M)$ as follows:

\begin{center}
        \includestandalone[width=0.35\textwidth]{Figures/Part2/p29}\\ 
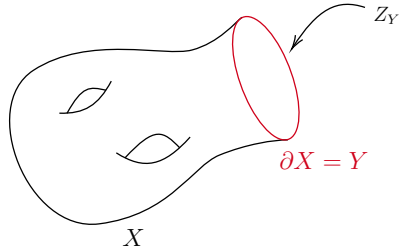
        \hypertarget{fig46}{Figure 46.} A vector space $Z_Y$ defined on the boundary $Y$ of a manifold $X$.
\end{center}
    \begin{equation}
        \forall \alpha \in C_Y = \text{Map} (Y, M),
    \end{equation}
    \begin{equation}
        Z(X) (\alpha) := \int_{C_X(\alpha)} \exp \left( \frac{i}{\hbar} E(f) \right) \in \; Z_Y,
    \end{equation}
    where 
    \begin{equation}
        C_X(\alpha) := \left\{ f \in C_X : f|_{\partial X = Y} = \alpha \right\}.
    \end{equation}
    Here, $Z_Y$ is some vector space of functions on $C_Y$.

    The idea is that even though $Z(X)$ may \textit{not} be rigorously defined, we can discuss "the expected" properties of $Z(X)$ and the corresponding \textit{vector space $Z_Y$ associated to the boundary $Y$}. This is actually where geometric quantization will come in handy.

\newpage
    \subsection{Axioms for Sigma Model}

    \begin{itemize}
        \item[$1)$] \textit{Orientation}: $\displaystyle Z_{-Y} \cong Z_Y^*$. Here, $-Y$  means $Y$ with the opposite orientation. 
        \begin{center}
                \includestandalone[width=0.35\textwidth]{Figures/Part2/p30.1}\\ 
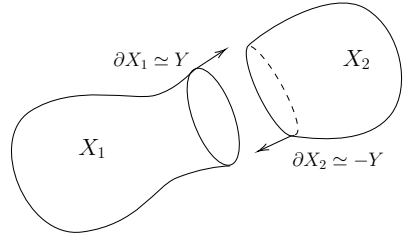
        \hypertarget{fig47}{Figure 47.} Boundaries with reversed orientation: $\partial X_1 = - \partial X_2$.
        \end{center}
        \begin{equation*}
            Y \longmapsto Z_Y \; \& \; -Y \longmapsto Z_Y^*. \quad (\text{dual of $Z_Y$})
        \end{equation*}
        \item[$2)$] \textit{Multiplication}: $\displaystyle Z_{Y_1 \amalg Y_2} = Z_{Y_1} \otimes Z_{Y_2}, $
        
        \begin{center}
        \includestandalone[width=0.4\textwidth]{Figures/Part2/p30.2}\\ 
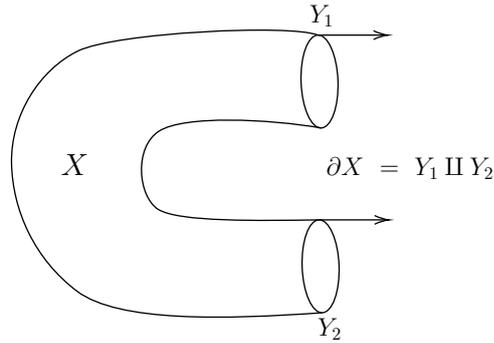
        \hypertarget{fig48}{Figure 48.} Both $Y_1$ and $Y_2$ constitute a boundary for $X$ through the disjoint union $\partial X = Y_1 \amalg Y_2$,
        \end{center}
        \begin{equation*}
            Y_1 \amalg Y_2 \longmapsto Z_{Y_1 \amalg Y_2} = Z_{Y_1} \otimes Z_{Y_2}.
        \end{equation*}
        
        \item[$3)$] \textit{Gluing}: $\displaystyle X = \frac{X_1 \amalg X_2}{Y}.$
        \begin{center}
        \includestandalone[width=0.4\textwidth]{Figures/Part2/p30.3}\\
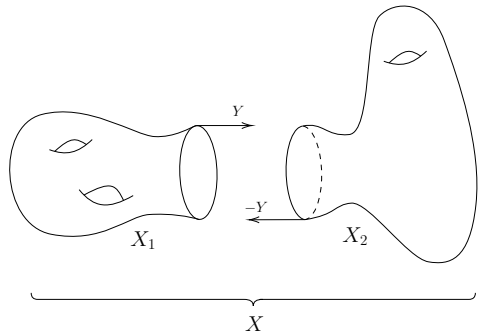
        \hypertarget{fig49}{Figure 49.} The manifolds $X_1$ and $X_2$ are glued along their boundaries.
        \end{center}
\newpage

        That means $X$ can be obtained by gluing $X_1$ and $X_2$ along their boundaries: Pick a diffeomorphism $\varphi: Y \xrightarrow{\simeq} Y$ and define an equivalence relation "$\sim$" by saying "$a\sim b$" in $Y$ if and only if $b = \varphi(a)$. And then we obtain the quotient space
        \begin{equation}
            X:= \frac{X_1 \amalg X_2}{\sim}.
        \end{equation}
        By the above axioms, we have
        \begin{equation}
            Y \longmapsto Z_Y \;, \; -Y \longmapsto Z_Y^*
        \end{equation}
        and 
        \begin{equation}
            Z(X_1) \in Z_Y \; , \; Z(X_2) \in Z_Y^*.
        \end{equation}
        We have a natural pairing
        \begin{equation}
            Z_Y \otimes Z_Y^* \xrightarrow{\langle \cdot, \cdot \rangle } \mathbb{C},
        \end{equation}
        and we impose
        \begin{equation}
            Z(X) = \langle Z(X_1), Z(X_2) \rangle.
        \end{equation}
    \end{itemize}
\subsection{Some Background from Category Theory}
    Now, we aim to reformulate these axioms by means of \textit{category-theoretic} tools and by introducing a suitable notion, so-called \textbf{Topological Quantum Field Theory}, that formalizes the above setup. Let's begin with some terminology and notation.

    \begin{definition}
        A \textbf{category} $\mathcal{C}$ consists of the following data:
        \begin{itemize}
            \item \textit{Objects}, denoted by $\textsf{Obj}(\mathcal{C})$.
            \item \textit{Morphisms}, $\forall A,B \in \textsf{Obj}(\mathcal{C})$, we have a set of morphisms ("arrows") from $A$ to $B$ ($A \xrightarrow{} B$), denoted by $\textsf{Mor}(A,B)$ such that
            \begin{itemize}
                \item[(a)] Composition morphism exists:
                \begin{equation}
                \begin{aligned}
                    \textsf{Mor} (A,B) \times \textsf{Mor} (B,C) &\longrightarrow{} \textsf{Mor}(A,C)\\
                    (f,g) &\longmapsto f \circ g.
                \end{aligned}
                \end{equation}
                \item[(b)] Composition is associative.
                \item[(c)] $\forall A \in \textsf{Obj}(\mathcal{C}), \; \exists\; \text{id}_A \in \textsf{Mor} (A,A)$.
            \end{itemize}
        \end{itemize}
    \end{definition}

    \begin{example}
        \textsf{Top}: The category of topological spaces, 
        \begin{itemize}
        \itemsep0em 
            \item \textit{Objects}: Topological spaces $X$.
            \item \textit{Morphisms}: continuous maps between topological spaces $f:X \rightarrow Y$.
        \end{itemize} 
    \end{example}

    \begin{example}
        $\textsf{Sets}$: The category of sets, 
        \begin{itemize}
        \itemsep0em 
            \item \textit{Objects}: Sets.
            \item  \textit{Morphisms}: Maps between sets $A \xrightarrow{\varphi} B$.
        \end{itemize}
    \end{example}

    \begin{example}
        $\textsf{{Vec}}_k$: The category of vector spaces over a field ${k}$.
        \begin{itemize}
        \itemsep0em 
            \item \textit{Objects}: Finite dimensional vector spaces .
            \item \textit{Morphisms}: $V \xlongrightarrow{L} W$, where $L$ is a linear transformation.
        \end{itemize}
    \end{example}

    \begin{example}
        $\textsf{Ab}$: The category of Abelian groups.
        \begin{itemize}
        \itemsep0em 
            \item \textit{Objects}: Abelian groups.
            \item \textit{Morphisms}: Group homomorphisms between abelian groups.
        \end{itemize}
    \end{example}

    \begin{example}
        $\textsf{{Cob}}(n+1)$: The category of $(n+1)$ dimensional cobordisms.
        \begin{itemize}
        \itemsep0em 
            \item \textit{Objects:} Closed, orientable, smooth $n$-manifolds.
            \item \textit{Morphisms}: Let $X, Y \in \textsf{Obj}(\text{Cob}(n+1))$.
        \end{itemize}
        
        A morphism from $X$ to $Y$ is defined by a \textbf{cobordism} from $X$ to $Y$.

        Note that a cobordism from $X$ to $Y$ is an compact, orientable, smooth $(n+1)$-dimensional manifold $W$ with $\partial W = X \amalg -Y$.

        \begin{center}
        \includestandalone[width=0.34\textwidth]{Figures/Part2/p32}\\ 
        \hypertarget{fig50}{Figure 50.} Cobordism from $X$ to $Y$.
        \end{center} 
        In that case, we denote a cobordism from $X$ to $Y$ simply by $W: X\Rightarrow Y$.
    \end{example}

    As one can realize, category theory is a sort of \textit{unifying theme} in mathematics.
    
    Now, it is natural to as how can we relate two given categories?

    \begin{definition}
        Let $\mathcal{C}$ and $\mathcal{C}'$ be two categories. A \textbf{covariant functor} $\mathcal{F}$ from $\mathcal{C}$ to $\mathcal{C}'$ is defined as follows.

        \begin{itemize}
            \item It is a map from $\textsf{Obj}(\mathcal{C})$ to $\textsf{Obj}(\mathcal{C}')$. That is
            \begin{equation}
                \forall A \in \text{Obj}(\mathcal{C}), \; \exists \mathcal{F}(A) \in \text{Obj}(\mathcal{C}'),
            \end{equation}
            such that $\mathcal{F}$ maps objects to objects.

            \item For any $A_1, A_2 \in \textsf{Obj}(\mathcal{C})$ and $m: A_1 \longrightarrow A_2 \in \textsf{Mor}(A_1, A_2)$, we have
            \begin{equation}
                \mathcal{F} (m): \mathcal{F}(A_1) \longrightarrow \mathcal{F}(A_2).
            \end{equation}

            It maps from morphisms to morphisms.
            \begin{equation}
                \forall m \in \textsf{Mor}(A_1, A_2), \; \mathcal{F} (m) \in \textsf{Mor} (\mathcal{F}(A_1), \mathcal{F}(A_2)),
            \end{equation}
            such that
            \begin{itemize}
                \item It maps identity morphism to identity morphism.
                \begin{equation}
                    \mathcal{F}(\text{id}_A) = \text{id}_{\mathcal{F}(A).}
                \end{equation}
                \item Given $m_1\in \textsf{Mor}(A,B)$ and $m_2 \in \textsf{Mor}(B,C)$, we have
                \begin{equation}
                \begin{aligned}
                    \mathcal{F}(m_2 \circ m_1) &= \mathcal{F}(m_2) \circ \mathcal{F}(m_1),
                \end{aligned}       
                \end{equation}
                which means it respects composition. 
                
                \[\begin{tikzcd}
\mathcal{F}(A) \arrow[r, "\mathcal{F}(m_1)"] \arrow[rr, "\mathcal{F}(m_2 \circ m_1)"', bend right] & \mathcal{F}(B) \arrow[r, "\mathcal{F}(m_2)"] & \mathcal{F}(C).
\end{tikzcd}\]             
                This is why the name "covariant" is used. Otherwise, it would be "contravariant".
            \end{itemize}
        \end{itemize}
    \end{definition}

    \begin{example}
        \textit{Forgetful Functor}: It basically forgets structures!
        \begin{equation}
            \textsf{Top} \xlongrightarrow{\mathcal{F}_o} \textsf{Sets}.
        \end{equation}
It sends a topological space $X$ to its underlying set. Similarly, it forgets the continuity of maps and considers them as maps between sets.
    \end{example}

    \begin{example}
        \textit{Homology Functor}:
        \begin{align}
            \textsf{Top} &\xlongrightarrow{H_i} \textsf{Ab}\\
            X &\longmapsto H_i(X) \quad (\text{i$^{th}$ homology group of} \; X)\\
            (f: X \rightarrow Y) &\longmapsto \left(f_*: H_i (X) \longrightarrow H_i(Y)\right),
        \end{align} where $f_*$ is the {\text{induced map}} on homology.
    \end{example}

    \subsection{Topological Quantum Field Theory (TQFT) as a functor}
    Using the terminology of the previous section (in the sense of Atiyah\cite{Atiyah:1989vu}), we now define a \textbf{TQFT} $\mathcal{Z}$ to be a \textit{functor} from $\textsf{Cob}(n+1)$ to $\textsf{Vect}_k$ acting on objects as
    \begin{align}
        \textsf{Cob} (n+1) &\xlongrightarrow{\mathcal{Z}} \textsf{Vect}_k\\
        X &\longmapsto \mathcal{Z}(X),
    \end{align}
    a finite-dimensional vector space over k, and $\mathcal{Z}$ maps a cobordism $W: X \Rightarrow Y$, as in \hyperlink{fig50}{Figure 50}, to a linear transformation between vector spaces
    \begin{equation}
        \mathcal{Z}(X) \xrightarrow{\mathcal{Z}(W)} \mathcal{Z}(Y).
    \end{equation}
    Observe that the axioms of the sigma model can be recovered from the definition of a functor.
    \begin{example}
        As above, if $W: X \Rightarrow  Y$ a {cobordism} from $X$ to $Y$, with $\partial W = X \amalg - Y$, we have
        \begin{equation}
        \begin{aligned}
            \mathcal{Z}(W) &\in \textsf{Hom} (\mathcal{Z}(X), \mathcal{Z}(Y))\\
            &\simeq \mathcal{Z}(X)^* \otimes \mathcal{Z}(Y).
        \end{aligned}   
        \end{equation}
        This is by recalling that for all finite dimensional vector spaces $V_1, V_2$,
        \begin{equation}
            \textsf{Hom}_k(V_1, V_2) \cong V_1^* \otimes V_2.
        \end{equation}
        If $X = \phi$, i.e. $W$ is a cobordism from $\phi$ to $Y$, with $\partial W = Y$, we have
    \begin{equation}
        Z(W) \in \underbrace{Z(\phi)^*}_k \otimes Z(Y) \underbrace{\simeq Z(Y)}_{\text{since } k \otimes V \simeq V}.
    \end{equation}
    Similarly, if $Y$ has reversed orientation, i.e. $W$ is a cobordism from $Y$ to $\phi$. We get
    \begin{equation}
        Z(W) \in \textsf{Hom}(Z(Y), \underbrace{Z(\phi)}_{k}) = Z(Y)^* ,
    \end{equation} 
    where
    \begin{center}
        \includestandalone[width=0.32\textwidth]{Figures/Part2/p35.2}
    \end{center} 
    Therefore, the "orientation" axiom is recovered!
    \end{example}
    Similarly, the other axioms can also be realized directly from the definition of functor. 

\chapter{Application: Knot theory and Witten's Work
}

Classically, a knot invariant can be built using \textit{algebraic, geometric, and statistical methods} \cite{Wu:1992zzb}. For example, an invariant of ambient isotopy—such as a number, group, or polynomial linked to a knot/link that remains unchanged under Reidemeister moves. 
\vspace{3pt}

Each approach is interesting historically and technically; however, all require a suitable 2D projection of the knot to evaluate the invariant. The algebraic method uses the \textit{braid group representation} for the knot (e.g., Jones' works \cite{Jones:1985dw,Jones:1989ed}), the geometric approach studies \textit{skein relations} \cite{PrasolovSossinsky1997, Wu:1992zzb}, and the statistical mechanical method constructs a \textit{lattice model} where the partition function yields a knot invariant. Some models produce relations similar to skein relations, giving rise to certain knot polynomials. This chapter focuses on Witten's 3D approach \cite{Witten:1988hf}.

The outline is as follows:

\begin{enumerate}
    \item A brief introduction to knot theory. (Main ref.: Wu’s paper \cite{Wu:1992zzb})
    \item Witten’s paper: Main setup and results. (Main ref.: \cite{Witten:1988hf})
\end{enumerate}

\newpage
\section{A Review of Knot Theory}

\subsection{Knots, Links, and Diagrams}

\begin{definition}
    \begin{enumerate}
        \item An \textbf{oriented knot} $K$ is an embedding of a circle $S^1$ into $\mathbb R^3$ (or into its one point compactification $S^3$).
        \item An \textbf{oriented link} $L$ with $m$-components is an embedding of a disjoint union of $m$ copies of $S^1$ into $\mathbb R^3$ (or $S^3$):
        \begin{equation}
        \begin{aligned}
            S^1 \quad\quad\ &\xlongrightarrow[]{ \ \varphi \ } \mathbb R^3, \quad \ \varphi(S^1)  \ \ =: K \\
            \underbrace{S^1 \amalg \cdots \amalg S^1}_{m} &\xlongrightarrow[]{\ \psi \ } \mathbb R^3, \quad \psi (\amalg S^1) =:L
        \end{aligned}
        \end{equation}
    \begin{center}
            \includestandalone[width=0.65\textwidth]{Figures/Part3/p1.1}\\
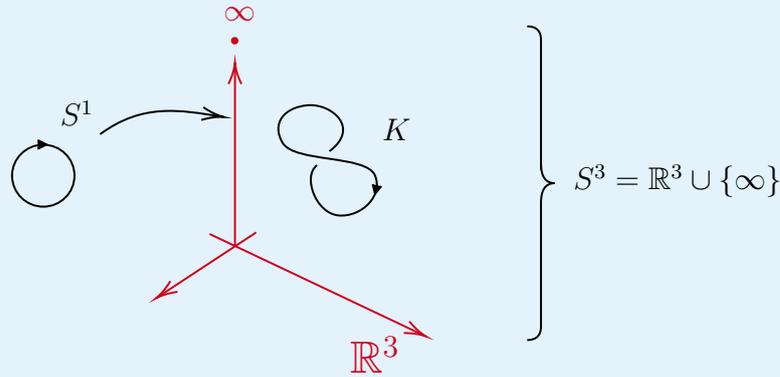
            \hypertarget{fig51}{Figure 51. }A figure-8 knot formed via an embedding of $S^1$ into $\mathbb{R}^3$
    \end{center}
        \begin{center}
            \includestandalone[width=0.5\textwidth]{Figures/Part3/p1.2}\\
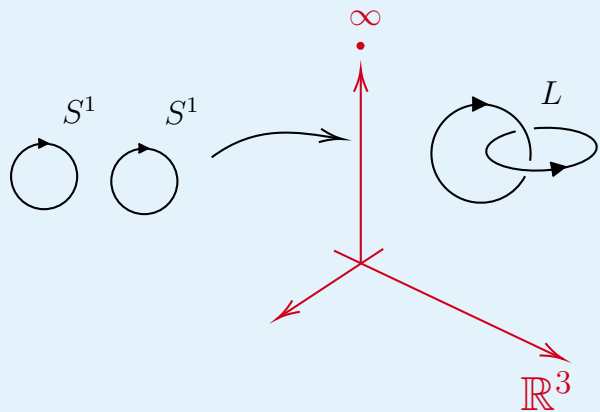
            \hypertarget{fig52}{Figure 52. } A 2-link (Hopf Link) formed via an embedding of two $S^1$s into $\mathbb{R}^3$.
    \end{center}
    \end{enumerate}
\end{definition}

\begin{remark}
    \begin{enumerate}
        \item The trivial embedding is called an \textbf{unknot}, i.e. $\varphi(S^1) = S^1$.
        \item For all knots $K$, we have a homeomorphism $K \cong S^1$. That is to say, each knot (closed, simple, smooth curve) is homeomorphic (topologically equivalent) to a circle.  Therefore, we need a stronger notion to distinguish two given knots $K_1$, $K_2$. This is the notion of isotopy (ambient) classes. For example, the unknot and the basic trefoil knot are both homeomorphic to the circle, but they are not isotopic.
    \begin{center}
            \includestandalone[width=0.35\textwidth]{Figures/Part3/p1.3}
    \end{center}
    \end{enumerate}
\end{remark}

\begin{definition}
    Let $K_1$ and $K_2$ be two knots. An \textbf{ambient isotopy} between $K_1$ to $K_2$ is a one parameter family $h_t$ of diffeomorphisms (smoothly depending on $t$)
    \begin{equation}
        h_t : [0,1] \times \mathbb R^3 \xlongrightarrow[\ \ \  ]{} \mathbb R^3, \quad h(t,x) = h_t(x).
    \end{equation}such that 
    \begin{equation}
    \begin{aligned}
        &1) \ h_0 = \text{id}_{\mathbb R^3}, \\
        &2) \ h_1(K_1) = K_2.
    \end{aligned}
    \end{equation}
    Here, $h_t \in \textsf{Diff}(\mathbb R^3)$ and $h_t : \mathbb R^3 \to \mathbb R^3$ is a diffeomorphism for all $t$. 
    
\end{definition}

\begin{observation}
    An isotopy $h_t$ defines a smooth deformation from $K_1$ to $K_2$.
    \begin{center}
            \includestandalone[width=0.8\textwidth]{Figures/Part3/p2.1}\\
            \hypertarget{fig53}{Figure 53. } Isotopies of $K_1 = S^1$.
    \end{center}
\end{observation}

We can define an \textit{equivalence relation} for knots as follows:

\begin{definition}
    $K_1 \sim K_2$ if and only if there exists an ambient isotopy going from $K_1$ to $K_2$.
\end{definition}

If there exists an ambient isotopy from $K_1$ to $K_2$, $K_1$ and $K_2$ are said to be \textbf{ambient isotopic.}

    \begin{center}
            \includestandalone[width=0.5\textwidth]{Figures/Part3/p2.2}\\
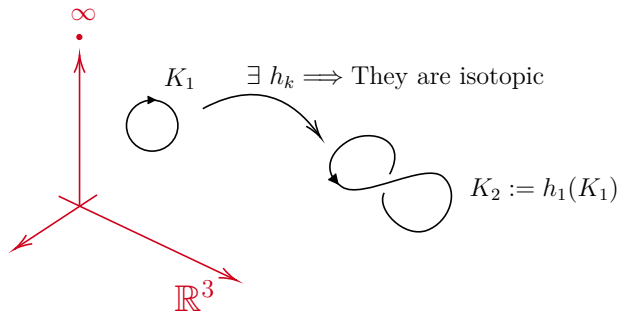
            \hypertarget{fig54}{Figure 54.} The unknot is isotopic to a figure-8 knot.
    \end{center}

\newpage
The study of knot theory aims to classify equivalence classes of knots defined by ambient isotopy. To understand which knots are isotopic (and to construct a convenient algorithm that can be used to detect isotopic behavior), we introduce the notion of "knot diagrams/projections". In other words, to study knots, we consider their suitable, perpendicular, 2D projections onto the plane.

\begin{definition}
    Let $K$ be a knot. A {\bf  knot diagram} $D(K)$ (or just use $K$ again) is the perpendicular projection of $K$ onto the plane such that
    \begin{enumerate}
        \item The projection maps tangent lines to tangent lines at all points $p \in K$.
        \item \emph{No more than two points} of $K$ are projected onto the same point.
        \item The set of "crossing points" is finite, and at each crossing, the projections of two tangents do not coincide.
        \begin{center}
            \includestandalone[width=0.47\textwidth]{Figures/Part3/p3.1}\\
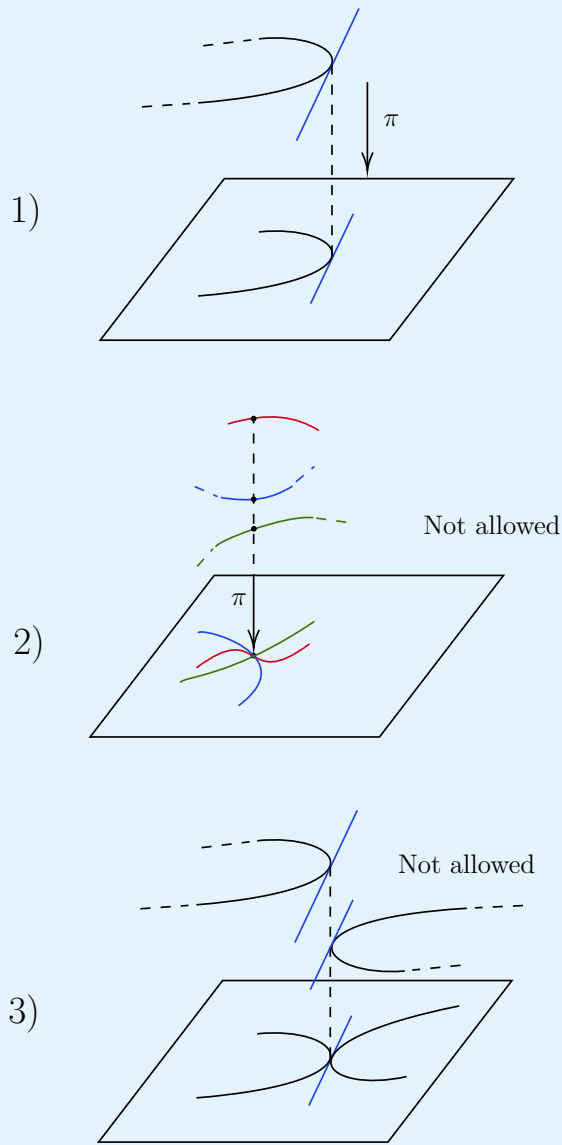
            \hypertarget{fig55}{Figure 55. } (1) The projection preserves the tangent lines. (2) Two points of $K$ cannot be mapped to the same point on the plane. (3) Projections of tangent lines onto the same line on the plane is not allowed.
        \end{center}
    \end{enumerate}
\end{definition}

\begin{example}
    The following figure indicates the three-dimensional nature of the theory and also gives information about the \emph{relative position} of two strip segments in space. Terminology: "overcrossing" and "undercrossing".
      \begin{center}
        \includestandalone[width=0.65\textwidth]{Figures/Part3/p3.2}\\
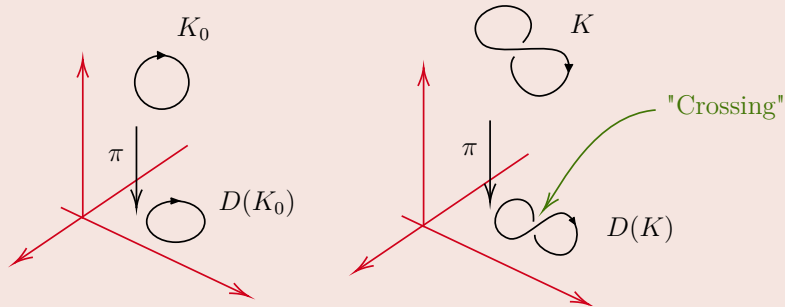
            \hypertarget{fig56}{Figure 56. } Knots in $\mathbb{R}^3$ are projected onto the $xy$-plane.
    \end{center}
    
\end{example}
\textbf{Convention:} Our sign convention will be as follows:
        \begin{center}
            \includestandalone[width=0.5\textwidth]{Figures/Part3/p4.1}\\
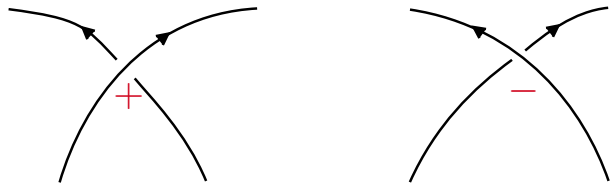
            \hypertarget{fig57}{Figure 57. } A $(+)$ sign is given whenever an arrow coming from the left passes \textit{over} the arrow coming from the right. In the converse, the crossing has a $(-)$ sign.
        \end{center}
        
\begin{example}
    Some examples of knots where the crossings are labeled:
        \begin{center}
            \includestandalone[width=0.75\textwidth]{Figures/Part3/p4.2}\\
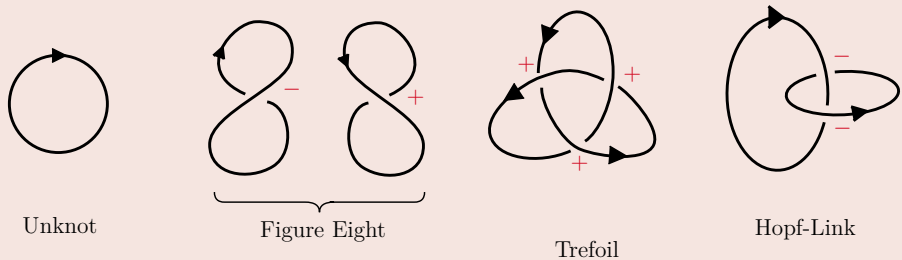
            \hypertarget{fig57}{Figure 57. }Some possible orientations of given knots. 
        \end{center}
\end{example} 

\begin{observation}
    Given an unoriented link $K$ consisting of $n$-components, one can "orientize" $K$ in $2^n$ different ways.\footnote{For example, in \hyperlink{fig57}{Figure 57} there are eight examples.} If $n=1$, it is just a knot, and there are $2^1 = 2$ ways to orient it as shown.    
\end{observation}

Why do we care about the planar projection? Because there exist some \textit{line moves} in the planar projection that produce isotopic knots in the space. These are called \textbf{Reidemeister moves}. \\
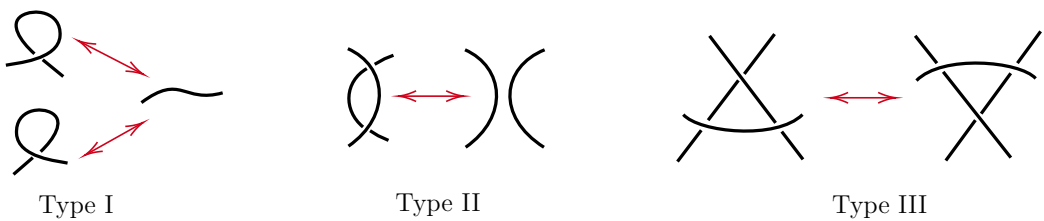
        \begin{center}
            \includestandalone[width=0.85\textwidth]{Figures/Part3/p4.3}
        \end{center}
In a more concrete picture, with orientations involved: \\ 

\begin{minipage}{0.29\textwidth}
    \begin{center}
            \includestandalone[width=0.7\textwidth]{Figures/Part3/p5.1}
        \end{center}
\end{minipage}
\begin{minipage}{0.29\textwidth}
    \begin{center}
        \includestandalone[width=0.9\textwidth]{Figures/Part3/p5.2}
    \end{center}
\end{minipage}
\begin{minipage}{0.29\textwidth}
        \begin{center}
            \includestandalone[width=0.4\textwidth]{Figures/Part3/p5.4}
        \end{center}
\end{minipage}\\
\begin{center}
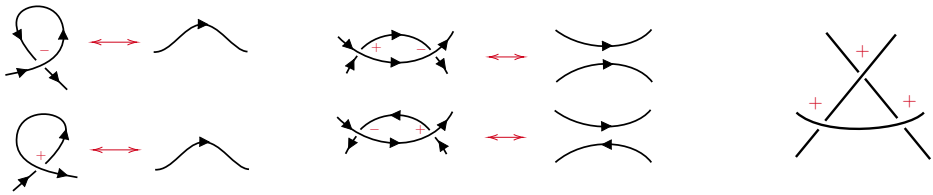
    \hypertarget{fig58}{Figure 58. } Type I (Left), Type II (Middle), and Type III (Right) Reidemeister moves.
\end{center}

\begin{definition}
    A quantity\footnote{e.g., a number, group, polynomial assigned to a given knot (link) $K$. } $Z$ assigned to a knot $K$ is said to be an \textbf{invariant of ambient isotopy} if it is invariant under all Reidemeister moves. Such a quantity is called an \textbf{invariant of regular isotopy} if it is invariant under type II and type III moves.
\end{definition}

Our approach to the problem of classifying knots will be to construct invariants for isotopy classes of knots (links), i.e., if $K_1 \sim K_2$ then $Z(K_1) = Z(K_2)$.\\

\textbf{Possible approaches to construct a knot invariant:} 
\begin{itemize}
    \item Algebraic approach: Finding the braid group representation of a given knot.
    \item Geometric approach: Using knot diagrams and skein relations.
    \item Statistical mechanics approach: Assign to a given knot diagram a lattice model and evaluate its partition function.
\end{itemize}
These are classical approaches. They all need the 2-dimensional projection of a given knot (i.e., knot projections), even if the theory itself is 3-dimensional. We also have the following approaches.

\begin{itemize}
    \item Witten's Chern-Simons approach: Assigning a Wilson line to a given knot, and computing its expectation value.
    \item Reshetikin and Turaev approach: Using quantum groups.
    \item many other modern approaches...
\end{itemize}
This part uses 3-dimensional Chern-Simons theory to recover the well-known knot polynomials. Due to the appearance of the partition function $Z$ in the construction, which is challenging to define, technical issues arise.

Our aim is to study Witten's approach. Before doing that, we focus on the geometric approach to warm up.

A question is how we detect the isotopy class of a knot via given knot diagrams?

\begin{theorem}
    \textbf{(Reidemeister)}
    
    Let $K_1$ and $K_2$ be two knots (links). Let $D(K_1)$ and $D(K_2)$ be knot diagrams for $K_1$ and $K_2$, respectively. Then 
$$
   \begin{array}{l}
   K_1 \; \text{and} \; K_2 \; \text{are}\\
   \text{(ambient) isotopic.}
  \end{array}
 \qquad \Longleftrightarrow \qquad
  \begin{array}{l}
    {D(K_2) \; \text{can be obtained from} \; D(K_1)} \\
    \text{through finitely many Reidemeister} \\
    \text{moves and the plane isotopy.}
  \end{array}
 $$
\end{theorem}

\begin{example}
    Let $D(K_1)$ be an unknot and $D(K_2)$ be a figure-8 knot with a crossing at the point $P$ as seen in \hyperlink{fig59}{Figure 59}. By using the first Reidemeister move at the crossing $P$, one has 

    \begin{center}
        \includestandalone[width=0.75\textwidth]{Figures/Part3/p7.1}\\
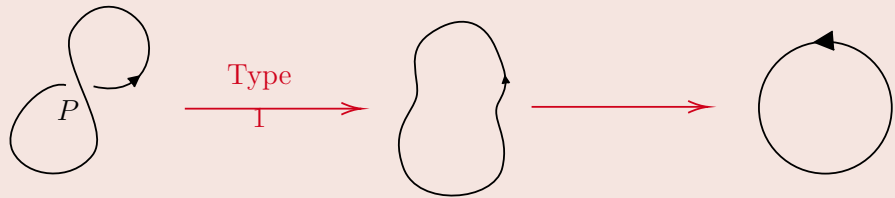
        \hypertarget{fig59}{Figure 59. } Isotopy between a figure-8 knot and the unknot is shown via Type I Reidemeister move.
    \end{center}

    Hence, $K_1$ and $K_2$ are ambient isotopic.    

    \begin{center}
        \includestandalone[width=0.75\textwidth]{Figures/Part3/p7.2}\\
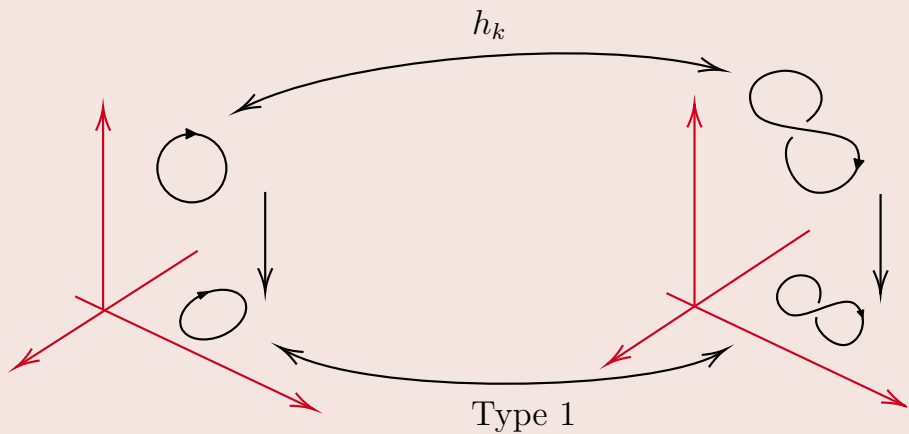
        \hypertarget{fig60}{Figure 60. } Isotopy between the unknot and the figure-8 knot.
    \end{center}
\end{example}

A naive construction for knot invariants is that there exists a natural "number" assigned to a given knot. Let $K$ be an oriented knot. We define the \textbf{writhe} $w(K)$ of $K$ by
\begin{equation}
    w(K) := n_+ - n_-,
\end{equation}
where $n_\pm$ denotes the number of $\pm$ crossings. This quantity is determined uniquely for a given knot $K$, and is independent of the orientation. 

\newpage
\begin{example} Let $K$ be a trefoil and $K'$  be the figure-8 knot as below.
   \begin{center}
        \includestandalone[width=0.25\textwidth]{Figures/Part3/p8.1}
    \quad 
        \includestandalone[width=0.25\textwidth]{Figures/Part3/p8.2}
    \end{center}Then $w(K)=3$ and $w(K')=-1$.
\end{example}

\begin{observation}
    $w$ is an invariant of regular isotopy as $w$ fails to be invariant under the Type-I move. Clearly, we have
    \begin{center}
        \includestandalone[width=0.35\textwidth]{Figures/Part3/p8.3}
    \end{center}
    
    (similarly for "-" crossing), when we resolve the crossing, the number of crossings has changed. 
\end{observation}

\subsection{The skein relation (for ``oriented" knots) and knot polynomials}

The \textbf{skein relation} is a recurrence relation relating the invariant of knots whose diagrams are identical except for a small neighborhood of a particular crossing.\\ 

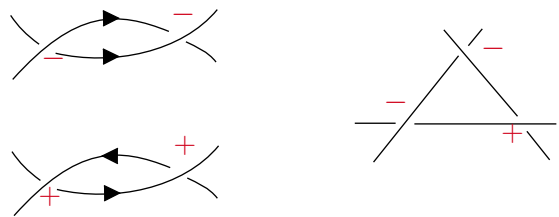
    \begin{center}
        \includestandalone[width=0.45\textwidth]{Figures/Part3/p9.1}\\ Such configurations cannot be disentangled by usual line moves.
    \end{center}
Three (independent) line configuration: 

    \begin{center}
        \includestandalone[width=0.5\textwidth]{Figures/Part3/p9.2}\\
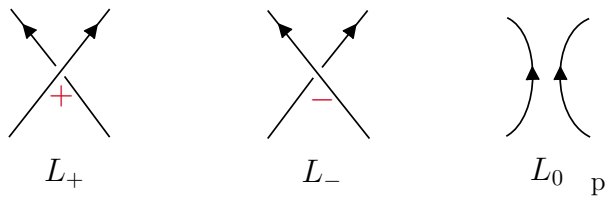
        \hypertarget{fig61}{Figure 61. } The three knots are identical except at a crossing point.
    \end{center}

    \begin{center}
        \includestandalone[width=0.4\textwidth]{Figures/Part3/p9.3}\\ 
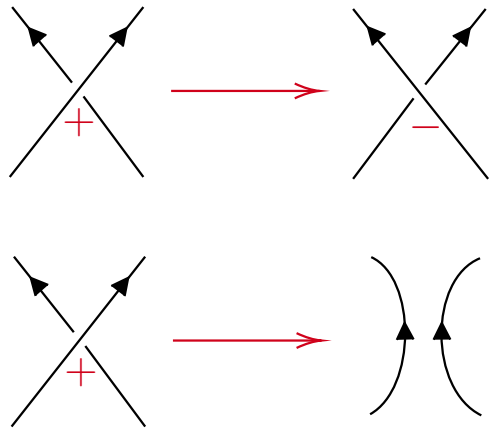
        \hypertarget{fig62}{Figure 62. } One cannot switch the lines to get one from the other.
    \end{center}
    
We denote the knots in \hyperlink{fig61}{Figure 61} by $L_+, L_-, L_0$ respectively and $P_+,P_-,P_0$ their associated invariants.

The simplest Skein relation reads
\begin{equation}
\begin{aligned}
    x P_{L_+}(x,y,z) + y P_{L_-} (x,y,z) &= z P_{L_0}(x,y,z), \\
    P_{\text{unknot}}(x,y,z) &= 1   ,
\end{aligned}
\end{equation}
where $P$ is a \textit{Laurent polynomial} of homogeneous degree $0$ in variables $(x,y,z)$.

\begin{example}
\begin{enumerate}
    \item Apply the Skein relation for two unlinked knots, written $2l$:
    \begin{equation}
        x \cdot 1+ y \cdot 1 = z P_{2l}(x,y,z)  \Longrightarrow \boxed{P_{2l}(x,y,z) = \dfrac{x+y}{z}.}
    \end{equation}
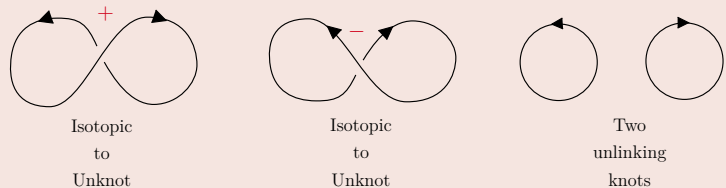
        \begin{center}
        \includestandalone[width=0.6\textwidth]{Figures/Part3/p10.1}
    \end{center}
    \item Now consider the Hopf link, with the three independent  configurations at a crossing:
    \begin{center}
        \includestandalone[width=0.7\textwidth]{Figures/Part3/p10.2}
    \end{center}

    Then, we have: 
    \begin{equation}
        x P_{\text{Hopf}}(x,y,z) + y \underbrace{P_{2l}(x,y,z)}_{\frac{x+y}{z}} = z \underbrace{P_{\text{unknot}}(x,y,z)}_{1} \Longrightarrow \boxed{ P_{\text{Hopf}}(x,y,z) = \dfrac{z^2 - xy - y^2}{xz}.}
    \end{equation}

    \item For the Trefoil knot, we have

    \begin{center}
        \includestandalone[width=0.60\textwidth]{Figures/Part3/p10.3}
    \end{center}
    
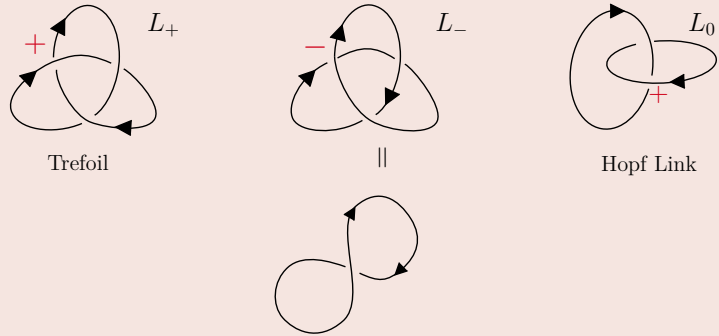
    \begin{equation}
        x P_{\text{trefoil}} + y \cdot 1 = z \times \frac{z^2 - xy - y^2}{xz} \Longrightarrow \boxed{P_{\text{trefoil}} = \frac{z^2 - 2xy - y^2}{x^2}.}
    \end{equation}
\end{enumerate}    
    
\end{example}

\begin{remark}
    For "unoriented" knots, skein relations can be written in a similar fashion, but with minor differences. We will skip that part in these notes.
\end{remark}

\paragraph{\textbf{Polynomial invariants.}} Each of the polynomials below is determined by the skein relation, and they are all ambient isotopy invariant (under Types I, II, III Reidemeister moves) unless otherwise stated. For more details, we refer the reader to \cite{PrasolovSossinsky1997}. 

\begin{enumerate}
    \item \emph{ Alexander (1928) - Conway (1970) polynomial}, $\Delta(t)$: Let 
    \begin{equation}
        x = 1, \quad y= -1, \quad z =- \sqrt t + \frac{1}{\sqrt t},
    \end{equation}
    in $P(x,y,z)$ with Skein relation 
    \begin{equation}
        \Delta_{L_+}(t) - \Delta_{L_-}(t) = \Big( \frac{1}{\sqrt t} - \sqrt t\Big) \Delta_{L_0}(t). \quad \quad (\Delta_{\text{unknot}}(t) =1)
    \end{equation}

    \item \emph{ Jones (1985) polynomial} $V(t)$: This is obtained by analyzing the braid group representations of knots using von Neumann algebras. Take
    \begin{equation}
        x = \frac{1}{t} ,\quad y = -t , \quad z = \sqrt t - \frac{1}{\sqrt t},
    \end{equation}
    and the skein relation reads 
    \begin{equation}
        \frac{1}{t} V_{L_+}(t) - t V_{L_-}(t) = \Big( \sqrt t - \frac{1}{\sqrt t} \Big) V_{L_0}(t). \quad \quad (V_{\text{unknot}}(t) = 1)
    \end{equation}

    \item \emph{HOMFLY(1985-87) polynomial} $P(t,z)$: Jones (1987) rederived this polynomial invariant using the Hecke algebra representation of the braid group. Take 
    \begin{equation}
        x = \frac{1}{t} , \quad y =-t,
    \end{equation}
    with the skein relation 
    \begin{equation}
        \frac{1}{t} P_{L_+}(t,z) - t P_{L_-}(t,z) = z P_{L_0}(t,z). \quad \quad (P_{\text{unknot}}(t,z) =1)
    \end{equation}

    \item \emph {The Akutsu-Wadati (1987) polynomial}: A new knot invariant derived from exactly solvable models in statistical mechanics. 

    \item \emph{Kauffman (1990) polynomial} $L(\alpha,z)$: This polynomial is regular isotopy invariant (i.e., only under Type II and III Reidemeister moves). The Skein relation reads as
    \begin{equation}
        L_{L_+}(\alpha,z) + L_{L_-}(\alpha,z) = z [ L_{L_0}(\alpha,z) + L_{D_{\text{unoriented}}}(\alpha,z) ],
    \end{equation}
    where the extra part in the parentheses on the right-hand side accounts for unoriented knots.
    
\end{enumerate}

\begin{theorem}
    (Kauffman 1988)
    
    Given an invariant of regular isotopy for "unoriented knots", then there exists an invariant of ambient isotopy for oriented knots.
\end{theorem}

Recall that given an oriented link $L$, $V_L(t)$ -the \textbf{Jones polynomial}- is determined by the following relations:
\begin{equation}
\begin{aligned}
    &(i) \quad \quad t^{-1} V_{L_+}(t) - t V_{L_-}(t) = (t^{1/2} - t^{-1/2}) V_{L_0}(t), \\
    &(ii) \quad \quad V( L \amalg \text{unknot)} = - (t^{-1/2} + t^{1/2}) V_L(t), \\
    &(iii) \quad \quad V_{\text{unknot}}(t) =1.
\end{aligned}
\end{equation}

\begin{example}
    One can see that 
    \begin{equation}
        V_{\text{Hopf}} (t) = - t^{-5/2} - t^{-1/2},
    \end{equation}
    and that 
    \begin{equation}
        V_{\text{trefoil}}(t) = -t^{-4} + t^{-3} + t^{-1}.
    \end{equation}
\end{example}

\section{Witten's Work}

\paragraph{Some notation and the setup.} Let us now focus on Edward Witten's work. For our setup, we consider a principal $SU(2)$ bundle over a closed, oriented, smooth 3-manifold $M$, together with a $2+1$ Chern-Simons theory on $M$,  and let
    \begin{equation}
        \begin{tikzcd}[sep=1.5cm]
            P \arrow[r, thick, "SU(2)"] & P \arrow[d, thick, "\pi"] \\ 
            {} & M \arrow[u, thick, "\sigma", bend right=60, swap]
        \end{tikzcd}
    \end{equation}
    On $P$, we have a Lie algebra-valued connection one form, $\omega\in \Omega'(P)\otimes \mathfrak{g}$; a covariant exterior derivative, $D$, which acts on $\theta$ as $ D\theta = d\theta \circ \mathsf{Hor}$; and a curvature 2-form $\Omega\in\Omega^2(P)\otimes \mathfrak{g}$ which is defined as
    \begin{equation}
        \Omega = D\omega = d\omega\circ\mathsf{Hor} = d\omega + \omega \wedge\omega.
    \end{equation}
    This is necessarily a trivial bundle, $P\cong M\times SU(2)$, so let $\sigma$ be the trivial section $\sigma:M\to P$ such that $\pi\circ\sigma = \text{id}_M$. Then, using $\sigma^*$, we can pull everything down to $M$. With these, we define the one-form, $A$, and its curvature two-form, $F$ as
    \begin{align}
        A := \sigma^*\omega \in \Omega^1(M) \otimes \mathfrak{g}, \qquad \text{and} \qquad
        F := \sigma^*\Omega \in \Omega^2(M; \mathfrak{g}),
     \end{align}
     where
     \begin{equation}
         F = d A + A\wedge A.
     \end{equation}
    Taking a  local coordinate chart $x=(x^i)_{i=1}^3$ on $M$, we get 
    \begin{align}
        A = A_\mu(x) d x^\mu, \qquad\ \text{and} \qquad
        F = \frac{1}{2} F_{\mu\nu} d x^\mu \wedge d x^\nu,
    \end{align}
    where
    \begin{equation}
        F_{\mu\nu} = \del_\mu A_\nu - \del_\nu A_\mu + [A_\mu, A_\nu].
    \end{equation}
    Now consider the Chern-Simons three-form on $M$, given by
    \begin{equation}
        Q_3(A,F) = \tr\lrp{A\wedge d A + \frac{2}{3}A\wedge A \wedge A},
    \end{equation}
    or equivalently
    \begin{equation}
        Q_3(A,F) = \tr\lrp{A\wedge F - \frac{1}{3}A\wedge A \wedge A}.
    \end{equation}
    Note that the exterior derivative of the Chern-Simons form is
    \begin{equation}
        d Q_3(A,F) = d \tr\lrp{A\wedge d A + \frac{2}{3}A\wedge A \wedge A} = \tr\lrp{F\wedge F} \in \Omega^4(N;\mathfrak{g}).
    \end{equation}
    Here, $N$ is a compact, oriented four-manifold whose boundary is $M$, i.e. $\del N = M$. Moreover, the quantity
    \begin{equation}
        P_{2n}(F) := \tr(F\wedge F),
    \end{equation}
    is called the \textbf{$n$th Chern form}. With all these, the action of the Chern-Simons theory is given by
    \begin{equation}
        S_\mathrm{CS}[A] = \frac{k}{4\pi}\int_M \tr\lrp{A\wedge d A+ \frac{2}{3}A \wedge A \wedge A},
    \end{equation}
    where $k\in \mathbb{N}^+$. In accordance with the path integral quantization formalism, the corresponding partition function is
    \begin{equation}
        Z(M) = \int_\mathcal{A} \mathcal{DA'} \exp(iS_\mathrm{CS}[A]).
    \end{equation}
    Here $\mathcal{A}$ is the space of all connections $\mathcal{A}\cong \Omega^1(M;\mathfrak{g})$.
    \begin{remark}
        \begin{enumerate}
            \item The integer $k$ is called the \textit{level}.
            \item The consistency of the QFT does not require the single-valuedness of the Chern-Simons action but that of $\exp(iS)$, where
            \begin{equation}
                S_\mathrm{CS} : \mathcal{A} \longrightarrow \mathbb{R}/ \mathbb{Z}.
            \end{equation}
            \item Let $\mathcal G$ be a group of $G$-equivariant diffeomorphisms, $\phi$, of $P$. Then, we saw before that 
            \begin{align}
                \mathcal{G} &= \mathrm{Map}(M; SU(2)) \notag \\ 
                &= \qty{g:M\rightarrow SU(2) \ | \ g \;\;\mathrm{ smooth}},
            \end{align}
             hence, $\mathcal G$ acts on $\mathcal{A}$ as follows: For all $g\in \mathcal G$ and $A \in \mathcal{A}$, we define
            \begin{equation}
                g\cdot A := g^{-1}Ag + g^{-1}d g.
            \end{equation}
            Locally, we write
            \begin{equation}
                A_\mu \longmapsto g^{-1}A_\mu g + g^{-1}\del_\mu g = g^{-1}(A_\mu + \del_\mu)g.
            \end{equation}
            \item We saw before that under gauge transformations,
            \begin{equation}
                F \longmapsto g\cdot F = g^{-1}F,
            \end{equation}
            hence we have 
            \begin{equation}
                \tr (g^{-1}Fg) = \tr(F),
            \end{equation}
            by similarity. This is the reason why $\tr(F^{\wedge 2})$  is invariant under gauge transformation.
            \item We need a term $k/4\pi$ in the dynamics of $S_\mathrm{CS}$ to make $\exp(iS_\mathrm{CS})$ invariant under the gauge transformation, which by a long calculation and recalling $\ker(e^{i\theta})=\qty{2\pi k | k\in \mathbb{Z}}$ , gives 
            \begin{equation}
                S_\mathrm{CS}[A] \longmapsto S_\mathrm{CS}[g\cdot A] = S_\mathrm{CS}[A] - 2\pi kN.
            \end{equation}
        \end{enumerate}
    \end{remark}
    
    The question now is how we compute $Z(M)$, which turns out to be a topological invariant. So, let us cut out $M$ by a Riemann surface $\Sigma$ and consider the above setup with a restriction on $\Sigma$.
    \begin{center}
        \includestandalone[width=0.7\textwidth]{Figures/Part3/p17}\\ 
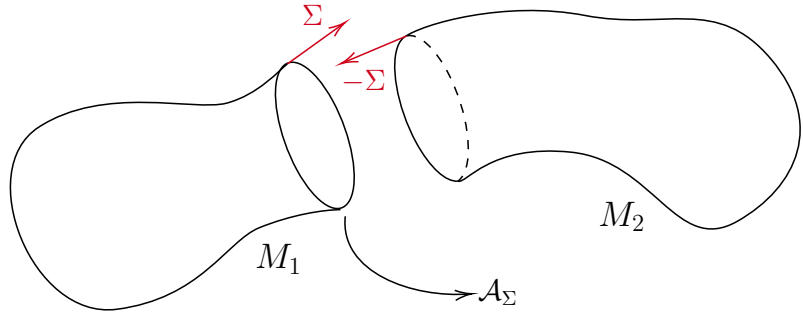
        \hypertarget{fig63}{Figure 63.} A manifold $M$ cut by a Riemann surface $\Sigma$, on which the space, $\mathcal{A}$, of all connections is defined.
    \end{center}
    Here, $M = (M_1 \coprod M_2) / \Sigma$. Recall that $Z(M_1)$ is an element of a certain vector space $Z_\Sigma$ assigned to $\Sigma$ by means of geometric quantization. Similarly, $Z(M_2)\in Z_{-\Sigma}$. Also, by the axioms of TQFT, we had $Z_{-\Sigma}=Z_\Sigma^*$. Hence,
    \begin{equation}
        Z(M) = \langle Z(M_1),Z(M_2) \rangle,
    \end{equation}
    by gluing action where $\langle \cdot, \cdot \rangle$ is the natural pairing on $Z_\Sigma$. Recall that at the end of Chapter ~\ref{ch:symplectic_geometry}, we saw 
    \begin{itemize}
        \item $\mathcal{A}_\Sigma$ is an infinite-dimensional symplectic manifold.
        \item Furthermore, $(\mathcal{A}_\Sigma, G, \mu)$ is a Hamiltonian G-space with a moment map
        \begin{equation}
            \begin{aligned}
                \mu : \mathcal{A} &\longrightarrow \mathfrak{g}_G^* \\ 
                A &\longmapsto F_A.
            \end{aligned}
        \end{equation}
        This is the curvature map!
        \item By symplectic reduction theorem and the results of Atiyah, the moduli space $\mathcal{M}:=\mu^{-1}(0)/G$ of flat connections on $M$ turns out to be a compact finite-dimensional symplectic manifold. 
    \end{itemize}
    Now, geometric quantization is available for $\mathcal{M}$, and one can canonically associate a vector space $Z_\Sigma$ to $\mathcal{M}$. This is a unique vector space up to isomorphisms (by the Hitchin construction). Then we have the following picture outlining the situation:
     \begin{center}
        \includestandalone[width=0.7\textwidth]{Figures/Part3/p18}\\
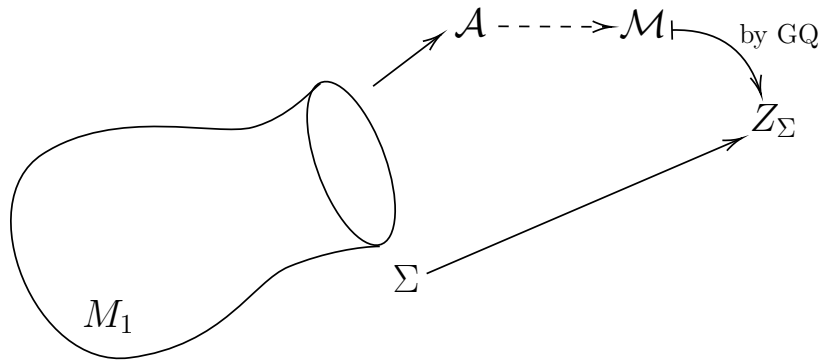
        \hypertarget{fig64}{Figure 64. } Assignment of a vector space $Z_\Sigma$ to the moduli space, $\mathcal{M}$.
    \end{center}   
    Here, to each $\Sigma$, we assign canonically a Hilbert space $Z_\Sigma$ such that  $Z(M_1)\in Z_\Sigma$ as we claimed. For a given $\Sigma$, one can talk explicitly about $Z_\Sigma$ in more detail in suitable cases. \\ 
    For $A = A_\mu d x^\mu = A_0d x^0 + A_id x^i$ where $i=1,2$, and noting that locally, $A_0=0$ on $\Sigma$, we write 
    \begin{equation}
        \begin{tikzcd}[sep=1cm]
           {\displaystyle S_\mathrm{CS}[A] = \frac{k}{2\pi}\int_M \epsilon^{\mu\nu\rho}\tr\lrp{A_\mu\del_\nu A_\rho - A_\mu\del_\rho A_\nu + \frac{2}{3}A_\mu [A_\nu, A_\rho]}} \arrow[d, thick, "\text{on } \Sigma\times \mathbb{R}"] \\ 
           {\displaystyle S_\mathrm{CS}[A]_\Sigma = \frac{k}{2\pi}\int dt \int_\Sigma \epsilon^{ij}\tr\lrp{A_i\frac{d}{d t}A_j}}.
        \end{tikzcd}.
    \end{equation}
    Now, \textbf{what if $\mathbf{M}$ includes a knot (link) $\mathbf{K}$?} 
    
    To obtain further information, first we need to consider the generalised version of $Z$ in the following sense: By introducing a functional $\mathcal{O}(A)$ on $\mathcal{A}$ associated to $A$, one can construct an invariant for the "data" defining $\mathcal{O}(A)$ as follows:
    \begin{equation}
        Z(M; \mathcal{O}) = \int_\mathcal{A}\mathcal{DA'}e^{iS_\mathrm{CS}}\mathcal{O}(A).
    \end{equation}
    Taking a link $L$, say with $n$-component, so that $L=\bigcup_{i=1}^nL_i$ in $M$ and letting $M=S^3$ in particular, we can define a function $\theta$ associated to $L$ as follows:
    \begin{itemize}
        \item For each $L_i$, assign an irreducible representation $R_i$ of $G=SU(2)$.
        \item For each loop $L_i$ and its assigned representation $R_i$, define the \textit{Wilson line operator} $W_{R_i}(L_i)$ by
        \begin{equation}
            W_{R_i}(L_i) := \tr_{R_i}\lrp{\mathcal{P}\exp\lrp{i\oint_{L_i} A}},
        \end{equation}
        where $\mathcal{P}$ is the \textit{path-ordering operator}. The term inside the trace is also called the \textit{holonomy} of $A$ over $L_i$, $\mathrm{Hol}_{L_i}(A)$.
        \item We now define the the operator
        \begin{equation}
            \mathcal{O}(A) := \prod_i W_{R_i}(L_i).
        \end{equation}
    \end{itemize}
     \begin{center}
        \includestandalone[width=0.55\textwidth]{Figures/Part3/p19} \\
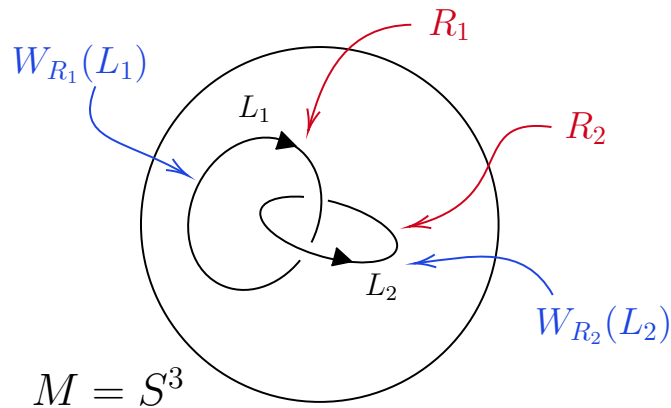
        \hypertarget{fig65}{Figure 65.} Linked loops $L_i$ embedded in $S^3$ and their corresponding Wilson operators.
    \end{center}     
    Then the partition function is given by
    \begin{equation}
        Z(M, \mathcal{O}) = \int_\mathcal{A}\mathcal{DA'}e^{iS_\mathrm{CS}}\prod_i W_{R_i}(L_i).
    \end{equation}
    When there are no knots or links on $M$, one reclaims the original $Z(M)$ without introducing an extra Wilson line term.

    \subsection{Rederiving the Knot Polynomials}
        Here the idea is to recover the suitable skein-like relations in terms of the partition function $Z$ that uniquely determines the associated knot polynomial. Our setup assumes $M=S^3$ with a knot $K$\footnote{Here, $K$ will only have one component.} in $M$. Considering a 3-ball   $D^3$ around a crossing $P$ with $\del D^3 = S^2$, we can let $S^2$ play the role of $\Sigma$ above. Then, after cutting out the manifold, we have the following figure. 
        \begin{center}
            \includestandalone[width=0.6\textwidth]{Figures/Part3/p20.1}\\
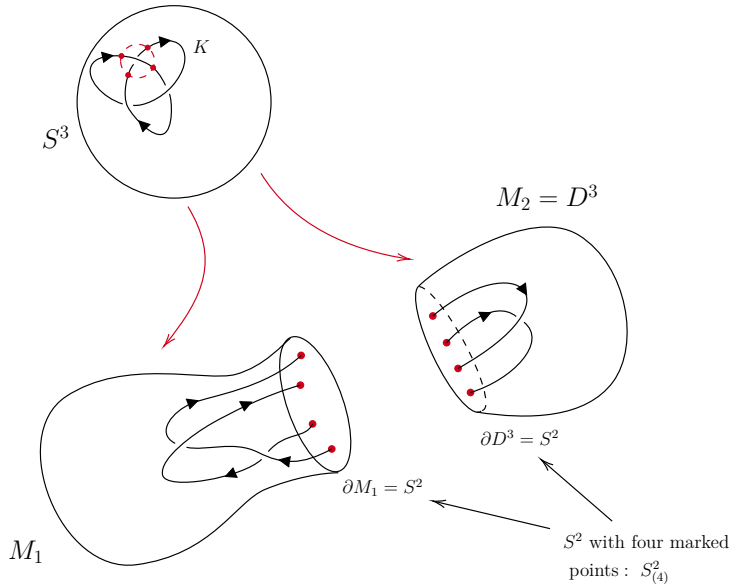
            \hypertarget{fig66}{Figure 66.} $S^3$ with an embedded knot is cut by a Riemann surface at the crossing point of the knot.
        \end{center}  
        \begin{observation}
            \begin{enumerate}
                \item Without touching anything, if we glue back each piece, then we obtain $S^3=M$ with the original knot $K$.
                \item Notice that while $M_1$ includes the complicated part of the knot, $M_2$ consists of parts of the original knot with some actual braiding. Possible cases are as follows:
                 \begin{center}
                    \includestandalone[width=0.80\textwidth]{Figures/Part3/p20.2}\\
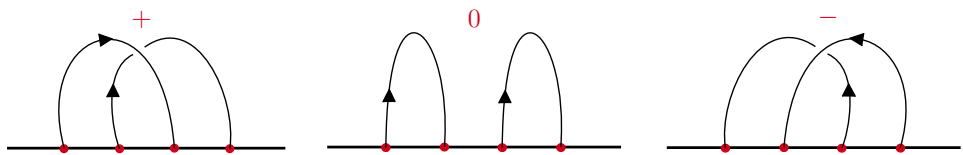
            \hypertarget{fig67}{Figure 67.} Each possible braiding for the part of the knot inside $M_2$.
                \end{center}                
                \item Each choice of possible braiding corresponds to one of the independent line configurations that appeared in our previous discussions on knot theory:
                \begin{center}
                    \includestandalone[width=0.60\textwidth]{Figures/Part3/p21.1}\\
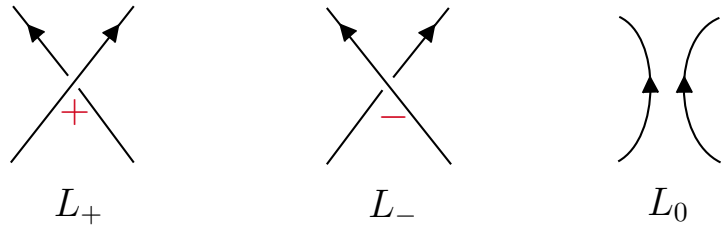
                    \hypertarget{fig68}{Figure 68.} Possible line configurations for the part of the knot inside $M_2$.
                \end{center}  
                \item The choice of braiding of the four points (or say labelling) determines different vectors in the assigned vector space $Z_{-S^2_{(4)}}=Z^*_{S^2_{(4)}}$. That is,
                \begin{equation}
                    Z(M_2,L_+), \ Z(M_2,L_0), \ Z(M_2,L_-) \in Z^*_{S^2_{(4)}}.
                \end{equation}
                \begin{center}
                    \includestandalone[width=0.66\textwidth]{Figures/Part3/p21.2}\\
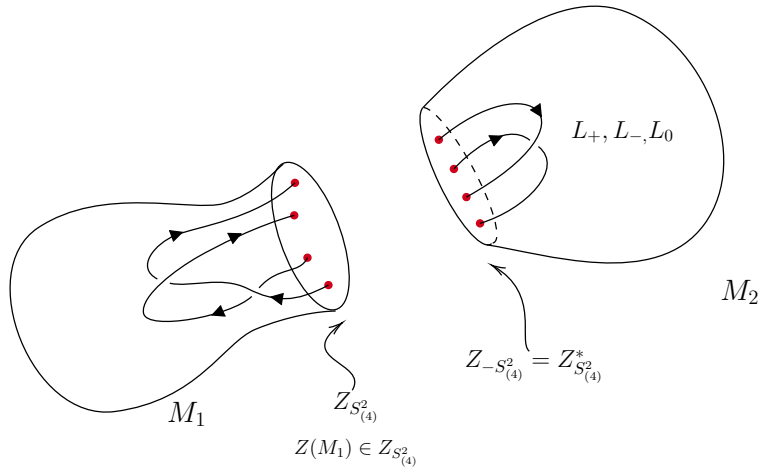
                    \hypertarget{fig69}{Figure 69.} Vector spaces assigned to the Riemann surface $S^2$ with four marked points.
                \end{center}  
            \end{enumerate}
        \end{observation}
        Here, a fact is that $\dim S^2_{(4)}\leq 2$. This can be seen by working in 2D-conformal field theory on the plane, using some representation-theoretical arguments, which we will not get into. For details, we refer to Kohno's book \cite{Kohno2002}, \textit{Conformal Field Theory and Topology.} \\ 

        \noindent All in all, the above vectors should obey the dependence relation due to dimensionality:
        \begin{equation}
            \alpha Z(M_2,L_+) + \beta Z(M_2,L_-) + \gamma Z(M_2,L_0)=0.
        \end{equation}
        When we glue back $M_1$ and $M_2$, we recover $M$, and depending on the choice of braiding, we obtain not just the original one, but $K$ with a modified line configuration at   the crossing. \\ 

        \noindent By using the pairing $Z(L_i) = \langle Z(M_1) , Z(M_2; L_i)\rangle $ for each $i$, we get
        \begin{equation}
             \alpha Z(L_0) + \beta Z(L_-) + \gamma Z(L_+)=0,
        \end{equation}
        where each term is assigned to one of the configurations in the figure below, respectively.
        \begin{center}
            \includestandalone[width=0.60\textwidth]{Figures/Part3/p22}\\
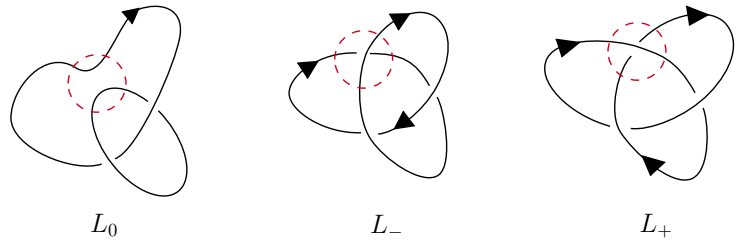
            \hypertarget{fig70}{Figure 70.} Crossings corresponding to each of the terms in the partition function.
        \end{center}  
        This is, in fact, the skein-like relation associated with the line configuration via the partition function $Z$. \\ 

        \noindent The determination of $\alpha,\beta, \gamma$ can be done explicitly as shown by Moore and Seiberg using the polynomial equations for rational conformal field theory, but we will not go into this. \\ 

        \noindent At the end of the day, one can recover the Jones polynomial as follows:
        \begin{equation}
            q^{-1}V(L_+) - qV(L_-) + (q^{1/2}-q^{-1/2})V(L_0)=0,
        \end{equation}
        where
        \begin{equation}
            q = \exp\lrp{\frac{2\pi i}{k+2}} \quad\text{and}\quad V(q) = Z(M,K).
        \end{equation}
        The case with $G=SU(2)$ and $M=S^3$ can then be generalized to any closed, oriented, smooth three-manifold $M$ by Witten's work, which essentially combines the previous arguments with the surgery machinery for 3-manifolds.


\bibliography{biblio}

@Inbook{daSilva2008,
author="da Silva, Ana Cannas",
title="Symplectic Forms",
bookTitle="Lectures on Symplectic Geometry",
year="2008",
publisher="Springer Berlin Heidelberg",
address="Berlin, Heidelberg",
pages="3--8",
isbn="978-3-540-45330-7",
doi="10.1007/978-3-540-45330-7_1",
url="https://doi.org/10.1007/978-3-540-45330-7_1"
}

@misc{HondaTopologicalQFT,
  author       = "K. Honda",
  title        = "{Lecture notes for MATH 635: Topological Quantum Field Theory}",
  howpublished = "\url{http://www.math.ucla.edu/honda/}",
  note         = "Accessed on 2025-11-30",
  year         = "n.d."
}

@misc{BlauGQ,
  author       = "Matthias Blau",
  title        = "{Lecture notes on Geometric Quantization}",
  howpublished = "\url{http://blau.itp.unibe.ch/lecturesGQ.pdf}",
  note         = "Accessed on 2025-11-30",
  year         = "n.d."
}

@article{Witten:1988hf,
    author = "Witten, Edward",
    editor = "Mitra, Asoke N.",
    title = "{Quantum Field Theory and the Jones Polynomial}",
    reportNumber = "IASSNS-HEP-88-33",
    doi = "10.1007/BF01217730",
    journal = "Commun. Math. Phys.",
    volume = "121",
    pages = "351--399",
    year = "1989"
}

@book{PrasolovSossinsky1997,
  author    = "V. V. Prasolov and A. B. Sossinsky",
  title     = "{Knots, Links, Braids and 3-Manifolds: An Introduction to the New Invariants in Low-Dimensional Topology}",
  series    = "Translations of Mathematical Monographs",
  volume    = "154",
  publisher = "American Mathematical Society",
  address   = "Providence, RI",
  year      = "1997",
  pages     = "239",
  isbn      = "978-0-8218-0898-6"
}

@article{Wu:1992zzb,
    author = "Wu, F. Y.",
    title = "{Knot theory and statistical mechanics}",
    doi = "10.1103/RevModPhys.64.1099",
    journal = "Rev. Mod. Phys.",
    volume = "64",
    pages = "1099--1131",
    year = "1992",
    note = "[Erratum: Rev.Mod.Phys. 65, 577--577 (1993)]"
}

@article{Jones:1989ed,
    author = "Jones, V. F. R.",
    title = "{On knot invariants related to some statistical mechanical models}",
    journal = "Pacific J. Math.",
    volume = "137",
    pages = "311--334",
    year = "1989"
}

@article{Jones:1985dw,
    author = "Jones, V. F. R.",
    title = "{A polynomial invariant for knots via von Neumann algebras}",
    doi = "10.1090/S0273-0979-1985-15304-2",
    journal = "Bull. Am. Math. Soc.",
    volume = "12",
    pages = "103--111",
    year = "1985"
}

@article{Atiyah:1989vu,
    author = "Atiyah, M.",
    title = "{Topological quantum field theories}",
    doi = "10.1007/BF02698547",
    journal = "Inst. Hautes Etudes Sci. Publ. Math.",
    volume = "68",
    pages = "175--186",
    year = "1989"
}

@book{Kohno2002,
  author    = "Toshitake Kohno",
  title     = "{Conformal Field Theory and Topology}",
  series    = "Translations of Mathematical Monographs",
  volume    = "210",
  publisher = "American Mathematical Society",
  address   = "Providence, RI",
  year      = "2002",
  pages     = "172",
  isbn      = "978-0-8218-2130-5"
}

@online{SchullerGeometricAnatomy,
  author = {Frederic P. Schuller},
  title  = {Lectures on the Geometric Anatomy of Theoretical Physics},
  year   = {2015},
  url    = {https://tales.mbivert.com/Lectures_on_Geometric_Anatomy_of_Theoretical_Physics.pdf},
  note   = {Lecture notes by Simon Rea; course given in 2013/14 at Friedrich-Alexander-Universität Erlangen-Nürnberg; accessed 2025-11-30}
}

@online{SchullerGATP_Lec19_21,
  author  = {Frederic P. Schuller},
  title   = {Lectures on the Geometric Anatomy of Theoretical Physics — Lectures 19–21},
  year    = {2015},
  url     = {https://www.youtube.com/playlist?list=PLPH7f_7ZlzxTi6kS4vCmv4ZKm9u8g5yic},
  note    = {Principal fibre bundles; Associated fibre bundles; Connections and connection 1-forms. Accessed 2025-11-30}
}
\bibliographystyle{alpha}

\end{document}